\documentclass[11pt,twoside,a4paper,openany]{memoir}
\maxsecnumdepth{subparagraph}
\setsecnumdepth{subsubsection}

\usepackage[T1]{fontenc}
\usepackage{eucal}
\usepackage{microtype}
\usepackage{csquotes}
\usepackage{amsthm}
\usepackage{mathtools}
\usepackage{amssymb}
\usepackage{bm}
\usepackage{cancel}

\usepackage{tikz}
\usetikzlibrary{cd}
\usetikzlibrary{calc,patterns,angles,quotes,positioning,fit,backgrounds,automata}

\usepackage[backend=biber,style=alphabetic,sorting=nyt,maxbibnames=10]{biblatex}
\bibstyle{apsrev}
\usepackage{enumitem}
\setlist[enumerate,1]{label={\roman*)}}
\usepackage[unicode=true, pdfusetitle, colorlinks=true, allcolors=blue, bookmarks=true]{hyperref}
\usepackage[left=2cm,right=2cm,top=2cm,bottom=2cm,headheight=15pt]{geometry}
\usepackage{pdfpages}
\usepackage{afterpage}
\usepackage[export]{adjustbox}
\usepackage{longtable}   
\usepackage{booktabs}   
\usepackage{array}
\usepackage{pdflscape}

\usepackage{stmaryrd}
\SetSymbolFont{stmry}{bold}{U}{stmry}{m}{n}
\usepackage{lmodern}
\usepackage[scaled=0.95]{zi4}

\usepackage{silence}
\theoremstyle{plain}
\newtheorem{theorem}{Theorem}[section]
\newtheorem*{theorem*}{Theorem}
\newtheorem{lemma}[theorem]{Lemma}
\newtheorem*{lemma*}{Lemma}
\newtheorem{proposition}[theorem]{Proposition}
\newtheorem*{proposition*}{Proposition}
\newtheorem{corollary}[theorem]{Corollary}
\newtheorem*{corollary*}{Corollary}
\theoremstyle{definition}
\newtheorem{definition}{Definition}[section]
\newtheorem*{definition*}{Definition}

\newtheorem*{example*}{Example}
\theoremstyle{remark}
\newtheorem{remark}{Remark}[section]
\newtheorem*{remark*}{Remark}

\newtheorem*{conjecture*}{Conjecture}

\newtheorem*{problem*}{Problem}

\newcommand*{\NN}{\mathbb{N}}
\newcommand*{\ZZ}{\mathbb{Z}}

\newcommand*{\RR}{\mathbb{R}}
\newcommand*{\CC}{\mathbb{C}}
\newcommand*{\A}{\mathbb{A}}
\newcommand*{\F}{\mathbb{F}}

\newcommand*{\tX}{\tilde{X}}
\newcommand*{\tY}{\tilde{Y}}
\newcommand*{\tZ}{\tilde{Z}}

\newcommand*{\dd}{\mathrm{d}}
\newcommand*{\hh}{\mathbf{h}}

\newcommand*{\parder}[2]{\frac{\partial#1}{\partial #2}}
\newcommand*{\der}[2]{\frac{\dd #1}{\dd #2}}
\DeclareMathOperator{\Leg}{Leg}

\newcommand{\leg}{\operatorname{leg}}
\newcommand{\Ad}{\operatorname{Ad}}
\newcommand{\ad}{\operatorname{ad}}
\newcommand{\Aut}{\operatorname{Aut}}

\newcommand{\Gau}{\operatorname{Gau}}
\newcommand{\gau}{\mathfrak{gau}}
\newcommand{\Diff}{\operatorname{Diff}}
\newcommand{\Conj}{\operatorname{Conj}}

\let\emphoriginal\emph

\let\emph\textbf
\let\vec\mathbf
\mathtoolsset{showonlyrefs}

\bibliography{biblio}

\AtEveryBibitem{
  \ifentrytype{online}{}{
  \ifentrytype{misc}{}{
   \ifentrytype{unpublished}{}{
    \clearfield{url}
  }}}
}

\DeclareSourcemap{
  \maps[datatype=bibtex]{
    \map[overwrite=true]{
      \step[fieldset=urldate, null]
      \step[fieldset=language, null]
      \step[fieldset=address, null]
     \step[fieldset=addendum, null]
     \step[fieldset=pagetotal, null]
     \step[fieldset=location, null]
    }
  }
}

\title{Yang--Mills Theory in the Multisymplectic Framework\\
\Large Master's Thesis\\
\large Master in Theoretical Physics\\
Autonomous University of Madrid and\\Institute of Theoretical Physics (CSIC-UAM)}

\author{Author: David Villa-Blanco
\\Supervisor: Manuel de León (ICMAT-CSIC)
\\With the assistance of: Rubén Izquierdo-López (ICMAT-CSIC)}
\date{\today}

\begin{document}
\pagenumbering{Alph}
\frontmatter
\maketitle
\thispagestyle{empty}
\newpage
\maxtocdepth{subsubsection}
\settocdepth{subsubsection}
\tableofcontents
\chapter*{Abstract}
\addcontentsline{toc}{chapter}{Abstract}
This Master’s thesis develops and applies the multisymplectic field-theory framework to
(pure) Yang–Mills theory. Building from the geometric foundations of jet bundles
and the Poincaré--Cartan forms, the work constructs the Lagrangian multisymplectic formalism
for first-order field theories, recovering the usual Euler--Lagrange field equations and
extending the Poisson brackets of symplectic mechanics to a full graded Lie algebra. It then establishes the Legendre
correspondence with the multisymplectic Hamiltonian formalism (in which Hamilton equations
are retrieved) including the case of singular Lagrangians (where a constraint algorithm
is reproduced). Breaking Lorentz covariance via a Cauchy surface slicing, the thesis then relates the
covariant multisymplectic picture to an instantaneous presymplectic (infinite-dimensional)
formulation in which the instantaneous De Donder equations and presymplectic
brackets are obtained in a compatible way. A comprehensive treatment of
symmetries then leads to multisymplectic and instantaneous momentum maps generalizations and
to several versions of Noether’s theorems. 

Specializing the formalism to
principal-bundle gauge theory, the thesis identifies the multi\-symplectic configuration space with the bundle
of connections, gauge symmetries with symmetries in the multi\-symplectic language and builds
the Yang--Mills Lagrangian. Lastly, among many computations, the standard tensorial Yang--Mills
equations and the Bianchi identity in local trivializations are retrieved, Yang--Mills conserved
currents and charges are obtained and the graded algebra is derived for a special class
of Hamiltonian forms.

\chapter*{Resumen}
\addcontentsline{toc}{chapter}{Resumen}
Este trabajo de fin de máster desarrolla y aplica el formalismo multisimpléctico de 
teorías de campos a la teoría de Yang--Mills (pura). A partir de los fundamentos
geométricos de los fibrados de jets y de las formas de Poincaré--Cartan, el trabajo
construye el formalismo lagrangiano multisimpléctico para teorías de campos de primer
orden, recuperando las ecuaciones de campo
de Euler--Lagrange habituales y extendiendo los corchetes de Poisson de la mecánica simpléctica
a un álgebra de Lie graduada completa. A continuación, establece la correspondencia de Legendre con el formalismo hamiltoniano
multisimpléctico -en el que se recuperan las ecuaciones hamiltonianas-,
incluyendo el caso de lagrangianos singulares (donde se reproduce un algoritmo de ligaduras).

Al romper la covarianza de Lorentz mediante un foliado por superficies de Cauchy, la tesis
relaciona la descripción multisimpléctica covariante con una formulación presimpléctica
instantánea (de dimensión infinita) en la que se obtienen las ecuaciones instantáneas de
De Donder y los corchetes instantáneos presimplécticos de manera compatible. Un tratamiento
exhaustivo de las simetrías conduce a una generalización de las aplicaciones momento en los marcos
multisimpléctico e instantáneo y a diferentes versiones de los teoremas de Noether.

Al especializar el formalismo a la teoría de gauge en fibrados principales,
la tesis identifica el espacio de configuración multisimpléctico con el fibrado de conexiones,
las simetrías de gauge con simetrías en el formalismo multisimpléctico, y construye el lagrangiano
de Yang--Mills. Por último, entre numerosos cálculos, se recuperan las ecuaciones tensoriales
habituales de Yang--Mills y la identidad de Bianchi en trivializaciones locales, se obtienen
las corrientes y cargas conservadas de Yang--Mills y se deriva el álgebra graduada para una
clase especial de formas hamiltonianas.

\let\cleardoublepage\clearpage
\mainmatter
\chapter{Introduction}\label{chap:introduction}
There is no doubt that \textbf{symplectic geometry} plays
a fundamental role in modern classical mechanics.
The phase space, identified with the cotangent bundle of the configuration manifold,
is naturally endowed with a symplectic structure and symmetries are canonically
described as Lie group actions that preserve the symplectic form (also called symplectic 
actions) \cite{pelayoHamiltonianSymplecticSymmetries2016}; Hamilton equations are integral curves
of specific vector fields on this symplectic manifold while many
standard results in classical mechanics (Noether's theorems and
conserved quantities, reduction by symmetries, Hamilton-Jacobi theory,
Liouville and Arnold-Liouville theorems...) or essential tools like
Poisson brackets are obtained in a canonical (coordinate-free) and elegant way
\cite{leonGradedPoissonGraded2025,leonDescriptionClassicalField2025}.
Fundamentally for the motivation of this thesis, symplectic mechanics has not only
recovered already known results in a geometric (rather than analytic) way, but has
also given \textbf{new insights} like coisotropic reduction, coisotropic regularization,
study of stability, geometric integrators, symplectic reduction, 
generalization of Noether theorems
\cite{izquierdo-lopezZooCoisotropicEmbeddings2025} and momentum maps,
differential-topological methods in the study of the phase portrait
\cite{abrahamFoundationsMechanics2008}, KAM theory (see 
\cite{fejoz13IntroductionKAM2016} for applications to the spatial three
body problem and celestial mechanics in general), Darboux theorem
\cite{Darboux1882}, Lagrangian submanifolds (see for example 
\cite{weinsteinSymplecticGeometry1981}), Gromov's non-squeezing theorem 
\cite{gromovPseudoHolomorphicCurves1985}, the highly interesting
-from the theoretical physics point of view- procedure of geometric
quantization (which could maybe serve as the natural bridge between the 
geometric view of gravity \cite{einsteinGrundlageAllgemeinenRelativitatstheorie1916}
and quantum mechanics) among many others.

Nobody either doubts about the success of \textbf{field theory} in modern theoretical
physics. The need for fields is clear once we note some crucial physical restrictions
like causality.
First, the laws of classical mechanics are instantaneous (even in the symplectic framework):
for example, if our Sun disappears right now, Newton 
would strongly argue all the planets stop immediately following their orbits
around it.
Clearly this would completely break causality and special relativity 
\cite{einsteinZurAllgemeinenMolekularen2005}. Hence, the need for local
laws of nature is clear, and indeed we have examples in the history of physics of such 
generalizations, like Maxwell (and Faraday) and Einstein did with the theories of
Coulomb and Newton respectively. In these field theories the interactions are
mediated by the respective \textit{gauge} fields, which evolve locally, i.e., field equations at a point
must only depend on the behaviour of the fields and their derivatives
at that given point (at least in first order theories).
Second, taking into account the quantum nature of physics at high energies (equivalently,
short scales) and noticing the existence of particle annihilation, particle-antiparticle
pairs..., relativistic quantum mechanics in which the number of particles is fixed is not
enough. Even the fact that a neutrino produced at a supernova $6.4$ $Mpc$ away 
\cite{guettaLowHighEnergy2023} is exactly identical to another one produced at some proton-proton
collision at CERN's LHC, hints at the existence of a `primordial'
field permeating the whole spacetime (in this case the neutrino's not yet fully understood
field). The theory resolving this second problem is quantum field theory (QFT) -in which
particles are treated as excitations of fields that extend throughout spacetime-
and even though its mathematical framework is still not fully rigorously deployed
(and famously though; note the Clay Mathematics Institute's Millennium Prize Problem about
the Yang--Mills existence and the mass gap. An attempt to find the essential features of
QFT to account in the axiomatization can be seen in \cite{wilczekQuantumFieldTheory1999}),
its astonishing degree of accuracy in its predictions makes undeniable the need for a rigorous
and geometric framework which could lead field theory to similar generalizations and insights
as the ones symplectic geometry has given to classical mechanics in the 1970s-1980s.

Lastly, undeniable also is the fact that multisymplectic field theory has not (yet)
permeated theoretical physics' literature at the same level as symplectic mechanics has. 
With this work hence \textbf{we aim towards a familiarization of theoretical physicists with the 
multisymplectic framework} -expecting to see an efflorescence of papers built around it- showing
its (relative) simplicity, elegance and enormously broad range of applications (covering all Lie
group theory, higher-order theories in a natural generalization (c.f 
\cite{leonGeneralizedClassicalMechanics1985})...) as its main `selling points'. We apply
the framework to (pure) Yang--Mills theory in particular due to the importance of this
theory in modern particle physics.

Regarding the structure of the thesis, we start Chapter \ref{chap:multisymplectic} by
defining the multisymplectic formalism for field theories as a
\textbf{generalization of the symplectic formalism for classical mechanics}. 
With this in mind, the canonical first jet bundle $J^1\pi$ is a generalization of the usual
tangent bundle -in which we recover a geometric equivalence to variational calculus and with it
Euler--Lagrange field equations- while the dual first jet bundle $(J^1\pi)^*$ defined in the
Hamiltonian setting, is a generalization of the cotangent bundle. In this last bundle -and upon
the natural generalization of the definition of a symplectic form in symplectic mechanics-
we find a canonical multisymplectic form upon which are found Hamilton equations.
For regular Lagrangians these bundles are locally diffeomorphic through the Legendre map and 
their main equations are proven to be equivalent.

After this, in Chapter \ref{chap:instantaneous_formalism}, we break Lorentz covariance in order to tackle
the \textbf{initial value problem}. To do this we introduce Cauchy surfaces, changing descriptions
from finite-dimensional (multisymplectic) to infinite-dimensional (instantaneous, or covariant)
presymplectic through integration along a Cauchy surface, which allow us to make a \textbf{time-space
decomposition} upon slicing.

Chapter \ref{chap:symmetries_and_Noether_theorems} deals with the multisymplectic formulation
of general symmetries and Noether's theorems for the Lagrangian, Hamiltonian and instantaneous
settings, finding both conserved currents and charges for the underneath physical theory. We also
introduce the multisymplectic and instantaneous momentum map as a generalization of the momentum map
that arises in symplectic mechanics in the context of gauge symmetries.

In Chapter \ref{chap:gauge_theory} we introduce principal fiber bundle gauge theory and explicitly
fine-tune our field configuration space $Y$ in a way such that there is an identification among
multisymplectic fields and geometric connections of the gauge theory (which are indeed our gauge fields).
Gauge symmetries are then promoted naturally to the multisymplectic symmetries of Chapter
\ref{chap:symmetries_and_Noether_theorems}.

Lastly, in Chapter \ref{chap:multisymplectic_yang_mills_computations} we
introduce multisymplectic Yang--Mills field theory through the usual Yang--Mills Lagrangian
and perform some computations once a local trivialization is chosen; we recover the
actual tensorial equations used by theoretical physicist on a daily basis such as
the Yang--Mills equation, the Bianchi identity... and as an original contribution not found
in the literature, we compute the Lie algebra of
the multisymplectic momentum map and show the graded Lie algebra structure of a special
class of Hamiltonian forms that resemble the gauge vector potential and field strength
for Yang--Mills theory. 
We close this last chapter with some words on further work not covered on this
thesis.

\chapter{Multisymplectic framework}\label{chap:multisymplectic}
Let us start with some preliminary definitions
\begin{definition}[Multisymplectic forms and multisymplectic manifold]
\label{def:multisymplectic_forms_and_manifold}
A \textbf{multi\-symplectic form} $\Omega$ on a manifold $M$ is a closed (i.e. $\dd
\Omega=0$) $k-$form with $k>1$ having the following non-degeneracy property
\[
	\iota_v\Omega=0\iff v=0,\quad\forall v\in T_xM,\,\forall x\in M.
\]
A \textbf{multisymplectic manifold} is a manifold equipped with a multisymplectic
form.

Note that for $k=2$ we recover the definition of symplectic forms and symplectic manifolds.
A closed $2-$form (without the non-degeneracy condition) is called \textbf{presymplectic}.
\end{definition}

To `apply' this generic definition we first need to build our manifold
framework, and in order to meet the main objective of this work
(and not only), i.e. Yang--Mills, this framework has to be suitable
and specifically designed for \textbf{first-order theories}. Namely, in
the case of the Lagrangian description, theories in which the
Lagrangian depends only up to the first derivatives of our field
$L\left(x^\mu,\phi^i,\partial_\mu\phi^i\right)$. The canonical
framework will be a set of fiber bundles (such that the fields
configuration and `velocities' are properly defined for each point
of spacetime) which are formed by the spacetime $X$, the
configuration space $Y$ (which form the configuration bundle
$\pi:Y\to X$) and the first-order bundle $Z\equiv J^1\pi$, which
we define in the next section.

\section{Lagrangian setting}\label{sect:lagrangian_setting}

Let us start defining our working spaces. First, for the `background' space 
$X$, we consider an orientable $(n+1)-$dimensional manifold (potentially with 
boundary $\partial X$) in which we consider a set of coordinates denoted
by $(x^\mu)$ and define a volume form $\eta:=\dd^{n+1}x=\dd x^0\wedge
\dd x^1\wedge\dots\wedge\dd x^n$. 
Then, for our field configuration space $Y$, we consider an 
$(n+1+m)-$dimensional manifold such that $\pi_{YX}\equiv\pi:Y\to X$ is a 
fiber bundle and so we can consider in this setting the fiberwise coordinates 
$(x^\mu,y^i)$ on $Y$ satisfying $\pi_{YX}(x^\mu,y^i)=(x^\mu)$ ($\mu\in\{0,\dots,n\}$), 
$i\in\{1,\dots,m\}$. As a useful notation, we introduce 

\[
	\dd^n x_\mu:=\iota_{\partial/\partial x^\mu}\dd^{n+1}x=(-1)^\mu \dd x^0\wedge\dots
	\wedge\widehat{\dd x^\mu}\wedge\dots \dd x^n,
\]
\[
	\dd^{n-1}x_{\mu\nu}:=
	\iota_{\partial/\partial x^\nu}\iota_{\partial/\partial x^\mu}\dd^{n+1}x\dots
\] 

We denote as $\Gamma(\pi_{YX}):=\{\phi:X\to Y\,:\, \pi_{YX}\circ\phi=
\mathrm{id}_X\}$ the space of sections $\phi$ of $\pi_{YX}$ and
consider for each point $x\in X$ the equivalence relation between two sections
$\phi,\psi\in\Gamma(\pi_{YX})$ defined by
\[
	\phi\sim\psi \iff \phi(x)=\psi(x),\, d_x\phi=d_x\psi,
\]
where $d_x\phi$ denotes the pushforward (or tangent) map on the base point $x$ for the
section $\phi$, which in this case is nothing more than ($\phi^i\equiv
y^i\circ\phi$)
\[
\begin{aligned}
d_x\phi \colon T_x X &\longrightarrow T_{\phi(x)} Y \\
v\equiv v^\mu(x) \frac{\partial}{\partial x^\mu}\bigg|_x &\longmapsto d_x\phi(v)=
v^\mu(x)\left(\frac{\partial}{\partial x^\mu}\bigg|_{\phi(x)}+
\partial_\mu\phi^i(x) \frac{\partial}{\partial y^i}\bigg|_{\phi(x)} \right).
\end{aligned}
\]
This pushforward map will be indistinctly denoted by $d_x\phi$, $T_x\phi$ and
$\phi_{*,x}$. 

With this equivalence relation we consider our last (for the moment) important
space definition, which generalizes the velocity space (not phase space) in the
usual symplectic framework of classical mechanics; this is the bundle which
encodes field's first (order) derivatives. It is not very difficult to
show that the quotient $\Gamma(\pi_{YX})/\sim$ -whose elements will be denoted
by $j^1_x\phi$ for $\phi\in\Gamma(\pi)$, i.e. the equivalence class of $\phi$ at
$x$ made of every other section $\psi$ satisfying that $d_x\phi$ and $d_x\psi$
are exactly the same map- can be endowed with a manifold
structure, so that we obtain the following
\begin{definition}[First order jet bundle]
	Using the equivalence relation defined above, the manifold 
	$Z\equiv J^1\pi:=\Gamma(\pi_{YX})/\sim\,\equiv\{j^1_x\phi\,|\, x\in X,
	\phi\in\Gamma(\pi)\}$ is called the 
	\textbf{first-order jet bundle}. See \cite{saundersGeometryJetBundles1989}
	for a complete analysis of jet bundles.
\end{definition}
$Z$ fibers canonically over $X$ and $Y$ considering the natural
projections $\pi_{ZX}$ and $\pi_{ZY}$ taking a point $j^1_x\phi\in Z$ to the
base point $x$ and the field value $\phi(x)$ respectively; hence, we can 
consider on $Z$ the coordinates $(x^\mu, y^i, z^i_\mu)$ acting as
\begin{align*}
x^\mu(j^1_x\phi)&=x^\mu(x),\\
y^i(j^1_x\phi)&=\phi^i(x),\\
z^i_\mu(j^1_x\phi)&=\partial_\mu\phi^i(x).
\end{align*} 
Using this coordinate chart, it is immediate that $\dim Z = \left(n+1+m+
(n+1)m\right)$

If needed, to implement \textbf{boundary conditions}, denoting 
$\partial Y\equiv\pi^{-1}(\partial X)$ and $\partial Z\equiv\pi_{ZX}
^{-1}(\partial X)$, we consider a subbundle $B\subseteq\partial Z$ of 
$\partial Z\to\partial X$ (i.e. $B\to\partial X$ is a fiber bundle) and, for 
whichever field equations we obtain, our solutions $\psi$ are
required to also satisfy the condition $j^1\psi(\partial X)\subset
B$, where $j^1\phi\in\Gamma(\pi_{ZX})$ are canonically defined as
$j^1\psi(x)=j^1_x\psi$.

A diagrammatic summary of the spaces involved up now is shown in Figure 
\ref{fig:lagrangian_setting}

\begin{figure}[htb]
  \centering
  \includegraphics[width=.7\textwidth]
  {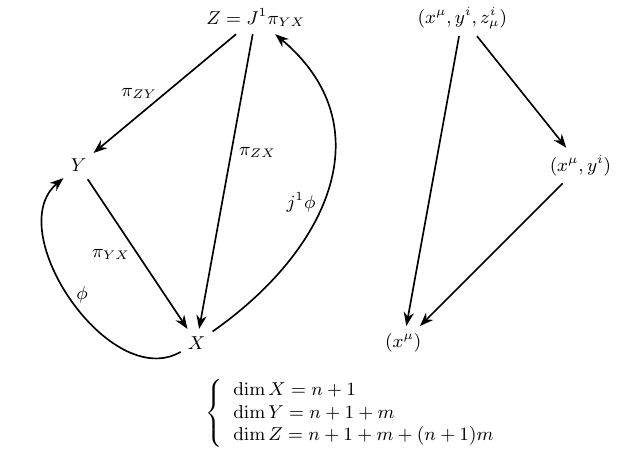}
  \caption{Involved fibrations in the Lagrangian setting.}
  \label{fig:lagrangian_setting}
\end{figure}

Although we stopped at the first-order jet bundle, we could have continued to
higher-order jets, for example the $\left(n+1+m+(n+1)m+\binom{n+2}{2}m\right)$-
dimensional second order jet bundle $J^2\pi$ with coordinates $(x^\mu,y^i,
z^i_\mu,z^i_{\mu\nu})$, where $z^i_{\mu\nu}(j^2_x\phi)=\frac{\partial^2
\phi^i}{\partial x^\mu\partial x^\nu}\bigg|_x$. This higher-order jets would
allow us to cover also higher-order theories where the Lagrangian could depend
on derivatives of order higher than one, i.e., $L\left(x^\mu,
\phi^i,\partial_\mu\phi^i,\partial_\mu\partial_\nu\phi^i,\dots\right)$.

Even though these generalizations will not be used as the manifold where our
Lagrangian will be defined, these higher-order
bundles allow us to define a total derivative
\[
	\der{}{x^\mu}=\frac{\partial}{\partial x^\mu}+z^i_\mu\frac{\partial}
	{\partial y^i}+z^i_{\mu\nu}\frac{\partial}{\partial z^i_\nu}+\dots
\]
which will be important in our Lagrangian setting.

\subsection{The vertical endomorphism \texorpdfstring{$S_\eta$}{S-eta}}
\label{subsect:the_vertical_endomorphism}

In order to fully complete the generalization of the symplectic structure to
field theories, we need to introduce the canonical structures appearing in these
brand new jet fiber bundles. In particular, in this subsection we define the
so-called vertical endomorphism.

Let us start considering two sections 
$\phi,\psi$ such that $\phi(x)=\psi(x)$ but having not equivalent jet prolongations;
by definition of these sections, both for $\phi$ and $\psi$ we have that 
$\dd\pi_{\phi(x)}(d_x\phi(v))=v$ $\forall v\in T_xX$ and so if we denote
$f_x:=d_x\phi-d_x\psi\neq 0$ we find that
\[
	\dd\pi_{\phi(x)}(f_x(v))=v-v=0,\quad\forall v\in T_xX
\]
which by definition means
$f_x(v)\in\mathcal{V}_{\phi(x)}\pi:=\mathrm{ker}\,\dd\pi_{\phi(x)}\subseteq
T_{\phi(x)}Y$, or equivalently, $f_x\in T^*_xX\otimes\mathcal{V}_{\phi(x)}\pi$
(see \cite{leonMethodsDifferentialGeometry1989}). In components $f_x=(\partial_\mu\phi^i-
\partial_\mu\psi^i)\left[\dd x^\mu|_x\otimes\frac{\partial}{\partial y^i}\big|_
{\phi(x)}\right]$. If we consider then the two points $j^1_x\phi,j^1_x\psi$
in $Z$ and define a point `difference' (in components) as
\begin{equation}
	\label{eq:affine_subtraction}
	j^1_x\phi-j^1_x\psi:=(0,0,\partial_\mu\phi^i-\partial_\mu\psi^i) 
\end{equation}
there is a clear components identification $j^1_x\phi-j^1_x\psi\equiv f_x$.

These remarks were made without any preferred pair of points $j^1_x\phi,
j^1_x\psi\in Z$, reminiscing a lack 
of privileged origin point, which is a defining characteristic of affine spaces.
Indeed the bundle $\pi_{ZY}:Z\to Y$ can be view as an affine bundle modelled
over the vector bundle $\pi^*(T^*X)\otimes_Y\mathcal{V}\pi\to Y$ where
$\pi^*(T^*X)$ denotes the pullback bundle over $Y$ w.r.t. the cotangent
fibration $T^*X\to X$. For completeness, let us recall what the pullback
of a fiber bundle is
\begin{definition}[Pullback bundle]
	\label{def:pullback-bundle}
	Let $\pi:E\to B$ be a general fiber bundle and $f:B^\prime\to B$ be
	a smooth map. We then define the \textbf{pullback bundle through}
	$\bm{f}$ as the projection $\pi^\prime:f^*E\to B^\prime$ where
	\[
		f^*E:=\{(b^\prime,e)\in B^\prime\times E\,|\, f(b^\prime)=\pi(e)\}
		\subseteq B^\prime\times E
	\]
	is equipped with the subspace topology of the product topology of $B^\prime
	\times E$ and where $\pi^\prime$ is the canonical projection onto the first
	factor. $f^*E$ is then a differentiable manifold which fibers to $B^\prime$
	through $\pi^\prime$
\end{definition}

And although we have arrived at this insight quite abruptly, the details can 
be checked on \cite{gotayMomentumMapsClassical2004,saundersAlternativeApproachCartan1987,
saundersGeometryJetBundles1989}. This affine bundle construction is
built in such a way that the $Z$ fibers over $Y$, $Z_y:=\pi_{ZY}^{-1}(y)$ are 
affine spaces modelled over the fibers $(\pi^*T^*X\otimes_Y\mathcal{V}\pi)_y=
T^*_{\pi(y)}X\otimes\mathcal{V}_y\pi$ (a vector space), giving an affine sense
to the `subtraction' defined on Equation \eqref{eq:affine_subtraction}, or
equivalently,
\[
	j^1_x\psi=j^1_x\phi+f_x
\]
as an affine sum. By construction, both points have the same $x$ and $y$ local
coordinates (same base point and field value), so any vector
$f_x\in(\pi^*T^*X\otimes_Y\mathcal{V}\pi)|_{\phi(x)}=
T^*_xX\otimes\mathcal{V}_{\phi(x)}\pi$ must point 
in the Z fiber direction (in the sense of this vector sum changing only the
$z$-component of points on $Z$).

Taking the derivative at $t=0$ of the curve $\gamma_{f_x}:\RR\to\pi^{-1}_{ZY}
(\pi_{ZY}(j^1_x\phi))$ defined by $\gamma_{f_x}(t)=j^1_x\phi+tf_x$
defines a canonical isomorphism $v:\pi^*T^*X\otimes_Y\mathcal{V}\pi\to
\mathcal{V}\pi_{ZY}$, called the \textbf{vertical lift}. In coordinates
it takes the following expression:
\[
	f_x^v:=f_{x,\mu}^i\frac{\partial}{\partial z^i_\mu}\bigg|_{\phi(x)},\quad\text{
		for}\, f_x=f_{x,\mu}^i\dd x^\mu|_x\otimes\frac{\partial}{\partial y^i}\bigg|_{\phi(x)}.
\]

With all this machinery, we are ready to define the vertical endomorphism.
First, if we consider $n+1$ vectors $V_0,\dots,V_n\in T_{j^1_x\phi}Z$
we note that for any of them $T_{j^1_x\phi}\pi_{ZY}(V_i)-T_x\phi\circ 
T_{j^1_x\phi}\pi_{ZX}(V_i)\in\mathcal{V}_{\phi(x)}\pi$
(cf. Appendix \ref{app:proof_of_belonging}).

Second, from the volume form $\eta$ we consider the family of covectors 
$\eta_i(x)\in T^*_xX$ for each $x\in X, i\in\{0,\dots,n\}$
defined as
\[
	\eta_i(x)=(-1)^{n+1-i}\iota_{T_{j^1_x\phi}\pi_{ZX}(V_0)}\dots\widehat{\iota_{T_{j^1_x\phi}
	\pi_{ZX}(V_i)}}\dots\iota_{T_{j^1_x\phi}\pi_{ZX}(V_n)}\,\eta(x).
\]
Finally,
\begin{definition}[Vertical endomorphism]
	We define the \textbf{vertical endomorphism} $S_\eta$ for our volume form
	$\eta$ as the (1,n+1)-rank tensor field defined on each point $z\in Z$, $S_\eta|_z\in T_z^*Z\otimes\overset{n)}{\dots}\otimes
	T_z^*Z\otimes\mathcal{V}_z\pi_{ZY}$ by
	\[
		(S_\eta)_{j^1_x\phi}(V_0,\dots,V_n)=\sum_{i=0}^n[\eta_i(x)\otimes
		T_{j^1_x\phi}\pi_{ZY}(V_i)-T_x\phi\circ T_{j^1_x\phi}\pi_{ZX}(V_i)]^v.
	\]	
	In coordinates,
	\[
		S_\eta=(\dd y^i-z^i_\mu\dd x^\mu)\wedge \dd^n x_\nu\otimes\frac{\partial}
		{\partial z^i_\nu}\equiv\theta^i\wedge \dd^n x_\nu\otimes\frac{\partial}
		{\partial z^i_\nu}
	\]
	where the $\theta^i:=\dd y^i-z^i_\mu\dd x^\mu$ are the so-called \textbf{contact forms} (cf. Subsection
	\ref{subsect:jet_prolongation_of_vector_fields} and \cite{saundersGeometryJetBundles1989}
	for the concept of the contact structure).
\end{definition}

\subsection{Jet prolongation of vector fields}
\label{subsect:jet_prolongation_of_vector_fields}	
Since the Lagrangian will be an object on $J^1\pi$ and our physical
field solutions will be sections $\phi\in\Gamma(\pi)$
of our \textbf{configuration bundle} $\pi:Y\to X$, we are interested
in filtering the whole space of sections $\sigma\in\Gamma(\pi_{ZX})$
of the jet bundle to only those of the form $\sigma=j^1\phi$ (called
\textbf{holonomic sections}), i.e. coming from an actual solution.

It turns out that these holonomic sections are characterized by the
contact forms $\theta^i$ previously defined (we can give a more
intrinsic definition of a contact form as follows: $(j^1\phi)^*
\theta=0,\forall\phi\in\Gamma(\pi)$ ), or more generally characterized
by the exterior algebraic ideal $\mathcal{C}:=\{\sum_{i=1}^m\alpha^i\wedge
\theta^i\,:\, (j^1\phi)^*\theta^i=0, \forall\phi\in\Gamma(\pi),\, 
\alpha^i\in\Omega(Z)\}$ generated by them.
Note that in the case of $\phi$ being a solution
of the soon to be defined Euler-Lagrange equations, $(j^1\phi)^*
\theta=0$ is nothing more than the usual sense of a vanishing 
\textbf{on-shell} condition in physics.

In order to characterize these contact forms, we have to first introduce
the following definition
\begin{definition}[Distributions and Cartan distribution]
	Let $M$ be a $m-$dimensional manifold. A $\bm{k}$
	\textbf{-dimensional distribution} $D$ on $M$ is a choice
	of a $k-$dimensional subspace $D(x)$ of $T_xM$ for every
	$x\in M$. We say that a vector field $X\in\mathfrak{X}(M)$
	\textbf{lies in the distribution} if $X(x)\in D(x)$ $\forall
	x\in M$. The distribution is called \textbf{involutive} if
	$[X,Y]$ lies in the distribution for all $X,Y\in\mathfrak{X}(M)$
	lying on the distribution. We also say that an immersed submanifold
	$\phi:N\subseteq M$ is an \textbf{integral manifold of} $\bm{D}$ if
	\[
		\dd_y\phi(T_yN)\subseteq D(\phi(y)),\quad\forall y\in N.
	\]
	A distribution is \textbf{completely integrable} if for every point $x\in M$
	there exists an integral manifold of the distribution of maximal dimension.
	
	In the case of the manifold $Z$, we define
	the \textbf{Cartan distribution} $\mathcal{D_C}$ as the
	distribution on $Z$ satisfying
	\[
		\iota_V\theta^i(z)=0,\quad \forall z\in Z,\forall V\in
		\mathcal{D_C}(z),
	\] 
	i.e. the annihilator distribution of the contact forms.
\end{definition}
A central result of distributions is the following theorem
\begin{theorem}[Frobenius theorem]
	\label{thm:Frobenius_theorem}
	A distribution is completely integrable if and only if it is involutive.
\end{theorem}
\begin{proof}
	See \cite{leonMethodsDifferentialGeometry1989}.
\end{proof}

\begin{proposition}\label{prop:holonomicity_iff_cartan}
	Let	$\sigma\in\Gamma(\pi_{ZX})$ be a section, then
	$\sigma=j^1\phi$ for some $\phi\in\Gamma(\pi_{YX})$
	(namely it is holonomic) if and only if $T_x\sigma(T_xX)\subseteq
	\mathcal{D_{C}}(\sigma(x))\,\forall x\in X$.
\end{proposition}
\begin{proof}
	The reader may refer to the Appendix
	\ref{app:proof_of_holonomicity_iff_cartan} for the full
	demonstration.	
\end{proof}
In short, a section $\sigma$ is holonomic if and only if it is an
integral section of the contact distribution. 

With symmetries in mind, it is natural to regard them as generated
by vector fields on $Y$ (cf. Chapter \ref{chap:symmetries_and_Noether_theorems}).
This makes it necessary to require the
canonical jet lift or prolongation of these
vector fields to vector fields on $Z$ to preserve holonomic sections,
as we expect sections coming from actual solutions of 
the e.o.m.s to flow into sections coming from actual solutions
if transformed by symmetries.
A vector field $X$ on $Z$ preserves all holonomic sections if and only if
it preserves the Cartan distribution
(i.e. that $\pounds_{X}Y$ lies in $\mathcal{D_{C}}$, for all
Y lying in $\mathcal{D_{C}}$ or equivalently -see 
\cite{leonSymmetriesClassicalField2004}- $\pounds_{X}\theta\in
\mathcal{C}$ for all $\theta \in\mathcal{C}$- both imply also
the preservation of $\mathcal{I(C)}$, the ideal of the exterior
algebra generated by the contact forms and their differentials).

Considering all this, we define
\begin{definition}[1-jet prolongation]
\label{def:1jet-prolongation}
	Let $\xi_Y\in\mathfrak{X}(Y)$. We define its \textbf{1-jet
	prolon\-gation} as the unique $\xi_Y^{(1)}\in\mathfrak{X}(Z)$
	projectable onto $\xi_Y$ (i.e. $T_z\pi_{ZY}(\xi_Y^{(1)}(z))=\xi_Y
	(\pi_{ZY}(z)),\,\forall z\in Z$) that preserves the Cartan
	distribution $(\pounds_{\xi_Y^{(1)}}\theta\in\mathcal{C},\,\forall
	\theta\in\mathcal{C})$. Locally, $\xi_{Y}^{(1)}$ is given by
	\cite{leonSymmetriesClassicalField2004}
	\[
		\xi_Y^{(1)}=\xi^\mu_Y(x,y)\frac{\partial}{\partial x^\mu}+
		\xi_Y^i(x,y)\frac{\partial}{\partial y^i}+\left(
		\der{\xi^i_Y(x,y)}{x^\mu}-z^i_\nu\der{\xi^\nu_Y(x,y)}{x^\mu}
		\right)\frac{\partial}{\partial z^i_\mu}.
	\]
	The definition of 1-jet prolongations will be vital for Chapter
	\ref{chap:symmetries_and_Noether_theorems} about symmetries.	
\end{definition}
The following proposition will be needed on Chapter \ref{chap:symmetries_and_Noether_theorems}
to show that symmetries of our Lagrangian form a Lie algebra
\begin{proposition}
\label{prop:symmetries_will_form_subalgebra}
	For all $\xi,\zeta\in\mathfrak{X}(Y)$, $[\xi,\zeta]^{(1)}=[\xi^{(1)},\zeta^{(1)}]$, i.e.,
	the $1$-jet prolongation is a Lie algebra homomorphism.
\end{proposition}
\begin{proof}
		Of course $[\xi^{(1)},\zeta^{(1)}]$ projects onto $[\xi,\zeta]$ and for any $\alpha
	\in\mathcal{C}$ contact form we have
	\[
		\pounds_{[\xi^{(1)},\zeta^{(1)}]}\alpha=
		\pounds_{\xi^{(1)}}\pounds_{\zeta^{(1)}}\alpha-
		\pounds_{\zeta^{(1)}}\pounds_{\xi^{(1)}}\alpha\in\mathcal{C}
	\]
	using that $\xi^{(1)}$ and $\zeta^{(1)}$ are 1-jet prolongations.
\end{proof}
\subsection{Lagrangian and Poincaré-Cartan forms}
Evidently, to describe a Lagrangian setting in the language of
differential geometry, at some point or another we would have
to introduce a Lagrangian density, and this is the time.

\begin{definition}[Lagrangian density, function and regularity]
\label{def:lagrangian_form_function_and_regularity}
A \textbf{Lagrangian density} $\mathcal{L}$ is an $(n+1)$-form on $Z$
with the requirement of being \textbf{semibasic} w.r.t. $\pi_{ZX}$,
i.e.
\[
	\iota_{v}\mathcal{L}=0,\quad\forall v\in\mathcal{V}\pi_{ZX}.
\]

The associated \textbf{Lagrangian function} $L$ then is the unique 
$L:Z\to\RR$ such that $\mathcal{L}=L\eta$ (note that we are making some
abuse of notation, as we have 
identified $\eta$ with its pullback to the jet bundle via the projection
$\pi_{ZX}^*\eta$). 

Lastly, $\mathcal{L}$ is said to be \textbf{regular} at a point whenever the
Hessian matrix
\[
	\left(\frac{\partial^2 L}{\partial z^i_\mu\partial z^j_\nu}\right)
	_{i,j,\mu,\nu}
\]
is regular at that point, and regular, whenever the Hessian is
regular at every point.
\end{definition}

Other important forms that can be built upon this Lagrangians are
the Poincaré-Cartan forms, which (in the case of regular Lagrangians
) will allow us to describe our equations of motion (e.o.m.) in an
equivalent way to Euler-Lagrange equations. The importance
of these forms and their role will be demonstrated in
Subsection \ref{subsect:euler_lagrange_and_de_donder_equations}.

\begin{definition}[Poincaré-Cartan forms]\label{def:poincare-cartan-forms}
	Considering the vertical endomorphism $S_\eta$ defined at
	Subsection \ref{subsect:the_vertical_endomorphism}, we define the
	\textbf{Poincaré-Cartan} $\bm{(n+1)}$\textbf{-form} as
	\[
		\Theta_L:=\mathcal{L}+(S_\eta)^*(\dd L),
	\]
	locally given by
	\begin{align*}
		\Theta_L&=\left(L-z^i_\mu\frac{\partial L}{\partial z^i_\mu}
		\right)\dd^{n+1}x+\frac{\partial L}{\partial z^i_\mu}\dd y^i\wedge
		\dd^nx_\mu \\
		&=\mathcal{L}+\frac{\partial L}{\partial z^i_\mu}\theta^i
		\wedge\dd^nx_\mu,
	\end{align*}
	where $(S_\eta)^*(\dd L)(V_0,\dots,V_n):=\dd L\left(S_\eta(V_0,\dots,V_n)\right)$.

	We also define the \textbf{Poincaré-Cartan} $\bm{(n+2)}$
	\textbf{-form} as
	\[
		\Omega_L:=-\dd\Theta_L,
	\]
	locally given by
	\[
		\Omega_L=-(\dd y^i-z^i_\mu \dd x^\mu)\wedge
		\left(
		\frac{\partial L}
		{\partial y^i}\dd^{n+1}x-\dd\left(\frac{\partial L}{\partial
		z^i_\mu}\right)\wedge \dd^nx_\mu
		\right).
	\]
\end{definition}

Boundary conditions on this description are implemented through the
condition of existence of an $(n+1)$-form $\Pi\in\Omega(Z)$ such
that $i_B^*\Theta_L=\dd\Pi$, where $i_B:B\to Z$ is the canonical 
inclusion. 

\subsection{Euler-Lagrange and De Donder equations}
\label{subsect:euler_lagrange_and_de_donder_equations}
Traditionally, for a section $\phi\in\Gamma(\pi)$ to be a physical
object, i.e. an object describing the dynamics of a physical
setting, apart from satisfying certain boundary conditions ($(j^1
\phi)(\partial X)\subseteq B$) it has to \textbf{extremise} (in
a variational calculus sense) the so-called \textbf{action integral}
\[
	S(\phi)=\int_{(j^1\phi)(C)}\mathcal{L}.
\] 
This variational approach for the extremality condition has an
equivalent geometric description through projectable vector fields
on $Y$.
\begin{definition}[Variations of a section]
	Let $\xi_Y\in\mathfrak{X}(Y)$ be a fixed projectable
	vector field with flow $\Phi_t^Y$ and let us denote by $\Phi_t^X$
	the flow of its projection. Then we define the \textbf{variations}
	w.r.t. $\xi_Y$ of any section $\phi\in\Gamma(\pi)$ as the sections $\phi_t\in
	\Gamma(\pi)$ defined by $\phi_t:=\Phi_t^Y\circ\phi\circ
	\Phi_{-t}^X$.
\end{definition}
Through this definition of variations, the principle of extremal
action materializes then as the condition \footnote{Where $!$ denotes that we are
imposing the equality, i.e., the left-hand side (lhs) must be equal to 
the right-hand side (rhs) of the equality; in this case, for the extremality condition to
hold by definition.}
\[
	\der{S(\phi_t)}{t}\bigg|_{t=0}\overset{!}{=}0,
\]
which can be rewritten as
\[
	\der{S(\phi_t)}{t}\bigg|_{t=0}=\der{}{t}\left[\int
	_{(j^1\phi_t)(C)}\mathcal{L}\right]\bigg|_{t=0}=\der{}{t}
	\left[\int_{C}(j^1\phi_t)^*\mathcal{L}\right]\bigg|_{t=0}
	\]
\[
	=\int_C\der{}{t}(j^1\phi_t)^*\mathcal{L}
	\big|_{t=0}=\int_C(j^1\phi)^*\pounds_{\xi_Y^{(1)}}(\mathcal{L})
	\overset{!}{=}0,
\]
where $\xi_Y^{(1)}$ is the
1-jet prolongation of $\xi_Y$ introduced in Definition 
\ref{def:1jet-prolongation} and in the last equality we used Lemma 2.4 of
\cite{leonSymmetriesClassicalField2004}.
Hence, we define
\begin{definition}[Extremal section of the action integral]
	\label{def:extremal_section}
	We say that $\phi\in\Gamma(\pi)$ is an \textbf{extremal section 
	of S} whenever
	\[
		\int_C(j^1\phi)^*\pounds_{\xi^{(1)}_Y}(\mathcal{L})=0
	\]
	for every projectable vector field $\xi_Y\in\mathfrak{X}(Y)$ and
	every compact $(n+1)$-dimensional submanifold $C$ of $X$.
\end{definition}

The extremality condition, after some manipulations (cf. 
Equation \eqref{eq:first_expression_extremality} of Appendix
\ref{app:some_manipulations_of_the_extremality_condition})
can be expressed as
\[
	\int_C(j^2\phi)^*\left[\frac{\partial L}{\partial y^i}-\der{}{x^\mu}
	\frac{\partial L}{\partial z^i_\mu}\right](\xi_Y^i-z^i_\nu
	\xi_Y^\nu)\eta+
	\int_{\partial C}(j^1\phi)^*\left[(\xi_Y^i-z^i_\nu\xi_Y^\nu)
	\frac{\partial L}{\partial z^i_\mu}+L\xi_Y^\mu\right]\dd^nx_\mu
	\overset{!}{=}0,
\]
i.e., Euler-Lagrange equations plus a boundary term. It can be
shown (again cf. Appendix
\ref{app:some_manipulations_of_the_extremality_condition})
that the boundary term is equal to
\[
	\int_{\partial C}(j^1\phi)^*(\iota_{\xi_Y^{(1)}}\Theta)
\]
for $\Theta=\Theta_L$, i.e., the Poincaré-Cartan $(n+1)-$form 
introduced in Definition \ref{def:poincare-cartan-forms}. This result
can also be proved in general for any $\Theta\in\Omega^{n+1}Z$ form
\textbf{Lepage-equivalent} to $\mathcal{L}$, i.e. such that
$\Theta-\mathcal{L}$ and $\iota_Vd\Theta$ belong to $\mathcal{I(C)}$
for all $V\in\mathcal{V}\pi_{ZY}$. However, for the rest of the thesis we
will only consider $\Theta_L$ for simplicity. 

Upon Definition \ref{def:extremal_section} we arrive at the very
well-known Euler-Lagrange equations
\begin{theorem}[Euler-Lagrange equations]
	Let $\phi\in\Gamma(\pi)$ be an extremal of $S$ and $C$ any compact
	$(n+1)$-dimensional submanifold of $X$ such that $\phi(C)$ is
	covered by a single patch with coordinates $(x^\mu,y^i)$. In the
	case that $\phi$ has fixed values at $\partial C$ w.r.t every
	possible variation, then $\phi$ satisfies the 
	\textbf{Euler-Lagrange equations}:
	\begin{equation}
	\label{eq:euler-lagrange_equations}
		(j^2\phi)^*\left(\frac{\partial L}{\partial y^i}-\der{}{x^\mu}
		\frac{\partial L}{\partial z^i_\mu}\right)=0,\quad\forall i\in\{
		1,\dots,m\}.
	\end{equation}
	In addition, a section $\phi\in\Gamma(\pi)$ is an extremal of $S$
	if and only if 
	\begin{equation}
		\label{eq:intrinsic_EL_equations}
		(j^1\phi)^*(\iota_\xi\Omega_L)=0,\quad\forall\xi\in
		\mathfrak{X}(Z),
	\end{equation}
	which are called the \textbf{intrinsic Euler-Lagrange equations} 
\end{theorem} 
\begin{proof}
	As we already proved, the extremality condition has the form
	\eqref{eq:first_expression_extremality}. If by hypothesis $\phi$
	has fixed values at $\partial C$, then every possible variation
	field $\xi_Y$ components must vanish on $(j^1\phi)(\partial C)$
	implying
	\[
		\int_C(j^2\phi)^*\left[\frac{\partial L}{\partial y^i}-\der{}{x^\mu}
		\frac{\partial L}{\partial z^i_\mu}\right](\xi_Y^i-z^i_\nu
		\xi_Y^\nu)\eta=0.
	\]
	As the components of the vector field are arbitrary, Euler-Lagrange
	equations must be met.
	For the if and only if statement cf. 
	\cite{leonSymmetriesClassicalField2004}.
\end{proof}
If we relax the holonomicity condition of the intrinsic Euler-Lagrange
equations and consider any section $\sigma\in\Gamma(\pi_{ZX})$, we
obtain the following set of equations
\begin{definition}[De Donder equations]
	Given a section $\sigma\in\Gamma(\pi_{ZX})$, the set of equations
	\begin{equation}\label{eq:De-Donder-equations}
		\sigma^*(\iota_\xi\Omega_L)=0,\quad\forall\xi\in\mathfrak{X}(Z)
	\end{equation}
	are called the \textbf{De Donder equations}.
\end{definition}
The importance of De Donder equations arise in the case
of \textbf{regular Lagrangians}; in such cases, holonomy is
recovered automatically and De Donder equations are proven to be
equivalent to Euler-Lagrange ones.
These two facts are collected in the following proposition
\begin{proposition}
\label{prop:De_Donder_eq_to_E-L}
	For a regular Lagrangian, if there exists a section $\sigma
	\in\Gamma(\pi_{ZX})$ satisfying the De Donder equations, then
	there exists a section $\phi\in\Gamma(\pi_{YX})$ such that
	$\sigma=j^1\phi$ and $\phi$ satisfies the Euler-Lagrange
	equations.
\end{proposition}
\begin{proof}
	See \cite{leonSymmetriesClassicalField2004}.
\end{proof}

Lastly, given a general Ehresmann connection $\Gamma$ on the bundle $\pi_{ZX}$ (i.e. a
distribution $H$ on $Z$ such that $T_zZ=H_z\oplus\mathcal{V}_z\pi_{ZX}$ for every point
$z\in Z$) we have an equivalent (geometric) description of the De Donder equations
in terms of the horizontal projector $\textbf{h}$, locally given $\forall z\in Z$ by
\[
\textbf{h}_z\left(\frac{\partial}{\partial x^\mu}\bigg|_z\right)=\frac{\partial}{\partial
x^\mu}\bigg|_z+\Gamma^i_\mu(z)\frac{\partial}{\partial y^i}\bigg|_z+\Gamma^i_{\mu\nu}(z)
\frac{\partial}{\partial z^i_\nu}\bigg|_z,
\]
\[
\textbf{h}_z\left(\frac{\partial}{\partial y^i}
\bigg|_z\right)=0,\quad \textbf{h}_z\left(\frac{\partial}{\partial z^i_\mu}\bigg|_z\right)=0,
\]
under the following definition of interior product for \textbf{h}
\begin{equation}
	\label{eq:interior_product_of_h}
	\iota_{\textbf{h}}\Omega(V_1,\dots,V_p):=\sum_{i=1}^p\Omega(V_1,\dots,\textbf{h}(V_i),
	\dots V_n),\quad\forall V_1,\dots,V_n\in\mathfrak{X}(Z)
\end{equation}
for any $p$-form $\Omega\in\Omega^p(Z)$. 

The equivalent equation is given in the following
proposition
\begin{proposition}
 Given a connection $\Gamma$ in $\pi_{ZX}$ satisfying the equation
 \begin{equation}
 \label{eq:geometric-De-Donder}
 	\iota_{\textbf{h}}\Omega_L=n\Omega_L,
 \end{equation}
 then, every horizontal local integral section $\sigma$ of $\Gamma$ (i.e. satisfying 
 $\textbf{h}\left(\frac{\partial}{\partial x^\mu}\big|_{\sigma(U)}\right)=\dd_z\sigma\left(\frac
 {\partial}{\partial x^\mu}\big|_U\right)$ for a given open set $U\subseteq Z$) is
 a solution of the De Donder equations. Equation
 \eqref{eq:geometric-De-Donder} will be called \textbf{geometric De
 Donder equation}.
\end{proposition}
\begin{proof}
	See \cite{leonSymmetriesClassicalField2004}.
\end{proof}
As this equivalent description (in the case of integrable connections or
holonomic connections, i.e., connections contained in the Cartan
distribution) is a pointwise condition (in the
sense of finding this connection, and so for each point on $Z$ a complement
to the vertical subspace), if the Lagrangian is not fully regular, we can
restrict this view to submanifolds of $Z$ in a form of an algorithm (as we
will see now) until finding a final submanifold in which the solution of this
field equation does exist. In the previous field equations perspective though,
for singular Lagrangians we can't proceed pointwise.

This geometric description can be stated both in this connection
perspective or in totally equivalent multivector field way (cf. for
example \cite{echeverria-enriquezMultivectorFieldsConnections1998}
to see this equivalence). In this work we will center
ourselves in the connection approach, in which boundary conditions are imprinted onto the connection itself.
As we restrict ourselves to fields satisfying $\sigma(\partial X)\subseteq B$, \textbf{h}
induces canonically a connection $\partial\textbf{h}$ on the bundle $\pi_{B\partial X}
:B\to\partial X$ under the additional condition $\textbf{h}_z\left(T_z B\right)\subseteq T_z B$ $\forall
z\in B$ (which just adds up to the equation $\iota_{\textbf{h}}\Omega_L=n\Omega_L$ restriction).

\subsection{The singular case}
\label{subsect:the_singular_case}
In the case of \textbf{singular Lagrangians} (i.e. not regular ones), this connection $\Gamma$ only exists
in general for a particular submanifold $Z_f\subset Z$ which is found through the following algorithm
(see \cite{deleonGeometricalApproachClassical1996}).

We start first with the total manifold $Z_1\equiv Z$ and construct $Z_2$ as the submanifold
of points such that a solution of Equation \eqref{eq:geometric-De-Donder} exists, i.e.
\[
	Z_2:=\{z\in Z_1\,|\,\exists\textbf{h}_z:T_z Z_1\to T_z Z_1\text{ linear projector such that } \textbf{h}_z^2=
	\textbf{h}_z, 
\]
\[
	\mathrm{ker}\textbf{h}_z=\mathcal{V}_z\pi_{ZX},	\iota_{\textbf{h}_z}\Omega_L(z)=n
	\Omega_L(z),\text{ and }\textbf{h}_{\tilde{z}}(T_{\tilde{z}}B)\subseteq T_{\tilde{z}}B\,
	\forall \tilde{z}\in B\}.
\] 
In the case that $Z_2$ is an embedded submanifold of $Z_1$ ($\mathrm{dim} Z_2>0$), and we have
$\textbf{h}_z(T_zZ_1)\subset T_z Z_2$ and $\textbf{h}_z(T_zB)\subset T_z B_2$ $\forall z\in Z_2$ (where
$B_2:=B\cap Z_2$), $\textbf{h}$ is the horizontal projector of an Ehresmann connection $\Gamma$
\textbf{defined on} $\bm{Z_2}$ satisfying the intrinsic field equation
(in the connection formulation).

The reason why these two additional conditions are imposed is that the resulting
horizontal local integral sections that solve the De Donder equations \textbf{must} live where
Equation \eqref{eq:geometric-De-Donder} is satisfied (i.e. $Z_2$ by definition)
and satisfy boundary conditions for $B_2$ in consequence. If a given 
horizontal local integral section $\sigma$ lives
in $Z_2$, $\textbf{h}\left(\frac{\partial}{\partial x^\mu}\big|_{\sigma(U)}\right)=\dd_z\sigma\left(\frac
{\partial}{\partial x^\mu}\big|_U\right)\in TZ_2$ for a given open set $U\subseteq X$,
i.e., \textbf{the connection must be tangent to} $\bm{Z_2}$, and
this is where the imposition $\textbf{h}_z(T_zZ_1)\subset T_z Z_2$ and
$\textbf{h}_z(T_zB)\subset T_z B_2$ $\forall z\in Z_2$ come from.
This is not guaranteed by the definition of $Z_2$ and so we need to go one step further.
Defining canonically $B_r$ at any subsequent step we construct
\[
	Z_3:=\{z\in Z_2\,|\,\exists\textbf{h}_z:T_z Z_1\to T_z Z_2 \text{ linear projector such that } \textbf{h}_z^2=
	\textbf{h}_z, 
\]
\[
	\mathrm{ker}\textbf{h}_z=\mathcal{V}_z\pi_{ZX},
	\iota_{\textbf{h}_z}\Omega_L(z)=n
	\Omega_L(z),\text{ and }\textbf{h}_{\tilde{z}}(T_{\tilde{z}} B_2)\subseteq T_{\tilde{z}}B_2\,\forall
	\tilde{z}\in B_2\},
\]
now, with the same reasoning, if $\textbf{h}_z(T_z Z_1)$ is not contained in $T_z Z_3$ and $\textbf{h}_z(T_z B)$ is not
contained in $T_z B_3$ $\forall z\in Z_3$ we continue with the third step. 
Proceeding further we obtain a sequence of embedded submanifolds of the form
\[
	\dots\hookrightarrow Z_3\hookrightarrow Z_2\hookrightarrow Z_1=J^1\pi=Z
\] 
up to the stabilization term in which, for some $f\in\NN$, we find $Z_{f+1}=Z_f$. This set will be denoted by $Z_f$ and
if $\mathrm{dim} Z_f>0$ (cf. \cite{deleonGeometricalApproachClassical1996}), then there exists a connection
$\Gamma$ in $\pi_{ZX}$ whose projector verifies the (geometric version of the) De Donder equations at the same 
time that satisfies the boundary conditions.
However, this is not equivalent to having a solution of the Euler-Lagrange equations. The existence of a local
section $\sigma$ satisfying the De Donder equations is equivalent to the horizontal distribution of the connection being
integrable (equivalently involutive by Frobenius Theorem
\ref{thm:Frobenius_theorem}) and this is equivalent to the connection being flat (i.e. with zero curvature, cf.
\cite{deleon1995ehresmann}). Note the crucial difference
with classical mechanics: as we deal with ordinary differential equations, integrability is guaranteed.

For singular Lagrangians, Proposition
\ref{prop:De_Donder_eq_to_E-L} can be substituted with an holonomicity condition for De Donder implying Euler-Lagrange.
With all this, we can build an additional algorithm upon $Z_f$ to find a submanifold in which Euler-Lagrange equations
are met.

\subsection{Hamiltonian forms and multisymplectic brackets}
\label{subsect:hamiltonian_forms_and_multisymplectic_brackets}
In this subsection we will introduce the concept of multisymplectic
brackets as the natural generalization of Poisson brackets in
symplectic classical mechanics and state a result relating regular
theories with multisymplectic Poincaré-Cartan frameworks.
We recall Definition \ref{def:multisymplectic_forms_and_manifold}
in which the concept of multisymplectic forms was introduced.

For $n=0$ (time-dependent Lagrangian mechanics), the structure $(Z,\Omega_L, dt)
$ is a cosymplectic manifold in the case of a regular Lagrangian 
(cf. \cite{leonMethodsDifferentialGeometry1989} to check the peculiarities of this type of structures).
For $n>0$ we can also characterize regular
Lagrangians in an equivalent way using Definitions
\ref{def:lagrangian_form_function_and_regularity} and
\ref{def:multisymplectic_forms_and_manifold} with the following proposition
\begin{proposition}\label{prop:reg_lagr_implies_multisymp}
	For $n>0$, the Lagrangian $L$ is regular iff $\Omega_L$ is a multisymplectic
	form
\end{proposition}
\begin{proof}
	See \cite{leonSymmetriesClassicalField2004}.
\end{proof}
For the following definitions, we will suppose $n>0$.
In the case of $(Z, \Omega_L)$ being a multisymplectic manifold.

We can now introduce Hamiltonian forms and vector fields in similarity to the symplectic
formalism (in that case Hamiltonian functions and vector fields)
\begin{definition}[Hamiltonian forms and vector fields]
	\label{def:hamiltonian_forms_and_vector_fields}
	A form $\alpha\in\Omega^n(Z)$ and a vector field $X_\alpha\in\mathfrak{X}(Z)$
	are called \textbf{Hamiltonian form} and \textbf{vector field}, respectively, whenever
	\[
		\dd\alpha=\iota_{X_\alpha}\Omega_L.
	\]
\end{definition}
Two Hamiltonian forms equivalent up to a closed form have the same
Hamiltonian vector fields, and for any given Hamiltonian form $\alpha$,
$X_\alpha$ is uniquely determined up to an element of the kernel of
\begin{align*}
\flat_{\Omega_L}\colon\mathfrak{X}(Z)&\to\Omega^{n+1}Z\\
X &\mapsto\iota_X\Omega_L.
\end{align*}
Hence for a regular Lagrangian $L$ (equivalently for $(Z,\Omega_L)$ being multisymplectic),
Proposition \ref{prop:reg_lagr_implies_multisymp}
implies (by non-degeneracy of $\Omega_L$) uniqueness of the vector field
$X_\alpha$ in the case that it exists.

We now define the multisymplectic bracket operation as a generalization
of the symplectic Poisson bracket operation which, as it is well known, governs the
dynamical evolution of symplectic Hamiltonian settings.
\begin{definition}[Multisymplectic bracket]\label{def:multisymplectic_bracket}
 	For two Hamiltonian forms $\alpha$ and $\beta$, we define its
 	\textbf{multisymplectic bracket} as
 	\[
 		\{\alpha,\beta\}=\iota_{X_\alpha}\iota_{X_\beta}\Omega_L,
 	\]
 	for any two $X_\alpha$, $X_\beta$ Hamiltonian vector fields of $\alpha$ and $\beta$
 	respectively, giving us an $n-$form.
 \end{definition} 
Hamiltonian forms are closed under this bracket, and the Hamiltonian
field of the bracket of two Hamiltonian forms is exactly the Lie 
bracket of the Hamiltonian vector fields associated with that two
Hamiltonian forms
\begin{proposition}\label{prop:field_of_bracket_is_Lie_bracket}
	Let $\alpha$ and $\beta$ be two Hamiltonian $n-$forms, then $\{\alpha,\beta\}$
	is a Hamiltonian $n-form$ and we have $X_{\{\alpha,\beta\}}=[X_\alpha,X_\beta]$
	for any two $X_\alpha$, $X_\beta$ Hamiltonian vector fields of $\alpha$ and $\beta$,
 	respectively.
\end{proposition}
\begin{proof}
	This follows straightforward using Cartan's magic formula for the Lie
	derivative of a vector field ($\pounds_X=\iota_X\dd-\dd\iota_X$) by just nothing
	first that, by Leibniz rule of the Lie derivative
	\[
		\pounds_{X_\alpha}(\iota_{X_\beta}\Omega_L)=\iota_{\pounds_{X_\alpha}
		X_\beta}\Omega_L+\iota_{X_\beta}(\pounds_{X_\alpha}\Omega_L)=
		\iota_{[X_\alpha,{X_\beta}]}\Omega_L+\iota_{X_\beta}(\pounds_{X_\alpha}
		\Omega_L)
	\]
	where we used the well known relation $\pounds_{X_\alpha}X_\beta
	=[X_\alpha,X_\beta]$; then, isolating the first term of the last expresion
	\[
		\iota_{[X_\alpha,X_\beta]}\Omega_L=\pounds_{X_\alpha}(\iota_{X_\beta}
		\Omega_L)-\iota_{X_\beta}(\pounds_{X_\alpha}\Omega_L)=\pounds_{X_\alpha}
		(\dd\beta)-\iota_{X_\beta}\cancel{dd\alpha}-\iota_{X_\beta}\iota_{X_\alpha}
		\cancel{\dd\Omega_L}
	\]
	\[
		=\dd\iota_{X_\alpha}\dd\beta+\iota_{X_\alpha}\cancel{\dd\dd\beta}=\dd\iota_{X_\alpha}
		\iota_{X_\beta}\Omega_L=\dd\{\alpha,\beta\}.
	\]
\end{proof}

\subsection{The graded Lie algebra of Hamiltonian forms}
\label{subsec:graded}

The bracket introduced in Definition \ref{def:multisymplectic_bracket} was nothing more
than the spearhead of a much richer algebraic structure, a result due to 
\cite{cantrijn1996hamiltonian} that needs of the proper introduction of 
\textbf{multivector fields} and the corresponding \textbf{Schouten--Nijenhuis} bracket
(a generalization of the Lie bracket for vector fields).
The sketch of the idea is the following: $\Omega_L$ is a $(n+2)$-form, hence contracting it with
a \textbf{single} vector field yields a $(n+1)$-form that define Hamiltonian $n-$forms up
to closed forms as mentioned in Definition \ref{def:hamiltonian_forms_and_vector_fields}.
Now the generalization is then straightforward. What would happen if instead we contract
$\Omega_L$ with higher-order vector fields (i.e. multivector fields)? Indeed in this way
we can obtain generalized Hamiltonian forms from degree $n$ \textbf{all the way down to}
degree $0$; these Hamiltonian forms happen to form a graded Lie algebra (up to closed forms)
upon a generalized multisymplectic bracket operation \cite{cantrijn1996hamiltonian}.

Before introducing all these concepts, let us state some
notation conventions. We denote by $\bigwedge^p Z$ the usual bundle of exterior $p-$forms
on $Z$ while $\bigvee_p Z$ denotes the bundle of $p-$vectors (contravariant skew-symmetric
$p$-tensors) while the modules of differential $p-$forms and $p-$vector fields will be
denoted by $\Omega^p(Z)$ and $\mathfrak{X}^p(Z)$ respectively. We also
denote by $\Omega^p_{cl}(Z)\subseteq\Omega^p(Z)$ the module of closed $p-$forms.
\subsubsection*{Multivector fields and the Schouten--Nijenhuis bracket}
 
Recall that $\bigvee_p Z$ denotes the bundle of $p$-vectors on $Z$, i.e.\
the bundle whose fibre at $z \in Z$ is $\bigwedge^p(T_zZ)$, and that
$\mathfrak{X}^p(Z) := \Gamma\left(\bigvee_p Z\right)$ is the module of
$p$-vector fields on $Z$, so that $\mathfrak{X}^1(Z) = \mathfrak{X}(Z)$ and
$\mathfrak{X}^\bullet(Z) = \bigoplus_{p \geq 0} \mathfrak{X}^p(Z)$ is a
graded commutative algebra under the exterior (wedge) product, i.e.
$U \wedge V = (-1)^{pq}\, V \wedge U$ for $U \in \mathfrak{X}^p(Z)$,
$V \in \mathfrak{X}^q(Z)$.
 
The Lie bracket of vector fields admits a unique extension to
$\mathfrak{X}^\bullet(Z)$ as a graded derivation of the wedge product $\wedge$; this
extension is the Schouten--Nijenhuis bracket, defined as follows

\begin{definition}[Schouten--Nijenhuis bracket]
	We define the \textbf{Schouten--Nijen\-huis bracket} as the unique
	$\RR-$bilinear operation
	\[
		[\,\cdot,\cdot\,] : \mathfrak{X}^p(Z) \times \mathfrak{X}^q(Z)
		\longrightarrow\mathfrak{X}^{p+q-1}(Z)
	\]
	such that
	\[
		[X_1\wedge\dots\wedge X_p,Q]=\sum_{i=1}^p(-1)^{i+1}X_1\wedge\dots
		\wedge\widehat{X_i}\wedge\dots\wedge X_p\wedge[X_i,Q]
	\]
	for all $X_1,\dots,X_p\in\mathfrak{X}(Z)$, $Q\in\mathfrak{X}^q(Z)$ where
	$[X_i,Q]:=\pounds_{X_i}Q$ is the Lie derivative of a multivector field
	w.r.t. a vector field defined as
	\[
		L_X(Y_1\wedge\dots\wedge Y_q):=\sum_{j=1}^qY_1\wedge\dots
		\wedge[X,Y_j]\wedge\dots\wedge Y_q
	\]
	where $Q$ can be always decomposed as a sum of multivector fields of the
	form $Y_1\wedge\dots\wedge Y_q$.
\end{definition}
\begin{remark}
	The Schouten--Nijenhuis bracket has the following properties. Let $P\in
	\mathfrak{X}^p(Z)$, $Q\in\mathfrak{X}^q(Z)$ and $R\in\mathfrak{X}^r(Z)$,
	then
	\begin{enumerate}
		\item $[P,Q]=(-1)^{pq}[Q,P]$,
		\item $[P,Q\wedge R]=[P,Q]\wedge R+(-1)^{pq+q} Q\wedge[P,R]$,
		\item $[X_1\wedge\dots\wedge X_p,Y_1\wedge\dots\wedge Y_q]=(-1)^{pq}
			\sum_{j=1}^q(-1)^{j+1}Y_1\wedge\dots\wedge\widehat{Y_j}\wedge\dots
			\wedge Y_q\wedge[Y_j,X_1\wedge\dots\wedge X_p]$
		\item $(-1)^{p(r-1)}[P,[Q,R]]+(-1)^{q(p-1)}[Q,[R,P]]+(-1)^{r(q-1)}[R,
		[P,Q]]=0$ (graded Jacobi identity).
	\end{enumerate}
	(see \cite{vaismanLecturesGeometryPoisson1994} for a proof of these
	properties and for a more detailed discussion of the Schouten--Nijenhuis bracket
	\footnote{In the literature there are several sign conventions for
	the definition of the Schouten-Nijenhuis bracket. We remark our adoption of
	the convention in \cite{vaismanLecturesGeometryPoisson1994}.})
\end{remark}

The drop in one unit, although not natural of a graded Lie algebra, is needed
in order to reproduce the behaviour of the usual Lie bracket of vector fields
(i.e. the Lie bracket of two vector fields must be a vector field). In order
to fix this, we could change the indexing of the graded structure of
$\mathfrak{X}^\bullet(Z)$. This kind of reindexing will not be done yet but
will be fundamental for Theorem \ref{thm:graded-lie-algebra}.
 
We also extend the interior product to multivector fields by declaring
\begin{equation}
  \label{eq:iota-multivector}
  \iota_{\xi_1 \wedge \dots \wedge \xi_p}\, \Omega
  := \iota_{\xi_p} \dots \iota_{\xi_1} \Omega ,
  \qquad \forall\xi_1, \dots, \xi_p \in \mathfrak{X}(Z), \forall\Omega\in \Omega(Z)
\end{equation}
so that $\iota_U : \Omega^r(Z) \to \Omega^{r-p}(Z)$ for
$U \in \mathfrak{X}^p(Z)$. It is easy to see that for all $U_a\in
\mathfrak{X}^a(Z)$ and $U_b\in\mathfrak{X}^b(Z)$
\begin{equation}
	\label{eq:anticommuting_multivectors_rule}
	\iota_{U_b\wedge U_a}=\iota_{U_a}\iota_{U_b}=(-1)^{ab}\iota_{U_b}\iota_{U_a}
\end{equation}
\subsubsection*{Hamiltonian forms of arbitrary degree}
In the following we extend Definition \ref{def:hamiltonian_forms_and_vector_fields} from
Hamiltonian $(n-)$forms and $(1-)$vector fields to Hamiltonian $\ell-$forms
and ($n+1-\ell$)-vector fields for all $\ell\in\{0,\dots,n\}$ in a completely parallel way
noting the following relation in the grading
\[
	\underbrace{n}_{\text{form}}+\underbrace{1}_{\text{vector}}=(n+2)-1=
	\underbrace{\ell}_{\text{form}}+\underbrace{(n+1-\ell)}_{\text{multi-vector}},
	\quad\forall\ell\in\{0,\dots,n\}.
\]
\begin{definition}[Hamiltonian forms and multivector fields]
\label{def:ham-forms-general}
Let $0 \leq \ell \leq n$. A form 
$\alpha\in \Omega^{\ell}(Z)$ and a multivector field $U\in\mathfrak{X}^{n+1-\ell}(Z)$
are called a \textbf{Hamiltonian} $\bm{\ell-}$\textbf{form} and an associated
\textbf{Hamiltonian} \bm{$(n+1-\ell)-$}\textbf{vector field}, respectively,
whenever
\[
  \dd\alpha = \iota_U \Omega_L .
\]
We write $\mathcal{H}^{\ell}(Z)$ for the space of Hamiltonian
$\ell$-forms and $\mathfrak{X}^{n+1-\ell}_{H}(Z)$ for the space of
Hamiltonian $(n+1-\ell)$-vector fields.
\end{definition}
 
For $\ell = n$ we have $U \in \mathfrak{X}^1(Z)
= \mathfrak{X}(Z)$ and Definition~\ref{def:ham-forms-general} reduces
verbatim to Definition \ref{def:hamiltonian_forms_and_vector_fields}; in particular
in the case of symplectic mechanics
$\mathcal{H}^0(Z) \subseteq C^\infty(Z)$ plays the role of the space of
observables of classical mechanics ($n=0$) with the Poisson bracket defined
canonically upon the symplectic form.
 
As in the case already discussed, $\bm{\alpha}$ \textbf{determines} $\bm{U}$
\textbf{only up to a multivector field which annihilates} $\bm{\Omega_L}$
(\textbf{even in the multisymplectic case}, i.e. for a regular Lagrangian,
non-degeneracy of $\Omega_L$ does not say anything about multivectors with degree
strictly greater than one), and $\bm{U}$ \textbf{determines} $\bm{\alpha}$
\textbf{only up to a closed} $\bm{\ell}$\textbf{-form}. Since closed forms are
trivially Hamiltonian (consider the $0$ multivector field), $\Omega^{\ell}_{cl}(Z)
\subseteq\mathcal{H}^{\ell}(Z)$ and we can eliminate this `up to closed forms'
redundancy by quotienting the graded spaces of Hamiltonian forms as
\[
  \widehat{\mathcal{H}}^{\ell}(Z)
  := \mathcal{H}^{\ell}(Z)/\Omega^{\ell}_{cl}(Z),
  \qquad
  \widehat{\mathcal{H}}^{\ast}(Z):=\bigoplus_{\ell=0}^{n}
  \widehat{\mathcal{H}}^{\ell}(Z),
\]
whose elements (equivalence classes) we denote by $\hat\alpha$. The elements of
$\widehat{\mathcal{H}}^{\ell}(Z)$ are in one-to-one correspondence with
equivalence classes of Hamiltonian $(n+1-\ell)$-vector fields, two of them
being identified when their difference annihilates $\Omega_L$.
 
\subsubsection*{The graded bracket}
 
The key computation is the following, which generalizes
Proposition \ref{prop:field_of_bracket_is_Lie_bracket} \footnote{Note the
minus sign on \eqref{eq:multibracket_to_multibraces}; it is needed
in order to reduce to the $\ell=n$ case of Proposition
\ref{prop:field_of_bracket_is_Lie_bracket}. Indeed consider that $U_1,U_2
\in\mathfrak{X}(Z)$, then by Equation \ref{eq:iota-multivector} we have 
$\iota_{[U_1,U_2]}\Omega_L=-\dd\iota_{U_1\wedge U_2}\Omega_L=-\dd\iota_{U_2}
\iota_{U_1}\Omega_L=\dd\iota_{U_1}\iota_{U_2}\Omega_L$.}
and is a direct consequence of the properties of the Schouten--Nijenhuis bracket together with
$\dd\Omega_L = 0$.
 
\begin{proposition}[\cite{cantrijn1996hamiltonian}]
\label{prop:SN-hamiltonian}
Let $U_1 \in \mathfrak{X}^{n+1-\ell_1}_{H}(Z)$ and
$U_2 \in \mathfrak{X}^{n+1-\ell_2}_{H}(Z)$ be Hamiltonian multivector
fields associated with $\alpha_1 \in \mathcal{H}^{\ell_1}(Z)$ and
$\alpha_2 \in \mathcal{H}^{\ell_2}(Z)$ with $\ell_1$ not necessarily equal
to $\ell_2$. Then $[U_1,U_2]$ is again
Hamiltonian, and
\begin{equation}
\label{eq:multibracket_to_multibraces}
  \iota_{[U_1,U_2]}\Omega_L=-\,\dd\,
  \iota_{U_1 \wedge U_2}\Omega_L=\dd[-\iota_{U_1 \wedge U_2}\Omega_L] .
\end{equation}
\end{proposition}
This proposition relates, upon the injective (multivector field)$\mapsto$
(Hamiltonian form) assignment, Schouten-Nijenhuis bracket with the operation
appearing on the right-hand side of \eqref{eq:multibracket_to_multibraces};
the definition of the generalized multisymplectic bracket is then
straightforward.
\begin{definition}[Generalized multisymplectic bracket]
\label{def:graded-bracket}
For $\alpha_i \in \mathcal{H}^{\ell_i}(Z)$ with some associated Hamiltonian
multivector fields $U_i$ ($i=1,2$), we set
\[
  \{\alpha_1, \alpha_2\} := -\,\iota_{U_1 \wedge U_2}\Omega_L .
\]
\end{definition}
The right-hand side does not depend on the choice of the $U_i$ associated
with the $\alpha_i$, so the bracket is well defined; moreover it vanishes
identically whenever $\ell_1 + \ell_2 < n$. 
\begin{remark}
	As this bracket is $\RR-$bilinear and zero if any of the arguments is a closed
	form, the algebraic structure induced by $\{\cdot,\cdot\}$ on
	$\mathcal{H}^\ast(Z):=\oplus_{\ell=0}^n\mathcal{H}^\ell(Z)$ \textbf{passes canonically
	to the quotient space} $\widetilde{\mathcal{H}}^\ast(Z)$ upon the definition
	\[
		\{\hat{\alpha},\hat{\beta}\}:=\widehat{\{\alpha,\beta\}}.
	\]
\end{remark}

A degree count gives
\[
  \deg \{\alpha_1, \alpha_2\}
  = (n+2) - \big[(n+1-\ell_1) + (n+1-\ell_2)\big]
  = \ell_1 + \ell_2 - n ,
\]
so that $\{\cdot,\cdot\}$ maps
$\mathcal{H}^{\ell_1}(Z) \times \mathcal{H}^{\ell_2}(Z)
\to \mathcal{H}^{\ell_1+\ell_2-n}(Z)$. 

Note again that this is not the desired
behaviour for a graded Lie algebra -which has to map degrees $(\ell_1,\ell_2)$
to degree $\ell_1+\ell_2$, i.e. being additive on the degrees- as we mentioned
already for the Schouten--Nijenhuis bracket; we expect then to introduce a
shift -in the degrees of definition for the Hamiltonian forms- by fixing a so-called
\textbf{grading}.

\subsubsection*{The grading}
We define $\mathcal{K}^s(Z):=\widehat{\mathcal{H}}^{n-s}(Z)$ for all integer $s\in
\{0,\dots,n\}$ and $\mathcal{K}^s(Z):=\{0\}$ for all integers $s>n$. Fix some integer $s$,
then if we consider some $\hat{\alpha}\in\mathcal{K}^s(Z)=
\widehat{\mathcal{H}}^{n-s}(Z)$, its degree as an element of
$\widehat{\mathcal{H}}^{n-s}(Z)$ is $\deg\alpha:=n-s$ and so as an element of
$\mathcal{K}^s(Z)$ is grade is $s=n-(n-s)=n-\deg\alpha$; hence if we consider
any $\tilde{\alpha}$ Hamiltonian form of degree $\deg\alpha\in\{0,\dots,n\}$
(in the usual grading), we fix the new grading by
\begin{equation}
	\label{eq:grading:fix}
	|\hat{\alpha}|:=n-\deg\alpha
\end{equation}
and $|\hat{\alpha}|=0$ for $\deg\alpha>n$.

Using this new grading, the generalized bracket becomes additive. Indeed, let
$\hat{\alpha}_1\in\widehat{\mathcal{H}}^{\ell_1}(Z)$ and
$\hat{\alpha}_1\in\widehat{\mathcal{H}}^{\ell_1}(Z)$ then
\[
  \big|\{\hat\alpha_1,\hat\alpha_2\}\big|
  =n-\deg\{\alpha,\beta\}= n - (\ell_1+\ell_2-n) = (n-\ell_1)+(n-\ell_2)
\]
\[
	=(n-\deg\alpha_1)+(n-\deg\alpha_2)= |\hat\alpha_1| + |\hat\alpha_2|.
\]

The job is not fully done yet. Considering \eqref{eq:anticommuting_multivectors_rule}
and Definition \ref{def:graded-bracket}, we note that the sign of permuting
our graded bracket depends on the degrees of the Hamiltonian multivector fields
representatives, as $|\hat{\alpha_1}|=n-\deg \alpha_1=n-\ell_1$ and
$\deg U_1=n+1-\ell_1$, then $\deg U_1=|\hat{\alpha_1}|+1$ and so
\[
	\{\hat{\alpha_1},\hat{\alpha_2}\}=(-1)^{(|\hat{\alpha_1}|+1)(|\hat{\alpha_2}|+1)}
	\{\hat{\alpha_2},\hat{\alpha_1}\}.
\]
In order to have a consistent sign rule between permutation of the bracket and
the graded Jacobi identity, we expect to have instead the \textbf{Koszul sign
rule}
\[
	\{\hat{\alpha_1},\hat{\alpha_2}\}=-(-1)^{|\hat{\alpha_1}||\hat{\alpha_2}|}
	\{\hat{\alpha_2},\hat{\alpha_1}\},
\]
which is obtained under the \textbf{final redefinition of the bracket} as
\[
	\{\hat{\alpha_1},\hat{\alpha_2}\}^*:=(-1)^{|\hat{\alpha_1}|}
	\{\hat{\alpha_1},\hat{\alpha_2}\}
\]
which of course leaves the additivity of the bracket untouched. Finally we obtain
the desired result
\begin{theorem}[\cite{cantrijn1996hamiltonian,leonDescriptionClassicalField2025,
leonGradedPoissonGraded2025}]
\label{thm:graded-lie-algebra}
$\left(\bigoplus_{s \geq 0}\mathcal{K}^{s}(Z),\ \{\cdot,\cdot\}^{\ast}\right)$
is a graded Lie algebra with respect to the grading \eqref{eq:grading:fix};
that is, $\{\cdot,\cdot\}^{\ast}$ is $\mathbb{R}$-bilinear, additive in the
degree, graded skew-symmetric,
\[
  \{\hat\alpha,\hat\beta\}^{\ast}
  = -(-1)^{|\hat\alpha||\hat\beta|}\{\hat\beta,\hat\alpha\}^{\ast},
\]
and satisfies the graded Jacobi identity
\[
  (-1)^{|\hat\alpha||\hat\gamma|}
  \{\{\hat\alpha,\hat\beta\}^{\ast},\hat\gamma\}^{\ast}
  +(-1)^{|\hat\beta||\hat\alpha|}
  \{\{\hat\beta,\hat\gamma\}^{\ast},\hat\alpha\}^{\ast}
  +(-1)^{|\hat\gamma||\hat\beta|}
  \{\{\hat\gamma,\hat\alpha\}^{\ast},\hat\beta\}^{\ast}
  =0.
\]
Moreover, the subalgebra consisting of the classes of Hamiltonian
$n$-forms, i.e.\ the elements with $|\hat\alpha| = 0$, is an ordinary Lie
algebra.
\end{theorem}
 
\begin{remark}[Recovering the previous picture]
The last statement is exactly the content of
Proposition \ref{prop:field_of_bracket_is_Lie_bracket}: in the degree-zero sector all the signs
$(-1)^{|\hat\alpha||\hat\beta|}$ are trivial and the graded structure
collapses to the Lie algebra of Hamiltonian $n$-forms, isomorphic (in the
regular Lagrangian case) to the Lie
algebra $(\mathfrak{X}^1_{H}(Z), [\cdot,\cdot])$ of Hamiltonian vector
fields. In this sense the bracket of
Subsection \ref{subsect:hamiltonian_forms_and_multisymplectic_brackets}
is the `top layer' of Theorem \ref{thm:graded-lie-algebra}.
\end{remark}

\section{Hamiltonian setting}
\label{sect:hamiltonian_setting}
In time-dependent classical mechanics (as $L:TQ\to\RR$), the canonical momentum $\frac{\partial L}
{\partial\dot{q}}$ is nothing more than a function from the tangent bundle $TQ$
to the space of momenta, which is identified with the cotangent bundle $T^*Q$ (i.e., the
dual bundle of $TQ$).
Hence, for the development of a Hamiltonian counterpart to that of the
Lagrangian setting living on $Z$, a canonical generalization could be built
around the dual $Z^*$ (generalizing the phase space of classical mechanics).

The dual $(J^1\pi)^*$ -cf. \cite{gotayMomentumMapsClassical2004} for a proper
discussion- is the vector bundle over $Y$ with the set of affine
maps from $J^1_y\pi$ to $\Omega_x^{n+1}$ being the fiber over $y\in\pi^{-1}(x)
\subset Y$ ($(J^1\pi)^*$ is then built as the disjoint union of these fibers).
Luckily for us, there is an alternative formulation of this affine dual bundle
considering first the following set of spaces
\begin{definition}
	Setting for $r\in\NN$
	\[
	(\Omega_r^{n+1}Y)_y:=\left\{\sigma\in\bigwedge^{n+1}T^*_yY\,|\,\iota_{V_1}\dots\iota_{V_r}\sigma=0,\,\forall V_i\in\mathcal{V}_y\pi,\, i\in\{1,\dots,
	r\}\right\},
	\]
	we consider the family of spaces of forms $\Omega_r^{n+1}Y:=\sqcup_{y\in Y}
	(\Omega_r^{n+1}Y)_y$ as the bundle over $Y$ with canonical projection
	\begin{align*}
		\nu_r:\Omega_r^{n+1}Y&\to Y\\
		\sigma\in(\Omega_r^{n+1}Y)_y&\mapsto y.	
	\end{align*}
\end{definition}
For $r=1$ these are just the $(n+1)$-forms of $Y$ annihilated by any $\pi-$vertical
vector (also called \textbf{semibasic forms}), hence, locally given by 
$p(x,y)\dd^{n+1}x$. For $r=2$, $\Omega_2^{n+1}Y$ includes \textbf{also} (as 
$\Omega_1^{n+1}Y\subset\Omega_2^{n+1}Y$) the forms $p^\mu_i(x,y)\dd y^i\wedge
\dd^nx_\mu$ and so on.
$Z^*$ (called the \textbf{dual first-order bundle}) is then built as the quotient
vector bundle
\[
	Z^*:=\Omega_2^{n+1}Y/\Omega_1^{n+1}Y
\]
over $Y$ where its elements are locally of the form $p^\mu_i(x,y)\dd y^i
\wedge \dd^nx_\mu$ (hence $Z^*$ has coordinates $(x^\mu,y^i,p^\mu_i)$ and so
dimensions $n+1+m+(n+1)m$ just as $Z$, and has canonical projection maps
$\pi_{Z^* Y}$ and $\pi_{Z^* X}$).
We can then define our Hamiltonian on $Z^*$ as
\begin{definition}[Hamiltonian and Hamiltonian function]
	\label{def:Hamiltonian_and_Hamiltonian_function}
	A \textbf{Hamiltonian} $h$ is a section of the canonical projection
	$\mu:\Omega_2
	^{n+1}Y\to Z^*$ given locally by
	\[h(x^\mu,y^i,p^\mu_i)=(x^\mu,y^i,p=-H(x^\mu,y^i,p^\mu_i),p^\mu_i)\]
	where $H$ is called a \textbf{Hamiltonian function}. 
\end{definition}

The advantage of working from the beginning with forms as points in the definition
of $\Omega_r^{n+1}Y$ is that we have a canonical
multisymplectic $n+2$ form on $\Omega^{n+1}_2Y$ (as opposed to the Lagrangian setting
where this form depended on the choice of Lagrangian and its multisymplectic nature on
this Lagrangian being regular)
\begin{definition}[Canonical Liouville and multisymplectic forms]
	$\Omega_2^{n+1}Y$ has a canonical $(n+1)-$form $\Theta$ defined as
	\[
		\Theta|_\omega(X_0,\dots,X_n):=\omega\left(\dd\nu_2(X_0),\dots,\dd\nu_2(X_n)
		\right),\quad\forall\omega\in\Omega_2^{n+1}Y
	\]	
	and called \textbf{multimomentum Liouville form} (or canonical multisymplectic potential).
	Then $\Omega:=-\dd\Theta$ is the $(n+2)$ \textbf{canonical multisymplectic} (cf. 
	\cite{leonSymmetriesClassicalField2004}) \textbf{form}. 

	They can be expressed locally as
	\begin{align*}
		\Theta&=p\dd^{n+1}x+p^\mu_i\dd y^i\wedge \dd^nx_\mu,\\
		\Omega&=-\dd p\wedge \dd^{n+1}x-\dd p^\mu_i\wedge\dd y^i\wedge \dd^nx_\mu.
	\end{align*}
\end{definition} 

These forms, once given a Hamiltonian, can be pullbacked to $Z^*$ as
\[
	\Theta_h:=h^*\Theta,\quad \Omega_h:=h^*\Omega,
\]
with local expressions
\begin{align*}
	\Theta_h&=-H \dd^{n+1}x+p^\mu_i \dd y^i\wedge \dd^nx_\mu\\
	\Omega_h&=\dd H\wedge \dd^{n+1}x-\dd p^\mu_i\wedge \dd y^i\wedge \dd^nx_\mu
\end{align*}
We will show that these forms are canonically related to the Poincaré-Cartan forms given 
by a Lagrangian density once we define the Legendre map; on their own, they lead us to a
generalization of Hamilton's equations, which are called Hamilton-De Donder-Weyl equations
(we will stick to Hamilton's equations for brevity)
\begin{remark}
	\label{rem:canonial_form_on_form_spaces}
	This canonical form is not tied to this particular structure and indeed it arises
	naturally in every space of forms. Let $M$ be an $m-$dimensional smooth manifold and
	consider the fiber bundle $\tau:\bigwedge^k M\to M$ of exterior $k-$forms
	(for any $k\in\{1,\dots,m\}$) given by $\tau(\omega(x))=x$ for all $x\in M$ where
	$\omega(x)\in\bigwedge^k T^*_xM=(\bigwedge^k M)_x$.
	Then, consider any $\alpha\in\bigwedge^k M$ (a point) fixed but
	arbitrary and let $X_1,\dots,X_k\in T_\alpha\bigwedge^kM$ be a fixed but arbitrary set
	of vectors tangent at $\alpha$; we define a canonical $k$-form $\Theta$ on $\bigwedge^k
	M$ as
	\[
		\Theta|_\alpha(X_1,\dots,X_k)=\alpha(\dd_\alpha\tau(X_1),
		\dots,\dd_\alpha\tau(X_k)).
	\]
\end{remark}

\begin{definition}[Hamilton's equations]
	For a given Hamiltonian $h$, we say that a section $\sigma\in\Gamma(\pi_{Z^*X
	})$ satisfy the \textbf{Hamilton's equations} (also called
	\textbf{Hamilton-De Donder-Weyl equations}) if
	\begin{equation}\label{eq:Hamilton-equations}
		\sigma^*(\iota_\xi\Omega_h)=0,\quad\forall\xi\in\mathfrak{X}(Z^*).
	\end{equation}
	Locally, if $\sigma(x^\mu)=\left(x^\mu,\sigma^i(x^\mu),\sigma_i^\mu(x^\mu)
	\right)$, Hamilton's equations take the usual form
	\[
		\frac{\partial\sigma^i}{\partial x^\mu}=\frac{\partial H}{\partial p^\mu
		_i},\quad\sum_{\mu=0}^n\frac{\partial\sigma_i^\mu}{\partial x^\mu}=-
		\frac{\partial H}{\partial y^i}.
	\]
\end{definition} 
Boundary conditions are set in a totally equivalent way as they were in the
Lagrangian setting ($B^*\subseteq\partial Z^*$ and so on)
\cite{leonSymmetriesClassicalField2004}; The interested reader is referred to \cite{IbortSpivak17} 
for a more detailed discussion. The same happens for the
geometric perspective to the Hamilton's equations materialized through a given
connection in $\pi_{Z^*X}$. 

Lastly for this Hamiltonian setting, a canonical way to lift $Y$ vector fields
to $Z^*$ is needed in order to rigorously define our Hamiltonian symmetries (as
we will see later).

Let $\xi_Y\in\mathfrak{X}(Y)$ and $\alpha=\nu_2^*\beta$ for some $\beta\in
\Omega_1^{n}Y$ semibasic $n-$form on $Y$, we then define 
\begin{definition}[$\alpha$-lift]
	\label{def:alpha_lift}
	For the objects just mentioned, we define the $\bm{\alpha}$\textbf{-lift} of $\xi_Y$ on
	$Y$ to $\Omega_2^{n+1}Y$ as the unique vector field $\xi_Y^\alpha$ that projects
	onto $\xi_Y$ through $\dd\nu_2$ and satisfies $\pounds_{\xi_Y^\alpha}\Theta=
	\dd\alpha$.

	In order for $\xi_Y^\alpha$ to be projectable onto a valid vector field on
	$Z^*$, $\xi_Y$ must be projectable onto $X$ (if not its flow would not preserve neither
	semibasic nor 1-vertical forms and then any posible lift onto $Z^*$ would not be
	properly defined). 
	Locally, working upon the defining
	conditions for the $\alpha-$lift of $\xi_Y=(\xi_Y^\mu,\xi_Y^i)$, is a
	straightforward computation to derive the components
	\begin{align*}
		\xi^p_Y&=-p\frac{\partial\xi^\mu}{\partial x^\mu}-p^\mu_i\frac{\partial\xi^i}
		{\partial x^\mu}+\frac{\partial\alpha^\mu}{\partial x^\mu}\\
		\xi^{p^\mu_i}_Y&=-p\frac{\partial\xi^\mu}{\partial y^i}-p^\mu_j\frac{\partial
		\xi^j}{\partial y^i}+p^\nu_i\frac{\partial\xi^\mu}{\partial x^\nu}-p^\mu_i\frac{
		\partial\xi^\nu}{\partial x^\nu}+\frac{\partial\alpha^\mu}{\partial y^i}
	\end{align*}
	of the vector field $\xi^\alpha_Y\in\mathfrak{X}(\Omega_2^{n+1}Y)$
	with components $(\xi_Y^\mu,\xi_Y^i,\xi_Y^p, \xi_Y^{p^\mu_i})$ and consequently,
	when $\xi_Y$ is projectable (i.e. $\partial\xi_Y^\mu/\partial y^i=0$), of 
	$\xi^{(1*)}_Y=(\xi^\mu_Y,\xi_Y^i,\xi_Y^{p^\mu_i})\in\mathfrak{X}(Z^*)$.
\end{definition}

\section{Legendre map}
As we hinted before, in the case of a regular Lagrangian, there exists a
(local) one-to-one map between our Lagrangian and Hamiltonian settings and
an equivalence between their respective equations. First, notice by the local
expression of $\Theta_L$ given in Definition \ref{def:poincare-cartan-forms} that
it satisfies $\iota_\xi\Theta_L=0$, $\forall\xi\in
\mathcal{V}\pi_{ZY}$ and that by proposition 2.6 of
\cite{leonSymmetriesClassicalField2004} (also immediate from the local expression
of $\Theta_L$), $\iota_\xi\iota_\zeta\Theta_L=0$, $\forall
\xi,\zeta\in\mathcal{V}\pi_{ZX}$, then we can properly define
\begin{definition}[Legendre transformation]
	\label{def:legendre_transformation}
	Given a Lagrangian function $L$, we define the \textbf{Legendre transformation}
	as the map $\mathrm{Leg}_L:Z\to\Omega_2^{n+1}Y$ that, for each $z\in Z$, 
	acts on a set of vectors $\xi_0,\dots,\xi_n\in T_{\pi_{ZY}(z)}Y$ as
	\[
		\left(\mathrm{Leg}_L(z)\right)(\xi_0,\dots,\xi_n):=(\Theta_L)_z(\hat{\xi}_0,\dots,
		\hat{\xi}_n)
	\] 	
	$\forall\hat{\xi}_0,\dots,\hat{\xi}_n\in T_z Z$ projectable into
	$\xi_0,\dots,\xi_n$. Locally, this map is given by
	\[
		\mathrm{Leg}_L(x^\mu,y^i,z^i_\mu)=\left(x^\mu,y^i,p=L-z^i_\mu\frac{\partial L}
		{\partial z^i_\mu},p^\mu_i=\frac{\partial L}{\partial z^i_\mu}\right).
	\]
	$p$ is called the energy component while $p^\mu_i$ the momentum (or multimomentum)
	component.
\end{definition} 
\begin{remark}
	Note that in the case of a regular Lagrangian, by the inverse function
	theorem, there is a (local diffeomorphic) one-to-one correspondence between field
	velocities and field momenta ($z^i_\mu\leftrightarrow p^\mu_i$).
	And once this transformation is defined, the canonical Legendre map is just the
	composition with the canonical projection $\mu:\Omega_2^{n+1}Y\to Z^*$.
\end{remark}
\begin{definition}[Legendre map]
	\label{def:legendre_map}
	Given a Lagrangian function $L$, we define the \textbf{Legendre map}
	$\mathrm{leg}_L:Z\to Z^*$ as $\mathrm{leg}_L:=\mu\circ\mathrm{Leg}_L$.
	Its local expression is clearly derived from the previous definition to be
	$p_i^\mu=\parder{L}{z^i_\mu}$.

	The Legendre map is a local diffeomorphism if and only if the Lagrangian is
	regular (cf. \cite{leonSymmetriesClassicalField2004}) while whenever
	$\mathrm{leg}_L$ (and hence $\mathrm{Leg}_L$) is globally a diffeomorphism, we say that
	the Lagrangian is \textbf{hyperregular} (which of course implies regularity). 
	\end{definition}

Note that Definition \ref{def:Hamiltonian_and_Hamiltonian_function} of the Hamiltonian
function gives $H=z^i_\mu\parder{L}{z^i_\mu}-L$ which recovers the usual $H=p_i\dot{q}^i-L$
formula of symplectic mechanics to obtain the Hamiltonian for a given Lagrangian, where $(q^i,p_i)$ are the
generalized coordinates and velocities in the phase space.

The Legendre map allows us to relate the already defined forms on $Z$ and
$Z^*$ through their pullbacks
\begin{proposition}
	Given a Hamiltonian $h$ we have the following relations
	\[
		(\mathrm{Leg}_L)^*\Theta=\Theta_L,\quad (\mathrm{Leg}_L)^*\Omega=\Omega_L
	\]
	\[
		(\mathrm{leg}_L)^*\Theta_h=\Theta_L,\quad (\mathrm{leg}_L)^*\Omega_h=\Omega_L
	\]
\end{proposition}
\begin{proof}
	The relations are straightforward using the local expression of the forms
	involved.
\end{proof}

Assuming that boundary conditions on $Z$ and $Z^*$ are compatible (i.e. $
\mathrm{leg}_L(B)=B^*$ and $(\mathrm{leg}_L)^*(\Pi^*)=\Pi$), we also have
a theorem (cf. \cite{leonSymmetriesClassicalField2004}) relating De Donder
and Hamilton's equations for regular Lagrangians
\begin{theorem}[\textbf{Equivalence theorem}]
	\label{thm:equivalence_theorem}
	For a regular Lagrangian $L$, if $\sigma\in\Gamma(\pi_{ZX})$ satisfies De
	Donder equations \eqref{eq:De-Donder-equations}, then $\sigma_2:=\mathrm{leg}_L
	\circ\sigma_1$ verifies Hamilton's equations \eqref{eq:Hamilton-equations}.

	Reciprocally, if $\sigma_2\in\Gamma(\pi_{Z^*X})$ satisfies Hamilton's equations,
	then the \textbf{locally defined} $\sigma_1:=\mathrm{leg}^{-1}_L\circ\sigma_2$
	satisfies De Donder equations.
\end{theorem}
Geometric equations are also equivalent for a regular Lagrangian. Given a
connection $\Gamma$ on the Lagrangian side, we have a connection on the
Hamiltonian side by just pushing through $\dd\,\mathrm{leg}_L$ the Lagrangian
horizontal vectors; the Lagrangian horizontal projector \textbf{h} is also pushed to
a Hamiltonian projector given by $\textbf{h}^*:=\dd\,\mathrm{leg}_L\circ
\textbf{h}\circ(\dd\,\mathrm{leg}_L)^{-1}$. 
\subsection{Almost regular Lagrangians and the constraint algorithm}
\label{subsect:almost_regular_lagrangians_and_the_constraint_algorithm}
Legendre equivalence can also be obtained between the Lagrangian and Hamiltonian
descriptions in the case of non-regular Lagrangians, although some conditions have to be met
(which will fall under the umbrella term of `almost regularity' for the Lagrangian).
This case is of particular interest due to Yang--Mills Lagrangian being almost regular.
\begin{definition}[Almost regular Lagrangian]
	\label{def:almost_regularity}
	Given a Lagrangian $L:Z\to\RR$, we will call it to be \textbf{almost regular} in the case
	that $\mathrm{Leg}_L(Z)=:\tilde{M_1}$ is a submanifold of $\Omega_2^{n+1}Y$ and
	$\mathrm{Leg}_L:Z\to\tilde{M_1}$ is a submersion with connected (in the topological sense)
	fibers. 
\end{definition}
We can also define $M_1:=\mathrm{leg}_L(Z)=\mu\circ\mathrm{Leg}_L(Z)=\mu(\tilde{M_1})$
such that we have the following remark
\begin{remark}
\label{rem:M1_equivalence}
If the Lagrangian is almost regular (which we will take as hypothesis for the
remaining of this subsection), the following is satisfied
\begin{itemize}
	\item $M_1$ is a submanifold of $Z^*$ and $\mathrm{leg}_L:Z\to M_1$ is also a 
	submersion with connected fibers.
	\item $\tilde{M_1}$ and $M_1$ are fibrations over $X$ and $Y$.
	\item $\mu_1:\tilde{M_1}\to M_1$ (the restriction to $\tilde{M_1}$ of the canonical
	projection $\mu:\Omega_2^{n+1}\to Z^*$) is a (global) \textbf{diffeomorphism} and hence
	we can define $h_1:=(\mu_1)^{-1}:M_1\to\tilde{M_1}$.
	\item Regularity implies almost regularity as $\mathrm{Leg}_L$ becomes a
	local diffeomorphism and fibers are trivially connected as being the preimage
	of a single point (which is of course connected).
\end{itemize}
\end{remark}
A diagrammatic summary with the relations involved up to this point can be seen in Figure 
\ref{fig:almost_regular_playground} where $\mathrm{Leg}_1$ and $\mathrm{leg}_1$ are
defined using the preimage of the inclusion maps as $\tilde{j}_1\circ\mathrm{Leg}_1=
\mathrm{Leg}_L$ and $j_1\circ\mathrm{leg}_1=\mathrm{leg}_L$

\begin{figure}[htb]
  \centering
  \includegraphics[width=.5\textwidth]{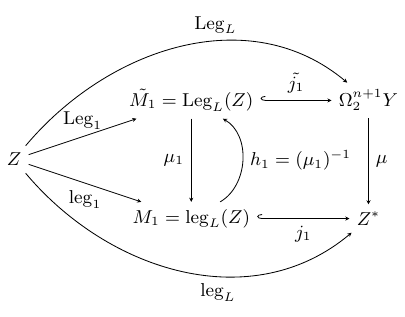}
  \caption{Visual summary of the almost regular `playground'.}
  \label{fig:almost_regular_playground}
\end{figure}

We can then pullback $\Omega$ to $M_1$ as $\Omega_{M_1}:=h_1^*(\tilde{j}_1^*\Omega)$
such that the geometric Hamilton's equations can be viewed as 
\begin{equation}
 \label{eq:geometric_almost_regular_Hamilton_equations}	
 \iota_{\tilde{\textbf{h}}}\Omega_{M_1}=n\Omega_{M_1}
\end{equation} 
for a given connection $\tilde{\textbf{h}}$ of the fibration $\pi_{M_1X}$ with boundary
conditions given by $\tilde{\textbf{h}}_{\tilde{z}}(T_{\tilde{z}}B_1^*)\subseteq
T_{\tilde{z}}B_1^*$, $\forall\tilde{z}\in B_1^*$, where $B_1^*:=B^*\cap M_1$
submanifold of $B^*$ (we define $B_r^*$ in a similar manner).
Of course, $\mathrm{leg}_1^*\Omega_{M_1}=\Omega_L$.

We then proceed in a parallel way to the algorithm of Subsection
\ref{subsect:the_singular_case} and hence define
\[
	M_2:=\{\tilde{z}\in M_1\,|\,\exists\tilde{\textbf{h}}_{\tilde{z}}:T_{\tilde{z}} M_1\to 
	T_{\tilde{z}} M_1\text{ linear projector such that } \tilde{\textbf{h}}_{\tilde{z}}^2=
	\tilde{\textbf{h}}_{\tilde{z}}, 
\]
\[
	\mathrm{ker}\tilde{\textbf{h}}_{\tilde{z}}=\mathcal{V}_{\tilde{z}}\pi_{M_1X},
	\iota_{\tilde{\textbf{h}}_{\tilde{z}}}\Omega_{M_1}(\tilde{z})=n\Omega_{M_1}
	(\tilde{z}),\text{ and }\tilde{\textbf{h}}_{\tilde{\tilde{z}}}(T_{\tilde{\tilde{z}}}
	B_1^*)\subseteq T_{\tilde{\tilde{z}}}B_1^*\,\forall \tilde{\tilde{z}}\in B_1^*\}.
\]
If $\tilde{\textbf{h}}_{\tilde{z}}(T_{\tilde{z}}M_1)\subset T_{\tilde{z}} M_2$ and 
$\tilde{\textbf{h}}_{\tilde{z}}(T_{\tilde{z}}B_1^*)\subset T_{\tilde{z}} B_2^*$
$\forall\tilde{z}\in M_2$ the algorithm stops and we have our connection
solution to the geometric Hamilton equation in the restricted $M_2\subseteq Z^*$.
If one of the conditions is not met, we proceed defining $M_3$, $M_4$ and so on
until stabilization on a given $M_f$ such that we find a connection on the
fibration $\pi_{M_f X}$ with horizontal projection $\tilde{\textbf{h}}$ satisfying
Equation \eqref{eq:geometric_almost_regular_Hamilton_equations} and boundary
conditions.

For $i\in\ZZ$, $i\geq2$, if we denote $\mathrm{leg}_i:=\mathrm{leg}_L|_{Z_i}$ where $Z_i$
is defined accordingly to Subsection \ref{subsect:the_singular_case} algorithm, then
$\mathrm{leg}_i(Z_i)=M_i$ and so both algorithms are equivalent, in the sense that
both stabilize in the same step; $\mathrm{leg}_f:Z_f\to M_f$ is a fibration with
the same fibers as $\mathrm{leg}_1$ and it can be proved that geometric De Donder
equation \eqref{eq:geometric-De-Donder} connection solutions and geometric 
Hamilton's equation connection solutions are equivalent (if projectable one to
another) \cite{leonSymmetriesClassicalField2004}.

\chapter{Instantaneous framework} \label{chap:instantaneous_formalism}

Up to now the (covariant) multisymplectic formulation -both in
the Lagrangian and Hamiltonian settings- has been a \textbf{
finite-dimensional} framework built around fiber bundles in which
our field solutions are identified as sections of the configuration bundle
built over the spacetime $X$. In this formulation, our equations give 
us a geometrically global picture of the physics behind it, but
not a `dynamical' picture, in the sense of having a solution that
evolves in time from a typically called `initial-value problem';
and although multisymplectic geometry feels like the natural
framework in which to generalize symplectic mechanics -having as we
saw a canonical multisymplectic form as the canonical generalization of the 
symplectic form of classical mechanics- (which has
long prove to be fruitful, cf. for example 
\cite{abrahamFoundationsMechanics2008}),
in order to have a canonical `time-arrow', we need to introduce the
so-called \textbf{instantaneous framework}.
Hence, we need to first break Lorentz covariance by
choosing a privileged time coordinate on our spacetime and then
consider our fields as defined on spatial slices of fixed time
(hence introducing a so-called slicing of our spacetime).
With this perspective though, we describe an 
\textbf{infinite-dimensional} framework in which -at least locally-
we have $X\cong I\times M$, for $I\subseteq\RR$ and $M$ a so-called
Cauchy surface. In this section we will assume $n>0$. 
\begin{definition}[Cauchy surface]
	Given a compact oriented $n-$dimensional manifold $M$ (with $n>0$)
	such that there exists an embedding $\tau:M\to X$ with
	$\tau(\partial M)\subseteq\partial X$ and $\tau(\mathring{M})
	\subseteq\mathring{X}$, the pair $(M,\tau)$ is usually called
	a Cauchy surface over $X$, defined up to an orientation
	and volume-preserving diffeomorphism of $M$.

	Of course, depending on the initial value problem we are trying
	to model, the set $M$ would be fixed accordingly, and so
	rather than the previous perspective, we will call \textbf{
	Cauchy surface over} $\bm{X}$ to a space $\tX$ of such kind
	of embeddings $\tau$ of the fixed $M$ to (the also fixed) $X$,
	i.e.
	\[
		\tX:=\{\tau:M\to X\,|\, (M,\tau) \text{ is a Cauchy
		surface over X}\}.
	\]
	This set depends on the physical theory involved.
\end{definition}
This set depends as we said on the initial value problem given, but of
course the problem must also characterize the initial configuration and
phase spaces $\tY$ and $\tZ$
\begin{definition}[Space of Cauchy data]
	We denote by $\tZ$ the set of embeddings $\gamma:M\to Z$
	such that $\exists\,\phi\in\Gamma(\pi)$ with
	\[
		\gamma=(j^1\phi)\circ\pi_{ZX}\circ\gamma,\quad
		\pi_{ZX}\circ\gamma\in\tX\quad\text{and}\quad
		\gamma(\partial M)\subseteq B.
	\]
	We also restrict $\tZ$ definition by imposing $\pi_{\tZ
	\tX}:\tZ\to\tX$ with $\pi_{\tZ\tX}
	(\gamma):=\pi_{ZX}\circ\gamma$ to be a locally trivial fibre
	bundle -i.e., $\forall\tau\in\tX$ and all open set $U\ni
	\tau$, $(\pi_{\tZ\tX})^{-1}(U)\cong U\times F$ for some
	differential manifold $F$ called typical fiber-. We call
	$\tZ$ then the \textbf{space of Cauchy data}.
\end{definition}

\begin{definition}[Space of Dirichlet data]
	We denote by $\tY$ the set of embeddings $\delta:M\to Y$
	of the form $\delta=\pi_{ZY}\circ\gamma$ for some $\gamma\in\tZ$
	such that, once again $\pi_{\tZ\tY}:\tZ\to\tY$ with $\pi_{\tZ\tY}
	(\gamma):=\pi_{ZY}\circ\gamma$ is a locally trivial fibration.

	We define lastly $\pi_{\tY\tX}$ to be the unique mapping from
	$\tY$ to $\tX$ such that $\pi_{\tZ\tX}=\pi_{\tY\tX}\circ\pi
	_{\tZ\tY}$ and we call $\tY$ the \textbf{space of Dirichlet data}.
\end{definition}
A visual summary of the spaces involved up now is shown in Figure
\ref{fig:instantaneous_setting}.
\begin{figure}[htb]
  \centering
  \includegraphics[width=.7\textwidth]{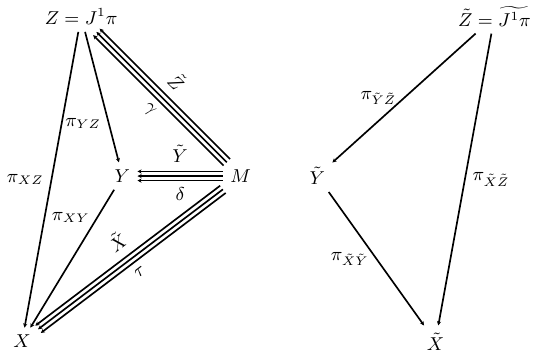}
  \caption{Involved fibrations in the instantaneous setting. Triple lines denote an embedding
  belonging to the corresponding tilde space drawn just above.}
  \label{fig:instantaneous_setting}
\end{figure}

It important to notice the double role of tangent vectors on $\tZ$.

Let $v\in T_\gamma\tZ$ for a fixed $\gamma\in\tZ$ and consider a
curve $\tilde{c}:I\subseteq\RR\to\tZ$ such that $\tilde{c}(0)=\gamma$ and
$\dot{\tilde{c}}(0)=v$. Of course, for any fixed $t\in I\subseteq\RR$,
we have a map $\tilde{c}(t):M\to Z$, as $\tilde{c}(t)\in\tZ$. Hence, for
a fixed $u\in M$ we obtain a map $\tilde{c}(t)(u)=:c^u(t):I\to Z$, which is of
course a curve on $Z$ satisfying $c^u(0)=\tilde{c}(0)(u)=\gamma(u)$ and
$\dot{c}^u(0)=\dot{\tilde{c}}(0)(u)=v(u)\in T_{\gamma(u)}Z$ (this
definition will be implicitly inverted i.e. $\tilde{c}(t)(u):=c^u(t):I\to Z$
in Proposition \ref{prop:integral_curves_on_tildeZ}). 

We then notice that $v$ is a map $v:M\to T_{\gamma(u)}Z$, for each
$u\in M$, so that we obtain the useful identification
\begin{equation}
	\label{eq:identification0}
	T_\gamma\tZ\cong\{v:M\to TZ\,|\, v(u)\in T_{\gamma(u)} Z\}=
	\Gamma(\gamma^*TZ),
\end{equation}
where we have used Definition \ref{def:pullback-bundle} of the
pullback bundle in the last equality.
This identification then gives us a way of \textbf{inducing vector fields on}
$\bm{\tZ}$ \textbf{from given vector fields on} $\bm{Z}$; 
for example, for $\xi\in\mathfrak{X}(Z)$
we define $\tilde{\xi}\in\mathfrak{X}(\tZ)$ for a given
$\gamma\in\tZ$ as the vector $\tilde{\xi}(\gamma)\in T_\gamma
\tZ\cong\Gamma(\gamma^*TZ)$
\[
	\tilde{\xi}(\gamma)(u):=\xi(\gamma(u)),\quad u\in M.		
\]

Identification \eqref{eq:identification0}
implies also the no less important remark that, as vectors
on $TZ$ act canonically on $\Omega^r(Z)$ forms
by the interior product, given $\alpha\in\Omega^r(Z)$ and 
$v\in T_\gamma\tZ$ we define canonically the interior product
$\iota_v\alpha$ as
\[
	(\iota_v\alpha)(z):=\iota_{v(u)}[\alpha(z)],\quad\text{for }
	z=\gamma(u),\,u\in M.
\]
Relying on these facts, upon integration (recall $M$ being
orientable and compact) we \textbf{crucially} induce $k-$forms on $\tZ$
from $(k+n)-$forms on $Z$ with the following definition
\begin{definition}\label{def:integration_upon_Cauchy_of_forms}
	Let $k\in\NN_0:=\NN\cup\{0\}$ and $\alpha\in\Omega^{k+n}Z$ such that $\exists
	\beta\in\Omega^{k+n-1}$ with $i_B^*\alpha=\dd\beta$, we define
	the $k-$form $\tilde{\alpha}$ on $\tZ$ as
	\begin{equation}
		\label{eq:tilde_forms}
		\iota_{\tilde{\zeta_1}}\dots\iota_{\tilde{\zeta}_k}\tilde
		{\alpha}(\gamma):=\int_M\gamma^*(\iota_{\zeta_1}\dots\iota_
		{\zeta_k}\alpha)-(-1)^k\int_{\partial M}\gamma^*(\iota_{\zeta_1}
		\dots\iota_{\zeta_k}\beta)
	\end{equation}
	for all $\gamma\in\tZ$ and all $\tilde{\zeta}_1,\dots\tilde
	{\zeta}_k\in T_\gamma\tZ$, where $\zeta_1\dots\zeta_k$ are
	exactly equal to $\tilde{\zeta}_1\dots\tilde{\zeta}_k$, the
	different notations are just for stressing they role as an
	element of $T_\gamma\tZ$ or $\Gamma(\gamma^*TZ)$. 
\end{definition}
Hence, in particular, the Poincaré-Cartan $(n+1)$ and $(n+2)$ forms
$\Theta_L$ and $\Omega_L$ on $Z$ induce on $\tZ$ the canonical
$1$-form and $2$-form $\widetilde{\Theta_L}$ and $\widetilde{\Omega_L}$ defined
by
\begin{align*}
	(\widetilde{\Theta_L})_\gamma(\tilde{\xi}):&=\int_M\gamma^*(\iota
	_\xi\Theta_L)+\int_{\partial M}\gamma^*(\iota_\xi\Pi),\\
(\widetilde{\Omega_L})_\gamma(\tilde{\xi}_1,\tilde{\xi}_2):&=\int_M
\gamma^*(\iota_{\xi_2}\iota_{\xi_1}\Omega_L).
\end{align*}

Upon \eqref{eq:identification0} we can keep extending in a canonical way
geometric tools from the multisymplectic to the instantaneous (`tilde')
formalism like the concept of integral curves, pushforward and pullback
maps or the exterior and Lie derivatives. The following is a concatenation
of results which proofs are straightforward and can be checked on
\cite{leonSymmetriesClassicalField2004}
\begin{proposition}\label{prop:integral_curves_on_tildeZ}
	Let $\xi\in\mathfrak{X}(Z)$ be fixed complete vector field and let
	$u\in M$ and $\gamma\in\tZ$ be fixed but arbitrary. If $c^u$ is an
	integral curve of $\xi$ through $\gamma$ i.e. $c^u(0)=\gamma(u)$
	and $\dot{c}^u(0)=\xi(\gamma(u))$	then $\tilde{c}$ is an integral
	curve of $\tilde{\xi}$ through $\gamma$, i.e. $\tilde{c}(0)=\gamma$
	and $\dot{\tilde{c}}(0)=\tilde{\xi}(\gamma)$.
\end{proposition}
In particular, given a diffeomorphism $F:Z\to Z$ and defining $\tilde{F}
:\tZ\to\tZ$ as $\tilde{F}(\gamma)(u)=F(\gamma(u))$ then
\begin{corollary}
	For $F:Z\to Z$ a given diffeomorphism, then $T\tilde{F}(\tilde{\xi})=
	\widetilde{TF(\xi)}$ $\forall\xi\in\mathfrak{X}(Z)$.
\end{corollary}
\begin{proposition}
	Let $F$ be as above, then for a fixed $\alpha\in\Omega^{n+k}Z$ such that
	$\exists\beta\in\Omega^{n+k-1}Z$ with $\iota_B^*\alpha=\dd\beta$, we have
	that $\tilde{F}^*\tilde{\alpha}=\widetilde{F^*\alpha}$ where $\widetilde
	{F^*\alpha}$ is given by Definition
	\ref{def:integration_upon_Cauchy_of_forms}.
\end{proposition}
\begin{proposition}
	Let $\xi\in\mathfrak{X}(Z)$ be a fixed vector field, then
	\[
		\pounds_{\tilde{\xi}}\tilde{\alpha}=\widetilde{\pounds_\xi\alpha}
	\]
	$\forall\alpha\in\Omega^{n+k}Z$, with $k\in\NN_0$ such that
	$\exists\beta\in\Omega^{n+k-1}Z$ where $\iota_B^*\alpha=\dd\beta$.
\end{proposition}
Lastly, and fundamentally for this instantaneous formalism, we have the
following result
\begin{proposition}\label{prop:presymplectic_proposition}
	Let $\alpha$ be a $n-$ or $(n+1)-$form, then $\widetilde{\dd\alpha}=\dd
	\tilde{\alpha}$. 
	In particular
	\[
		\widetilde{\Omega_L}=-\dd\widetilde{\Theta_L},
	\]
	and hence $(\tZ,\widetilde{\Omega_L})$ is a presymplectic manifold.
\end{proposition}

\section{Instantaneous De Donder equations}
As we introduced in the beginning of this chapter, one of the goals
of this instantaneous formalism is to materialize a canonical arrow
of time, i.e., a dynamical picture of the field equations. In
particular, we will end this section with a fundamental theorem that
proves the existence of an equivalence between multisymplectic
De Donder equations formulation and a sort of presymplectic equation
for $\widetilde{\Omega_L}$ that specifically exploits this fact
discovered on Proposition \ref{prop:presymplectic_proposition}.

So first, we start defining a foliation of our spacetime $X$.
\begin{definition}[Slicing]
	\label{def:slicing}
	Let $c_{\tX}:I\subseteq\RR\to\tX$ be a curve on $\tX$. We call
	$c_{\tX}$ a \textbf{slicing} or that it \textbf{splits} $X$ if
	the canonical map
	\begin{align*}
		\Phi:I\times M&\to X\\
		(t,u)&\mapsto c_{\tX}(t)(u)	
	\end{align*}
	is a diffeomorphism, hence $X\cong(I\subseteq\RR)\times M$. 
\end{definition}
For a fixed $t\in I\subseteq$, $\Phi(t,\cdot)\in\tX$, clearly foliating
$X$ as time flows. With this slicing property, as $\Phi:I\times M\to X$
becomes a diffeomorphism, we can canonically identify a time-space
decomposition in our spacetime. Indeed, fixing $(t,u^1,\dots,u^n)$
local coordinates on $I\times M$ and the considering the \textbf{fixed}
coordinates on $X$ given by
\begin{align*}
	x^0&:=t\circ\Phi^{-1},\\
	x^n&:=u^n\circ\Phi^{-1},
\end{align*}
if $\partial_t$ generates the tangent space to $I\subseteq\RR$
(time direction, essentially) and $\partial_{u^\mu}$ the tangent
space to $M$ (space direction, essentially), the
time-space canonical splitting is materialized as
\[
	\Phi_{*,(t,u)}\left(\frac{\partial}{\partial t}\bigg|_{(t,u)}\right)
	=\frac{\partial}{\partial x^0}\bigg|_{\Phi(t,u)},\quad
	\Phi_{*,(t,u)}\left(\frac{\partial}{\partial u^\mu}\bigg|_{(t,u)}
	\right)=\frac{\partial}{\partial x^\mu}\bigg|_{\Phi(t,u)}.
\]
Remarking first that the identification \eqref{eq:identification0} could
also have been done on the $\tX$ context, we have the following
definition
\begin{definition}[Infinitesimal slicing]
	Given a point $\tau\in\tX$, we call a tangent vector $v\in T_\tau\tX$
	an \textbf{infinitesimal slicing} provided that $v(u)\in
	\Gamma(\tau^*TX)$ is transverse to $\textrm{Im}\,\tau$, i.e.
	\[
		v(u)\notin T_{\tau(u)}\tau(M),\quad\forall u\in M.
	\]
\end{definition}

With this two previous definitions and considering the natural
\textbf{projection of curves} in $\tZ$ to curves in $\tX$ by $\pi_{ZX}$
as $c_{\tX}=\pi_{ZX}\circ c_{\tZ}$, we note the following relation
between curves on $\tZ$ and sections of $\pi_{ZX}$ given by the following
equations
\begin{equation}\label{eq:section-curve-identification_rhs}
	\sigma(c_{\tX}(t)(u))=:c_{\tZ}(t)(u),
\end{equation}
\begin{equation}\label{eq:section-curve-identification_lhs}
	\sigma(c_{\tX}(t)(u)):=c_{\tZ}(t)(u).
\end{equation}
Indeed, given a section $\sigma\in\Gamma(\pi_{ZX})$ and a curve
$c_{\tX}$ on $\tX$, we define a curve $c_{\tZ}$ on $\tZ$ using Equation
\eqref{eq:section-curve-identification_rhs}, while, conversely, given a curve
$c_{\tZ}$ on $\tZ$ such that its projection $c_{\tX}$ splits $X$, then $\sigma$
defined by Equation \eqref{eq:section-curve-identification_lhs} is a \textbf{local}
section of $\pi_{ZX}$. We will call this the \textbf{section/curve identification}.

And finally, with this identification, we state the desired theorem
\begin{theorem}
	\label{thm:instantaneous_de_donder}
	Given $\sigma\in\Gamma(\pi_{ZX})$ satisfying De Donder equations
	\eqref{eq:De-Donder-equations} then its curve identification $c_{\tZ}$
	(note that $\dot{c}_{\tZ}\in\mathfrak{X}(\tZ)$) satisfies
	\begin{equation}
		\label{eq:presymplectic_equation}
		\iota_{\dot{c}_{\tZ}}\widetilde{\Omega_L}=0.
	\end{equation}
	In particular, given a connection $\Gamma$ on $\pi_{ZX}$ satisfying
	the geometric De Donder equation \eqref{eq:geometric-De-Donder} for
	its horizontal projector $\textbf{h}$, if we have an horizontal local integral
	section $\sigma$ of $\Gamma$ then the solution to
	\eqref{eq:presymplectic_equation} is given by the horizontal lift
	through $\textbf{h}$ (i.e. through $\sigma$) of an infinitesimal slicing.

	Conversely given a curve $c_{\tZ}$ which projection splits $X$, its
	(local) section identification verifies the De Donder equations.
\end{theorem}
\begin{proof}
	See \cite{binzGeometryClassicalFields1988}.	
\end{proof}
Note that for a singular Lagrangian, the connection lifting of an infinitesimal
slicing is not guaranteed (as the connection is not guaranteed to even exist
globally), hence, the need of a submanifold restriction algorithm just like
the one given in the Subsection \ref{subsect:the_singular_case} is clear.
The general procedure is similar to the restriction algorithms already presented in
Subsections \ref{subsect:the_singular_case} and
\ref{subsect:almost_regular_lagrangians_and_the_constraint_algorithm}; we restrict
to submanifolds where a tangent vector satisfying the De Donder equations exists and
such that vector projects onto an infinitesimal slicing (cf.
\cite{leonSymmetriesClassicalField2004} for the details).

\section{Instantaneous presymplectic bracket}
In finite-dimensional symplectic geometry $(M,\omega)$, by definition of the symplectic form,
its non-degeneracy condition implies that the canonical \textbf{flat} map
\begin{align*}
	\flat_{\omega}\colon\Gamma(TM)&\to\Gamma(T^*M)\\
	\xi&\mapsto \iota_{\xi}\omega
\end{align*}
has trivial kernel, and as the domain and codomain are fiber-wise vector spaces of the same
dimension, the flat map has an inverse (\textbf{sharp}) map by the rank-nullity
theorem given by
\begin{align*}
	\sharp_{\omega}\colon\Gamma(T^*M)&\to\Gamma(TM)\\
	\alpha&\mapsto \sharp_{\omega}(\alpha)
	:=(\flat_{\omega})^{-1}(\alpha).
\end{align*}
In the infinite dimensional cases (like $(\tZ,\widetilde{\Omega_L})$), the rank-nullity theorem
is not applicable, and hence tangent and cotangent bundles are not guaranteed to be isomorphic.
If $(M,\omega)$ in the infinite dimensional case is a symplectic manifold, injectivity of
the flat map is met by the non-degeneracy condition; in this case we differentiate between
\textbf{strong symplectic forms} when $\omega$ induces an isomorphism between the tangent and
the cotangent bundles and \textbf{weak symplectic forms} when $\omega$ induces a flat map
that is only injective.

From Proposition \ref{prop:presymplectic_proposition}, we know that
$(\tZ,\widetilde{\Omega_L})$ is only guaranteed to be a (infinite-dimensional) presymplectic
manifold, and not a symplectic one, hence, even injectivity of the flat map is not assured.

The existence of the sharp map, in the (finite-dimensional) symplectic context $(M,\omega)$,
implies the existence of a unique Hamiltonian vector field (cf. the following definition)
for every function $\tilde{f}\in C^\infty(M)$. In the presymplectic
context of $(\tZ,\widetilde{\Omega_L})$, neither the existence nor the uniqueness of this
Hamiltonian vector field is guaranteed (uniqueness only up to an element of the kernel of
the $\flat$ map); In any case, the following
definition allow us to introduce a well defined presymplectic bracket related
to the multisymplectic one already defined in Subsection
\ref{subsect:hamiltonian_forms_and_multisymplectic_brackets}. 
\begin{definition}[Hamiltonian function and vector field]
	Let $\tilde{f}\in C^\infty(\tZ)$ such that $\exists\tilde{\xi}\in\mathfrak{X}
	(\tZ)$ with $\iota_{\tilde{\xi}}\widetilde{\Omega_L}=\dd f$, then $f$ is
	called a \textbf{Hamiltonian function} while $\tilde{\xi}$ is the
	\textbf{Hamiltonian vector field} associated to $\tilde{f}$.
\end{definition}
Recalling Definition \ref{def:hamiltonian_forms_and_vector_fields} of Hamiltonian
$n-$forms and vector fields in the multisymplectic setting, we have a relation with
the above defined counterpart for the instantaneous formalism given by the
following definition
\begin{proposition}
	\label{prop:relation_Hamiltonian_forms_and_functions}
	Let $\alpha$ be a Hamiltonian $n$-form for $\Omega_L$ which is exact on
	$\partial Z$ ($B\subseteq\partial Z$) and let $X_\alpha$ be its corresponding Hamiltonian vector
	field, then the $0-$form $\tilde{\alpha}$ (i.e. function) given by Definition
	\ref{def:integration_upon_Cauchy_of_forms} and the induced vector field
	$\tilde{X}_{\tilde{\alpha}}$ by $X_\alpha$ are a Hamiltonian function and
	vector field for $\widetilde{\Omega_L}$ respectively.
\end{proposition}
\begin{proof}
	See Appendix
	\ref{app:proof_of_proposition_ref_prop_relation_hamiltonian_forms_and_functions}.
\end{proof}

As we said in the beginning of this section, given $\tilde{f}\in C^\infty(\tZ)$, (if
it exists) its associated Hamiltonian vector field $\tilde{X}_{\tilde{f}}$ is a 
`representative' defined only up to an element of
$\textrm{ker}(\flat_{\widetilde{\Omega_L}})$.
Due to $\widetilde{\Omega_L}$ being a form (hence multilinear) we can use it
canonically to define a presymplectic bracket exactly the same way as it is done
in the symplectic context even though of this lack of uniqueness (which could
cause the bracket of two Hamiltonian functions to not being uniquely defined).
\begin{definition}[Presymplectic bracket]
	Let $(\tilde{f},\tilde{X}_{\tilde{f}})$ and $(\tilde{g},\tilde{X}_{\tilde{f}})$
	be a pair of Hamiltonian functions and vector fields for $\widetilde{\Omega_L}$,
	then we define the \textbf{presymplectic bracket} as
	\[
		\{\tilde{f},\tilde{g}\}:=\iota_{\tilde{X}_{\tilde{f}}}\iota_{\tilde{X}_{\tilde{g}}}\widetilde{\Omega_L}=
		\widetilde{\Omega_L}(\tilde{X}_{\tilde{g}},\tilde{X}_{\tilde{f}}).
	\]
\end{definition}
Multisymplectic and presymplectic brackets (as expected) are connected through
Cauchy surface integration of Definition \ref{def:integration_upon_Cauchy_of_forms}
in a consistent way
\begin{proposition}
	\label{prop:brackets_commute_with_tilde}
	Let $\alpha_1$ and $\alpha_2$ be two Hamiltonian $n-forms$ that are exact on $\partial Z\supseteq B$ and which Hamiltonian
	vector fields are tangent to $B$, then 
	\[
		\widetilde{\{\alpha_1,\alpha_2\}}=
		\{\tilde{\alpha}_1,\tilde{\alpha}_2\}.
	\]
\end{proposition}
\begin{proof}
	As $i_B^*\Theta_L=\dd\Pi$ by boundary conditions, then $i_B^*\Omega_L=0$. As
	$X_1$ and $X_2$ are tangent to $B$, we have that their Lie bracket is also
	tangent to $B$ and hence
	\[
		i^*_B\{\alpha_1,\alpha_2\}
		\overset{\ref{def:multisymplectic_bracket}}{=}i_B^*\iota_{X_1}\iota_{X_2}
		\Omega_L=\iota_{X_1|_B}\iota_{X_2|_B}i_B^*\Omega_L=\iota_{X_1|_B}(0|_B)=0
	\]
		where in the last equality we used precisely this last fact mentioned. Once noted
	this, the corresponding $\beta$ for $\{\alpha_1,\alpha_2\}$ if not exactly zero has to be
	closed, and then by Stokes theorem is easy to see $\forall\gamma\in\tZ$ that
	\[
		\widetilde{\{\alpha_1,\alpha_2\}}|_\gamma=\int_M\gamma^*\{\alpha_1,\alpha_2\}
		-\cancel{\int_{M}\gamma^* \dd\beta}=\int_M\gamma^*\iota_{X_1}\iota_{X_2}
		\Omega_L=\widetilde{\Omega_L}(\tilde{X}_2,\tilde{X}_1)|_\gamma=\{\tilde{\alpha}_1,
		\tilde{\alpha}_2\}|_\gamma.
	\]
\end{proof}

\chapter{Symmetries and Noether's theorems} \label{chap:symmetries_and_Noether_theorems}
Noether's theorems have changed the way theoretical physicists face
building new theories. They state a relation between \textbf{symmetries} of a
theory and conservation laws. These relations allow to compute invariant
(conserved quantities) from symmetries or restrict the theory (the Lagrangian)
given a set of symmetries known to be conserved. 

We have two Noether's theorems \cite{noether_invariant_1971}; the first
one deals with global symmetries and conserved quantities
(or more generally conserved currents, one per each symmetry generator).
The second one deals instead with localizable (local and independent)
symmetries and relations between the equations of motion.

Here we state the formulation of Noether's theorems for the multisymplectic
and the instantaneous frameworks. We shall suppose that the Lagrangian is
regular unless stated otherwise.
\section{Symmetries for the Euler-Lagrange equations}
\label{sec:symmetries_for_euler_lagrange_equations}
In a nutshell, we say that a theory is invariant under a (infinitesimal)
symmetry if the symmetry leaves invariant the action, i.e. the (linear)
variation of the action induced by the symmetry is identically zero 
$\delta S:=\der{S(\phi_t)}{t}\big|_{t=0}\overset{!}{=}0$. 
As the reader could remember, we had some
equivalent expressions for $\delta S$ already computed in the Subsection
\ref{subsect:euler_lagrange_and_de_donder_equations}. Let $C$ be \textbf{any}
compact submanifold of $X$, we then have
\begin{align*}
	\delta S:&=\der{S(\phi_t)}{t}\bigg|_{t=0}=\int_C(j^1\phi)^*
	\pounds_{\xi_Y^{(1)}}(\mathcal{L})=\int_C(j^1\phi)^*
	\pounds_{\xi_Y^{(1)}}(\Theta_L)\\
	&=\int_C(j^2\phi)^*\left[\frac{\partial L}{\partial y^i}-\der{}{x^\mu}
	\frac{\partial L}{\partial z^i_\mu}\right](\xi_Y^i-z^i_\nu
	\xi_Y^\nu)\eta+
	\int_{\partial C}(j^1\phi)^*(\iota_{\xi_Y^{(1)}}\Theta_L),
\end{align*}
(\textbf{valid for any} $\bm{\Theta}$ \textbf{Lepage-equivalent to} $\bm{\mathcal{L}}$
but we will stick
to Poincaré-Cartan $(n+1)-$forms specifically for the rest of the chapter).
From the last equality of the first row, we identify $\pounds_{\xi_Y^{(1)}}
\Theta_L\in\mathcal{I(C)}$ as an equivalent definition for $\xi_Y$ being a
symmetry (note that this implies the on-shell variation of the Lagrangian being
exactly zero), while from the second row expression (variation is equal to
E-L plus a boundary term),
we note that if $\phi$ is a solution of the Euler-Lagrange equations and
$\xi_Y$ is a symmetry (i.e. $\delta S=0$), then $\int_{\partial C}
(j^1\phi)^*(\iota_{\xi_Y^{(1)}}\Theta)=0$ or equivalently, by Stokes' theorem
\begin{equation}
	\label{eq:preservation_already_proven}
	\int_{C}(j^1\phi)^*(\dd\iota_{\xi_Y^{(1)}}\Theta)=0,\quad\forall\phi\in\Gamma(\pi)
	\text{ solution to E-L's equations,}
\end{equation}
which then implies (remember that $C$ is any compact submanifold of $X$) that
$\iota_{\xi_Y^{(1)}}\Theta_L$ is conserved. This motivates the following definition for
\textbf{conserved quantities} (also called \textbf{conserved currents})
\begin{definition}[Conserved quantity for E-L's equations and associated momentum]
	\label{def:conserved_quantity_for_EL_equations}
	A $n-$form on $Z$ is called a \textbf{conserved quantity for the
	Euler-Lagrange equations} if $(j^1\phi)^*\dd\alpha=0$ for every solution
	$\phi\in\Gamma(\pi)$ of the E-L's equations
	\eqref{eq:euler-lagrange_equations}; $\tilde{\alpha}$ defined by
	\eqref{eq:tilde_forms} is called its associated \textbf{momentum}.
\end{definition}
\begin{remark}
	\label{remark:preserved_quantities_plus_closed_or_contact_forms}
	Noting how we went from $\int_{\partial C}(j^1\phi)^*(\iota_{\xi_Y^{(1)}}
	\Theta)=0$ to the conserved quantity through Stokes' theorem, is easy to
	realise that \textbf{conserved quantities are preserved under the addition of
	closed forms or forms belonging to} $\bm{\mathcal{I(C)}}$.
\end{remark}

We can not forget about boundary conditions in the case that $X$ has a boundary;
here we have two additional requirements.
As we already defined, boundary conditions are encoded in a subbundle
$B\subseteq\partial Z$ such that we define admissible solutions of E-L's equations
as solutions $\sigma\in\Gamma(\pi)$ satisfying $(j^1\phi)(\partial X)\subseteq B$.
In the case of \textbf{projectable symmetries} (which will \textbf{not} be
demanded to our soon-to-be defined symmetries but makes the motivation easier)
the flow of admissible symmetries should send admissible solutions to
admissible solutions while conserved quantities should also be well behaved in
$B\subseteq\partial B$ in a sense; as we will see 
(in Proposition \ref{prop:conservation_along_solutions})
this last bit is equivalent to conserved quantities being exact (and not only
closed) on $B$.

The \textbf{first extra condition} is stated through the requirement that $(j^1\phi_t)
(\partial X)\subseteq B$ for all $t$ in the time domain of definition of the
flow of the symmetry which is equivalent to $\bm{\xi_{Y}^{(1)}}$
\textbf{being tangent to} $\bm{B}$ (see Appendix
\ref{app:some_symmetries_computations} for a proof).

The \textbf{second extra condition}, in the case of the conserved quantity being
$\iota_{\xi_Y^{(1)}}\Theta_L$, i.e. coming from the Lagrangian,
is met whenever $\bm{\pounds_{\xi_Y^{(1)}|_B}\Pi=0}$ (see Appendix
\ref{app:some_symmetries_computations} for a proof).

Once all this motivation is done, we are ready to define our symmetries.
There will come in three types, \textbf{symmetries of the Lagrangian} (which in
theoretical physics are usually called exact symmetries -we already gave some
hints on why its called that way), \textbf{Noether symmetries} (in theoretical physics'
perspective are symmetries in which the Lagrangian's variation is not
exactly zero but a divergence term) and \textbf{Cartan symmetries} (which will rigorously
motivate the statement of symmetries that flow E-L's solutions to E-L's solutions).
These symmetries form a hierarchy of increasing generality: every Lagrangian symmetry is a 
Noether symmetry, and every Noether symmetry induces a Cartan symmetry via $1$-jet 
prolongation; also, each one of them will have an associated Noether's current, giving us
a conserved quantity (also called current).
\subsection*{Symmetries of the Lagrangian}
\begin{definition}[Symmetries of the Lagrangian]
\label{def:symmetries_of_the_lagrangian}
A vector field $\xi_Y\in\mathfrak{X}(Y)$ is called a \textbf{symmetry of
the Lagrangian} or \textbf{variational symmetry} if the following 
conditions hold: 
\textup{(i)} $\pounds_{\xi_Y^{(1)}}\Theta_L\in\mathcal{I(C)}$, \quad
\textup{(ii)} $\xi_Y^{(1)}$ is tangent to $B$ \quad and \quad
\textup{(iii)} $\pounds_{\xi_Y^{(1)}|_B}\Pi=0$.
\end{definition}
Recalling Proposition \ref{prop:symmetries_will_form_subalgebra}, the next proposition is
obvious
\begin{proposition}
	Lagrangian symmetries form a Lie subalgebra of $\mathfrak{X}(Y).$
\end{proposition}

In order to actually make computations, we have to introduce a function $F$ on $Z$
\begin{definition}[F function]
	We define the $\bm{F}$ \textbf{function} as the function $F\in C^\infty(Z)$ such that
	\begin{equation}
	\label{eq:F_function}
		\pounds_{\xi_Y^{(1)}}\mathcal{L}-F\eta\in\mathcal{I(C)},
	\end{equation}
	with local expression
	\begin{equation}
	\label{eq:F_local_expression}
		F=\xi_Y^{(1)}(L)+\left(\frac{\partial\xi_Y^\mu}{\partial x^\mu}+
		z^i_\nu\frac{\partial\xi_Y^\nu}{\partial y^i}\right)L.
	\end{equation}
\end{definition}
A long computation then gives a local expression for $\pounds_{\xi_Y^{(1)}}\Theta_L$
(cf. \cite{leonSymmetriesClassicalField2004})

\begin{equation}
\begin{aligned}
	\label{eq:lengthy_computation}
	\pounds_{\xi^{(1)}_{Y}}\Theta_L
	&=
	F\eta
	+
	\frac{\partial F}{\partial z^i_\mu}\,\theta^i\wedge \dd^n x_\mu
	-
	z^j_\nu\left(
	\frac{\partial \xi^\nu_Y}{\partial y^j}\frac{\partial L}{\partial z^i_\mu}
	-
	\frac{\partial \xi^\mu_Y}{\partial y^j}\frac{\partial L}{\partial z^i_\nu}
	\right)\theta^i\wedge \dd^n x_\mu \\
	&-
	\frac{\partial \xi^\nu_Y}{\partial y^j}\frac{\partial L}{\partial z^i_\mu}\,
	\theta^i\wedge dy^j\wedge \dd^{\,n-1}x_{\mu\nu}.
\end{aligned}
\end{equation}
Considering \ref{eq:lengthy_computation}, the following result is obvious once noted that
the terms containing $\theta^i$ belong to $\mathcal{I(C)}$.
\begin{proposition}
	\label{prop:Fzero_iff_symmetry}
	If $\xi_Y\in\mathfrak{X}(Y)$ is a symmetry of the Lagrangian, then the function $F$
	defined by \eqref{eq:F_function} is equal to zero. The converse is also true in the case that
	$\xi_Y$ satisfy the boundary conditions $ii)$ and $iii)$ of Definition
	\ref{def:symmetries_of_the_lagrangian}.
\end{proposition}
Also, \ref{eq:lengthy_computation} gives us a nice computational checking tool equation to
see whenever a given $\textbf{projectable}$ vector field $\xi_Y$ is a symmetry (note
that by definition a vector field $\xi_Y$ is projectable to $X$ if the $\xi_Y^\mu$ terms
do not depend on $y$)
\begin{proposition}
	\label{prop:useful_proposition_for_checking_symmetries}
	Given a symmetry of the Lagrangian $\xi_Y$ (hence $F=1$), if $\dim X=0$ (i.e. $n=0$) or
	$\xi_Y$ is projectable, then $\pounds_{\xi^{(1)}_{Y}}\Theta_L=0$ and Equation
	\eqref{eq:F_local_expression} gives
	\begin{equation}
	\label{eq:useful_tool_for_checking_symmetries}	
		\xi_Y^{(1)}(L)=-\sum_\mu\der{\xi^\mu_Y}{x^\mu}L
	\end{equation}
\end{proposition}
\begin{proof}
See Appendix \ref{app:some_symmetries_computations}.
\end{proof}
Finally
\begin{theorem}[Noether's theorem for symmetries of the Lagrangian]
	\label{thm:noether_theorem_symmetries_of_the_lagrangian}
	If $\xi_Y$ is a symmetry of the Lagrangian then $\iota_{\xi_Y^{(1)}}\Theta_L$ is a
	conserved quantity which is exact on $B$.	
\end{theorem}
\begin{proof}
	The exactness part is already proved in Appendix \ref{app:some_symmetries_computations},
	while the preservation part was already proven in the symmetries' motivation part
	(specifically on Equation \eqref{eq:preservation_already_proven}).
	Which is not proven yet (cf. \cite{leonSymmetriesClassicalField2004})
	is the useful local expression for the conserved quantity
	\begin{equation}
		\label{eq:useful_local_expression_for_the_preserved_quantity}
		\iota_{\xi_Y^{(1)}}\Theta_L=\left(\left[L-z^i_\mu\frac{\partial L}{\partial z^i_\mu}
		\right]\xi_Y^\nu+\frac{\partial L}{\partial z^i_\nu}\xi_Y^i\right)\dd^nx_\nu
		\,-\,\frac{\partial L}{\partial z^i_\mu}\xi_Y^\nu dy^i\wedge
		\dd^{n-1}x_{\mu\nu}.
	\end{equation}
\end{proof}
\subsection*{Noether's symmetries}
As mentioned before, this type of symmetries correspond to symmetries in which the Lagrangian
varies as a divergence term:
\begin{definition}[Noether's symmetries]
A vector field $\xi_Y\in\mathfrak{X}(Y)$ is called a \textbf{Noether's symmetry} or
\textbf{divergence symmetry} if the following 
conditions hold:
\begin{enumerate}
 	\item There is a form $\alpha\in\Omega^n(Y)$ such that its pullback to $Z$ is exact on $B$, i.e., 
 	$\iota^*_B\alpha=\dd\beta$, for some $\beta\in\Omega^{n-1}(B)$
 	(where $\pi_{ZY}^*\alpha\equiv\alpha$),
 	\item $\pounds_{\xi_Y^{(1)}}\Theta_L-\dd\alpha\in\mathcal{I(C)}$,
	\item $\xi_Y^{(1)}$ is tangent to $B$,
	\item $\pounds_{\xi_Y^{(1)}|_B}\Pi=0$.
 \end{enumerate} 
\end{definition}
Clearly, symmetries of the Lagrangian are also Noether's symmetries, and as for the
symmetries of the Lagrangian, we have the following result for Noether's symmetries
\begin{proposition}
	The set of Noether's symmetries form a subalgebra of $\mathfrak{X}(Y)$ containing the
	subalgebra of symmetries of the Lagrangian.
\end{proposition}
As for the symmetries of the Lagrangian, we have some useful local results. First, using
that $dy^i=\theta^i+z^i_\mu dx^\mu$ by definition of the contact forms, every $n-$form
on $Y$ including $\alpha$ can be written locally as
\[
	\alpha=\alpha^\mu \dd^nx_\mu+\theta,\quad\text{for }\theta\in\mathcal{I(C)}.
\]
From this expression, and using once again $dy^i=\theta^i+z^i_\mu dx^\mu$ we have
that
\[
	\dd\alpha\equiv\sum_\mu\left(\frac{\partial\alpha^\mu}{\partial x^\mu}+ z^i_\mu
	\frac{\partial\alpha^\mu}{\partial y^i}\right)\eta\quad\mod\mathcal{I(C)}
\] 
and as from the previous section we know that $\pounds_{\xi_Y}\Theta_L\equiv F\eta$
$\mod\mathcal{I(C)}$ (where $F$ is defined upon $\xi_Y$ using \eqref{eq:F_function},
if we define
\begin{equation}
	\label{eq:F_tilde}
	\tilde{F}:=F-\sum_\mu\left(\frac{\partial\alpha^\mu}{\partial x^\mu}+ z^i_\mu
	\frac{\partial\alpha^\mu}{\partial y^i}\right)
\end{equation}
the following result is trivial
\begin{proposition}
	If $\xi_Y\in\mathfrak{X}(Y)$ is a Noether's symmetry, then $\tilde{F}$ defined upon
	$\xi_Y$ and $\alpha$ by \eqref{eq:F_tilde} is equal to zero. 
\end{proposition}
We also have
\begin{proposition}[Noether symmetries]
Let $\xi_Y$ be a Noether symmetry. If $\dim X=1$ ($n=0$), then
\begin{equation}
\label{eq:noether-symmetry}
\mathcal{L}_{\xi_Y^{(1)}}\Theta_L=\dd\alpha,
\end{equation}
while in general, if $\xi_Y$ is projectable, the equation above is also satisfied and
furthermore
\[
\xi_Y^{(1)}(L)=\sum_{\mu}\left(\der{\alpha^\mu}{x^\mu}-\der{\xi^\mu_Y}{x^\mu}L
\right).
\]
\end{proposition}
\begin{proof}
	See \cite{leonSymmetriesClassicalField2004}.
\end{proof}
Now we arrive at the Noether's theorem of this Noether's symmetries. As admissible
solutions of E-L's equations satisfy $(j^1\phi)(\partial X)\subseteq B$, then for every
compact submanifold $C$ of $X$, we have also that $(j^1\phi)(\partial X)\subseteq B$ and
so Remark \ref{remark:preserved_quantities_plus_closed_or_contact_forms} ends up motivating
the condition of $\alpha$ being exact (hence closed) on $B$ as
\begin{theorem}[Noether's theorem for Noether's symmetries]
	If $\xi_Y$ is a Noether's symmetry then the quantity $\iota_{\xi_Y^{(1)}}\Theta_L-\alpha$
	is a conserved quantity which is exact on the boundary.
\end{theorem}
\begin{proof}
	See \cite{leonSymmetriesClassicalField2004}.
\end{proof}
\subsection*{Cartan symmetries}
As already mentioned in the motivation, Cartan symmetries are those most general kind of
symmetries that at least flow solutions of E-L's equations into each other. These live
on $Z$, and do not even need to be projectable into vector fields on $Y$ or $X$. Once
introduced the other two kinds of symmetries, the only different requirement is that
its flow must preserve $\mathcal{I(C)}$
\begin{definition}[Cartan's symmetries]
	A vector field $\xi_Z\in\mathfrak{X}(Z)$ is called a \textbf{Cartan's symmetry} whenever
	\begin{enumerate}
	\item Its flow preserves $\mathcal{I(C)}$ (i.e. $\pounds_{\xi_Z}\subseteq\mathcal{I(C)})$,
 	\item There is a form $\alpha\in\Omega^n(Z)$ which is exact on $B$, i.e., 
 	$\iota^*_B\alpha=\dd\beta$, for some $\beta\in\Omega^{n-1}(B)$,
 	\item $\pounds_{\xi_Z}\Theta_L-\dd\alpha\in\mathcal{I(C)}$,
	\item $\xi_Z$ is tangent to $B$,
	\item $\pounds_{\xi_Z|_B}\Pi=0$.
 \end{enumerate} 
\end{definition}
The $1-$jet prolongation of Noether's symmetries are Cartan symmetries, while the projection
(if it exists) of Cartan symmetries are Noether's symmetries. We also have naturally
\begin{proposition}
	Cartan symmetries form a Lie subalgebra of $\mathfrak{X}(Z)$.
\end{proposition}
As we motivated the following proposition is true for every Cartan symmetry
\begin{proposition}
	The flow of Cartan symmetries map solutions of the E-L's equations to solutions
	of the E-L's equations.
\end{proposition}
This proposition is fundamental to see that Cartan symmetries are the less restrictive
symmetries that are really symmetries in the sense of transforming solutions
of E-L's equations into solutions of E-L's equations.
\begin{proof}
	See \cite{leonSymmetriesClassicalField2004}.
\end{proof}
Lastly
\begin{theorem}[Noether's theorem for Cartan's symmetries]
	If $\xi_Z$ is a Cartan symmetry, then the $\iota_{\xi_Z}\Theta_L-\alpha$ is a conserved
	quantity which is exact on $B$.
\end{theorem}
\section{Symmetries for De Donder equations}

Up to now we used constantly the fact that if $\phi$ is a solution to E-L equations then
$(j^1\phi)^*\theta=0$, for all $\theta\in\mathcal{I(C)}$, but in the context of De Donder
equations, solutions are sections $\sigma\in\Gamma(\pi_{ZX})$ which of course in general
do not annihilate contact forms through their pullback (as $\sigma^i_\mu\neq\frac{\partial
\sigma^i}{\partial x^\mu}$ in general). This of course reflects on symmetries for 
the De Donder equations being more restrictive that the ones of E-L's equations, but, luckily
for us, the generalizations of the three kind of symmetries and conserved quantities
are straightforward 
\begin{definition}[Conserved quantities for the De Donder equations]
	A $n-$form on $Z$ is called a \textbf{conserved quantity for the
	De Donder equations} if $\sigma^*\dd\alpha=0$ for every solution
	$\sigma\in\Gamma(\pi_{ZX})$ of the De Donder equations
	\eqref{eq:De-Donder-equations}; $\tilde{\alpha}$ defined by
	\eqref{eq:tilde_forms} is called its associated \textbf{momentum}.
\end{definition}
\begin{remark}
	\label{remark:preserved_quantities_plus_closed}
	Note that \textbf{conserved quantities for the De Donder equations remain conserved
	under the addition of closed forms} (but not under contact forms of course).
\end{remark}
Given an integrable connection solution of the geometric De Donder equation
\eqref{eq:geometric-De-Donder} then we have another way of characterizing 
conserved quantities
\begin{proposition}[\cite{leonSymmetriesClassicalField2004}]
	Let \textbf{h} be an integrable connection solution of geometric De Donder equation
	\eqref{eq:geometric-De-Donder}, then $\alpha$ is a conserved quantity for the De Donder
	equations if and only if $\dd\alpha$ is annihilated by any $n+1$ horizontal tangent vectors
	at each point. 
\end{proposition}

The three kind of symmetries are defined as follow
\begin{definition}[Symmetries for the De Donder equations]\label{def:5.6}
We have the following definitions of symmetries for the De Donder equations
\begin{enumerate}
	\item $\xi_Y\in\mathfrak{X}(Y)$ is said to be a \textbf{symmetry of the Lagrangian},
	or a \textbf{variational symmetry} if $\pounds_{\xi_Y^{(1)}}\Theta_L = 0$,
	$\xi_Y^{(1)}$ is tangent to $B$ and verifies $\pounds_{\xi_Y^{(1)}|_B}\Pi=0$.
	\item $\xi_Y\in\mathfrak{X}(Y)$ is said to be a \textbf{Noether's symmetry}, or a
	\textbf{divergence symmetry} if $\pounds_{\xi_Y^{(1)}}\Theta_L = \dd\alpha$
	where $\alpha$ is the pullback to $Z$ of an $n$-form on $Y$ (that must be exact
	$\alpha=\dd\beta$ on $B$), $\xi_Y^{(1)}$ is tangent to $B$ and verifies 
	$\pounds_{\xi_Y^{(1)}|_B}\Pi = 0.$
	\item A vector field $\xi_Z$ on $Z$ is a \textbf{Cartan's symmetry} if $\pounds_{\xi_Z}\Theta_L=
	\dd\alpha$ where $\alpha$ is an $n$-form on $Z$ (that is exact $\alpha=\dd\beta$ on $B$).
	
	Equivalently, if there is an $n$-form $\alpha'$ such that $\iota_{\xi_Z}\Omega_L=
	\dd\alpha^\prime$
	(we can take for example $\alpha^\prime=\iota_{\xi_Z}\Theta_L-\alpha$). This
	formulation is equivalent to $\bm{\xi_Z}$ \textbf{being a Hamiltonian vector field}.
	
	We also need $\xi_Z$ being tangent to $B$ and verifying $\pounds_{\xi_Z|_B}\Pi = 0.$
\end{enumerate}
\begin{remark}
	For a Noether's symmetry, it can be computed locally that $\alpha(x,y,z)=
	\alpha^\mu(x,y)\dd^nx_{\mu}|_{(x,y,z)}$, while for a would-be Cartan's symmetry, we have the
	following computational check to see if it is a symmetry
	\[
		\iota_{\xi_Z}\Omega_L=\dd\alpha^\prime\implies \dd\iota_{\xi_Z}\Omega_L=\pounds_{\xi_Z}
		\Omega_L=0
	\]
	Of course the relation among the three symmetries is the same that we had for the E-L's
	equations symmetries. Then, as Cartan's symmetries also include the other two, we will only focus on them
\end{remark}
\end{definition}
\begin{theorem}[Noether's theorem for Cartan symmetries]
	If $\xi_Z\in\mathfrak{X}(Z)$ is a Cartan's symmetry for the De Donder equations
	with $\pounds_{\xi_z}\Theta_L=\dd\alpha$, then $\iota_{\xi_Z}\Theta_L-\dd\alpha$ is a
	conserved quantity for the De Donder equations that is exact on $B$.
\end{theorem}
\begin{proof}
	See \cite{leonSymmetriesClassicalField2004}.
\end{proof}
\begin{proposition}
	\label{prop:cartan_flow_de_donder}
	The flow of Cartan's symmetries maps solutions of the De Donder equations into
	solutions of the De Donder equations.
\end{proposition}
\begin{proof}
	See \cite{leonSymmetriesClassicalField2004}.
\end{proof}
\subsection*{The singular Lagrangian case}
Considering the constraint algorithm of Subsection \ref{subsect:the_singular_case}, we can
define symmetries for the final constraint submanifold $Z_f\subset Z$. We first note that
the following result holds for diffeomorphisms $\Psi:Z\to Z$ which preserve the
Poincaré-Cartan $(n+2)-$form and are $\pi_{ZX}$ projectable (namely ther is a unique
$\psi:X\to X$ diffeomorphism such that $\pi_{ZX}\circ\Psi=\psi\circ\pi_{ZX}$)
\begin{proposition}[\cite{leonSymmetriesClassicalField2004}]
	If any $\Psi$ as above satisfies $\Psi(B)\subseteq B$ then it restricts to
	a diffeomorphism $\Psi_a:Z_a\to Z_a$ for any $a\in\NN$ step in the constraint algorithm
	(hence also for the final step).
\end{proposition}
As the flow generated by a vector field $\xi_Z\in\mathfrak{X}(Z)$ is a diffeomorphism,
the following corollary is straightforward
\begin{corollary}
	Let $\xi_Z\in\mathfrak{X}(Z)$ a $\pi_{ZX}-$projectable vector field that is tangent to B
	and such that $\pounds_{\xi_Z}\Omega_L=0$, then $\xi_Z$ is tangent to $Z_f$. In particular,
	$\bm{\pi_{ZX}-}$\textbf{projectable Cartan's symmetries are tangent to} $\bm{Z_f}$.
\end{corollary}
Considering the inclusion $i_f:Z_f\to Z$, we denote $\Theta_{M_f}\equiv i_f^*(\Theta_L)$, and
$\Omega_{M_f}\equiv i_f^*(\Omega_L)$; we then define Cartan symmetries for our final constraint
submanifold as follows
\begin{definition}[Cartan's symmetries in $Z_f$]
	We say that $\xi_{Z_f}\in\mathfrak{X}(Z_f)$ is a \textbf{Cartan's symmetry} for
	$\bm{(Z_f,\Omega_{Z_f})}$ if it is tangent to $B\cap Z_f$, and there exists $\alpha
	\in\Omega^n(Z_f)$ such that $\pounds_{\xi_{Z_f}}\Theta_{Z_f}=\dd\alpha_{Z_f}$.	
\end{definition}
\begin{remark}
	Note that a $\pi_{ZX}-$projectable Cartan's symmetry for the De Donder equations
	restricted to $Z_f$ is a Cartan symmetry for $(Z_f,\Omega_{Z_f})$.
\end{remark}


\section{Symmetries for the Hamilton's equations}
\label{sect:symmetries_for_hamilton_equations}
Everything we have done so far lives on the Lagrangian side, but in Section
\ref{sect:hamiltonian_setting} we built a whole Hamiltonian counterpart on $Z^*$ with
its own canonical forms $\Theta_h,\Omega_h$ and its own equations
\eqref{eq:Hamilton-equations}. The equivalence Theorem \ref{thm:equivalence_theorem}
tells us that, for a regular Lagrangian, both descriptions carry exactly the same
dynamics, so we should expect them to carry exactly the same symmetries as well. This
subsection makes that expectation precise.
 
There is one technical point that deserves attention before writing any definition.
Cartan's symmetries were the ones living directly on the total space of the theory and
they will keep doing so (now on $Z^*$), but Noether's symmetries were vector fields on
$Y$, and in the Lagrangian setting we could canonically push them up to $Z$ through the
$1-$jet prolongation $\xi_Y\mapsto\xi_Y^{(1)}$. On the Hamiltonian side there is no jet
prolongation available, and the canonical replacement is precisely the $\alpha-$\textbf{lift}
of Definition \ref{def:alpha_lift}: recall that, given $\xi_Y\in\mathfrak{X}(Y)$ and
$\alpha$ the pullback to $\Omega^{n+1}_2Y$ of a semibasic $n-$form on $Y$, the
$\alpha-$lift $\xi_Y^\alpha$ is the unique vector field on $\Omega^{n+1}_2Y$ projecting
onto $\xi_Y$ and satisfying $\pounds_{\xi_Y^\alpha}\Theta=\dd\alpha$, and that it descends
to a vector field $\xi_Y^{(1*)}\in\mathfrak{X}(Z^*)$ whenever $\xi_Y$ is
$\pi-$projectable. Note that the defining condition of the $\alpha-$lift is already
\textit{half} of a Cartan condition, but written upstairs with the \textbf{canonical}
$(n+1)-$form $\Theta$; what the definition below adds is the compatibility with the
particular Hamiltonian $h$ that we have chosen.
\begin{definition}[Symmetries for the Hamilton's equations]
\label{def:symmetries_for_hamilton_equations}
Given a Hamiltonian $h$, we have the following definitions of symmetries for the
Hamilton's equations
\begin{enumerate}
	\item A vector field $\xi_Y\in\mathfrak{X}(Y)$ is said to be a \textbf{Noether's
	symmetry}, or a \textbf{divergence symmetry}, if there exists a $\pi-$semibasic
	$n-$form on $Y$ whose pullback $\alpha$ to $\Omega^{n+1}_2Y$ is exact on $B^*$
	(i.e. $i^*_{B^*}\alpha=\dd\beta$ for some $\beta\in\Omega^{n-1}(B^*)$) and such that
	\begin{enumerate}
		\item the $\alpha-$lift $\xi_Y^\alpha$ of $\xi_Y$ to $\Omega^{n+1}_2Y$ is
		projectable to a vector field $\xi_Y^{(1*)}\in\mathfrak{X}(Z^*)$,
		\item $\pounds_{\xi_Y^{(1*)}}\Theta_h=\dd\alpha$, $\xi_Y^{(1*)}$ is tangent to
		$B^*$ and verifies $\pounds_{\xi_Y^{(1*)}|_{B^*}}\Pi^*=0$.
		Note we are \textbf{abusing notation} by using both $\alpha$ to denote the pullback
		of the original form on $Y$ both to $\Omega^{n+1}_2Y$ and $Z^*$.
	\end{enumerate}
	\item A vector field $\xi_{Z^*}\in\mathfrak{X}(Z^*)$ is said to be a
	\textbf{Cartan's symmetry} if
	\[
		\pounds_{\xi_{Z^*}}\Theta_h=\dd\alpha
	\]
	for some $\alpha\in\Omega^n(Z^*)$ which is exact on $B^*$, $\xi_{Z^*}$ is tangent
	to $B^*$ and verifies $\pounds_{\xi_{Z^*}|_{B^*}}\Pi^*=0$.
\end{enumerate}
\end{definition}
As it was the case for the De Donder equations, Noether's symmetries induce Cartan's
symmetries on $Z^*$, so once again \textbf{we may
concentrate on Cartan's symmetries}, which are the most general ones (conversely, 
if we have a Cartan's symmetry $\xi_{Z^*}\in\mathfrak{X}(Z^*)$ and there is a vector
field $\xi_Y\in\mathfrak{X}(Y)$ such that $\xi_{Z^*}$ comes from the $\alpha$-lift of
$\xi_Y$ as before, then $\xi_Y$ is a Noether's symmetry).
\begin{remark}
	\label{rem:cartan_hamiltonian_check}
	Exactly as in Definition \ref{def:5.6} (because $\Omega_h=-\dd\Theta_h$), Cartan's
	condition can be rewritten in the equivalent (and much more computable) form
	\[
		\iota_{\xi_{Z^*}}\Omega_h=\dd\alpha^\prime,\qquad
		\alpha^\prime:=\iota_{\xi_{Z^*}}\Theta_h-\alpha,
	\]
	i.e. \textbf{Cartan's symmetries are exactly the Hamiltonian vector fields of}
	$\bm{\Omega_h}$ (which are tangent to $B^*$ and preserve $\Pi^*$ there). In
	particular, a quick necessary check for a would-be Cartan's symmetry is
	\[
		\pounds_{\xi_{Z^*}}\Omega_h=\dd\iota_{\xi_{Z^*}}\Omega_h=0 .
	\]
\end{remark}
For actual computations, and recalling the local expression
$\Theta_h=(-H\,dx^\mu+p^\mu_i dy^i)\wedge \dd^nx_\mu$ of Section
\ref{sect:hamiltonian_setting}, a completely analogous computation to the one leading to
\eqref{eq:useful_local_expression_for_the_preserved_quantity} gives, for
$\xi_{Z^*}=\xi^\mu\partial_{x^\mu}+\xi^i\partial_{y^i}+\xi^\mu_i\partial_{p^\mu_i}$,
\begin{equation}
	\label{eq:useful_local_expression_hamiltonian_current}
	\iota_{\xi_{Z^*}}\Theta_h=\left(-H\xi^\nu+p^\nu_i\xi^i\right)\dd^nx_\nu
	-p^\mu_i\xi^\nu dy^i\wedge \dd^{n-1}x_{\mu\nu},
\end{equation}
which is exactly \eqref{eq:useful_local_expression_for_the_preserved_quantity} after
the substitutions $p^\mu_i=\partial L/\partial z^i_\mu$ and
$-H=p=L-z^i_\mu\partial L/\partial z^i_\mu$ dictated by the Legendre transformation
(Definition \ref{def:legendre_transformation}); this is of course not a coincidence, and it is
exactly the content of the next subsection. Note that the $p^\mu_i$ components
$\xi^\mu_i$ of the symmetry ($\xi_{Z^*}$) do not appear in the current, only in the condition defining
it.
 
Finally we have the two results that make all of this worth the effort, whose proofs are
completely analogous to their De Donder counterparts
\begin{theorem}[Noether's theorem for Cartan's symmetries of the Hamilton's equations]
	\label{thm:noether_theorem_hamilton}
	If $\xi_{Z^*}\in\mathfrak{X}(Z^*)$ is a Cartan's symmetry with
	$\pounds_{\xi_{Z^*}}\Theta_h=\dd\alpha$, then
	\[
		\sigma^*\dd\left(\iota_{\xi_{Z^*}}\Theta_h-\alpha\right)=0
	\]
	for every solution $\sigma\in\Gamma(\pi_{Z^*X})$ of the Hamilton's equations
	\eqref{eq:Hamilton-equations}, i.e. $\iota_{\xi_{Z^*}}\Theta_h-\alpha$ is a
	\textbf{conserved quantity for the Hamilton's equations}, which is moreover exact
	on $B^*$.
\end{theorem}
\begin{proof}
	See \cite{leonSymmetriesClassicalField2004}.
\end{proof}
\begin{proposition}
	\label{prop:cartan_flow_hamilton}
	The flow of Cartan's symmetries maps solutions of the Hamilton's equations into
	solutions of the Hamilton's equations.
\end{proposition}
\begin{proof}
	Identical to the one of Proposition \ref{prop:cartan_flow_de_donder} for the De
	Donder equations, see \cite{leonSymmetriesClassicalField2004}.
\end{proof}
\subsection{The Legendre map and the symmetries}
We now make precise the correspondence we have been hinting. \textbf{Throughout this
subsection we assume that $L$ is hyperregular}, so that $\mathrm{leg}_L:Z\to Z^*$ is a
global diffeomorphism (Definition \ref{def:legendre_map}) and vector fields can be
pushed forward through it; boundary conditions on both sides are assumed to be
compatible, i.e. $\mathrm{leg}_L(B)=B^*$ and $(\mathrm{leg}_L)^*\Pi^*=\Pi$, as we
demanded in Section \ref{subsect:almost_regular_lagrangians_and_the_constraint_algorithm}.
The proofs are out of the scope of this thesis (and they are straightforward), so we will
only state them (for proofs we refer the reader to \cite{leonSymmetriesClassicalField2004}).
\begin{proposition}
	\label{prop:legendre_cartan_correspondence}
	If $\xi_Z\in\mathfrak{X}(Z)$ is a Cartan's symmetry for the De Donder equations,
	then $\dd\,\mathrm{leg}_L(\xi_Z)$ is a Cartan's symmetry for the Hamilton's equations.
	The converse is also true.
\end{proposition}
\begin{lemma}
	\label{lem:legendre_alpha_lift}
	If $\xi_Y\in\mathfrak{X}(Y)$ is a Noether's symmetry for the De Donder equations,
	with $\pounds_{\xi^{(1)}_Y}\Theta_L=\dd\alpha$, then $\dd\,\mathrm{Leg}_L(\xi^{(1)}_Y)$
	is precisely the $\alpha-$lift of $\xi_Y$.
\end{lemma}
\begin{proposition}
	\label{prop:legendre_noether_correspondence}
	Every Noether's symmetry for the De Donder equations is a Noether's symmetry for the
	Hamilton's equations. The converse is also true.
\end{proposition}
\begin{remark}
	\label{rem:legendre_currents_correspond}
	The practical consequence of these results is that
	\textbf{for a hyperregular Lagrangian one may compute the Noether currents on
	whichever side is more convenient}: applying $(\mathrm{leg}_L)^*$ to the conserved
	quantity of Theorem \ref{thm:noether_theorem_hamilton} gives
	\[
		(\mathrm{leg}_L)^*\left(\iota_{\dd\,\mathrm{leg}_L(\xi_Z)}\Theta_h-\tilde{\alpha}
		\right)=\iota_{\xi_Z}\Theta_L-\alpha,
	\]
	so the Lagrangian and the Hamiltonian currents are the same object read in
	the two distint sets of coordinates, exactly as
	\eqref{eq:useful_local_expression_hamiltonian_current} and
	\eqref{eq:useful_local_expression_for_the_preserved_quantity} suggested.
	
	\textbf{However}, note that these set of results will not be used
	in our computations because, as already said, our Yang-Mills
	Lagrangian is almost regular (hence not regular neither hyperregular).
\end{remark}

\subsection*{The almost regular case}
The correspondence above breaks down as soon as $\mathrm{leg}_L$ stops being a
diffeomorphism, and unfortunately (or rather, interesting physically)
\textbf{this is precisely our case}: the Yang--Mills Lagrangian is almost regular but not
regular (Proposition \ref{prop:YM_almost_regular}). The way out is the same one we
already took for the De Donder equations: to restrict everything to the final constraint
submanifold produced by the algorithm of Subsection
\ref{subsect:almost_regular_lagrangians_and_the_constraint_algorithm}.
 
Recall from that subsection the objects $M_1=\mathrm{leg}_L(Z)\subseteq Z^*$,
$h_1=(\mu_1)^{-1}$ and $\Omega_{M_1}:=h_1^*(\tilde{\jmath}_1^*\Omega)$, together with
the sequence $M_1\supseteq M_2\supseteq\cdots\supseteq M_f$ and the fibration
$\mathrm{leg}_f:Z_f\to M_f$. Denoting by $j_f:M_f\hookrightarrow M_1$ the canonical
inclusion, we set
\[
	\Theta_{M_f}:=j_f^*\left(h_1^*(\tilde{\jmath}_1^*\Theta)\right),\qquad
	\Omega_{M_f}:=j_f^*\Omega_{M_1}=-\dd\Theta_{M_f},
\]
so that $(M_f,\Omega_{M_f})$ is the Hamiltonian mirror of the pair $(Z_f,\Omega_{Z_f})$
introduced at the end of the previous section, and we define
\begin{definition}[Cartan's symmetries in $M_f$]
	\label{def:cartan_symmetries_Mf}
	We say that $\xi_{M_f}\in\mathfrak{X}(M_f)$ is a \textbf{Cartan's symmetry} for
	$\bm{(M_f,\Omega_{M_f})}$ if it is tangent to $M_f\cap B^*$ and there exists
	$\alpha_{M_f}\in\Omega^n(M_f)$ such that
	$\pounds_{\xi_{M_f}}\Theta_{M_f}=\dd\alpha_{M_f}$.
\end{definition}
\begin{proposition}
	\label{prop:cartan_symmetries_Mf_to_Zf}
	If $\xi_{M_f}$ is a Cartan's symmetry of $(M_f,\Omega_{M_f})$, then any vector field
	$\xi_{Z_f}\in\mathfrak{X}(Z_f)$ such that $\dd\,\mathrm{leg}_f(\xi_{Z_f})=\xi_{M_f}$ is a
	Cartan's symmetry of $(Z_f,\Omega_{Z_f})$.
\end{proposition}
\begin{proof}
	See \cite{leonSymmetriesClassicalField2004}.
\end{proof}
\begin{proposition}
	\label{prop:noether_theorem_Mf}
	If $\xi_{M_f}$ is a Cartan's symmetry of $(M_f,\Omega_{M_f})$ with
	$\pounds_{\xi_{M_f}}\Theta_{M_f}=\dd\alpha_{M_f}$, then
	$\iota_{\xi_{M_f}}\Theta_{M_f}-\alpha_{M_f}$ is a conserved quantity for the
	Hamilton's equations restricted to $M_f$, and the analogous statement holds on
	$(Z_f,\Omega_{Z_f})$ for the De Donder equations.
\end{proposition}
\begin{proof}
	Word by word the proof of Theorem \ref{thm:noether_theorem_hamilton}, replacing
	$(Z^*,\Theta_h,\Omega_h)$ by $(M_f,\Theta_{M_f},\\\Omega_{M_f})$ and using that
	solutions of the geometric Hamilton equations are, by construction of the algorithm,
	tangent to $M_f$.
\end{proof}
\begin{remark}[Why this is the version we will actually use]
	\label{rem:almost_regular_symmetries_yang_mills}
	Note two comments.
	First, Proposition
	\ref{prop:cartan_symmetries_Mf_to_Zf} goes only in one direction and the lift
	$\xi_{Z_f}$ is \textbf{not unique}: two lifts differ by a vector field tangent to
	the fibers of $\mathrm{leg}_f$, which are exactly the directions along which the
	Lagrangian does not see the velocities (in the sense that the corresponding
	multimomenta are zero and so we can not apply the inverse function theorem to locally
	relate them one to one). For Yang--Mills these are the directions of
	the symmetric part $A^a_{(\mu\nu)}$ of the jet coordinates, which is precisely the
	piece of $Z$ killed by $\mathrm{Leg}_L$ (cf. the proof of Proposition
	\ref{prop:YM_almost_regular}). In other words, \textbf{the currents are
	honestly defined downstairs on} $\bm{M_f}$, and all the ambiguity sits in the
	Lagrangian picture.
	
	Second, since $\mathrm{leg}_f$ is a fibration with connected fibers, a function
	(or form) on $Z_f$ that is constant along those fibers descends to $M_f$, so the
	conserved quantities of Proposition \ref{prop:noether_theorem_Mf} are in one-to-one
	correspondence on both sides.
\end{remark}
 
\section{Symmetries in the instantaneous framework}
\label{sect:symmetries_instantaneous}
So far every conservation statement we have proven has been of the form
``$\sigma^*\dd\alpha=0$'': a \textit{local} conservation law, one $(n+1)-$dimensional
divergence that vanishes, with no distinguished notion of time and hence no
\textit{conserved charge}. That is exactly what the instantaneous framework of Chapter
\ref{chap:instantaneous_formalism} was built to provide, so it is here where the word
``conserved'' finally recovers its usual dynamical meaning: a number attached to each
Cauchy slice which does not change as the slice moves in time.
 
The starting point is that symmetries of presymplectic systems are remarkably simple to
describe. Let $(N,\omega)$ be a presymplectic manifold (i.e. $\omega$ a closed $2-$form,
possibly degenerate); $\xi\in\mathfrak{X}(N)$ is a symmetry of the presymplectic system if
there exists a function $G\in C^\infty(N)$ satisfying $\iota_\xi\omega=\dd G$. Suppose that
there a solution of the presymplectic system, i.e. a vector field $U\in\mathfrak{X}(N)$
such that $\iota_U\omega=0$, then
\begin{equation}
	\label{eq:presymplectic_noether_one_line}
	0=\iota_U\iota_\xi\omega=\dd G(U)=U(G),
\end{equation}
so $G$ is constant along the solutions. The whole first Noether's theorem in the
instantaneous framework is contained in this one line; everything else is bookkeeping
relating $G$ with the currents we already know. Applying this to
$(\tilde{Z},\widetilde{\Omega_L})$ -which is presymplectic by Proposition
\ref{prop:presymplectic_proposition}- we define
\begin{definition}[Cartan's symmetries in the space of Cauchy data]
	\label{def:presymplectic_cartan_symmetry}
	A vector field $\tilde{\xi}\in\mathfrak{X}(\tilde{Z})$ is a \textbf{Cartan's
	symmetry} of the presymplectic system $(\tilde{Z},\widetilde{\Omega_L})$ if there
	exists $\tilde{G}\in C^\infty(\tilde{Z})$ such that
	\[
		\iota_{\tilde{\xi}}\widetilde{\Omega_L}=\dd\tilde{G},
	\]
	that is, if $\tilde\xi$ is a Hamiltonian vector field with Hamiltonian function
	$\tilde{G}$ in the same sense of a symplectic system.
\end{definition}
\begin{proposition}[Noether's theorem in the instantaneous framework]
	\label{prop:noether_presymplectic}
	Let $\tilde{\xi}$ be a Cartan's symmetry of $(\tilde{Z},\widetilde{\Omega_L})$ with
	Hamiltonian function $\tilde{G}$, and let $c_{\tilde{Z}}$ be a solution of the
	instantaneous De Donder equations \eqref{eq:presymplectic_equation}. Then
	$\tilde{G}\circ c_{\tilde{Z}}$ is constant.
\end{proposition}
\begin{proof}
	This is \eqref{eq:presymplectic_noether_one_line} with $U=\dot{c}_{\tilde{Z}}$
	\[
		\der{}{t}\left(\tilde{G}\circ c_{\tilde{Z}}\right)=
		\dot{c}_{\tilde{Z}}(\tilde{G})=\iota_{\dot{c}_{\tilde{Z}}}\dd\tilde{G}=
		\iota_{\dot{c}_{\tilde{Z}}}\iota_{\tilde{\xi}}\widetilde{\Omega_L}=
		-\iota_{\tilde{\xi}}\iota_{\dot{c}_{\tilde{Z}}}\widetilde{\Omega_L}=0,
	\]
	where the last equality is precisely Theorem \ref{thm:instantaneous_de_donder}.
\end{proof}
It only remains to check that the Cartan's symmetries we already know produce Cartan's
symmetries in this sense, and to identify the corresponding $\tilde{G}$.
\begin{proposition}
	\label{prop:cartan_symmetry_induces_presymplectic}
	Let $\xi_Z\in\mathfrak{X}(Z)$ be a Cartan's symmetry for the De Donder equations,
	i.e. $\pounds_{\xi_Z}\Theta_L=\dd\alpha$. Then the induced vector field
	$\tilde{\xi}_{\tilde{Z}}\in\mathfrak{X}(\tilde{Z})$, defined by
	$\tilde{\xi}_{\tilde{Z}}(\gamma):=\xi_Z\circ\gamma$, is a Cartan's symmetry of the
	presymplectic system $(\tilde{Z},\widetilde{\Omega_L})$, with Hamiltonian function
	the momentum $\widetilde{\iota_{\xi_Z}\Theta_L-\alpha}$ of its associated conserved
	quantity.
\end{proposition}
\begin{proof}
	See \cite{leonSymmetriesClassicalField2004}.
\end{proof}
\begin{remark}
	\label{rem:noether_current_is_hamiltonian_form}
	The proposition above is the cleanest
	possible statement of the first Noether's theorem in the multisymplectic language:
	\textbf{the Noether current of a Cartan's symmetry is exactly a Hamiltonian}
	$\bm{n-}$\textbf{form (Definition \ref{def:hamiltonian_forms_and_vector_fields}),
	and the symmetry itself is its Hamiltonian vector field}. The dictionary
	``symmetry $\leftrightarrow$ conserved quantity'' is therefore the dictionary
	``Hamiltonian vector field $\leftrightarrow$ Hamiltonian form'' of Subsection
	\ref{subsect:hamiltonian_forms_and_multisymplectic_brackets}, and the ambiguity of
	the current up to closed forms (Remark
	\ref{remark:preserved_quantities_plus_closed}) is nothing but the ambiguity in
	Definition \ref{def:hamiltonian_forms_and_vector_fields}. Combining this with
	Proposition \ref{prop:noether_presymplectic}, the conserved charge attached to a
	Cartan's symmetry is just its momentum.
\end{remark}
\subsection*{Conservation of the momenta along solutions}
We now prove the result that we owed the reader since the motivation of Section
\ref{sec:symmetries_for_euler_lagrange_equations}, namely the precise sense in which a
conserved quantity gives a conserved number, and why we insisted so much on conserved
quantities being \textbf{exact} (and not merely closed) on $B$.
\begin{proposition}
	\label{prop:conservation_along_solutions}
	Let $\alpha$ be a conserved quantity for the De Donder equations which is exact on
	$B\subseteq\partial Z$, $\iota_B^*\alpha=\dd\beta$, and let $c_{\tilde{Z}}$ be a
	solution of the instantaneous De Donder equations
	\eqref{eq:presymplectic_equation} whose projection $c_{\tilde{X}}$ to $\tilde{X}$
	splits $X$ in the sense of Definition \ref{def:slicing}. Then
	$\tilde{\alpha}\circ c_{\tilde{Z}}$ is constant, i.e. the function
	\[
		t\longmapsto\int_M c_{\tilde{Z}}(t)^*\alpha-\int_{\partial M}
		c_{\tilde{Z}}(t)^*\beta
	\]
	does not depend on $t$, and hence is a \textbf{conserved quantity}.
\end{proposition}
\begin{proof}
	Pick $t_1<t_2$ in the domain of the solution curve and denote
	$M_1:=c_{\tilde{X}}(t_1)$, $M_2:=c_{\tilde{X}}(t_2)$. As $c_{\tilde{X}}$ splits $X$,
	we can consider the piece $U\subseteq X$ identified with $M\times[t_1,t_2]$, where
	$M_1$ is identified with $M\times\{t_1\}$, $M_2$ with $M\times\{t_2\}$, and we
	denote by $V$ the boundary piece corresponding to $\partial M\times[t_1,t_2]$.
	Since $\alpha$ is a conserved quantity, $c_{\tilde{Z}}(t)^*\dd\alpha=0$ for all $t$,
	so integrating over $U$ and applying Stokes' theorem we obtain
	\[
		0=\int_{M_2}c_{\tilde{Z}}(t)^*\alpha+\int_{V}c_{\tilde{Z}}(t)^*\alpha-
		\int_{M_1}c_{\tilde{Z}}(t)^*\alpha .
	\]
	Now, $V$ lies inside the boundary, where $\alpha=\dd\beta$, and from
	$0=\partial\partial U=\partial M_2+\partial V-\partial M_1$ we get, applying
	Stokes' theorem once more,
	\[
		\int_{V}c_{\tilde{Z}}(t)^*\alpha=\int_{\partial V}c_{\tilde{Z}}(t)^*\beta=
		\int_{\partial M_1}c_{\tilde{Z}}(t)^*\beta-\int_{\partial M_2}
		c_{\tilde{Z}}(t)^*\beta .
	\]
	Substituting in the previous equation and grouping the terms of each slice together
	we obtain $\tilde{\alpha}(c_{\tilde{Z}}(t_2))=\tilde{\alpha}(c_{\tilde{Z}}(t_1))$,
	which is the claim as $t_1$ and $t_2$ were arbitrary.
\end{proof}
\begin{corollary}
	\label{cor:conserved_charge_lagrangian_symmetry}
	In particular, if $\xi_Y$ is a symmetry of the Lagrangian for the De Donder
	equations, the preceding formula applied to the conserved quantity
	$\iota_{\xi^{(1)}_Y}\Theta_L$ says that
	\[
		\int_{M}c_{\tilde{Z}}(t)^*\iota_{\xi^{(1)}_Y}\Theta_L+\int_{\partial M}
		c_{\tilde{Z}}(t)^*\iota_{\xi^{(1)}_Y}\Pi
	\]
	is conserved along every solution of the instantaneous De Donder equations
	\eqref{eq:presymplectic_equation} whose projection $c_{\tilde{X}}$ splits $X$.
\end{corollary}
\begin{remark}
	\label{rem:why_exactness_on_B}
	This closes the circle opened in the motivation of Section
	\ref{sec:symmetries_for_euler_lagrange_equations}. The two ``extra conditions'' we
	imposed there were not technical decoration: the tangency of $\xi^{(1)}_Y$ to $B$
	guarantees that the flow of the symmetry maps admissible solutions to admissible
	solutions, while $\pounds_{\xi^{(1)}_Y|_B}\Pi=0$ is exactly what makes the conserved
	quantity \textbf{exact} on $B$, which by Definition
	\ref{def:integration_upon_Cauchy_of_forms}/\eqref{eq:tilde_forms} is the hypothesis needed for the
	momentum $\tilde{\alpha}$ to be defined at all, and by Proposition
	\ref{prop:conservation_along_solutions} the hypothesis needed for it to be
	conserved. Had we only asked for closedness, the boundary contribution
	$\int_V c^*\alpha$ in the proof above would have been an uncontrolled flux of the
	current through $\partial M$ and no conservation statement would survive.
\end{remark}
\section{Localizable symmetries and the second Noether's theorem}
Everything up to this point has been (a geometric version of) the \textbf{first}
Noether's theorem: one symmetry, one current. We finally turn to the second one, which
concerns symmetries that can be turned on and off independently at different regions of
spacetime -what a theoretical physicist calls \textbf{local} or \textbf{gauge}
symmetries- and which, instead of producing new information, produce \textbf{identities}
among the field equations. This is the case of interest for us, as the gauge
symmetries of Section \ref{sec:gauge_group} are of exactly this
type.
\begin{definition}[Localizable symmetry]
	\label{def:localizable_symmetry}
	A symmetry of the Lagrangian $\xi_Y$ is said to be \textbf{localizable} if
	$\xi^{(1)}_Y$ vanishes on $\partial Z$ and, for every pair of open sets $U,U'$ in
	$X$ with disjoint closures, there exists another symmetry of the Lagrangian
	$\zeta_Y$ such that
	\[
		\xi^{(1)}_Y=\zeta^{(1)}_Y\quad\text{on }\pi_{ZX}^{-1}(U)\qquad\text{and}\qquad
		\zeta^{(1)}_Y=0\quad\text{on }\pi^{-1}_{ZX}(U')\cup\partial Z .
	\]
\end{definition}
\begin{remark}
	The definition is just the rigorous version of the sentence ``the symmetry depends
	on arbitrary functions on the spacetime''. For a \textbf{rigid} (global) symmetry, such
	as a spacetime translation or an internal rotation with a constant parameter,
	prescribing the symmetry on an open set already prescribes it everywhere, so no such
	$\zeta_Y$ can exist unless $\xi_Y$ vanishes identically. For a \textbf{gauge}
	symmetry, on the contrary, the parameter is an arbitrary (say compactly supported)
	section, and $\zeta_Y$ is obtained simply by multiplying the parameter by a bump
	function which is $1$ on $U$ and $0$ on $U'$; the whole point of Definition
	\ref{def:localizable_symmetry} is that \textbf{the symmetry algebra is infinite
	dimensional and locally supported}.
\end{remark}
\begin{theorem}[Second Noether's theorem \cite{leonSymmetriesClassicalField2004}]
	\label{thm:second_noether_theorem}
	Let $\xi_Y$ be a localizable symmetry and let $c_{\tilde{Z}}$ be a solution of the
	instantaneous De Donder equations \eqref{eq:presymplectic_equation}. Then the following
	equation (called \textbf{Noether's identity}) is satisfied
	\[
		\widetilde{\left(\iota_{\xi^{(1)}_Y}\Theta_L\right)}(c_{\tilde{Z}}(t))=0
	\]
	for all $t$. 
\end{theorem}
In other words, the conserved charge associated with a localizable
symmetry ($\iota_{\xi^{(1)}_Y}\Theta_L$) does not merely stay constant: it
\textbf{vanishes identically on-shell}.

\section{Momentum map}
\label{sect:momentum_map}
In symplectic mechanics, when a Lie group $G$ acts on the phase space preserving the
symplectic form, all the conserved quantities produced by Noether's theorem can be
packaged into a single object, the momentum map $J:M\to\mathfrak{g}^*$, which is not
only elegant but the starting point of Marsden-Weinstein reduction (see
\cite{baguisMarsdenWeinsteinReductionSymplectic2001} for example). We close this
chapter by doing exactly the same here, with the only (expected) difference that in
field theory a conserved quantity is an $n-$form rather than a function, so our momentum
map will take values in $\mathfrak{g}^*\otimes\Omega^n(Z)$: \textbf{one conserved current
per generator of the symmetry group}.
\subsection*{Lie group actions as symmetries}
If $G$ is a Lie group acting on $Y$, the action can be lifted to an action of $G$ on
$Z$, and the infinitesimal generator of the lifted action is the $1-$jet prolongation of
the infinitesimal generator of the original one, i.e., for each $\xi\in\mathfrak{g}$,
$\xi_Z=\xi_Y^{(1)}$ (where $\xi$ acts on $Y$ through the exponential map and 
$\xi_Y(y):=\der{}{t}\big|_{t=0}y\cdot\exp(t\xi)$).

\begin{definition}[Group of symmetries of the Lagrangian]
	\label{def:group_of_symmetries_of_the_lagrangian}
	We say that a Lie group $G$ \textbf{acts as a group of symmetries of the Lagrangian}
	if it defines an action on $Y$ that projects onto a compatible action on $X$, whose
	$1-$jet prolongation preserves $B$, and such that the flow $\phi_t^Z$ of $\xi_Z$
	verifies
	\[
		(\phi_t^Z)^*\mathcal{L}=\mathcal{L},\qquad(\phi_t^Z)^*\Pi=\Pi,\quad\forall t.
	\]
\end{definition}
The fact that the action is fibered implies that $\xi_Y$ is a projectable vector field;
therefore the condition $(\phi_t^Z)^*\mathcal{L}=\mathcal{L}$, infinitesimally expressed as
$\pounds_{\xi_Z}\mathcal{L}=0$, together with the two direct consequences of the
definition $i)$ $\xi_Z$ is tangent to $B$ and $ii)$ $\pounds_{(\xi_Z)|_B}\Pi=0$
states exactly that $\xi_Y$ is a symmetry of the Lagrangian in the sense of Definition
\ref{def:symmetries_of_the_lagrangian}, and by Proposition
\ref{prop:useful_proposition_for_checking_symmetries} (recall that $\xi_Y$ is
projectable) we even have $\pounds_{\xi_Z}\Theta_L=0$ on the nose. \textbf{So a group of
symmetries of the Lagrangian is just a whole Lie algebra worth of symmetries of the
Lagrangian, parametrized by} $\mathfrak{g}$.
\subsection*{The multisymplectic momentum map}
\begin{definition}[Momentum map]
	\label{def:momentum_map}
	Let $G$ act on $Y$ as a group of symmetries of the Lagrangian. The
	\textbf{momentum map} is the map
	\[
		J:Z\longrightarrow\mathfrak{g}^*\otimes\bigwedge^nZ
	\]
	defined by $\langle J(z),\xi\rangle:=\left(\iota_{\xi_Z}\Theta_L\right)_z$ for every
	$\xi\in\mathfrak{g}$. Equivalently, and this is the way we shall use it, $J$ assigns
	to each $\xi\in\mathfrak{g}$ the $n-$form
	\[
		J^\xi:=\iota_{\xi_Z}\Theta_L\in\Omega^n(Z),
	\]
	which is precisely the Noether current of Theorem
	\ref{thm:noether_theorem_symmetries_of_the_lagrangian}; in particular $J^\xi$ is a
	conserved quantity and we shall denote its associated momentum by
	$\widetilde{J^\xi}$.
\end{definition}
\begin{remark}
	\label{rem:momentum_map_on_B}
	On $B$, since $\iota_B^*\Theta_L=\dd\Pi$ and $\pounds_{(\xi_Z)|_B}\Pi=0$, we have
	$\iota_{(\xi_Z)|_B}\dd\Pi=-\dd\iota_{(\xi_Z)|_B}\Pi$ and therefore
	\[
		J^\xi|_B=\iota_{\xi_Z}\dd\Pi=-\dd\left(\iota_{\xi_Z}\Pi\right),
	\]
	which exhibits explicitly the exactness on the boundary demanded in Proposition
	\ref{prop:conservation_along_solutions}, with $\beta=-\iota_{\xi_Z}\Pi$.
\end{remark}
The local expression of $J^\xi$ is, of course, the one we already computed in
\eqref{eq:useful_local_expression_for_the_preserved_quantity}, so all the momentum map
computations of Chapter \ref{chap:multisymplectic_yang_mills_computations} reduce to
plugging the generators $\xi_Z$ into that formula.
\begin{proposition}
	\label{prop:dJ_equals_iota_omega}
	For every $\xi\in\mathfrak{g}$,
	\[
		\dd J^\xi=\iota_{\xi_Z}\Omega_L,
	\]
	that is, $J^\xi$ is a Hamiltonian $n-$form with Hamiltonian vector field $\xi_Z$ in
	the sense of Definition \ref{def:hamiltonian_forms_and_vector_fields}.
\end{proposition}
\begin{proof}
	As $\xi_Y$ is projectable, $\pounds_{\xi_Z}\Theta_L=0$ by Proposition
	\ref{prop:useful_proposition_for_checking_symmetries}, so that by Cartan's magic
	formula and $\Omega_L=-\dd\Theta_L$
	\[
		0=\pounds_{\xi_Z}\Theta_L=\iota_{\xi_Z}\dd\Theta_L+\dd\iota_{\xi_Z}\Theta_L=
		-\iota_{\xi_Z}\Omega_L+dJ^\xi .
	\]
\end{proof}
Proposition \ref{prop:dJ_equals_iota_omega} is the exact analogue of the symplectic
identity $\iota_{\xi_M}\omega=dJ^\xi$, and it is what allows us to feed the momentum map
into the bracket machinery of Subsection
\ref{subsect:hamiltonian_forms_and_multisymplectic_brackets}
\begin{proposition}
	\label{prop:momentum_map_algebra}
	For every $\xi,\zeta\in\mathfrak{g}$,
	\begin{equation}
		\label{eq:momentum_map_algebra}
		\left\{J^\xi,J^\zeta\right\}=J^{[\xi,\zeta]}-\dd(\iota_{\xi_Z}\iota_{\zeta_Z}\Theta_L)
	\end{equation}
	(namely $\left\{J^\xi,J^\zeta\right\}=J^{[\xi,\zeta]}$ modulo exact $n$-forms),
	provided that $\xi\mapsto\xi_Y$ is a homomorphism of Lie algebras; if the
	identification is an anti-homomorphism (which is the case for a left action) a
	global minus sign appears.
\end{proposition}
\begin{proof}
	By Proposition \ref{prop:dJ_equals_iota_omega} the vector field $\zeta_Z$ is a
	Hamiltonian vector field for $J^\zeta$, so using Definition
	\ref{def:multisymplectic_bracket} we have
	\[
		\{J^\xi,J^\zeta\}=\iota_{\xi_Z}\iota_{\zeta_Z}\Omega_L=\iota_{\xi_Z}dJ^\zeta
		=\pounds_{\xi_Z}J^\zeta-\dd(\iota_{\xi_Z}J^\zeta),
	\]
	and so Equation \eqref{eq:momentum_map_algebra} follows after noticing
	\[
		\pounds_{\xi_Z}J^\zeta=\pounds_{\xi_Z}(\iota_{\zeta_Z}\Theta_L)=
		\iota_{\pounds_{\xi_Z}\zeta_Z}\Theta_L+\iota_{\zeta_Z}\pounds_{\xi_Z}\Theta_L
		=\iota_{[\xi_Z,\zeta_Z]}\Theta_L+0=\iota_{[\xi_Y^{(1)},\zeta_Y^{(1)}]}\Theta_L
	\]
	\[
		\overset{\text{Prop. }\ref{prop:symmetries_will_form_subalgebra}}{=}
		\iota_{[\xi_Y,\zeta_Y]^{(1)}}\Theta_L\overset{(*)}{=}\iota_{[\xi,\zeta]^{(1)}_Y}\Theta_L=
		\iota_{[\xi,\zeta]_Z}\Theta_L\overset{\text{Prop. }\ref{prop:dJ_equals_iota_omega}}{=}
		J^{[\xi,\zeta]},
	\]
	where $\iota_{\zeta_Z}\pounds_{\xi_Z}\Theta_L=0$ follows from $\xi_Y$ being projectable
	and $(*)$ follows from the hypothesis of $\xi\mapsto\xi_Y$ being an homomorphism of Lie
	algebras -in the case of it being an anti-homomorphism, we would get an overall minus
	sign in Equation \eqref{eq:momentum_map_algebra}.
	Observe that the argument never used uniqueness of the Hamiltonian vector fields.
\end{proof}
\begin{remark}
Note that \textbf{no regularity of} $\bm{L}$ \textbf{is
needed}, which matters for us since Yang--Mills is only almost regular. Note also that
Equation \eqref{eq:momentum_map_algebra} do not have the exact term in the case
that both $\xi_Y$ and $\zeta_Y$ are $\pi_{YX}$ vertical. This fact is straightforward once
considered the local expressions of $\Omega_L$ (cf. Definition
\ref{def:poincare-cartan-forms}) and of the jet prolongations (cf. Definition
\ref{def:1jet-prolongation}, where because of verticality $\xi_Y^\mu=0$) and will be
specifically met in Subsection \ref{sub:top_forms_ell_n_}, implying hence the closeness
of the algebra.
\end{remark}
\subsection*{The momentum map in the instantaneous framework}
If $G$ acts on $Y$ as a group of symmetries of the Lagrangian, it induces an action on
$\tilde{Z}$ defined pointwise on the image of every embedding $\gamma\in\tilde{Z}$, and
for $\xi\in\mathfrak{g}$ the vector field $\tilde{\xi}_{\tilde{Z}}\in\mathfrak{X}(\tilde{Z})$
induced by $\xi_Z$ is, by Proposition
\ref{prop:cartan_symmetry_induces_presymplectic}, a Cartan's symmetry of the
presymplectic system $(\tilde{Z},\widetilde{\Omega_L})$.
\begin{definition}[Instantaneous momentum map]
	\label{def:instantaneous_momentum_map}
	The presymplectic form $\widetilde{\Theta_L}$ induces a momentum map
	\[
		\tilde{J}:\tilde{Z}\longrightarrow\mathfrak{g}^*,\qquad
		\tilde{J}^\xi:=\langle\tilde{J},\xi\rangle=\iota_{\tilde{\xi}_{\tilde{Z}}}
		\widetilde{\Theta_L}:\tilde{Z}\longrightarrow\mathbb{R},
	\]
	which, unwinding Definition \ref{def:integration_upon_Cauchy_of_forms}, is explicitly
	\begin{equation}
		\label{eq:instantaneous_momentum_map_local}
		\tilde{J}^\xi(\gamma)=\int_M\gamma^*\left(\iota_{\xi_Z}\Theta_L\right)+
		\int_{\partial M}\gamma^*\left(\iota_{\xi_Z}\Pi\right),\qquad\gamma\in\tilde{Z}.
	\end{equation}
\end{definition}
Comparing \eqref{eq:instantaneous_momentum_map_local} with Corollary
\ref{cor:conserved_charge_lagrangian_symmetry} we immediately get
$\tilde{J}^\xi=\widetilde{J^\xi}$: \textbf{the instantaneous momentum map is the
momentum of the multisymplectic one}, so it is a conserved charge.

Repeating verbatim
the argument of Proposition \ref{prop:dJ_equals_iota_omega} in $\tilde{Z}$, or simply
applying Proposition \ref{prop:presymplectic_proposition} to Proposition
\ref{prop:dJ_equals_iota_omega}, we also have
\begin{proposition}
	\label{prop:instantaneous_momentum_map}
	For every $\xi\in\mathfrak{g}$,
	\[
		\dd\tilde{J}^\xi=\iota_{\tilde{\xi}_{\tilde{Z}}}\widetilde{\Omega_L},
	\]
	i.e. $\tilde{J}^\xi$ is a Hamiltonian function with Hamiltonian vector field
	$\tilde{\xi}_{\tilde{Z}}$ in the sense of symplectic mechanics; consequently, by Proposition
	\ref{prop:noether_presymplectic}, $\tilde{J}^\xi$ is constant along every solution
	of the instantaneous De Donder equations \eqref{eq:presymplectic_equation}, i.e., on-shell.
\end{proposition}
Furthermore, by Proposition \ref{prop:brackets_commute_with_tilde} the multisymplectic
and the presymplectic brackets are compatible with the Cauchy surface integration, so
Proposition \ref{prop:momentum_map_algebra} descends to
\[
	\left\{\tilde{J}^\xi,\tilde{J}^\zeta\right\}=\widetilde{\left\{J^\xi,J^\zeta\right\}}
	=\tilde{J}^{[\xi,\zeta]},
\]
i.e., \textbf{the charges close the same algebra as the symmetry
generators}. This is the field theoretical version of the classical statement that a
momentum map is a Poisson map, and it is the structure that organizes the whole
computation of conserved quantities.

\chapter{Gauge theory} \label{chap:gauge_theory}

Up to now we have developed the multisymplectic machinery
(Chapter~\ref{chap:multisymplectic}), its instantaneous counterpart
(Chapter~\ref{chap:instantaneous_formalism}) and the theory of symmetries and conserved
quantities (Chapter~\ref{chap:symmetries_and_Noether_theorems}) for a completely
general first order field theory, i.e.\ for an arbitrary configuration bundle
$\pi\colon Y\to X$. In order to reach the main goal of this work, Yang--Mills
theories, it only remains to identify \textbf{which} configuration bundle encodes a
gauge field and to translate the notion of gauge transformation into the language
of the previous chapters.

Nowadays, familiarization with fiber bundles (principal fiber bundles in the case of gauge
theories) among particle phenomenologists is limited beyond some general concepts as
(principal) connections, field strengths, covariant derivatives among others. Of course, this is not
always the case; indeed, an anecdote regarding this is the documented encounter (see
\cite{yangSelectedPapers194519802005}) between the physicist Chen Ning Yang and the
mathematician Shiing-Shen Chern in 1975.
Thanks to the mathematician Jim Simons, Yang learned about the beauty of fiber-bundle theory
and how his gauge fields are exactly geometric connections on principal fiber bundles.
Motivated by this, he drove to the house of Chern in El Cerrito (near Berkeley university)
to talk about it (among many other things).
In this encounter, Yang, years after developing Yang--Mills theory in collaboration with
Robert Mills, asked Chern how could mathematicians have dreamt about fiber bundle
concepts out of nowhere, without any reference to the physical world. Chern answered that
these concepts were indeed \textbf{natural} and \textbf{real}.
As David Bleecker mentions on \cite{Bleecker81}, `to the uninitiated, it would seem that the use
of (principal) fiber bundles and connections to describe the basic forces of nature is a
half-baked scheme devised by some clique of mathematicians bent on producing an application
for their work' but indeed is not.

Through this chapter, we will dive on the whole principal fiber bundle structure motivated
to both \textbf{show the algebraic richness of the mathematical formulation of gauge theory}
(and hence of modern theoretical physics) and the \textbf{identifications with the usual
physical gauge `computational tools'} (like gauge potentials, field strength tensors,
gauge transformations and much more concepts used by phenomenologist every day). But first,
let us try to show some light on the physical nature behind principal fiber bundles and how
these are the canonical mathematical structure to describe the notion of `choosing a gauge'.

For every spacetime point ($x\in X$), gauge inner nature of a theory is encapsulated 
in a vector space $V$ (for example $\mathbb{C}$ representing the phase of a wave function in
quantum mechanics), in the case of gauge fields (which will be the only case possible
in this work about \textbf{pure} Yang--Mills theory, also called \textbf{gluodynamics}),
the Lie algebra $V\equiv \mathfrak{g}$. 
Hence, a natural arena for our gauge theory will be one that associates for each spacetime
point a copy of $\mathfrak{g}$. This clearly reminds of a fiber bundle construction in
which \textbf{locally}, the bundle is identify with the base manifold times a typical fiber;
in this case we want the typical fiber to be a Lie group $G$ such that the tangent (vertical)
space resembles $\mathfrak{g}$, i.e., the inner `directions'. The naive approach of
just considering the space $X\times G$ can't describe phenomena like instantons, magnetic
monopoles, string theory and in general, situations in where the base manifold $X$ has a
non-trivial topology. 

Fibre bundles can be thought as a twisted, non-trivial ensemble of products between a
base manifold and a fiber manifold. In the case of principal fiber bundles the fiber manifold
is a Lie group $G$ while in \textit{vector fiber bundles} the fiber manifold is a vector space
(vector fiber manifolds were already devised on Subsection
\ref{subsect:the_vertical_endomorphism}; in that case fibers were vector spaces of the form
$T^*_{\pi(y)}X\otimes\mathcal{V}_y\pi$). In the end, as mentioned by Mark J.D. Hamilton on \cite{hamiltonMathematicalGaugeTheory2018}
`It is not an overstatement to say that gauge theory is ultimately the theory of principal
bundles and associated vector bundles'. In pure gauge theories (like pure Yang--Mills), we
fix the principal fiber bundle and then the associated vector bundle is canonically fully determined
by it, by the so-called \textbf{adjoint action} and the Lie algebra $\mathfrak{g}$. In full gauge
theories (like quantum chromodynamics with Lie group SU(3)) matter fields need an additional
structure, a \textbf{representation}; then we have the associated vector bundle of pure gauge
plus another associated vector bundle determined by the principal fiber bundle, by the vector
space $V$ of the representation and by the representation itself $\rho$. The ultimate reason
for the interaction between matter fields and gauge fields (for example once quantized,
in chromodynamics the interaction between quarks and gluons) is that both are related to the
\textbf{same principal bundle} \cite{hamiltonMathematicalGaugeTheory2018}.

Considering the principal fiber bundle, as every fiber bundle, we can consider local sections
that associate to base points $x\in X$ a \textit{choice} of a point on the fiber above it.
These local sections allow physicist to work locally on $X$ with object defined globally
on $P$ (like the U(1) connection on electromagnetism) at the price of having to make that
choice of local internal $\textbf{frame}$ for the gauge nature of the object. Of course,
physically measurable quantities should not depend on the specific choice of point above the fiber
neither should there be a specific `privileged' point. The nature of the free transitive action of
the Lie group over the principal fiber bundle avoid this defect allowing also a way of
translating local descriptions on $X$ of objects living on $P$ between two different physicists.

The physical (gauge) field of a Yang--Mills theory is a \textbf{gauge potential}, whose
geometric incarnation is a \textbf{principal connection} on a principal fiber bundle
$P\to X$ modelling the internal (colour) degrees of freedom. The subtlety, which
makes gauge theories genuinely richer than, say, a scalar field, is twofold.
First, connections are not sections of a bundle over $X$ in an obvious way, so it
is not immediately clear how they fit into the framework of
Chapter~\ref{chap:multisymplectic}; the answer is the \textbf{bundle of connections}
$C(P)=J^{1}P/G$ (Section~\ref{sec:bundle_of_connections}), whose sections are
exactly the connections and which turns out to carry precisely the affine
structure already met in Chapter~\ref{chap:multisymplectic}. Second, the theory is
invariant under an infinite dimensional group of \textbf{gauge transformations}
(Section~\ref{sec:gauge_group}); these are the localizable symmetries governed by
the second Noether theorem, and they are responsible for the Yang--Mills
Lagrangian being singular --- indeed \textbf{almost regular} --- so that the
constraint algorithm of Chapter~\ref{chap:multisymplectic} is unavoidable.

In this chapter we set up all these geometric ingredients; the explicit
computation of the Yang--Mills field equations, momenta and conserved currents is
carried out in Chapter \ref{chap:multisymplectic_yang_mills_computations}. Standard
references for the differential geometry of principal fiber bundles and connections are
\cite{sontzPrincipalBundlesClassical2019,kobayashiFoundationsDifferentialGeometry1996,KMS93}, and for the variational, physics-oriented viewpoint
\cite{Bleecker81}. The realization of gauge fields inside the
multisymplectic formalism follows \cite{Garcia77,CLMM01}; an alternative
multisymplectic treatment of the Yang--Mills equations, keeping the connection on
the total space of $P$ rather than reducing, is \cite{Helein14}.

After all this motivation, we introduce mathematically the concept of principal fiber bundles
\section{Principal bundles}\label{sec:principal_bundles}

Let $G$ be a (finite-dimensional) Lie group (non abelian in general) 
with dimension $N$ -some physical examples
are $G=U(1)$ for electromagnetism, $G=SU(2)$ for the weak interaction, $G=SU(3)$ for the
strong interaction (QCD), and $G=U(1)\times SU(2)\times SU(3)$ for the Standard Model-
then

\begin{definition}[Principal fiber bundle (PFB) \cite{Bleecker81}]\label{def:PFB}
A \textbf{principal fiber bundle (PFB)} consists of a manifold $P$ (called the total space),
a Lie group $G$, a base manifold $X$, and a projection map $\pi_P:P\longrightarrow X$
such that the following conditions hold.
\begin{enumerate}
\item[(A)] For each $g\in G$ there is a diffeomorphism $R_g:P\longrightarrow P$, $R_g(p)=pg$,
such that $(pg_1)g_2=p(g_1g_2)$ for all $g_1,g_2\in G$ and $p\in P$, and if $e\in G$ is the
identity element, then $pe=p$ for all $p\in P$. We require the function
$P\times G\longrightarrow P$, $(p,g)\mapsto pg$
to be a differentiable map. We suppose that if $pg=p$ for some $p\in P$ and
$g\in G$, then $g=e$. We summarize this paragraph by saying that
$\bm{G}$ \textbf{acts freely (and smoothly) on} $\bm{P}$ \textbf{to the right}.

\item[(B)] The map $\pi_P:P\longrightarrow X$ is surjective, and $\pi_P^{-1}(\pi_P(p))=
\{pg:g\in G\}$, which is, by definition, the \textbf{orbit} of $G$ through $p$. Note that
$X\cong P/G$ under the right action identification.

Note that for each fixed $p\in\pi_P^{-1}(x)$ (i.e. chosen, not canonically)
there is a diffeomorphism
\[
G\longrightarrow\pi_P^{-1}(x),
\qquad
g\longmapsto pg.
\]
Thus, \textbf{all the fibers} $\bm{\pi_P^{-1}(x)}$ \textbf{are diffeomorphic to} $\bm{G}$,
but there is no canonical identification of $\pi_P^{-1}(x)$ with $G$, and hence no
natural group structure on $\pi_P^{-1}(x)$.

\item[(C)] For each $x\in X$ there is an open set $U_\alpha$ with $x\in U_\alpha$
and a diffeomorphism
\begin{align*}
  T_\alpha:\pi_P^{-1}(U_\alpha)&\longrightarrow U_\alpha\times G\\
  p&\longmapsto \left(\pi_P(p),t_\alpha(p)\right),
\end{align*}
where $t_\alpha:\pi_P^{-1}(U_\alpha)\longrightarrow G$ has the property $t_\alpha(pg)=t_\alpha(p)g$
for all $g\in G$ and $p\in\pi_P^{-1}(U_\alpha)$ (i.e. is $\bm{G}$\textbf{-equivariant}). 
The map $T_\alpha$ is called a \textbf{local trivialization (LT)}, or 
(in \textit{physics language}) a \textbf{choice of gauge}.
\end{enumerate}
\end{definition}
Once built the PFB, we need a way of translating different gauge choices as mentioned in the
motivation of the beginning of the chapter
\begin{definition}[Transition functions \cite{Bleecker81}]
  Let $T_\alpha:\pi_P^{-1}(U_\alpha)\to U_\alpha\times G$ and
  $T_\beta:\pi_P^{-1}(U_\beta)\to U_\beta\times G$ be two LTs of a PFB $\pi_P:P\to X$ with
  Lie group $G$. The \textbf{transition function} (also called \textbf{cocycle}) from
  $T_\alpha$ to $T_\beta$ is the map
  \begin{align*}
    g_{\alpha\beta}:U_\alpha\cap U_\beta&\to G\\
    x=\pi_P(p)&\longmapsto g_{\alpha\beta}(x):=t_\alpha(p)t_\beta(p)^{-1},
  \end{align*}
  defined this way such as ensuring $T_\alpha\circ T_\beta^{-1}$ maps $(x,t_\beta(p))$ for
  some $p\in\pi_P^{-1}(U_\alpha\cap U_\beta)$ to $(x,g_{\alpha\beta}(x)t_\beta(p))$.
  Note that $g_{\alpha\beta}(x)$ is independent of the choice of $p\in\pi_P^{-1}(x)$ and that
  each trivialization induces a canonical local section of $\pi_P$ given by
  $s_\alpha:=T_\alpha^{-1}(x,e):U_\alpha\to P_x$ for every $x\in U_\alpha$. In this
  way, once chosen two trivializations, the cocycle relates their canonical local sections
  as $s_\beta(x)=s_\alpha(x)\cdot g_{\alpha\beta}(x)$ for all $x\in U_\alpha\cap U_\beta$.
\end{definition}
\textbf{The transition functions} $\bm{g_{\alpha\beta}}$ \textbf{encode how two observers'
gauge choices are related}.
\begin{remark}
  Due to the properties of $s_\alpha$ and $s_\beta$ maps defined on Definition \ref{def:PFB}
  the following is satisfied for transition functions
  \begin{enumerate}
    \item $g_{\alpha\alpha}(x)=e$ for all $x\in U_\alpha$, 
    \item $g_{\alpha\beta}(x)g_{\beta\gamma}(x)=g_{\alpha\gamma}(x)$ for all $x\in
    U_\alpha\cap U_\beta\cap U_\gamma$. 
    \end{enumerate}
  Another important remark is that $P$ itself can be entirely built from these transition
  functions. Indeed, consider an open cover $\{U_\alpha\}$ of $X$ and for each pair of open
  sets a transition function $g_{\alpha\beta}$ satisfying the two (cocycle) conditions above,
  then it can be proved (cf. \cite{sontzPrincipalBundlesClassical2019} for example) that
  defining the set $P$ as
  \begin{equation}
    P := \left(\bigsqcup_{\alpha} U_{\alpha}\times G\right)\Big/\sim,
    \qquad
    (x,g)_{\beta}\sim(x,\,g_{\alpha\beta}(x)\,g)_{\alpha}
    \quad\text{for } x\in U_{\alpha}\cap U_{\beta}.
  \end{equation}
  where the obvious projection to $X$ is compatible with $\sim$ and descent to a proper
  PFB projection $\pi_P$; then $P$ becomes a PFB once considered the quotient topology on it.
  Note that in this construction \textbf{the whole twisted nature of
  the PFB is encoded in the transition functions} used to build it. Indeed, if all transition
  functions map to the identity element $e\in G$ then $P\cong X\times G$ globally.
  This fact as already shown \textbf{is always satisfied locally} by a LT so that
  \textbf{most explicit computations may be carried out in a single trivialization}.
\end{remark}

Because the action is free and fiber-preserving, differentiating it yields, for
every $\xi\in\mathfrak{g}=T_eG$, the \textbf{fundamental} vector field
$\xi_{P}\in\mathfrak{X}(P)$,
\begin{equation}
\xi_{P}(p):=\der{}{t}\Big|_{t=0}\,\left(p\cdot\exp(t\xi)\right)\equiv\der{}{t}\Big|_{t=0}\,
c(t),
\end{equation}
which is \textbf{vertical}, $\xi_{P}(p)\in \mathcal{V}_{p}\pi_{P}$ (since due to G preserving fibers,
the curve $c(t)$ never leaves the fiber above $\pi_P(p)$ and so $\pi_P\circ c(t)$ is a
constant map). For the rigorous definition of the exponential map on Lie theory and other
basic Lie concepts the reader can consult \cite{kirillovIntroductionLieGroups2008}.

\begin{remark}\label{remark:fundamental_vector_field}
  The map
  $\mathfrak{g}\to \mathcal{V}_{p}\pi_{P}$, $\xi\mapsto\xi_{P}(p)$, is a linear isomorphism for
  every $p\in P$: \textbf{the vertical subspaces of a PFB are canonically
  copies of the Lie algebra}. 
  Hence, every vertical vector $v\in\mathcal{V}_p\pi_P$ for every $p\in P$ is of the form
  $v=\xi_P(p)$ for a unique $\xi\in\mathfrak{g}$.
\end{remark}

Throughout, $G$ is a (finite dimensional) Lie group of dimension $N$ with Lie algebra
$\mathfrak{g}=T_{e}G$, and we fix a basis $\{e_{a}\}_{a=1}^{\dim{G}=N}$ of $\mathfrak{g}$ with structure
constants $[e_{b},e_{c}]=C^{a}_{bc}\,e_{a}$
-and \textbf{not} $\mathrm{i}C^{a}_{bc}\,e_{a}$ as usual in physics; this is just a
convention that depends on choosing the generators of the algebra to be
\textit{anti-hermitian} (mathematical perspective) or \textit{hermitian} (physical
perspective, usually in QFT contexts). We will stick to the mathematical convention-.
Note that $C^a_{cb}=-C^a_{bc}$ due the skew-symmetry property of the Lie bracket and
\[
\sum_{b=1}^{N}C^a_{cb}C^b_{de}+C^a_{eb}C^b_{cd}+C^a_{db}C^b_{ec}=0,\quad\forall
a,c,d,e\in\{1,\dots,N\}
\]
due to the Jacobi identity of the Lie bracket. 

We will also use this section to introduce the crucial concept of vector-valued forms because
-as we will see in Section \ref{sect:connections}- the canonical way of introducing connections
on a PFB is equivalent to a formulation using this type of forms.
\begin{definition}[Vector-valued forms \cite{Bleecker81} and forms in a vector fiber bundle]
  \label{def:vector-valued-forms}
  Let $V$ be a vector space with basis $v_1,\dots,v_m$ and consider $\alpha^1,\dots\alpha^m
  \in\Omega^k(M)$ $k-$forms living in a differentiable manifold $M$. Then, the map $x\mapsto
  \alpha^1_x v_1+\dots\alpha^m_x v_m\in \bigwedge^kT_x^*M\otimes V$ is called a $\bm{V-}$
  \textbf{valued} $\bm{k-}$\textbf{form}. The space of $V-$valued $k-$forms is denoted by
  $\Omega^k(M;V)$.

  In general, given a vector fiber bundle $\pi:E\to M$ with model fiber the vector space $V$, an
  $\bm{E}$\textbf{-valued differential} $\bm{k}$\textbf{-form} is a smooth section of the
  tensor product bundle $E\otimes\bigwedge^k T^*M$ over $M$, i.e. an element of
  $\Gamma(E\otimes\bigwedge^k T^*M)$. This general definition reduces to the vector-valued
  one on each fiber.

  Implicitly in this definition is the fact that all usual operators on forms on $M$ are
  naturally endowed in this vector-valued structure by just acting on the forms on the
  vector basis sum decomposition; as an example, the exterior derivative acts on a given
  $\alpha\in\Omega^k(M;V)$ -that decomposes as $\alpha=\alpha^1 v_1+\dots+\alpha^m v_m$- as
  \[
    \dd\alpha=(\dd\alpha^1)v_1+\dots(\dd\alpha^m)v_m.
  \]
  Note that $\Omega^k(M;V)\cong\Omega^k(M)\otimes_{\RR} V$ as $\RR-$vector spaces hence
  notations $\alpha^k v_k$ and $\alpha^k\otimes v_k$ will be used interchangeably. 

  Examples of spaces of vector-valued forms are $\Omega(P;\mathfrak{g})$,
  $\Omega(G;\mathfrak{g})$...
\end{definition}

Is also important to recall the group and algebra-level adjoint maps, which pervade the
whole chapter (you can consult \cite{sontzPrincipalBundlesClassical2019} for a more detailed
introduction of these representations and the Maurer--Cartan form).

\begin{definition}[Adjoint representations and Maurer--Cartan form]\label{def:adjoint}
For each $g\in G$, conjugation $C_{g}\colon G\to G$, $C_{g}(h):=g\,h\,g^{-1}$, fixes
the identity $e$; its differential there is the \textbf{adjoint representation of the
group},
\begin{equation}\label{eq:Ad_def}
\Ad\colon G\longrightarrow GL(\mathfrak{g}),\qquad
\Ad_{g}:=T_{e}C_{g},
\end{equation}
which is a Lie group homomorphism. Differentiating $\Ad$ at $e$ gives the
\textbf{adjoint representation of the algebra},
\begin{equation}\label{eq:ad_def}
\ad\colon\mathfrak{g}\longrightarrow\mathfrak{gl}(\mathfrak{g}),\qquad
\ad_{\xi}(\zeta)=[\xi,\zeta].
\end{equation}
Finally, the \textbf{(left)
Maurer--Cartan form} $\theta_{G}\in\Omega^{1}(G;\mathfrak{g})$ is defined by
$\theta_{G}|_{g}:=(L_{g^{-1}})_{*,g}$, i.e.\ it left-translates a tangent vector on $g$ back
to $e$. For a matrix group these are simply
$\Ad_{g}\zeta=g\,\zeta\,g^{-1}$, $\ad_{\xi}\zeta=[\xi,\zeta]$ (the matrix commutator) and
$\theta_{G}=g^{-1}\dd g$ (note that things are trivial if the Lie group is abelian).
The fact that the adjoint representation coincides with the bracket of the Lie algebra,
i.e., that $\ad_{\xi}(\zeta)=[\xi,\zeta]$ is proved in Theorem 0.3.9 of \cite{Bleecker81}.
\end{definition}

\section{Connections, the connection form and the adjoint bundle}\label{sect:connections}
A connection is a coherent, $G$-invariant prescription of ``horizontality'' that
allows one to differentiate and to compare fibers.

\begin{definition}[Principal connection]\label{def:connection}
A \textbf{principal connection} (or PFB connection) on $P$ is a $G$-invariant distribution
$H\subset TP$ complementary to the vertical one, i.e., 
\begin{equation}
T_{p}P=\mathcal{V}_{p}\pi_{P}\oplus H_{p},\qquad
(R_{g})_{*}H_{p}=H_{pg}\quad\forall g\in G,
\end{equation}
called the \textbf{horizontal} distribution. Equivalently (cf.
\cite{sontzPrincipalBundlesClassical2019}), a connection is a
$\mathfrak{g}$-valued $1$-form $\omega\in\Omega^{1}(P;\mathfrak{g})$, the
\textbf{connection form}, such that
\begin{enumerate}
  \item $\omega(\xi_{P})=\xi$ for all $\xi\in\mathfrak{g}$
        (it reproduces the fundamental fields), and
  \item $R_{g}^{*}\omega=\Ad_{g^{-1}}\circ\,\omega$ for all $g\in G$
        (it is $\bm{\mathrm{Ad}}$\textbf{-equivariant}),
\end{enumerate}
the two viewpoints being related by $H_{p}=\ker\omega_{p}$, and from the form perspective
of the PFB connection $H_p$ is defined as the set of vectors annihilated by $\omega$.
\end{definition}
\begin{remark}
  Any PFB admits a connection. C.f. \cite{kobayashiFoundationsDifferentialGeometry1996} for
  a proof of this fact.
\end{remark}
A principal connection is thus a very particular Ehresmann connection on
$\pi_{P}$; the two extra conditions above are exactly the $G$-invariance built into
Definition~\ref{def:connection}. Given a local section
$s_{\alpha}\colon U_{\alpha}\subseteq X\to P$ (which naturally induces a local trivialization
$T_\alpha:\pi^{-1}(U)\to U\times G$ by
$T_\alpha(s_\alpha(x)g):=(x,g)$), the connection is
described on $U_{\alpha}$ by its pullback; this is what physicist call the
\textbf{gauge potential}
\begin{equation}
A_{\alpha}:=s_{\alpha}^{*}\omega\in\Omega^{1}(U_{\alpha};\mathfrak{g}),\qquad
A_{\alpha}=A^{a}_{\mu}(x)\,\dd x^{\mu}\otimes e_{a},\quad x\in U_\alpha.
\end{equation}
On an overlap $U_{\alpha}\cap U_{\beta}$, writing $g:=g_{\alpha\beta}$ and denoting
by $\theta_{G}$ the (left) Maurer--Cartan form of $G$, the potentials are related by
(cf. \cite{Bleecker81} recalling that $s_\beta(x)=s_\alpha(x)\cdot g_{\alpha\beta}(x)$)
\begin{equation}\label{eq:gauge_transf_A}
A_{\beta}=\Ad_{g^{-1}}\circ A_{\alpha}+g^{*}\theta_{G},
\end{equation}
which for a matrix group (physics `notation') reads $A_{\beta}=g^{-1}A_{\alpha}\,g+g^{-1}\dd g$. The
inhomogeneous Maurer--Cartan term $g^{-1}\dd g$ is what prevents $A$ from being
tensorial and it is the origin of the covariant derivative below.
Note that in the case of $G=\mathrm{U}(1)$ (abelian), we get (up to the $\mathrm{i}$ dependent
of the convention on the Lie algebra basis mentioned just before Definition
\ref{def:vector-valued-forms}) the usual $A^\prime=A+\nabla\alpha$, so things seem to be going on the
right track.

\begin{remark}
The affine nature already appears here: recalling Remark \ref{remark:fundamental_vector_field}
is easy to show that the difference of two connection forms
$\omega_{1}-\omega_{2}$ is horizontal (i.e. it vanishes upon vertical vector fields)
and $\Ad$-equivariant, which is called \textbf{tensorial}, so the space of connections is
an affine space modelled on the vector space of tensorial $\mathfrak{g}$-valued $1$-forms.
This is the germ of Theorem \ref{thm:connections} below.
\end{remark}

Physicists work always with the gauge potential on a given neighborhood of spacetime, so
one could ask when does a form on $\Omega(P;\mathfrak{g})$ projects down to a proper form
defined on the whole spacetime $X$ and not just locally. For this, we have to introduce
first the concept of an \textit{associated vector bundle} to a PFB $P$ and a representation
$\rho:G\to\mathrm{GL}(V)$ of $G$ over a $\mathbb{K}-$vector space $V$ (for $\mathbb{K}$
either $\CC$ or $\RR$); the \textit{adjoint bundle} is the instance relevant here.

Recall that a representation such as $\rho$ is a group homomorphism hence $\rho(g_1g_2)=
\rho(g_1)\circ\rho(g_2)$; we also have that $\rho(g^{-1})=\rho(g)^{-1}$.

\begin{theorem}[\cite{hamiltonMathematicalGaugeTheory2018}]
\label{thm:associated}
Let $P$ be a PFB with Lie group $G$ and $\rho:G\to\mathrm{GL}(V)$ representation of a
$\mathbb{K}-$vector space. Consider the right action of $G$ over $P\times V$ as
\begin{align*}
  (P\times V)\times G&\longrightarrow P\times V\\
  (p,v,g)&\longmapsto (p,v)\cdot g:=(pg,\rho(g)^{-1}v)
\end{align*}
and then the quotient space by that right action $(P\times V)/G:=E$. If we consider
the canonical projection $\pi_E:E\to X$ by $\pi_E((p,v))=\pi_P(p)$ then $(E,\pi_E)$ is a
$\mathbb{K}-$vector bundle over $X$ with fibers $E_x\cong V$.
\end{theorem}

\begin{definition}[Associated vector bundles and the adjoint bundle]\label{def:associated}
$E=(P\times V)/G=:P\times_\rho V$ with the projection over $X$ $\pi_E$ defined on Theorem
\ref{thm:associated} is called the \textbf{vector bundle associated} to the PFB $P$ and the
representation $\rho$ on $V$.

In particular, taking $V=\mathfrak{g}$ with the adjoint representation $\Ad$ of $G$
\eqref{eq:Ad_def} yields the so-called \textbf{adjoint bundle} $\ad P:=P\times_{\Ad}
\mathfrak{g}$.
\end{definition}

Note that we have
\begin{equation}\label{eq:adP}
\ad P:=P\times_{\Ad}\mathfrak{g}\;\cong\;(\mathcal{V}\pi_{P})/G,
\end{equation}
under the identification given by the action $\xi_P\cdot g = (R_g)_*(\xi_P)$ on $\mathcal{V}\pi_{P}$. To see this
consider the isomorphism $\mathfrak{g}\cong\mathcal{V}_p\pi_P$ 
of Remark \ref{remark:fundamental_vector_field} once a $p\in P$ is fixed, then $P\times
\mathfrak{g}\cong\mathcal{V}\pi_P$.
As the adjoint bundle relates $(p,\xi)\sim(pg,\Ad_{g^{-1}}\xi)$, we just have to prove
that 
\begin{equation}
  \label{eq:adjoint_bundle_is_vertical_equivariant_vector_fields_on_P}
  (R_g)_*(\xi_P(p))=(\Ad_{g^{-1}}\xi)_P(pg),
\end{equation}
which is done on Equation \eqref{eq:proof_adjoint_bundle_is_vertical_proof} of
Appendix \ref{app:some_proofs_on_gauge_theory}.

After all this machinery of associated vector bundles to a PFB and a representation of the
fiber vector space, we have the following remark.

\begin{remark}[Which $P-$forms become a legitimate (physical) object on X? Horizontal and
equivariant!]
\label{rem:physical_forms}
Imagine we have a $k-$form $\Omega\in\Omega^k(P;\mathfrak{g})$ with values in $\mathfrak{g}$
and we want to see if this form can be projected down to $X$ to a well defined object
(i.e. not depending on information on the fiber above each point on $X$).
Evaluating this $k$-form on $X$ at $x$, on vectors $v_1,\dots,v_k\in T_xX$, requires
in general two arbitrary choices on $P$: a point $p\in\pi_P^{-1}(x)$, and, for each $v_i$,
a \textbf{lift} $w_i\in T_pP$ that projects into $v_i$, i.e., $\dd_p\pi_P(w_i)=v_i$
(any two such lifts differ by a vertical vector, by Remark
\ref{remark:fundamental_vector_field}).

Horizontality removes the dependence on the lift: if $w_i'=w_i+\xi_P(p)$ is another lift of
the same $v_i$, multilinearity and $\Omega(\xi_P,\cdot)=0$ give
\[
\Omega_p(w_1,\dots,w_i,\dots)-\Omega_p(w_1,\dots,w_i',\dots)
=\Omega_p(w_1,\dots,\xi_P(p),\dots)=0,
\]
so the value depends only on $v_1,\dots,v_k$, not on the chosen lifts.

Equivariance then removes the dependence on the point $p\in\pi_P^{-1}(x)$: if
$(R_g)_*w_i$ lifts $v_i$ at $p\cdot g$, then $\Omega_{p\cdot g}((R_g)_*w_1,\dots)
=\Ad_{g^{-1}}\Omega_p(w_1,\dots)$, so the class $[p,\Omega_p(w_1,\dots,w_k)]$ in the
associated bundle $P\times_{\Ad}\mathfrak{g}$ (defined in Section
\ref{sec:bundle_of_connections}) is independent of $p$ as well. Together,
both choices are eliminated and $\bm{\Omega}\in\Omega^k(P;\mathfrak{g})$ \textbf{defines a
genuine} $\bm{\mathrm{ad} P}$\textbf{-valued} $\bm{k}$\textbf{-form} on $X$,
$\Omega\in\Omega^k(X;\ad P)$. This is exactly
where $\omega$ fails: without horizontality, changing the
lift changes the value by $-\xi$ (from $\omega(\xi_P)=\xi$), so ``$\omega$ evaluated
at $x\in X$'' is simply not well defined -only $A_\alpha=s_\alpha^*\omega$, tied to a
choice of section, makes sense, and the inhomogeneous term in
\eqref{eq:gauge_transf_A} is exactly the bookkeeping of this lift-dependence when
changing from $s_\alpha$ to $s_\beta$.
\end{remark}

\section{Curvature, the field strength and the exterior covariant derivative}\label{sec:curvature}

The failure of the horizontal distribution to be involutive is measured by the
curvature, which is the geometric version of the (local) field strength used by
physicists.

Let $\omega$ be a PFB connection, then every $X\in\mathfrak{X}(P)$ can be decomposed as
$X=X^V+X^H$ where $X^V\in\mathcal{V}\pi_P$ and $\omega(X^H)=0$ (note that this decomposition
depends clearly on $\omega$) and so we define

\begin{definition}[Horizontal part of a $\mathfrak{g}-$valued form and exterior covariant derivative]
  Let $\omega$ be a PFB connection and $\varphi\in\Omega^k(P;\mathfrak{g})$ for any
  $k\in\{0,\dots,(n+1)+\dim G\}$, we define the \textbf{horizontal part} of
  $\varphi$, $\varphi^H\in\Omega^k(P,\mathfrak{g})$ as
  \[
    \varphi^H(X_1,\dots,X_k)=\varphi(X_1^H,\dots,X_k^H),\quad\forall X_1,\dots,X_k\in\mathfrak{X}(P).
  \]

  Once defined this, we also define the \textbf{exterior covariant derivative} as the
  derivation $\bm{D^\omega}$ acting on $\varphi$ as $D^\omega\varphi:=(\dd\varphi)^H
  \in\Omega^{k+1}(P,\mathfrak{g})$. Clearly
  $D^\omega$ depends on the connection form chosen, but sometimes the superscript $^\omega$
  is omitted in the literature.
\end{definition}
With the exterior covariant derivative, we define the curvature of the PFB connection
\begin{definition}[curvature of a PFB connection]
  Given a PFB connection $\omega$, we define its \textbf{curvature} $\bm{\Omega^\omega}$ as
  the $2-$form with values in $\mathfrak{g}$ given by $\Omega^\omega:=D^\omega\omega
  \in\Omega^2(P;\mathfrak{g})$.
\end{definition}
A simpler way to compute the curvature of a PFB connection arises once introduced
the \textbf{canonical Lie bracket of} $\bm{\mathfrak{g}-}$\textbf{valued forms}:
let $\varphi\in\Omega^j(P;\mathfrak{g})$,
$\psi\in\Omega^k(P;\mathfrak{g})$ we define the bracket of $\mathfrak{g}-$valued forms as
\[
  [\varphi,\psi]^\wedge(X_1,\dots,X_{j+k}):=\frac{1}{j!k!}\sum_{\sigma}(-1)^\sigma\left[
    \varphi(X_{\sigma(1)},\dots,X_{\sigma(j)}),\psi(X_{\sigma(j+1)},\dots,X_{\sigma(j+k)})
  \right],
\]
where the bracket in the rhs is the usual Lie bracket on $\mathfrak{g}$ (the $\wedge$
superscript of the bracket will almost always be omitted, only left to distinguish with
the usual bracket if necessary). 
In particular for two $1-$forms, we have $[\varphi,\psi](X,Y)=[\varphi(X),\psi(Y)]-
[\varphi(Y),\psi(X)]$.
Then, one can prove that

\begin{theorem}[Structure equation]
  \label{thm:structure_equation}
  Given a PFB connection $\omega$, its curvature form $\Omega^\omega$ is given by
  \begin{equation}
    \label{eq:structure_equation}
    \Omega^\omega=\dd\omega+\frac{1}{2}[\omega,\omega].
  \end{equation}
\end{theorem}
\begin{proof}
  See \cite{hamiltonMathematicalGaugeTheory2018}.
\end{proof}

In the same way that a local section $s_\alpha:U_\alpha\subseteq X\to P$ (which uniquely
determines a trivialization) pulled the connection to the gauge potential, we pull now the
curvature to obtain
\begin{definition}[Field strength]
  Let $\Omega^\omega$ be the curvature of a PFB connection $\omega$ and $s_\alpha$ a local
  section on $\pi_P$. We define the \textbf{field strength} $\bm{F_\alpha}:=s_\alpha^*
  \Omega^\omega\in\Omega^2(U_\alpha;\mathfrak{g})$ which by Theorem
  \ref{thm:structure_equation} is equal to
  \[
    F_\alpha=dA_\alpha+\frac{1}{2}[A_\alpha,A_\alpha].
  \]
  Using that locally $A_\alpha=A^a_\mu(x)\,\dd x^\mu\otimes e_a$, is easy to compute that the
  field strength is locally given by
  \begin{equation}
    \label{eq:F_components}
    F_\alpha=\frac{1}{2}F_{\mu\nu}^a(x)\dd x^\mu\wedge\dd x^\nu\otimes e_a,\quad
    F_{\mu\nu}^a(x):=\partial_\mu A^a_\nu(x)-\partial_\nu A^a_\mu(x)+C^a_{bc}A^b_\mu(x)
    A^c_\nu(x),
  \end{equation}
  with $F^a_{\nu\mu}=-F^a_{\mu\nu}$.
\end{definition}
Differentiating \eqref{eq:structure_equation} and using
$[\omega\wedge[\omega\wedge\omega]]=0$ (the graded Jacobi identity) yields the
\textbf{Bianchi identity} (cf. \cite{hamiltonMathematicalGaugeTheory2018} for a proof)
\begin{equation}\label{eq:bianchi}
D^\omega\Omega^\omega=\dd\Omega^\omega+[\omega,\Omega^\omega]=0,
\end{equation}
which of course also implies $D_{A_\alpha} F_\alpha:=s_\alpha^*D^\omega \Omega^\omega=
\dd F_\alpha+[A_\alpha,F_\alpha]=0$.

Under \eqref{eq:gauge_transf_A}, noting that Lie bracket on $\mathfrak{g}-$valued forms
is preserved under pullback (in particular $(R_g)_*$), the field strength transforms
tensorially, without any inhomogeneous term,
\begin{equation}\label{eq:gauge_transf_F}
F_{\beta}=\Ad_{g^{-1}}\circ F_{\alpha}
\quad\Big(\text{matrix groups: } F_{\beta}=g^{-1}F_{\alpha}\,g\Big).
\end{equation}

\begin{remark}[Electromagnetism]
For $G=U(1)$ the algebra is abelian, $C^{a}_{bc}=0$, so \eqref{eq:F_components}
reduces to $F_{\mu\nu}=\partial_{\mu}A_{\nu}-\partial_{\nu}A_{\mu}$,
\eqref{eq:gauge_transf_A} to $A_{\beta}=A_{\alpha}+\dd\lambda$ and $D^\omega$ to $\dd$:
one recovers Maxwell's electromagnetic tensor and the familiar $U(1)$ gauge freedom,
while the Bianchi identity \eqref{eq:bianchi} becomes the homogeneous Maxwell equations
$\dd F=0$ (i.e. $\partial_{[\rho}F_{\mu\nu]}=0)$. The non-abelian novelty is entirely encoded in the quadratic term
$C^{a}_{bc}A^{b}_{\mu}A^{c}_{\nu}$,
which makes the gauge field self-interacting.
\end{remark}

\begin{remark}[The key first-order fact]\label{rem:antisymmetry}
The field strength \eqref{eq:F_components} depends on the first derivatives of the
potential \textbf{only through the antisymmetric combination}
$\partial_{[\mu}A^{a}_{\nu]}$. This innocent-looking observation is the root of the
singular (almost regular) character of the Yang--Mills Lagrangian, as we shall see
in Section \ref{sec:YM_lagrangian}.
\end{remark}

\subsection*{The exterior covariant derivative for a general representation}
\label{subsec:D_general_rep}
 
Everything above was written for $\mathfrak{g}$-valued forms, but neither the
definition of $\varphi^{H}$ nor that of $D^{\omega}$ used anything about
$\mathfrak{g}$ beyond it being a vector space. The natural generality --- and the
one we shall actually need in Chapter
\ref{chap:multisymplectic_yang_mills_computations}, where the multimomenta take
values in $\ad^{*}P$ and the multivector fields in $\bigwedge^{k}\ad P$ --- is
that of an arbitrary representation.
 
\begin{definition}[Tensorial forms of type $\rho$]
  \label{def:tensorial_forms}
  Let $(V,\rho)$ be a finite-dimensional representation of $G$ and let
  $\rho_{*}:\mathfrak{g}\to\mathfrak{gl}(V)$ denote the induced representation of
  $\mathfrak{g}$ (the differential of $\rho$ at the identity). A form $\varphi\in\Omega^{k}(P;V)$ is called
  \textbf{tensorial of type} $\bm{\rho}$ if it is
  \begin{enumerate}
    \item \emph{horizontal}: $\iota_{X}\varphi=0$ for every $X\in\mathcal{V}\pi_{P}$
          (equivalently $\varphi^{H}=\varphi$), and
    \item \emph{equivariant}: $R_{g}^{*}\varphi=\rho(g^{-1})\circ\varphi$ for all $g\in G$.
  \end{enumerate}
  We write $\Omega^{k}_{\rho}(P;V)$ for the space of such forms. In an equivalent way of
  the one presented in Remark \ref{rem:physical_forms} the pullback along local sections
  gives a canonical isomorphism
  \begin{equation}
    \label{eq:tensorial_iso}
    \Omega^{k}_{\rho}(P;V)\;\cong\;\Omega^{k}\bigl(X;P\times_{\rho}V\bigr),
  \end{equation}
  so a tensorial form of type $\rho$ \emph{is} a form on the base with values in the
  associated bundle.
\end{definition}
 
The definition $D^{\omega}\varphi:=(\dd\varphi)^{H}$ makes sense verbatim for
$V$-valued forms, and on tensorial forms it becomes computable; it can be proved that
the following equation is satisfied for tensorial forms of type $\rho$:
  \begin{equation}
  \label{eq:exterior_covariant_derivative_on_tensorial_forms}
    D^{\omega}\varphi=\dd\varphi+\rho_{*}(\omega)\wedge\varphi,
  \end{equation}
  where
  \[
    \bigl(\rho_{*}(\omega)\wedge\varphi\bigr)(X_{0},\dots,X_{k})
    :=\sum_{i=0}^{k}(-1)^{i}\,
      \rho_{*}\bigl(\omega(X_{i})\bigr)\,
      \varphi(X_{0},\dots,\widehat{X_{i}},\dots,X_{k}).
  \]
Note that the connection form $\omega$ is \textbf{not} horizontal, hence not tensorial
and then \eqref{eq:exterior_covariant_derivative_on_tensorial_forms} is not applicable
(recall that we have \eqref{eq:structure_equation} nevertheless).
\subsubsection*{The two representations we shall use}
 
\paragraph{Adjoint, $V=\mathfrak{g}$, $\rho=\Ad$.}
Here $\rho_{*}(\xi)=\ad_{\xi}=[\xi,\cdot\,]$ and
\eqref{eq:exterior_covariant_derivative_on_tensorial_forms} reduces to
$D^{\omega}\varphi=\dd\varphi+[\omega,\varphi]$, the case treated above; the
associated bundle is $\ad P:=P\times_{\Ad}\mathfrak{g}$. In components, for
$\varphi=\varphi^{a}\otimes e_{a}$,
\begin{equation}
  \label{eq:D_adjoint_components}
  (D_{A}\varphi)^{a}=\dd\varphi^{a}+C^{a}_{bc}\,A^{b}\wedge\varphi^{c}.
\end{equation}
 
\paragraph{Coadjoint, $V=\mathfrak{g}^{*}$, $\rho=\Ad^{*}$.}
With $\langle\Ad^{*}_{g}\mu,\eta\rangle:=\langle\mu,\Ad_{g^{-1}}\eta\rangle$ one gets
$\langle\rho_{*}(\xi)\mu,\eta\rangle=-\langle\mu,[\xi,\eta]\rangle$, so for
$\mu=\mu_{a}\otimes e^{a}$
\begin{equation}
  \label{eq:D_coadjoint_components}
  (D_{A}\mu)_{a}=\dd\mu_{a}+C^{c}_{ab}\,A^{b}\wedge\mu_{c}.
\end{equation}
Note the position of the structure index: it is \emph{down} on the object and the
constants appear as $C^{c}_{ab}$. This is precisely the derivative that acts on the
multimomenta $p_{a}{}^{\mu\nu}$ in the Yang--Mills field equations
$D_{\nu}p_{a}{}^{\mu\nu}=0$ as we will see on Chapter
\ref{chap:multisymplectic_yang_mills_computations}.
This behaviour is easily generalizable to $\wedge^j\mathfrak{g}\otimes
\wedge^k\mathfrak{g}^*$ which in components reads
\begin{equation}
  \label{eq:D_general_case}  
  (D_\mu\zeta)^{a_1\dots,a_j}_{b_1,\dots,b_k}=\partial_\mu
  \zeta^{a_1\dots,a_j}_{b_1,\dots,b_k}+\sum_{l=1}^jC^{a_l}_{cd}A^c_\mu
  \zeta^{a_1\dots,d,\dots,a_j}_{b_1,\dots,b_k}+\sum_{m=1}^kC^{c}_{b_md}A^d_\mu
  \zeta^{a_1,\dots,a_j}_{b_1,\dots,c,\dots,b_k}.
\end{equation}

\section{The bundle of connections}\label{sec:bundle_of_connections}

We now perform the identification that lets the whole apparatus of
Chapter~\ref{chap:multisymplectic} act on gauge fields: realizing connections as
sections of a bundle over $X$.

Let $J^{1}P$ be the first order jet bundle of $\pi_{P}\colon P\to X$, with source
projection $\pi_{1,0}\colon J^{1}P\to P$ and target projection $\pi_{1}\colon
J^{1}P\to X$. The right action of $G$ on $P$ prolongs canonically to a right action
on $J^{1}P$ by $g\cdot j^{1}_{x}s:=j^{1}_{x}(R_{g}\circ s)$, and so we can consider
quotienting by $G$. This quotient bundle will be our configuration bundle
\begin{definition}[Bundle of connections]\label{def:bundle_connections}
The \textbf{bundle of connections} of $P$ is the quotient
\begin{equation}
C(P):=J^{1}P/G,
\end{equation}
which is a fiber bundle $\pi\colon C(P)\to X$ over $X$ (we will see in a moment why do we
use exactly $\pi$ for the projection like $\pi:Y\to X)$.
\end{definition}

The two projections fit into the commutative diagram
\begin{equation*}
\begin{tikzcd}[row sep=2.4em, column sep=3.2em]
J^{1}P \arrow[r, "\rho"] \arrow[d, "\pi_{1,0}"'] & C(P)=J^{1}P/G \arrow[d, "\pi"] \\
P \arrow[r, "\pi_{P}"'] & X
\end{tikzcd}
\end{equation*}
where $\rho$ is the quotient projection.

The fundamental result (see for example \cite{Garcia77,CLMM01,KMS93}), is
the following.

\begin{theorem}[Connections as sections]\label{thm:connections}
There is a one-to-one correspondence between principal connections on $P$ and
(global) sections of $\pi\colon C(P)\to X$. Moreover $\pi\colon C(P)\to X$ is an
\textbf{affine} bundle modelled on the vector bundle $T^{*}X\otimes\ad P\to X$, with
$\ad P$ the adjoint bundle \eqref{eq:adP}.
\end{theorem}

Comparing with Chapter~\ref{chap:multisymplectic}, this affine bundle structure is
\textbf{exactly} the one already encountered there: recall that $\pi_{ZY}\colon Z\to
Y$ was shown to be an affine bundle modelled on $\pi^{*}T^{*}X\otimes_{Y}\mathcal{V}\pi$
in Chapter~\ref{chap:multisymplectic}. Here $C(P)\to X$ is affine over
$T^{*}X\otimes\ad P$ in precisely the same fashion. This is no coincidence: the
adjoint bundle plays, for the gauge potential, the role that the vertical bundle
$\mathcal{V}\pi$ played for a general field.

In view of Theorem~\ref{thm:connections} we \textbf{set}
\begin{equation}\label{eq:identification}
\boxed{\;Y:=C(P)=J^{1}P/G\;}\qquad\Longrightarrow\qquad
Z:=J^{1}\pi=J^{1}\bigl(C(P)\bigr),
\end{equation}
so that a gauge field $\omega$ is (identified with) a section $\phi\in\Gamma(\pi)$ and all of
Chapters~\ref{chap:multisymplectic}--\ref{chap:symmetries_and_Noether_theorems}
applies verbatim. In a local trivialization the fiber coordinates on $Y=C(P)$ are
the potential components $A^{a}_{\mu}$, so the general fiber index of
Chapter~\ref{chap:multisymplectic} becomes a pair $i\leftrightarrow(a,\mu)$ and the
fiber dimension is $m=(n+1)\dim\mathfrak{g}$; the jet coordinates on
$Z=J^{1}(C(P))$ are the derivatives $\partial_{\nu}A^{a}_{\mu}$. Explicitly the
coordinate dictionary with Chapter~\ref{chap:multisymplectic} is
\begin{equation}\label{eq:dictionary}
y^{i}\;\longleftrightarrow\;A^{a}_{\mu},\qquad
z^{i}_{\nu}\;\longleftrightarrow\;A_{\mu\nu}^a,
\end{equation}
and this last pulls to $\parder{A^{a}_{\mu}}{x^{\nu}}$ by a $\pi-$section.
Under \eqref{eq:dictionary} the field strength \eqref{eq:F_components} defines a
bundle map (affine along the fibers of $\pi_{ZY}$, quadratic along $Y$)
\begin{equation}
\mathcal{F}\colon Z=J^{1}\bigl(C(P)\bigr)\longrightarrow
\textstyle\bigwedge^{2}T^{*}X\otimes\ad P,\qquad
\mathcal{F}(j^{1}_{x}A)=F,
\end{equation}
which, by Remark~\ref{rem:antisymmetry}, factors through the antisymmetrization of
the jet coordinates $z^{i}_{\nu}\leftrightarrow\partial_{\nu}A^{a}_{\mu}$.

\section{The gauge group and gauge symmetries}\label{sec:gauge_group}
We finally translate ``gauge transformation'' into the language of symmetries of
Chapter \ref{chap:symmetries_and_Noether_theorems}; because of that, this section is
specially relevant to this work.

Let us begin motivating the concept of a \textit{gauge transformation}. A gauge transformation
should be a diffeomorphism between two physically equivalent internal frames, hence it should
act along the same fiber of $\pi_P$ (for each base $x\in X$ point). These two requirements
(working on fibers and not moving base points) translate into the following definition
\begin{definition}[Gauge group]\label{def:gauge_group}
The \textbf{group of gauge transformations} of a PFB $\pi_P:P\to X$ is the group of
vertical and $G-$equivariant diffeomorphism of $P$, i.e., 
\[
  \Gau(P):=\{\Phi\in\mathrm{Diff}(P)\,|\,\Phi(pg)=\Phi(p)g\,\forall p\in P\,\forall g\in G,\quad
  \pi_P\circ\Phi=\pi_P\}.
\]
\end{definition}
Consider now a $1-$parameter family of diffeomorphisms $\Phi_t$ such that
$\Phi_0=\mathrm{id}_P$.
For each fixed $p\in P$, consider the curve $\Phi_t(p)$ on $P$, then
$\xi(p):=\der{}{t}\big|_{t=0}\Phi_t(p)$ defines a vector on $T_p P$ and extending to the
whole $P$, a vector field on $P$, i.e., $\xi\in\mathfrak{X}(P)$. By this construction $\xi$
resembles an infinitesimal gauge transformation (note that infinitesimal transformations
are the type of symmetries we have considered on Chapter
\ref{chap:symmetries_and_Noether_theorems}). From the properties of the $\Phi_t$ we infer
that $\xi$ is vertical and $G-$invariant (see Appendix
\ref{app:some_proofs_on_gauge_theory}), and as this forms a subset of $\mathfrak{X}(P)$,
we define
\begin{definition}[Infinitesimal gauge transformations]
  We define the Lie algebra of \textbf{infinitesimal gauge transformations} as
  \[
    \mathfrak{gau}(P):=\{\xi\in\mathfrak{X}(P)\,|\,\dd\pi_P(\xi)=0,\quad (R_g)_*\xi=\xi\,
    \forall g\in G\}.
  \]
\end{definition}
Is easy to see then that $\mathfrak{gau}(P)\cong\Gamma(\mathcal{V}\pi_P/G)\cong\Gamma(\ad P)$.

We now define an (right, active) action of $\Gau(P)$ over PFB connections given by
\[
  \Phi\cdot\omega:=\Phi^*\omega;
\]
note that $\Phi\cdot\omega$ is a PFB connection (a fact proven in Appendix
\ref{app:some_proofs_on_gauge_theory}).

\subsection*{The gauge group in local terms}
 
Definition~\ref{def:gauge_group} is intrinsic but not operative: it says nothing about
how a gauge transformation looks in a local trivialization. The following elementary
result supplies the missing dictionary, and it is what allows us to write the local
transformation law \eqref{eq:gauge_transf_A_active} below.
 
\begin{proposition}[Gauge transformations as equivariant maps]
\label{prop:gauge_group_equivariant_maps}
Let
\[
  C^{\infty}_{\Ad}(P,G):=\bigl\{\tau\in C^{\infty}(P,G)\;\big|\;
  \tau(pg)=g^{-1}\tau(p)g=C_{g^{-1}}\!\bigl(\tau(p)\bigr)
  \ \ \forall p\in P,\ \forall g\in G\bigr\}.
\]
Then
\begin{equation}\label{eq:gauge_tau}
  \tau\;\longmapsto\;\Phi_{\tau},\qquad \Phi_{\tau}(p):=p\cdot\tau(p),
\end{equation}
is a group isomorphism $C^{\infty}_{\Ad}(P,G)\to\Gau(P)$, where
$C^{\infty}_{\Ad}(P,G)$ carries the pointwise product; explicitly
$\Phi_{\tau}\circ\Phi_{\sigma}=\Phi_{\tau\sigma}$ and
$\Phi_{\tau}^{-1}=\Phi_{\tau^{-1}}$.
\end{proposition}
 
\begin{proof}
Let $\Phi\in\Gau(P)$. Since $\pi_{P}\circ\Phi=\pi_{P}$, the points $p$ and $\Phi(p)$
lie in the same fiber; as $G$ acts freely and transitively on each fiber there is a
\emph{unique} $\tau(p)\in G$ with $\Phi(p)=p\cdot\tau(p)$, and $\tau$ is smooth because
$\Phi$ is and the action is free and proper. Equivariance of $\Phi$ gives
\[
  p\,\tau(p)\,g=\Phi(p)\,g=\Phi(pg)=pg\,\tau(pg)
  \quad\Longrightarrow\quad
  \tau(pg)=g^{-1}\tau(p)g ,
\]
cancelling $p$ by freeness. Conversely, for $\tau\in C^{\infty}_{\Ad}(P,G)$ the map
$\Phi_{\tau}$ satisfies $\pi_{P}\circ\Phi_{\tau}=\pi_{P}$ and
\[
  \Phi_{\tau}(pg)=pg\,\tau(pg)=pg\,g^{-1}\tau(p)g=p\,\tau(p)\,g=\Phi_{\tau}(p)\,g ,
\]
so it is vertical and $G$-equivariant. Finally
\[
  (\Phi_{\tau}\circ\Phi_{\sigma})(p)=\Phi_{\tau}\bigl(p\,\sigma(p)\bigr)
  =p\,\sigma(p)\,\tau\bigl(p\,\sigma(p)\bigr)
  =p\,\sigma(p)\,\sigma(p)^{-1}\tau(p)\,\sigma(p)
  =p\,\bigl(\tau(p)\sigma(p)\bigr),
\]
which is $\Phi_{\tau\sigma}(p)$; taking $\sigma=\tau^{-1}$ gives the inverse.
\end{proof}
 
Given a local section $s_{\alpha}\colon U_{\alpha}\to P$ we call
\begin{equation}\label{eq:local_tau}
  \tau_{\alpha}:=\tau\circ s_{\alpha}\colon U_{\alpha}\longrightarrow G
\end{equation}
the \textbf{local gauge transformation} determined by $\Phi_{\tau}$ in the frame
$s_{\alpha}$. This is the object physicists actually write.
 
\begin{proposition}[Pullback of a connection by a gauge transformation]
\label{prop:pullback_omega_gauge}
Let $\omega$ be a principal connection and $\Phi_{\tau}\in\Gau(P)$. Then
$\Phi_{\tau}^{*}\omega$ is again a principal connection, and
\begin{equation}\label{eq:pullback_omega_gauge}
  \Phi_{\tau}^{*}\omega=\Ad_{\tau^{-1}}\!\circ\,\omega+\tau^{*}\theta_{G}
  \ \in\ \Omega^{1}(P;\mathfrak{g}).
\end{equation}
Consequently, in the frame $s_{\alpha}$ the gauge potential transforms as
\begin{equation}\label{eq:gauge_transf_A_active}
  A_{\alpha}\;\longmapsto\;
  s_{\alpha}^{*}\bigl(\Phi_{\tau}^{*}\omega\bigr)
  =\Ad_{\tau_{\alpha}^{-1}}\!\circ A_{\alpha}+\tau_{\alpha}^{*}\theta_{G}
  \qquad\Bigl(\text{matrix groups: }
  \tau_{\alpha}^{-1}A_{\alpha}\tau_{\alpha}+\tau_{\alpha}^{-1}\dd\tau_{\alpha}\Bigr).
\end{equation}
\end{proposition}
\begin{proof}
  See Appendix \ref{app:some_proofs_on_gauge_theory}.
\end{proof} 

\begin{remark}[One computation, two readings]\label{rem:one_computation}
The proof of Proposition~\ref{prop:pullback_omega_gauge} is, line by line,
the same computation that produces the transition law
\eqref{eq:gauge_transf_A}. In both cases one differentiates a map of the
form
\[
  x\longmapsto \mu\bigl(\sigma(x),\,h(x)\bigr)\in P ,
  \qquad \sigma\in\Gamma(U,P),\quad h\in C^{\infty}(U,G),
\]
and the Leibniz rule for the action $\mu$ splits the result into a
$P$-derivative, which $\omega$ converts into $\Ad_{h^{-1}}$, and a
$G$-derivative, which $\omega$ converts into $h^{*}\theta_{G}$. The
\emph{only} difference is the meaning of $h$:
\begin{itemize}
  \item in \eqref{eq:gauge_transf_A}, $h=g_{\alpha\beta}$ compares
        \textbf{two frames} $s_{\alpha},s_{\beta}$ for \textbf{one}
        connection --- this is the \textbf{passive} reading. We are
        describing the same point in $C(P)$ from two different perspectives;
  \item in \eqref{eq:gauge_transf_A_active}, $h=\tau_{\alpha}$ is read in
        \textbf{one frame} $s_{\alpha}$ and \textbf{moves} the connection
        --- this is the \textbf{active} reading. We are really
        transforming in $C(P)$.
\end{itemize}
The formulas are formally identical, but only the active one produces
vector fields on the configuration bundle (as we will see), and hence symmetries in the
sense of Chapter~\ref{chap:symmetries_and_Noether_theorems}. Keeping the
two apart is the single most common source of sign errors in this subject,
and it is why we spell out the convention in
Remark~\ref{rem:sign_convention_gauge} below.
\end{remark}

The next observation says what kind of object the collection
$\{\tau_{\alpha}\}$ of local gauge transformations \eqref{eq:local_tau}
really is. It is \emph{not} a $G$-valued function on $X$: it is a section
of a bundle of groups.

\begin{remark}[Associated bundles with non-linear fiber]\label{rem:associated_general}
Theorem \ref{thm:associated} and Defini\-tion \ref{def:associated} were stated
for a representation $\rho\colon G\to\mathrm{GL}(V)$, but the construction
never uses the linear structure of $V$: given \emph{any} smooth left action
$\lambda\colon G\times F\to F$, $(g,f)\mapsto g\cdot f$, of $G$ on a manifold
$F$, the right action $(p,f)\cdot g:=(pg,\,g^{-1}\!\cdot f)$ on $P\times F$ is
again free and proper, and the quotient
\[
  P\times_{G}F:=(P\times F)/G,
  \qquad \pi_{F}\bigl([p,f]\bigr):=\pi_{P}(p),
\]
is a fiber bundle over $X$ with typical fiber $F$; a local trivialization
$T_{\alpha}=(\pi_{P},t_{\alpha})$ of $P$ induces
$[p,f]\mapsto(\pi_{P}(p),t_{\alpha}(p)\cdot f)$, whose transition functions are
$(x,f)\mapsto(x,g_{\alpha\beta}(x)\cdot f)$. Exactly as in the linear case, the
fibers are diffeomorphic to $F$ non canonically --- the identification
$\iota_{p}\colon f\mapsto[p,f]$ satisfies $\iota_{pg}=\iota_{p}\circ\lambda_{g}$
--- so \textbf{any structure on $F$ preserved by the action descends
canonically to each fiber}: if the $\lambda_{g}$ are linear one recovers a
vector bundle (Theorem~\ref{thm:associated}), and if they are \emph{group
automorphisms} one obtains a bundle of groups, which is the case we need
below. The correspondence between sections and equivariant maps is also
unchanged: $\Gamma(P\times_{G}F)\cong\{\sigma\in C^{\infty}(P,F):
\sigma(pg)=g^{-1}\!\cdot\sigma(p)\}$, via $s_{\sigma}(x)=[p,\sigma(p)]$ for any
$p\in\pi_{P}^{-1}(x)$, and the local representatives
$\sigma_{\alpha}:=\sigma\circ s_{\alpha}$ glue by
$\sigma_{\beta}=g_{\alpha\beta}^{-1}\!\cdot\sigma_{\alpha}$. We shall only use
this generality once, for $F=G$ with the conjugation action, in
Corollary~\ref{cor:tau_transition}.
\end{remark}

\begin{corollary}[Local gauge transformations glue by conjugation]
\label{cor:tau_transition}
Let $\tau\in C^{\infty}_{\Ad}(P,G)$ and let $\tau_{\alpha}=\tau\circ
s_{\alpha}$, $\tau_{\beta}=\tau\circ s_{\beta}$ be its local expressions in
two frames. Then, on $U_{\alpha}\cap U_{\beta}$,
\begin{equation}\label{eq:tau_transition}
  \tau_{\beta}=g_{\alpha\beta}^{-1}\,\tau_{\alpha}\,g_{\alpha\beta}
  =\Conj_{g_{\alpha\beta}^{-1}}\!\circ\,\tau_{\alpha}.
\end{equation}
Consequently $\Gau(P)\cong C^\infty_{\Ad}(P,G)\cong\Gamma\bigl(\Conj P\bigr)$, where
\[
  \Conj P:=P\times_{\Conj}G
\]
is the \textbf{group bundle} (or \textbf{adjoint group bundle}) associated
to $P$ through the conjugation action $\Conj\colon G\to\Aut(G)$,
$\Conj_{g}(h)=ghg^{-1}$; its fibers are groups isomorphic to $G$ and its
sections are multiplied pointwise.
\end{corollary}

\begin{proof}
Using $s_{\beta}=s_{\alpha}\cdot g_{\alpha\beta}$ and the equivariance
$\tau(pg)=g^{-1}\tau(p)g$ of Proposition~\ref{prop:gauge_group_equivariant_maps},
\[
  \tau_{\beta}(x)=\tau\bigl(s_{\alpha}(x)g_{\alpha\beta}(x)\bigr)
  =g_{\alpha\beta}(x)^{-1}\tau\bigl(s_{\alpha}(x)\bigr)g_{\alpha\beta}(x).
\]
The identification $C^{\infty}_{\Ad}(P,G)\cong\Gamma(\Conj P)$ is the
standard correspondence between $G$-equivariant maps $P\to F$ and sections
of $P\times_{G}F$ (Theorem~\ref{thm:associated} with $F=G$ and the
conjugation action, which is by group automorphisms, so the fibers of
$\Conj P$ inherit a group structure).
\end{proof}

\subsection*{The infinitesimal gauge algebra}

We now make the isomorphism $\gau(P)\cong\Gamma(\ad P)$ used above
completely explicit, because the whole of
Chapter~\ref{chap:multisymplectic_yang_mills_computations} is written in
terms of the local components $\xi^{a}(x)$ that will appear in the following.
\begin{lemma}[Three descriptions of an infinitesimal gauge transformation]
\label{lem:gau_three_descriptions}
The following three $\RR$-vector spaces are canonically isomorphic:
\begin{enumerate}
  \item[(i)] $\gau(P)=\{\xi\in\mathfrak{X}(P)\ :\ \dd\pi_{P}(\xi)=0,\
        (R_{g})_{*}\xi=\xi\ \forall g\in G\}$, the vertical
        $G$-invariant vector fields on $P$;
  \item[(ii)] $C^{\infty}_{\Ad}(P,\mathfrak{g})
        =\{\widehat\xi\in C^{\infty}(P,\mathfrak{g})\ :\
        \widehat\xi(pg)=\Ad_{g^{-1}}\widehat\xi(p)\}$, the
        $\Ad$-equivariant $\mathfrak{g}$-valued functions on $P$;
  \item[(iii)] $\Gamma(\ad P)$, the sections of the adjoint bundle
        \eqref{eq:adP}.
\end{enumerate}
The isomorphisms are
\begin{equation}\label{eq:gau_dictionary}
  \xi(p)=\widehat\xi(p)_{P}(p),
  \qquad
  \sigma_{\xi}(x)=\bigl[p,\widehat\xi(p)\bigr]\in(\ad P)_{x}
  \quad\text{for any }p\in\pi_{P}^{-1}(x).
\end{equation}
In a frame $s_{\alpha}$ the corresponding \textbf{local gauge parameter} is
\begin{equation}\label{eq:local_xi}
  \xi_{\alpha}:=\widehat\xi\circ s_{\alpha}\colon U_{\alpha}\to\mathfrak{g},
  \qquad
  \xi_{\alpha}=\xi^{a}(x)\,e_{a},
\end{equation}
and on an overlap it transforms \textbf{tensorially},
\begin{equation}\label{eq:xi_transition}
  \xi_{\beta}=\Ad_{g_{\alpha\beta}^{-1}}\circ\,\xi_{\alpha}.
\end{equation}
\end{lemma}

\begin{proof}
(i)$\leftrightarrow$(ii). By Remark~\ref{remark:fundamental_vector_field}
every vertical vector field is of the form $\xi(p)=\widehat\xi(p)_{P}(p)$
for a unique smooth $\widehat\xi\colon P\to\mathfrak{g}$. Using
\eqref{eq:adjoint_bundle_is_vertical_equivariant_vector_fields_on_P},
\[
  (R_{g})_{*}\xi(p)
  =(R_{g})_{*}\bigl(\widehat\xi(p)_{P}(p)\bigr)
  =\bigl(\Ad_{g^{-1}}\widehat\xi(p)\bigr)_{P}(pg),
  \qquad
  \xi(pg)=\widehat\xi(pg)_{P}(pg),
\]
and $G$-invariance $(R_{g})_{*}\xi=\xi$ is therefore \emph{equivalent}
to $\widehat\xi(pg)=\Ad_{g^{-1}}\widehat\xi(p)$, again by the injectivity
of $\zeta\mapsto\zeta_{P}(pg)$.
(ii)$\leftrightarrow$(iii) is Theorem~\ref{thm:associated} with
$V=\mathfrak{g}$ and $\rho=\Ad$: equivariant maps $P\to\mathfrak{g}$ are
exactly sections of $P\times_{\Ad}\mathfrak{g}$. Finally
\eqref{eq:xi_transition} is \eqref{eq:gau_dictionary} evaluated on
$s_{\beta}=s_{\alpha}g_{\alpha\beta}$.
\end{proof}

\begin{remark}[The bracket, and a sign that must be taken into account]
\label{rem:bracket_sign}
$\Gau(P)$ is an infinite dimensional Lie group; by
Proposition~\ref{prop:gauge_group_equivariant_maps} it is isomorphic, as a
group, to $C^{\infty}_{\Ad}(P,G)$ with the \textbf{pointwise} product. Its
Lie algebra is therefore $C^{\infty}_{\Ad}(P,\mathfrak{g})\cong
\Gamma(\ad P)$ with the \textbf{pointwise} bracket
\begin{equation}\label{eq:pointwise_bracket}
  [\sigma,\varsigma](x):=\bigl[\xi(x),\zeta(x)\bigr]_{\mathfrak{g}},
  \qquad\text{locally}\qquad
  [\xi,\zeta]^{a}(x)=C^{a}_{bc}\,\xi^{b}(x)\zeta^{c}(x),
\end{equation}
where $\sigma(x)=[s(x),\xi(x)]$ and $\varsigma(x)=[s(x),\zeta(x)]$ given a local section
$s\in\Gamma(\pi)$ and $\xi,\zeta$ $\mathfrak{g}$-valued functions.
This -and not the vector field bracket- is the bracket we use
throughout. Indeed, realizing $\gau(P)$ inside $\mathfrak{X}(P)$ as in
Lemma~\ref{lem:gau_three_descriptions}(i) \textbf{reverses the sign}:
\begin{equation}\label{eq:bracket_sign}
  \bigl[\,\xi,\zeta\,\bigr]_{\mathfrak{X}(P)}
  \;=\;-\,\bigl[\widehat\xi,\widehat\zeta\,\bigr]^{\wedge},
  \qquad \xi,\zeta\in\gau(P).
\end{equation}
To see this, write $X_{f}(p):=f(p)_{P}(p)$ for arbitrary smooth
$f\colon P\to\mathfrak{g}$. Expanding in a basis and using that
$\zeta\mapsto\zeta_{P}$ is a Lie algebra homomorphism for the right action,
\[
  [X_{f},X_{h}]=X_{[f,h]+X_{f}(h)-X_{h}(f)} .
\]
If $f=\widehat\xi$ and $h=\widehat\zeta$ are equivariant, then
\[
  X_{\widehat\xi}(\widehat\zeta)(p)
  =\der{}{t}\Big|_{t=0}\widehat\zeta\bigl(p\exp(t\widehat\xi(p))\bigr)
  =\der{}{t}\Big|_{t=0}\Ad_{\exp(-t\widehat\xi(p))}\widehat\zeta(p)
  =-\bigl[\widehat\xi,\widehat\zeta\bigr](p),
\]
and substituting gives $[X_{\widehat\xi},X_{\widehat\zeta}]
=X_{[\widehat\xi,\widehat\zeta]-[\widehat\xi,\widehat\zeta]
-[\widehat\xi,\widehat\zeta]}=X_{-[\widehat\xi,\widehat\zeta]}$, which is
\eqref{eq:bracket_sign} by definition of $(X_\cdot)$. This is precisely what makes
$\xi\mapsto\xi_{C}$ a homomorphism rather than an anti-homomorphism in
Theorem~\ref{thm:gauge_prolongation}.
\end{remark}

\subsection*{The correspondence \texorpdfstring{$\omega\leftrightarrow\phi_{\omega}$}{omega <-> phi} made explicit}

Theorem~\ref{thm:connections} was stated without construction. We need the
construction itself, because the action of $\Gau(P)$ on $C(P)$ is defined
on jets and we must be able to read off what it does to connections. The
whole point is the following elementary identification.

\begin{lemma}[Jets are horizontal complements]\label{lem:jets_are_complements}
Fix $p\in P$ with $x=\pi_{P}(p)$. The map
\[
  j^{1}_{x}s\;\longmapsto\;\mathrm{im}\bigl(T_{x}s\bigr)=:H_p\subset T_{p}P
\]
is a bijection between $\pi_{1,0}^{-1}(p)\subset J^{1}P$ (the $1$-jets at
$x$ of local sections of $\pi_{P}$ through $p$) and the set of linear
complements of $\mathcal{V}_{p}\pi_{P}$ in $T_{p}P$. Under the prolonged
right action $g\cdot j^{1}_{x}s=j^{1}_{x}(R_{g}\circ s)$ this bijection
intertwines with
\[
  H_{p}\;\longmapsto\;(R_{g})_{*}H_{p}\ \subset T_{pg}P .
\]
\end{lemma}

\begin{proof}
A jet $j^{1}_{x}s$ with $s(x)=p$ is the same thing as a linear map
$T_{x}s\colon T_{x}X\to T_{p}P$ with
$T_{p}\pi_{P}\circ T_{x}s=\mathrm{id}_{T_{x}X}$. i.e. a right splitting of
the exact sequence
$0\to\mathcal{V}_{p}\pi_{P}\to T_{p}P\to T_{x}X\to0$. Splittings (in this
case $T_x s$) of a short exact sequence correspond bijectively to complements of
the kernel (the kernel of $T_p\pi_P$ i.e. $\mathcal{V}_p\pi_P$), the complement being
the image of the splitting and the splitting being the inverse of $T_{p}\pi_{P}|_{H_{p}}$.
Equivariance is $\mathrm{im}\,T_{x}(R_{g}\circ s)=(R_{g})_{*}\mathrm{im}\,T_{x}s$.
\end{proof}

\begin{proposition}[Connections as sections, explicitly]
\label{prop:connections_as_sections_explicit}
Let $x\in X$. Because $G$ acts freely on $P$, a $G$-orbit in
$\pi_{1}^{-1}(x)\subset J^{1}P$ meets each fiber $\pi_{1,0}^{-1}(p)$,
$p\in P_{x}$, in exactly one point. Hence, by
Lemma~\ref{lem:jets_are_complements},
\[
  C(P)_{x}=\pi^{-1}(x)
  \;\cong\;
  \left\{
  \begin{array}{c}
    \text{$G$-invariant fields }p\mapsto H_{p}\ (p\in P_{x})\\[2pt]
    \text{of complements of }\mathcal{V}_{p}\pi_{P}\text{ in }T_{p}P
  \end{array}
  \right\},
\]
and \textbf{consequently a section} $\bm{\phi\in\Gamma(\pi)}$ \textbf{is exactly} a
$G$-invariant horizontal distribution $H\subset TP$, i.e. \textbf{a principal connection}
(Definition \ref{def:connection}). We write $\phi_{\omega}\in\Gamma(\pi)$
for the section determined by the connection $\omega$, so that
\[
  \phi_{\omega}(x)=\bigl[\,j^{1}_{x}s\,\bigr],
  \qquad s\ \text{any local section with}\ \mathrm{im}\,T_{x}s=H_{s(x)}
  =\ker\omega_{s(x)} .
\]
Moreover, in the frame (local section) $s_{\alpha}$ the (isomorphic)
\textbf{horizontal lift} is
\begin{align}\label{eq:horizontal_lift_local}
  T_{x}X\ni v\;&\longmapsto\;
  T_{x}s_{\alpha}(v)-\bigl(A_{\alpha}(v)\bigr)_{P}\bigl(s_{\alpha}(x)\bigr)
  \ \in H_{s_{\alpha}(x)},\\
  &A_{\alpha}=s_{\alpha}^{*}\omega=A_\mu^a(x)\dd x^\mu\otimes e_a,
\end{align}
so that \textbf{the fiber coordinates of} $\bm{\phi_{\omega}(x)}$ \textbf{in the induced
chart are precisely the components} $\bm{A^{a}_{\mu}(x)}$ \textbf{of the gauge potential}.
Finally, two such sections differ by a $G$-invariant horizontal
$\Ad$-equivariant $\mathfrak{g}$-valued $1$-form on $P$, i.e.\ by an
element of $\Omega^{1}(X;\ad P)=\Gamma(T^{*}X\otimes\ad P)$: this is the
affine structure asserted in Theorem~\ref{thm:connections}.
\end{proposition}

\begin{proof}
  See Appendix \ref{app:some_proofs_on_gauge_theory}.
\end{proof}

\subsection*{The induced action on the bundle of connections}

\begin{lemma}[The prolongation of a gauge transformation]
\label{lem:PhiC_welldefined}
Let $\Phi\in\Gau(P)$. Then
\[
  \widehat{\Phi}\colon J^{1}P\longrightarrow J^{1}P,
  \qquad
  \widehat{\Phi}\bigl(j^{1}_{x}s\bigr):=j^{1}_{x}(\Phi\circ s),
\]
is well defined, is a diffeomorphism, commutes with the prolonged right
$G$-action, and therefore descends to a diffeomorphism
\[
  \Phi_{C}\colon C(P)\longrightarrow C(P),
  \qquad
  \Phi_{C}\bigl(\rho(j^{1}_{x}s)\bigr):=\rho\bigl(j^{1}_{x}(\Phi\circ s)\bigr),
\]
which is \textbf{vertical}, $\pi\circ\Phi_{C}=\pi$. Moreover
$\Phi\mapsto\Phi_{C}$ is a group \textbf{homomorphism}
$\Gau(P)\to\Diff(C(P))$.
\end{lemma}

\begin{proof}
$\Phi\circ s$ is again a local section of $\pi_{P}$, because
$\pi_{P}\circ\Phi\circ s=\pi_{P}\circ s=\mathrm{id}_X$; and $j^{1}_{x}(\Phi\circ
s)$ depends only on $j^{1}_{x}s$ by the chain rule, so $\widehat{\Phi}$ is
well defined, with inverse $\widehat{\Phi^{-1}}$. For the equivariance,
$G$-equivariance of $\Phi$ gives $\Phi\circ R_{g}=R_{g}\circ\Phi$, hence
\[
  \widehat{\Phi}\bigl(j^{1}_{x}s\cdot g\bigr)
  =j^{1}_{x}\bigl(\Phi\circ R_{g}\circ s\bigr)
  =j^{1}_{x}\bigl(R_{g}\circ\Phi\circ s\bigr)
  =\widehat{\Phi}\bigl(j^{1}_{x}s\bigr)\cdot g ,
\]
so $\widehat{\Phi}$ descends to the quotient $C(P)=J^{1}P/G$. Verticality
holds because $\pi_{1}\circ\widehat{\Phi}=\pi_{1}$ (the base point $x$ is
untouched). Finally
$\widehat{\Phi\circ\Psi}(j^{1}_{x}s)=j^{1}_{x}(\Phi\circ\Psi\circ s)
=\widehat{\Phi}\bigl(\widehat{\Psi}(j^{1}_{x}s)\bigr)$, so
$(\Phi\circ\Psi)_{C}=\Phi_{C}\circ\Psi_{C}$ (group homomorphism).
\end{proof}

The following proposition is the one that was missing, and it is what
forces the convention of Remark~\ref{rem:sign_convention_gauge}: the
canonical map $\Phi_{C}$ is \textbf{not} the map $\omega\mapsto
\Phi^{*}\omega$, but rather, its inverse.

\begin{proposition}[What $\Phi_{C}$ does to a connection]
\label{prop:PhiC_is_pullback}
For every $\Phi\in\Gau(P)$ and every principal connection $\omega$,
\begin{equation}\label{eq:PhiC_is_pullback}
  \Phi_{C}\circ\phi_{\omega}=\phi_{(\Phi^{-1})^{*}\omega}.
\end{equation}
\end{proposition}

\begin{proof}
By Proposition~\ref{prop:connections_as_sections_explicit},
$\phi_{\omega}(x)$ is the orbit of a pair $(p,H_{p})$ with
$H_{p}=\ker\omega_{p}$, and by Lemma~\ref{lem:jets_are_complements}
$\Phi_{C}$ sends it to the orbit of $\bigl(\Phi(p),T_{p}\Phi(H_{p})\bigr)$.
Hence $\Phi_{C}\circ\phi_{\omega}=\phi_{\omega'}$, where $\omega'$ is the
connection whose horizontal distribution is $H'=T\Phi(H)$. Now
$\Phi^{*}\omega'$ is a principal connection by
Proposition~\ref{prop:pullback_omega_gauge}, and
\[
  v\in\ker(\Phi^{*}\omega')_{p}
  \iff T_{p}\Phi(v)\in\ker\omega'_{\Phi(p)}=T_{p}\Phi(H_{p})
  \iff v\in H_{p}=\ker\omega_{p},
\]
where the last iff condition is true because $T_{p}\Phi$ is an isomorphism.
Two principal connection forms with the same kernel coincide: they agree on the
horizontal part (both vanish) and on the vertical part (both send $\zeta_{P}\mapsto\zeta$,
by Definition~\ref{def:connection}(i)), and $T_{p}P=\mathcal{V}_{p}\pi_{P}
\oplus H_{p}$. Therefore $\Phi^{*}\omega'=\omega$, i.e.\
$\omega'=(\Phi^{-1})^{*}\omega$.
\end{proof}

\begin{remark}[Convention, fixed once and for all]
\label{rem:sign_convention_gauge}
We have two natural actions of $\Gau(P)$ on the space $\mathcal{A}$ of
principal connections, and they are inverse to each other:
\begin{align*}
  \Phi\cdot\omega&:=\Phi^{*}\omega
  &&\textbf{right}\text{ action, since }(\Phi\circ\Psi)^{*}
    =\Psi^{*}\circ\Phi^{*},\\
  \Phi\ast\omega&:=(\Phi^{-1})^{*}\omega
  &&\textbf{left}\text{ action, since }((\Phi\circ\Psi)^{-1})^*=(\Phi^{-1})^*\circ(\Psi^{-1})^*.
\end{align*}
By Proposition~\ref{prop:PhiC_is_pullback} the \emph{canonical} map
$\Phi_{C}$ of Lemma~\ref{lem:PhiC_welldefined} realizes the second one.
Following \eqref{eq:gauge_transf_A_active} --- and standard physics usage
--- we adopt the \textbf{first}:
\[
  \boxed{\ \Phi\cdot\omega:=\Phi^{*}\omega\ }
\]
\textbf{Accordingly}, the transformation of the configuration bundle $Y=C(P)$ that
we shall use throughout is
\begin{equation}\label{eq:PhiY_def}
  \Phi_{Y}:=\bigl(\Phi^{-1}\bigr)_{C}=\bigl(\Phi_{C}\bigr)^{-1},
  \qquad\text{so that}\qquad
  \Phi_{Y}\circ\phi_{\omega}=\phi_{\Phi^{*}\omega},
\end{equation}
a vertical (trivial by $\Phi_C$ being vertical) automorphism of $\pi\colon Y\to X$ over
$\mathrm{id}_{X}$. Note
that $\Phi\mapsto\Phi_{Y}$ is an \emph{anti}-homomorphism,
$(\Phi\circ\Psi)_{Y}=\Psi_{Y}\circ\Phi_{Y}$, exactly as it must be for a
right action. The opposite choice would replace $\Phi$ by $\Phi^{-1}$
everywhere, i.e.\ $\xi$ by $-\xi$, and would turn
\eqref{eq:infinitesimal_gauge} into $\delta_{\xi}A=-D_{A}\xi$ and
\eqref{eq:gauge_transf_A_active} into
$\Ad_{\tau_{\alpha}}\!\circ A_{\alpha}+(\tau_\alpha^{-1})^*\theta_G$.
It changes no orbit, no invariant and no field equation, but it changes the
sign of every current below, so it must be declared.
\end{remark}

\subsection{Infinitesimal gauge transformations on \texorpdfstring{$Y$}{Y} and \texorpdfstring{$Z$}{Z}}

\begin{theorem}[Gauge transformations and their prolongations]
\label{thm:gauge_prolongation}
Let $\xi\in\gau(P)\cong\Gamma(\ad P)$ and let $\{\Phi_{t}\}\subset\Gau(P)$
be its flow, $\Phi_{0}=\mathrm{id}_{P}$. Define its \textbf{generator} as
\begin{equation}\label{eq:xi_C_def}
  \xi_{C}:=\der{}{t}\Big|_{t=0}\bigl(\Phi_{t}\bigr)_{Y}
        =-\der{}{t}\Big|_{t=0}\bigl(\Phi_{t}\bigr)_{C} .
\end{equation}
Then:
\begin{enumerate}
  \item $\xi_{C}$ is a $\pi$-\textbf{vertical} vector field on $Y=C(P)$, and in the
        fibered coordinates $(x^{\mu},A^{a}_{\mu})$ of
        Proposition \ref{prop:connections_as_sections_explicit},
        \begin{equation}\label{eq:xi_C_local}
          \xi_{C}=\bigl(D_{\mu}\xi\bigr)^{a}(x)\parder{}{A^{a}_{\mu}},
          \qquad
          \bigl(D_{\mu}\xi\bigr)^{a}(x)
          :=\parder{\xi^{a}(x)}{x^{\mu}}+C^{a}_{bc}A^{b}_{\mu}(x)\xi^{c}(x);
        \end{equation}
  \item $\xi\mapsto\xi_{C}$ is a \textbf{homomorphism of Lie algebras}
        from $\Gamma(\ad P)$ with the pointwise bracket
        \eqref{eq:pointwise_bracket} into $\Gamma(\mathcal{V}\pi)=:
        \mathfrak{X}^{\mathcal{V}(\pi)}(C(P))$:
        \[
          [\xi_{C},\zeta_{C}]=[\xi,\zeta]_{C},
        \]
        and (for $G$ connected) its kernel consists of the locally constant 
        sections with values in the centre $\mathfrak{z}(\mathfrak{g})$; in
        particular it is \textbf{injective} whenever
        $\mathfrak{z}(\mathfrak{g})=0$, e.g.\ for $G$ semisimple;
  \item the canonical prolongation $j^{1}\xi_{C}\in\mathfrak{X}(Z)$,
        $Z=J^{1}(C(P))$, is the infinitesimal gauge transformation acting
        on the first order field variables, with local expression
        \begin{equation}\label{eq:j1xi_C_local}
          \begin{aligned}
          j^{1}\xi_{C}
          &=\bigl(D_{\mu}\xi\bigr)^{a}(x)\parder{}{A^{a}_{\mu}}
          +\frac{\dd\bigl(D_{\mu}\xi\bigr)^{a}}{\dd x^{\nu}}(x)
           \parder{}{A^{a}_{\mu\nu}},\\[2pt]
          \frac{\dd(D_{\mu}\xi)^{a}}{\dd x^{\nu}}(x)
          &=\partial_{\nu}\partial_{\mu}\xi^{a}(x)
           +C^{a}_{bc}\bigl(A^{b}_{\mu\nu}(x)\xi^{c}(x)
           +A^{b}_{\mu}(x)\partial_{\nu}\xi^{c}(x)\bigr),
          \end{aligned}
        \end{equation}
        where $A^{a}_{\mu\nu}$ denotes the jet coordinate corresponding to
        $\partial_{\nu}A^{a}_{\mu}$ under the dictionary
        \eqref{eq:dictionary}.
\end{enumerate}
\end{theorem}
\begin{proof}
The full proof is given in Section \ref{app:the_big_proof} of
Appendix \ref{app:some_proofs_on_gauge_theory}; the reader is encouraged to
consult it, since some notation (e.g.\ $\delta_\xi A_\alpha$) is introduced
there.
\end{proof}

\begin{remark}[Passive versus active, once more]
Equation \eqref{eq:infinitesimal_gauge} is formally identical to
\eqref{eq:gauge_transf_A} with $g_{t}(x)=\exp(t\xi_{\alpha}(x))$ in the
role of the transition function $g_{\alpha\beta}$; by
Remark~\ref{rem:one_computation} this is no accident. But only the active
reading produces the vector field \eqref{eq:xi_C_local} on $Y=C(P)$, which
is the object that the variational theory of
Chapter~\ref{chap:symmetries_and_Noether_theorems} consumes.
\end{remark}

The field strength, in contrast, transforms tensorially. Since
$F_\alpha=\dd A_\alpha+\tfrac{1}{2}[A_\alpha,A_\alpha]$ and the bracket of two
$\mathfrak{g}$-valued $1$-forms is symmetric, $\delta_{\xi}F_\alpha=\dd(\delta_{\xi}A_\alpha)
+[A_\alpha,\delta_{\xi}A_\alpha]=D_{A_\alpha}(\delta_{\xi}A_\alpha)$, so by
\eqref{eq:infinitesimal_gauge} and $D_{A_\alpha}D_{A_\alpha}\sigma=[F_\alpha,\sigma]$ for
$\sigma\in\Gamma(\ad P)$,
\begin{equation}\label{eq:infinitesimal_gauge_F}
  \delta_{\xi_\alpha}F_\alpha=D_{A_\alpha}\bigl(D_{A_\alpha}\xi_\alpha\bigr)=[F_\alpha,\xi_\alpha],
  \qquad\text{in components}\qquad
  j^{1}\xi_{C}\bigl(F^{a}_{\mu\nu}\bigr)=C^{a}_{bc}F^{b}_{\mu\nu}\xi^{c},
\end{equation}
the infinitesimal version of the finite law
$F\mapsto\Ad_{g^{-1}}\!\circ F$ of \eqref{eq:gauge_transf_F}, which is proved in
Subsection \ref{sub:useful_identities_for_gauge_symmetries_computations}.

\section{The Yang--Mills Lagrangian}\label{sec:YM_lagrangian}

We close the chapter by writing down the Yang--Mills Lagrangian on $Z=J^{1}(C(P))$
and pointing out its degeneracy, which is what forces the constraint algorithm of
Chapter~\ref{chap:multisymplectic} into play. The detailed computation of the field
equations, momenta and currents is deferred to
Chapter~\ref{chap:multisymplectic_yang_mills_computations}.

The Yang--Mills Lagrangian requires two extra pieces of data, both external to the
gauge structure:
\begin{enumerate}
  \item a (pseudo-)Riemannian metric $g$ on $X$, whose associated volume form we
        identify with $\sqrt{|g|}\eta$ for $\eta$ fixed in Chapter~\ref{chap:multisymplectic}, and
  \item an $\Ad$-invariant, non-degenerate bilinear form $\kappa$ on
        $\mathfrak{g}$ (the Killing form when $G$ is semisimple, or $-\operatorname{tr}$
        for matrix groups), with components $\kappa_{ab}=\kappa(e_{a},e_{b})=\kappa_{ba}$.
\end{enumerate}
Having the metric $g$, we can define the Hodge star operator $\star$

\begin{definition}[Hodge star]\label{def:hodge_star}
Let $(X,g)$ be an oriented $n+1$-dimensional (pseudo-)Riemannian manifold
with volume form $\eta$. The metric induces a fiberwise non-degenerate
symmetric pairing (an inner product) on each $\bigwedge^{k}T^{*}_{x}X$, determined on
decomposables by
\begin{equation}\label{eq:metric_on_forms}
  \bigl\langle
    \alpha^{1}\wedge\dots\wedge\alpha^{k},\;
    \beta^{1}\wedge\dots\wedge\beta^{k}
  \bigr\rangle_{g}
  :=\det\bigl(g^{-1}(\alpha^{i},\beta^{j})\bigr)_{i,j=1}^{k},
\end{equation}
and extended to $\bigwedge^{k}T^{*}_{x}X$ by linearity, 
where $g^{-1}$ is the induced metric on $T^{*}_{x}X$. The \textbf{Hodge
star operator} is then the unique bundle morphism
\[
  \star\colon\Lambda^{k}T^{*}X\longrightarrow\Lambda^{n+1-k}T^{*}X,
  \qquad 0\le k\le n+1,
\]
such that
\begin{equation}\label{eq:hodge_defining}
  \alpha\wedge\star\beta=\langle\alpha,\beta\rangle_{g}\,\eta
  \qquad\text{for all }\alpha,\beta\in\Lambda^{k}T^{*}_{x}X,\ x\in X .
\end{equation}
\end{definition}
 
Existence and uniqueness are immediate from elementary linear algebra: for fixed
$\beta$ the left-hand side of \eqref{eq:hodge_defining} is linear in
$\star\beta$, and the wedge pairing
$\Lambda^{k}T^{*}_{x}X\times\Lambda^{n+1-k}T^{*}_{x}X\to\Lambda^{n+1}T^{*}_{x}X
\cong\RR$ (the last isomorphism being $\eta$) is non-degenerate, so
\eqref{eq:hodge_defining} determines $\star\beta$ uniquely; smoothness and
$C^{\infty}(X)$-linearity follow because everything is algebraic in $g$ and
$\eta$. Note that $\star$ requires an orientation, which we have already
assumed in fixing $\eta$.
 
\begin{remark}[Coordinate expression]\label{rem:hodge_coordinates}
Write $\varepsilon_{\mu_{0}\dots\mu_{n}}$ for the Levi-Civita \emph{symbol}
(totally skew-symmetric, $\varepsilon_{0\dots n}=1$) and
$\eta=\sqrt{\lvert g\rvert}\,\dd x^{0}\wedge\dots\wedge\dd x^{n}$ with
$\lvert g\rvert:=\lvert\det(g_{\mu\nu})\rvert$. If
$\beta=\tfrac{1}{k!}\beta_{\mu_{1}\dots\mu_{k}}
\dd x^{\mu_{1}}\wedge\dots\wedge\dd x^{\mu_{k}}$ with
$\beta_{\mu_{1}\dots\mu_{k}}$ totally skew-symmetric, then
\begin{equation}\label{eq:hodge_components}
  (\star\beta)_{\nu_{1}\dots\nu_{n+1-k}}
  =\frac{\sqrt{\lvert g\rvert}}{k!}\,
   \beta^{\mu_{1}\dots\mu_{k}}\,
   \varepsilon_{\mu_{1}\dots\mu_{k}\nu_{1}\dots\nu_{n+1-k}},
\end{equation}
indices being raised with $g^{\mu\nu}$. Moreover
\begin{equation}\label{eq:star_star}
  \star\star\big|_{\Lambda^{k}}=(-1)^{k(n+1-k)}(-1)^{s}\,\mathrm{id},
\end{equation}
$s$ being the number of negative eigenvalues of $g$; in Lorentzian
signature on a four dimensional $X$ this gives $\star\star=-\mathrm{id}$ on
$2$-forms.
\end{remark}
 
\begin{definition}[Hodge star on bundle-valued forms]\label{def:hodge_bundle_valued}
Let $E\to X$ be a vector bundle. Since $\star$ is $C^{\infty}(X)$-linear
and acts only on the form factor of
$\Omega^{k}(X;E)=\Gamma\bigl(\Lambda^{k}T^{*}X\otimes E\bigr)$, the operator
\[
  \star:=\star\otimes\mathrm{id}_{E}\colon
  \Omega^{k}(X;E)\longrightarrow\Omega^{n-k}(X;E),
  \qquad
  \star(\beta\otimes s):=(\star\beta)\otimes s ,
\]
is well defined on decomposables and extends uniquely by
$C^{\infty}(X)$-linearity. We use it for $E=\ad P$.
\end{definition}
 
\begin{remark}[Applying $\star$ to the curvature]\label{rem:hodge_on_F}
Let $F\in\Omega^{2}(X;\ad P)$ (see Remark \ref{rem:physical_forms}) be the curvature,
with local representatives
$F_{\alpha}\in\Omega^{2}(U_{\alpha};\mathfrak{g})$ obtained through the
trivialization $s_{\alpha}$, written in the basis $\{e_{a}\}$ as
\begin{equation}\label{eq:F_as_two_form}
  F_{\alpha}
  =\tfrac{1}{2}\,F^{a}_{\mu\nu}\,
   \dd x^{\mu}\wedge\dd x^{\nu}\otimes e_{a},
  \qquad F^{a}_{\mu\nu}=-F^{a}_{\nu\mu},
\end{equation}
the factor $\tfrac12$ compensating the double counting of the
unrestricted sum over $(\mu,\nu)$. Because the $F^{a}_{\mu\nu}$ are
functions and $\star$ is $C^{\infty}$-linear in the form slot and the
identity on the algebra slot,
\begin{equation}\label{eq:star_F_local}
  \star F_{\alpha}
  =\tfrac{1}{2}\,F^{a}_{\mu\nu}\,
   \star\bigl(\dd x^{\mu}\wedge\dd x^{\nu}\bigr)\otimes e_{a},
  \qquad
  (\star F)^{a}_{\nu_{1}\dots\nu_{n-2}}
  =\tfrac{1}{2}\sqrt{\lvert g\rvert}\,F^{a\,\mu\nu}
   \varepsilon_{\mu\nu\nu_{1}\dots\nu_{n-2}} .
\end{equation}
This is consistent globally: under a change of trivialization
$F_{\beta}=\Ad_{g_{\beta\alpha}}\!\circ F_{\alpha}$ by
\eqref{eq:gauge_transf_F}, and $\Ad$ acts on the $\mathfrak{g}$ factor
while $\star$ acts on the form factor, so the two commute and the local
$\star F_{\alpha}$ glue to a global $\star F\in\Omega^{n+1-2}(X;\ad P)$.
\end{remark}

\begin{definition}[Yang--Mills Lagrangian]\label{def:YM_lagrangian}
The \textbf{Yang--Mills Lagrangian} is the first order Lagrangian $L=L_{\mathrm{YM}}\,\eta$
on $Z=J^{1}(C(P))$ with Lagrangian function
\begin{equation}\label{eq:YM_lagrangian}
L_{\mathrm{YM}}=-\tfrac{1}{4}\,\kappa_{ab}\,F^{a}_{\mu\nu}F^{b}_{\rho\sigma}\,
g^{\mu\rho}g^{\nu\sigma}\,\sqrt{\lvert g\rvert}
\;=\;-\tfrac{1}{2}\,\bigl\langle F,\star F\bigr\rangle_{\kappa},
\end{equation}
where $F^{a}_{\mu\nu}$ is given by \eqref{eq:F_components} (hence
$L_{\mathrm{YM}}$ is a genuine function of $(x^{\mu},y^{i},z^{i}_{\mu})$ through the
dictionary \eqref{eq:dictionary}), $\star$ is the Hodge star of $g$ and
$\langle\,\cdot\,,\cdot\,\rangle_{\ad P}$ pairs $\ad P$-values with $\kappa$.
\end{definition}

Since $F$ is first order in $A$, the Lagrangian \eqref{eq:YM_lagrangian} is a
\textbf{first order} Lagrangian, so it fits the framework of
Chapter~\ref{chap:multisymplectic} with no need for higher order jets. Moreover, by
the tensoriality \eqref{eq:gauge_transf_F} of $F$ and the $\Ad$-invariance of
$\kappa$, $L_{\mathrm{YM}}$ is invariant under gauge transformations, so every
$\xi\in\Gamma(\ad P)$ is a symmetry in the sense of
Chapter~\ref{chap:symmetries_and_Noether_theorems} (cf.\
Theorem \ref{thm:gauge_prolongation}). That $L_{\mathrm{YM}}$ is the \textbf{most general}
first order gauge invariant Lagrangian (a Utiyama-type statement) is proved in
\cite{CLMMR20}.
\begin{remark}[The Lagrangian in invariant form]\label{rem:YM_invariant_form}
Combining $\star$ with $\kappa$ gives the $\RR$-valued pairing
\[
  \bigl\langle\phi,\psi\bigr\rangle_{\ad P}:=\kappa_{ab}\,\varphi^{a}\wedge\psi^{b}
  \in\Omega^{k+l}(X),
  \qquad
  \varphi\in\Omega^{k}(X;\ad P),\ \psi\in\Omega^{l}(X;\ad P),
\]
well defined because $\kappa$ is $\Ad$-invariant and hence descends to the
associated bundle. With \eqref{eq:hodge_defining} and
\eqref{eq:F_as_two_form},
\begin{equation}\label{eq:YM_lagrangian_invariant}
  -\tfrac{1}{2}\,\bigl\langle F,\star F\bigr\rangle_{\ad P}
  =-\tfrac{1}{4}\,\kappa_{ab}F^{a}_{\mu\nu}F^{b\,\mu\nu}\,\eta
  =\mathcal{L}_{\mathrm{YM}},
\end{equation}
which is \eqref{eq:YM_lagrangian} written without indices, and the field
equations and Bianchi identity become $D_{A}\star F=0$ and $D_{A}F=0$ as we will
prove on Chapter \ref{chap:multisymplectic_yang_mills_computations}.
\end{remark}

\begin{proposition}[Almost regularity]\label{prop:YM_almost_regular}
The Yang--Mills Lagrangian \eqref{eq:YM_lagrangian} is \textbf{almost regular}
(see Definition \ref{def:almost_regularity}), but it is \textbf{not}
regular (see Definition \ref{def:lagrangian_form_function_and_regularity}).
\end{proposition}

\begin{proof}[Proof (sketch)]
By Remark~\ref{rem:antisymmetry}, $L_{\mathrm{YM}}$ depends on the jet coordinates
$z^{i}_{\nu}\leftrightarrow\partial_{\nu}A^{a}_{\mu}$ only through the antisymmetric
part $\partial_{[\mu}A^{a}_{\nu]}$. Hence the fiber derivatives, i.e. the
multimomenta (which will be computed explicitly in Chapter
\ref{chap:multisymplectic_yang_mills_computations})
\begin{equation}\label{eq:YM_momenta}
p^{\mu\nu}_{a}=\parder{L_{\mathrm{YM}}}{A^a_{\mu\nu}}
=\sqrt{\lvert g\rvert}\,\kappa_{ab}\,F^{b\,\mu\nu}
\end{equation}
are \textbf{antisymmetric} in $(\mu,\nu)$, so their symmetric part vanishes
identically. Equivalently, the Hessian
$\bigl(\partial^{2}L_{\mathrm{YM}}/\partial z^{i}_{\mu}\partial z^{j}_{\nu}\bigr)$ is
degenerate, and $L_{\mathrm{YM}}$ is not regular. The relations
$p^{\mu\nu}_{a}+p^{\nu\mu}_{a}=0$ cut out the image
$\widetilde{M}_{1}=\Leg_{L}(Z)$ as a submanifold of $\Omega^{n+1}_{2}Y$, and the
Legendre transformation $\Leg_{L}\colon Z\to\widetilde{M}_{1}$ is a submersion whose
fibers are the affine subspaces obtained by varying the (dynamically irrelevant)
symmetric part of the jet; these fibers are connected. By the very definition of
almost regularity of Chapter~\ref{chap:multisymplectic}, $L_{\mathrm{YM}}$ is almost
regular.
\end{proof}

Proposition~\ref{prop:YM_almost_regular} is precisely the promised nontrivial
example of an almost regular Lagrangian announced in
Chapter~\ref{chap:multisymplectic}: the constraint algorithm developed there is a necessity
already for the most fundamental gauge theory.

\chapter{Multisymplectic Yang--Mills computations} \label{chap:multisymplectic_yang_mills_computations}

Finally, in this last chapter we apply the whole apparatus of multisymplectic pure gauge theory,
symmetries and preserved quantities to Yang--Mills theory, finding the usual Yang--Mills equations,
Bianchi identity, conserved currents and charges and some Hamiltonian forms. We also compute some
brackets of the graded Lie algebra associated with them. To close the chapter and the thesis,
we highlight a number of open problems and outline possible directions for future research.
\section{Initial clarifications}
\label{sec:initial_clarifications}
Before going into the details of the computations, let us make clear some notational stuff and
compute some key expressions that we will use all the time. 

First, as we will deal with long tensorial computations, \textbf{dummy index relabeling} will
be done recurrently. We will not indicate all relabel correspondences explicitly but we will
use the symbol $\equiv$ to indicate that a dummy index relabeling has been done when possible.
Second, \textbf{some of the long computations that are not interesting perse will be archive
in the Appendix} \ref{app:auxiliary}; we will reference them when needed.
\textbf{Explicit dependence on the points of the manifolds where the
tensor live will be omitted in order to make computations more clear, except in the case of
the dependence in} $\bm{x}$; in the case of the components of our gauge symmetries, namely
$\xi^a(x)$, this is done to remark its infinite-dimensional nature (as mentioned in
on the Section \ref{app:the_big_proof}) and its localizability nature (we can always consider
bump functions in order to make these symmetries localizable in the sense of Definition
\ref{def:localizable_symmetry}), while in the case of regular tensor living on $Z$ (most
of the cases) like $F^a_{\mu\nu}(x)$, $p_a^{\mu\nu}(x)$ etc, this will be done when we
pullback them by an holonomic section $\phi_\omega(x)=(x^\mu,A^a_\mu(x),
\partial_\nu A^a_\mu(x))$. Hence expressions like $\partial_\mu F^a_{\rho\sigma}(x)$ will
refer actually by chain rule to
\begin{align}
	\partial_\mu F^a_{\rho\sigma}(x)&=(j^1\phi_\omega)^*\left(
	\parder{F^a_{\rho\sigma}(z)}{x^\mu}+A^b_{\nu\mu}(z)
	\parder{F^a_{\rho\sigma}(z)}{A^b_{\nu}}+A^b_{\nu\tau\mu}(z)
	\parder{F^a_{\rho\sigma}(z)}{A^b_{\nu\tau}}\right)\\
	&=(j^1\phi_\omega)^*
	\left(\der{F^a_{\rho\sigma}(z)}{x^\mu}\right),
\end{align}
which will appear in computations like the Yang--Mills e.o.m.
By the same reasoning, every function $f\in C^\infty(Z)$ (including $F$, $p$ etc) will
satisfy
\begin{equation}
	\label{eq:pullbacking_dform}
	(j^1\phi_\omega)^*(df(z))=(j^1\phi_\omega)^*\left(\der{f(z)}{x^\mu}\dd x^\mu\right)=
	\partial_\mu f(x)\,\dd x^\mu.
\end{equation}
Recall that, working locally on $Z=J^1(C(P))$ with fiber coordinates
$(x^\mu,A^a_\mu,A^a_{\mu\nu})$, our Yang--Mills Lagrangian has the form
\[
	\mathcal{L}_{YM}(z)=-\frac{1}{4}\sqrt{|g(x)|}\kappa_{ab}g^{\mu\rho}(x)g^{\nu\sigma}(x)
	F^a_{\mu\nu}(z)F^b_{\rho\sigma}(z)\underbrace{\dd^{n+1}x|_z}_{\eta|_z}\equiv L_{YM}(z)\eta|_z,
	\quad \forall z\in Z 
\]
\[
	\text{where }F^a_{\mu\nu}(z):=A^a_{\nu\mu}(z)-A^a_{\mu\nu}(z)+C^a_{bc}A^b_\mu(z)A^c_\nu(z),
	\quad x=\pi_{ZX}(z),
\]
and $C^a_{bc}$ are the structure constants associated with the Lie algebra basis $\{e_a\}$.
Note that $\Ad$-invariance of the bilinear form $\kappa$ implies that the tensor
$C_{abc}=\kappa_{ad}C^d_{bc}$ is totally antisymmetric and that $F^a_{\mu\nu}$ is antisymmetric
in the pair of index $(\mu,\nu)$.

Some straightforward computations that will be substituted on our big computations to make
them more manageable are
\begin{description}
 	\item[Contact forms $\theta^a_\mu$\,:]$\quad\theta^a_\mu \;=\,\dd A^a_\mu
 	- A^a_{\mu\nu}\,\dd x^\nu.$
 	\item[Total derivative $\der{}{x^\nu}$\,:]$\quad\frac{\dd}{\dd x^\nu}
  		\;=\; \frac{\partial}{\partial x^\nu}
      + A^a_{\mu\nu}\,\frac{\partial}{\partial A^a_\mu}
      + A^a_{\mu\sigma\nu}\,\frac{\partial}{\partial A^a_{\mu\sigma}}.$
	\item[Partial derivatives of the $F$ tensor\,:]
	\[
		\frac{\partial F_{\rho\sigma}^{a}}
		{\partial A_{\mu\nu}^{c}}
		=
		\delta^{a}_{c}\,\delta^{\mu}_{\sigma}\,\delta^{\nu}_{\rho}
		-
		\delta^{a}_{c}\,\delta^{\mu}_{\rho}\,\delta^{\nu}_{\sigma}
		=
		\delta^{a}_{c}
		\left(
		\delta^{\mu}_{\sigma}\delta^{\nu}_{\rho}
		-
		\delta^{\mu}_{\rho}\delta^{\nu}_{\sigma}
		\right),
	\]
	\[
		\frac{\partial F_{\rho\sigma}^{a}}
		{\partial A_{\lambda}^{c}}
		=
		\frac{\partial}{\partial A_{\lambda}^{c}}
		\left[
		C^{a}_{de}\,A_{\rho}^{d}A_{\sigma}^{e}
		\right]
		=
		C^{a}_{de}\,\delta^{d}_{c}\,\delta^{\lambda}_{\rho}\,A_{\sigma}^{e}
		+
		C^{a}_{de}\,A_{\rho}^{d}\,\delta^{e}_{c}\,\delta^{\lambda}_{\sigma}
		\equiv
		C^a_{cd}(\delta^\lambda_\rho A^d_\sigma-\delta^\lambda_\sigma A^d_\rho).
	\]	
	\item[Multimomenta $p^a_{\mu\nu}$\,:] The whole computation is done in
	\ref{eq:explicit_computation_of_multimomentum}. Note the $(\mu,\nu)$ antisymmetry
	inherited from $F$
		\[
			p_{a}^{\mu\nu}:=\frac{\partial L_{YM}}{\partial A^{a}_{\mu\nu}}=
			\sqrt{|g|}\kappa_{ab}\,F^{b\,\mu\nu}\equiv \sqrt{|g|}\kappa_{ab}g^{\mu\rho}g^{\nu\sigma}
			F^b_{\rho\sigma}\equiv\sqrt{|g|}F_a^{\mu\nu}.
		\]
	\item[Partial derivative of $L_{YM}$ w.r.t. $A^a_{\mu}$\,:] The whole computation is done in
	\ref{eq:pd_of_L_wrt_field_components}
	\[
		\frac{\partial L_{YM}}{\partial A^{a}_{\mu}}=-\sqrt{|g|}\kappa_{cd}C^c_{ab}A^b_\nu
		F^{d\mu\nu}\equiv-C^c_{ab}A^b_\nu p_c^{\mu\nu}\equiv C^c_{ab}A^b_\nu p_c^{\nu\mu}.
	\]
 \end{description}
\section[Poincaré-Cartan forms, Legendre transformations and the constraint algorithm][Poincaré--Cartan forms and the constraint algorithm]{Poincaré-Cartan forms, Legendre transformations and the constraint algorithm}
\label{sec:poincare_cartan_forms_legendre_transformations_and_the_constraint_algorithm}
Is straightforward to express the Poincaré-Cartan forms on $Z$ of Definition
\ref{def:poincare-cartan-forms} as
\[
	\Theta_{L_{YM}}(z)=L_{YM}(z)\eta|_z+\parder{L_{YM}(z)}{A^a_{\mu\nu}}\theta^a_\mu\wedge
	\dd^nx_\nu|_z=L_{YM}(z)\eta|_z+p_a^{\mu\nu}(z)\theta^a_\mu\wedge\dd^nx_\nu|_z,
\]
\begin{align}
	\Omega_{L_{YM}}(z)&=-\theta^a_\mu\wedge\left(\parder{L_{YM}(z)}{A^a_\mu}\dd^{n+1}x-
	\dd\,\left(
	\parder{L_{YM}(z)}{A^a_{\mu\nu}}\right)\dd^nx_\nu\right)|_z\\
	&=C^c_{ab}A^b_\nu(z)p_c^{\mu\nu}(z)\theta^a_\mu\wedge\dd^{n+1}x|_z+\theta^a_\mu\wedge
	\dd p_a^{\mu\nu}
	\wedge\dd^nx_\nu|_z
\end{align}
where $\dd p_a^{\mu\nu}$ is computed in \ref{eq:explicit_computation_of_multimomentum}.
For the Legendre transformations, we will denote by $(x^\mu,A^a_\mu,p,\hat{p}_a^{\mu\nu})$
the $\Omega^{n+1}_2 C(P)$ coordinates in order to differentiate $\hat{p}_a^{\mu\nu}$ (the
component) and $p_a^{\mu\nu}$ (the multimomenta with values in $Z$). Then, the Legendre
transformations read
\[
	\Leg_{L_{YM}}\left(x^\mu,A^a_\mu,A^a_{\mu\nu}\right)=\left(x^\mu,A^a_\mu,p=L_{YM}-
	A^a_{\mu\nu}\,p_a^{\mu\nu},\hat{p}_a^{\mu\nu}=p_a^{\mu\nu}\right),
\]
\[
	\leg_{L_{YM}}\left(x^\mu,A^a_\mu,A^a_{\mu\nu}\right)=\left(x^\mu,A^a_\mu,
	\hat{p}_a^{\mu\nu}=p_a^{\mu\nu}\right).
\]
From the antisymmetry of $p_a^{\mu\nu}$ we can read two things. First, we can express $p$
in  more manageable way as follows
\[
	p=p=L_{YM}-A^a_{\mu\nu}\,p_a^{[\mu\nu]}=L_{YM}-A^a_{[\mu\nu]}\,p_a^{[\mu\nu]}=
	L_{YM}-\frac{1}{2}\left(A^a_{\mu\nu}-A^a_{\nu\mu}\right)p_a^{\mu\nu},
\]
and using the definition of $F^a_{\mu\nu}$ to replace the term $(A^a_{\mu\nu}-A^a_{\nu\mu})$,
we have
\begin{align*}
	p&=L_{YM}+\frac{1}{2}F^a_{\mu\nu}p_a^{\mu\nu}-\frac{1}{2}C^a_{bc}A^b_\mu A^c_\nu
	p_a^{\mu\nu}=L_{YM}-2L_{YM}-\frac{1}{2}C^a_{bc}A^b_\mu A^c_\nu p_a^{\mu\nu}\\
	&=-L_{YM}-\frac{1}{2}C^a_{bc}A^b_\mu A^c_\nu p_a^{\mu\nu}.
\end{align*}
Second, we observe that the submanifold $\tilde{M}_1:=\Leg_{L_{YM}}(Z)$ of the constraint
algorithm of Subsection \ref{subsect:almost_regular_lagrangians_and_the_constraint_algorithm}
is given by $\{\hat{p}_a^{\mu\nu}=-\hat{p}_a^{\nu\mu}\}$ with $p=-L_{YM}-1/2C^a_{bc}A^b_\mu A^c_\nu$
fixed; equivalently by the submanifold of $\Omega^{n+1}_2 Y$ where the $\hat{p}_a^{\mu\nu}$ coordinates
satisfy the $N\frac{(n+1)(n+2)}{2}$ constraints $\hat{p}_a^{(\mu\nu)}=0$ and $p$ is held fixed.
As $M_1:=\leg_{L_{YM}}(Z)\cong\tilde{M}_1$, $M_1$ is
equivalently defined. In both manifolds, the number of constraints is equal to
$N\frac{(n+1)(n+2)}{2}$ and so $\dim(\tilde{M}_1)=\dim(M_1)=\dim(Z^*)-\sharp\{\text{of
constraints}\}=(n+1)+(n+1)N+\frac{n(n+1)}{2}N$. Yang--Mills Lagrangian is not regular,
as we already said, but a special one, since the constraint algorithm ends in the first step
(at $M_1$ as we will discuss on Section \ref{sec:the_connection_approach_to_yang_mills_equations}).
\section{Yang--Mills equations}
\label{sec:yang_mills_equations}
Given a connection $\omega$ and the corresponding section $\phi_\omega\in\Gamma(C(P))$ due to
Proposition \ref{prop:connections_as_sections_explicit} with $j^1\phi_\omega(x)=
(x^\mu,A^a_{\mu}(x),\partial_\nu A^a_\mu(x))$, where $A_a^\mu(x)$ are the components of
$A_{\alpha}=s_{\alpha}^{*}\omega=A_\mu^a(x)\dd x^\mu\otimes e_a$ once chosen a local
trivialization $T_\alpha$ of our PFB, Euler-Lagrange equations for our Yang--Mills
Lagrangian read from \ref{eq:euler-lagrange_equations} 
\[
	(j^2\phi_\omega)^*\left(\parder{L_{YM}}{A^a_\mu}-\der{}{x^\nu}\parder{L_{YM}}
	{A^a_{\mu\nu}}\right)=(j^2\phi_\omega)^*\left(-C^c_{ab}A^b_\nu p_c^{\mu\nu}-\der{}{x^\nu}
	p_a^{\mu\nu}\right)
\]
\[
	=-\left[\partial_\nu p_a^{\mu\nu}(x)+C^c_{ab}A^b_\nu(x)p_c^{\mu\nu}(x)\right]=0.
\]
Then, since $p_a^{\mu\nu}$ (the form it defines) inherits tensoriality from $F^a_{\mu\nu}$
(which does from $\Omega^\omega$), we can use the component form of Equation
\eqref{eq:D_coadjoint_components} to express Yang--Mills equations in components as
\begin{equation}
	\label{eq:Yang-Mills_equations}
	\boxed{
	D_\nu p_a^{\mu\nu}(x)=0,\quad\forall\mu\in\{0,\dots,n\},a\in\{1,\dots,N\}
	}\,.
\end{equation}
We can also express the equation in term of the field strength (we omit the $x$ dependence
for clarity)
\begin{align*}
	D_\nu p_a^{\mu\nu}(x)=D_\nu\left[\sqrt{|g|}\kappa_{ab}F^{b\mu\nu}\right]=0\implies
	\partial_\nu\left[\sqrt{|g|}F^{b\mu\nu}\right]+\sqrt{|g|}C^b_{cd}A^c_\nu F^{d\mu\nu}=0.
\end{align*}
Using the well known fact
\[
	\partial_\nu\left[\sqrt{|g(x)|}F^{b\mu\nu}(x)\right]=\sqrt{|g(x)|}\nabla_\nu F^{b\mu\nu}(x),
\]
where $\nabla$ is the covariant derivative in the spacetime $X$, we have finally
\[
	\mathcal{D}_\nu F^{a\mu\nu}(x)\equiv\nabla_\nu F^{b\mu\nu}(x)+
	C^b_{cd}A^c_\nu(x)F^{d\mu\nu}(x)=0\quad\forall\mu\in\{0,\dots,n\},a\in\{1,\dots,N\}.	
\]

It remains to check that \eqref{eq:Yang-Mills_equations} is exactly the intrinsic
equation $D_A\star F=0$ announced in Remark \ref{rem:YM_invariant_form}. The bridge is the
following $(\ad P)^*$-valued $(n-1)$-form, which will reappear in Subsection
\ref{sub:the_generalization_of_hamiltonian_forms}:
\begin{equation}
	\label{eq:bbF_definition}
	\F_a:=\frac{1}{2}\,p^{\mu\nu}_a(x)\,\dd^{n-1}x_{\mu\nu}\ \in\ \Omega^{n-1}(X;(\ad P)^*).
\end{equation}
 
\begin{lemma}
	\label{lem:F_is_star_F}
	$\F_a=\kappa_{ab}\,\star\! F^b$.
\end{lemma}
 
\begin{proof}
	Let $\beta=\frac12\beta_{\mu\nu}\dd x^\mu\wedge\dd x^\nu$ be any $2$-form. Expanding
	$\dd^{n-1}x_{\mu\nu}=\iota_{\partial_\nu}\iota_{\partial_\mu}\eta$ in the coordinate
	basis gives
	$\dd^{n-1}x_{\mu\nu}=\frac{1}{(n-1)!}\,\varepsilon_{\mu\nu\nu_1\dots\nu_{n-1}}
	\dd x^{\nu_1}\wedge\dots\wedge\dd x^{\nu_{n-1}}$, whence
	\[
		\frac12\sqrt{|g|}\,\beta^{\mu\nu}\dd^{n-1}x_{\mu\nu}
		=\frac{1}{(n-1)!}\left[\frac{\sqrt{|g|}}{2!}\,\beta^{\mu\nu}
		\varepsilon_{\mu\nu\nu_1\dots\nu_{n-1}}\right]\dd x^{\nu_1}\wedge\dots\wedge\dd x^{\nu_{n-1}}
		=\star\beta
	\]
	by the coordinate expression of Remark \ref{rem:hodge_coordinates} with $k=2$. Applying this to
	$\beta=F^b$ and using $p^{\mu\nu}_a=\sqrt{|g|}\,\kappa_{ab}F^{b\,\mu\nu}$ from
	Section \ref{sec:initial_clarifications} yields the claim.
\end{proof}
 
\begin{proposition}
	\label{prop:YM_is_DstarF}
	For every $x\in X$,
	\[
		\big(D_A\F\big)_a=\big(D_\nu p^{\mu\nu}_a\big)\,\dd^n x_\mu ,
	\]
	and consequently
	\[
		D_\nu p^{\mu\nu}_a=0\quad\forall\mu,a
		\qquad\Longleftrightarrow\qquad
		D_A\!\star\! F=0 .
	\]
\end{proposition}
 
\begin{proof}
	The covariant exterior derivative on $(\ad P)^*$-valued forms is the coadjoint
	version of \eqref{eq:D_general_case}, namely Equation \eqref{eq:D_coadjoint_components}
	for $V=(\ad P)^*$. We will use twice the identity
	\eqref{eq:crucial_identity_superpotential}.
 
	First, we have
	\[
		\dd\F_a=\frac12\,\partial_\alpha p^{\mu\nu}_a\,\dd x^\alpha\wedge\dd^{n-1}x_{\mu\nu}
		=\frac12\left(\partial_\nu p^{\mu\nu}_a\,\dd^n x_\mu-\partial_\mu p^{\mu\nu}_a\,\dd^n x_\nu\right)
		=\partial_\nu p^{\mu\nu}_a\,\dd^n x_\mu ,
	\]
	the last step relabelling $\mu\leftrightarrow\nu$ in the second summand and using
	$p^{\nu\mu}_a=-p^{\mu\nu}_a$. Second,
	\[
		C^c_{ab}\,A^b\wedge\F_c
		=\frac12\,C^c_{ab}A^b_\alpha p^{\mu\nu}_c\,\dd x^\alpha\wedge\dd^{n-1}x_{\mu\nu}
		=C^c_{ab}A^b_\nu p^{\mu\nu}_c\,\dd^n x_\mu .
	\]
	Adding the two,
	\[
		(D_A\F)_a=\Big(\partial_\nu p^{\mu\nu}_a+C^c_{ab}A^b_\nu p^{\mu\nu}_c\Big)\dd^n x_\mu
		=\big(D_\nu p^{\mu\nu}_a\big)\,\dd^n x_\mu .
	\]
	Finally, by Lemma \ref{lem:F_is_star_F},
	$(D_A\F)_a=\kappa_{ab}(D_A\!\star\! F)^b$ while $\kappa$ is $D_A$-parallel because it is
	$\mathrm{Ad}$-invariant and $\kappa$ is non-degenerate, so
	$(D_A\F)_a=0\iff D_A\!\star\! F=0$.
\end{proof}
 
\begin{remark}
	\label{rem:no_sqrtg_needed}
	Proposition~\ref{prop:YM_is_DstarF} makes the passage through
	$\partial_\nu(\sqrt{|g|}F^{b\mu\nu})=\sqrt{|g|}\nabla_\nu F^{b\mu\nu}$ unnecessary:
	the metric enters only through the definition of $p^{\mu\nu}_a$, and the Levi-Civita
	connection of $X$ never appears, as it should be, since $\star F$ is an
	$\mathrm{ad}^*P$-valued form and $D_A$ only differentiates the bundle index.
\end{remark}

\subsection{Intrinsic Yang--Mills equations}
\label{sub:intrinsic_yang_mills_equations}
Recalling \eqref{eq:intrinsic_EL_equations} in the general framework, the intrinsic Yang-Mills
equations have the form
\[
	(j^1\phi_\omega)^*(\iota_\zeta\Omega_{L_{YM}})=0,\quad\forall\zeta\in\mathfrak{X}(Z)
\]
with $Z=J^1(C(P))$. Using $\Omega_{L_{YM}}$ expression of Section
\ref{sec:poincare_cartan_forms_legendre_transformations_and_the_constraint_algorithm}
with $\dd p_a^{\mu\nu}$ computed in \ref{eq:explicit_computation_of_multimomentum} we have
\begin{align*}
\Omega_{L_{YM}}
&=\left[C^c_{ab}A^b_\nu p_c^{\mu\nu}+\kappa_{ab}F^{b\mu\rho}
\partial_\rho(\sqrt{|g|})\right]\theta^a_\mu\wedge d^{n+1}x\\
&+
\sqrt{|g|}\,\kappa_{ab}\,g^{\mu\rho}g^{\nu\sigma}
\theta^a_\mu\wedge
\left[
dA_{\sigma\rho}^{b}-dA_{\rho\sigma}^{b}
+C^{b}_{cd}\left(A_{\rho}^{c}(dA_{\sigma}^{d})-A_{\sigma}^{c}(dA_{\rho}^{d})\right)
\right]
\end{align*}
and all the information of the intrinsic Yang-Mills equations is contained in the following
three cases
\begin{enumerate}
	\item $\zeta=\parder{}{A^z_{\eta\theta}}$:
		\begin{align*}
			(j^1\phi_\omega)^*(\iota_\zeta\Omega_{L_{YM}}) &= (j^1\phi_\omega)^*
			\left[
			\sqrt{|g|} \kappa_{ab} g^{\mu\rho} g^{\nu\sigma} \left(\delta_z^b \delta^\eta_\sigma
			\delta_\rho^\theta-\delta_z^b \delta_\rho^\eta \delta_\sigma^\theta 
			\right)
			\theta_\mu^a \wedge\dd^n x_\nu \right] \\
			&= (j^1\phi_\omega)^*
			\left[
			\sqrt{|g|}\kappa_{az} \left( g^{\mu\theta} g^{\nu\eta}-g^{\mu\eta}g^{\nu\theta}
			\right) \theta_\mu^a \wedge\dd^n x_\nu \right]
		\end{align*}	
	which clearly goes to zero as is the pullback of a contact form by an holonomic section.
	The interesting part in this case is that, recalling \ref{eq:explicit_hessian} we have
	exactly
	\[
		(j^1\phi_\omega)^*(\iota_\zeta\Omega_{L_{YM}})=\frac{\partial^{2}L_{YM}}
		{\partial A^{a}_{\mu\nu}\,\partial A^{z}_{\eta\theta}}\theta_\mu^a \wedge
		\dd^n x_\nu=0;
	\]
	in particular
	\begin{itemize}
		\item $\zeta=\parder{}{A^z_{(\eta\theta)}}\longrightarrow \iota_\zeta\Omega_{L_{YM}}=0$.
		These are specifically the vectors spaning the kernel of the Hessian; no information
		given by the e.o.m.
		\item $\zeta=\parder{}{A^z_{[\eta\theta]}}\longrightarrow \iota_\zeta\Omega_{L_{YM}}\neq0$,
		$(j^1\phi_\omega)^*\iota_\zeta\Omega_{L_{YM}}=0.$ This is the restriction imposed
		on $\partial_{[\theta}A^z_{\eta]}(x)$.
	\end{itemize}
	\item $\zeta=\parder{}{A^z_{\eta}}$: Using $\theta^a_\mu=dA^a_\mu-A^a_{\mu\nu}\dd x^\nu$
	we have
		\begin{align*}
			(j^1\phi_\omega)^*(\iota_\zeta\Omega_{L_{YM}})=(j^1\phi_\omega)^*
			\big[&
				C^c_{ab}A^b_\nu p_c^{\mu\nu}\delta^a_z\delta^\eta_\mu\dd^{n+1}x+
				\delta^a_z\delta^\eta_\mu\dd p_a^{\mu\nu}\wedge\dd^nx_\nu\\
				&-\cancel{(\iota_\zeta
				\dd p_a^{\mu\nu})\theta^a_\mu\wedge\dd^nx_\nu},
			\big]
		\end{align*}
		where the last term cancels when pullbacking because it being a contact form. Now,
		applying the pullback to $\dd^{n+1}x$ is trivial, while for $\dd p_a^{\mu\nu}$ we use
		\eqref{eq:pullbacking_dform}, obtaining
		\begin{align*}
			(j^1\phi_\omega)^*(\iota_\zeta\Omega_{L_{YM}}) &= C^c_{zb}A^b_\nu(x)p_c^{\eta\nu}(x)
			\dd^{n+1}x+\partial_\nu p_z^{\eta\nu}(x)\dd^{n+1}x=(D_\nu p_z^{\eta\nu})=0,
		\end{align*}
		i.e. specifically Yang--Mills equations.
	\item $\zeta=\parder{}{x^\eta}$: Using $\theta^a_\mu=dA^a_\mu-A^a_{\mu\nu}\dd x^\nu$
	and cancelling contact forms once again we have
		\begin{align*}
			(j^1\phi_\omega)^*(\iota_\zeta\Omega_{L_{YM}})&=
			(j^1\phi_\omega)^*
			\left[
			C_{ab}^c A_\nu^b p_c^{\mu\nu} (-A_{\mu\eta}^a)\dd^{n+1} x
			-\cancel{C_{ab}^c A_\nu^b p_c^{\mu\nu} \theta_\mu^a \wedge\dd^n x_\nu}\right.\\
			&- A_{\mu\eta}^a\dd p_a^{\mu\nu}\wedge\dd^n x_\nu - \cancel{\theta_\mu^a
			\wedge\iota_\zeta(\dd p_a^{\mu\nu}\wedge\dd^n x_\nu)}
			\left.\right]\\
			&=-\partial_\eta A^a_\mu(x)
			\left[
			\partial_\nu p_a^{\mu\nu}(x)+C_{ab}^c A_\nu^b(x)
			p_c^{\mu\nu}(x)
			\right]\dd^{n+1}x=0,
		\end{align*}
		namely $n+1$ (running on $\eta$ free index) linear combinations of Yang-Mills
		equations. This does not add any more information as because of $\mathrm{ii})$
		all terms in all linear combinations are exactly zero.	
\end{enumerate}
\section{De Donder equations}
As we know, De Donder equations have the same form as the intrinsic Euler--Lagrange equations
but for a general section $\sigma\in\Gamma(\pi_{ZX})$ that may not be holonomic
(hence does not annihilate contact forms when pullbacking in general). However, due to the
specific construction of the jet bundle, there is a connection $\omega$ such that
$\sigma(x)=(x^\mu,A^a_\mu(x),S^a_{\mu\nu}(x))$, where as we said in general
$S^a_{\mu\nu}\neq\partial_\nu A^a_\mu(x)$; then, we declare the following notation
\begin{equation}
  \label{eq:YM_pullback_contacto_no_holonomo}
  \sigma^{*}\theta^{a}_{\mu}
  =\sigma^{*}\bigl(\dd A^{a}_{\mu}-A^{a}_{\mu\nu}\,\dd x^{\nu}\bigr)
  =\partial_{\lambda}A^{a}_{\mu}(x)\,\dd x^{\lambda}-S^{a}_{\mu\nu}(x)\,\dd x^{\nu}
  =\underbrace{\bigl(\partial_{\nu}A^{a}_{\mu}(x)-S^{a}_{\mu\nu}(x)\bigr)}
    _{=:\;\Delta^{a}_{\mu\nu}(x)}\,\dd x^{\nu}.
\end{equation}
Using $\Omega_{L}=C^{c}_{ab}A^{b}_{\nu}\,p_{c}^{\mu\nu}\;\theta^{a}_{\mu}\wedge\dd^{n+1}x
+\theta^{a}_{\mu}\wedge\dd p_{a}^{\mu\nu}\wedge\dd^{n}x_{\nu}$ from Section
\ref{sec:poincare_cartan_forms_legendre_transformations_and_the_constraint_algorithm},
we divide the computation in three cases
\begin{enumerate}
	\item $\zeta=\dfrac{\partial}{\partial A^{z}_{\eta,\theta}}$: The beginning of the
	computation is equivalent to the case $\mathrm{i})$ of Subsection
	\ref{sub:intrinsic_yang_mills_equations} and so we have
		\begin{align*}
			\sigma^*(\iota_\zeta\Omega_{L_{YM}}) &= \sigma^*
			\left[
			\sqrt{|g|}\kappa_{az}\left( g^{\mu\theta} g^{\nu\eta}-g^{\mu\eta}g^{\nu\theta}
			\right) \theta_\mu^a\wedge\dd^n x_\nu \right].
		\end{align*}	
		Using then \eqref{eq:YM_pullback_contacto_no_holonomo} it is easy to see that 
		\[
			\sigma^*(\iota_\zeta\Omega_{L_{YM}})=\sqrt{|g(x)|}\kappa_{az}
			\left(
				\Delta^{a\theta\eta}(x)-\Delta^{a\eta\theta}(x)
			\right)\dd^{n+1}x=0\iff \Delta^a_{[\mu\nu]}(x)=0,
		\]
		which then implies
		\[
			\boxed{
			S^a_{[\mu\nu]}(x)=\partial_{[\nu}A^a_{\mu]}(x)\quad\forall x\in X
			}
		\]
		i.e. holonomicity but only on the antisymmetric components of our section $\sigma$;
		in general there is no restriction to $S^a_{(\mu\nu)}(x)$.
		\item $\zeta=\dfrac{\partial}{\partial A^{z}_{\eta}}$: Once again we use Subsection
		\ref{sub:intrinsic_yang_mills_equations} to see
		\begin{align*}
			\sigma^*(\iota_\zeta\Omega_{L_{YM}}) &= \sigma^*
			\left[
				C^c_{zb}A^b_\nu p_c^{\eta\nu}\dd^{n+1}x+
				\dd p_z^{\eta\nu}\wedge\dd^nx_\nu-(\iota_\zeta
				\dd p_a^{\mu\nu})\theta^a_\mu\wedge\dd^nx_\nu
			\right]\\
			&=C^c_{zb}A^b_\nu(x)p_c^{\eta\nu}(x)\dd^{n+1}x+
				\partial_\nu p_z^{\eta\nu}(x)\dd^{n+1}x-\sigma^*(\iota_\zeta
				\dd p_a^{\mu\nu})\Delta^a_{\mu\nu}\dd^{n+1}x,
		\end{align*}
		where the dependence of $x$ is through $\sigma(x)$. Let us compute separately
		the term $\sigma^*(\iota_\zeta\dd p_a^{\mu\nu})$
		\begin{align*}
		  \sigma^{*}\bigl[\iota_{\zeta}\dd p_{a}^{\mu\nu}\bigr]
		  &=\sigma^{*}\Bigl[\sqrt{|g|}\,\kappa_{ab}\,g^{\mu\rho}g^{\nu\sigma}\,C^{b}_{cd}
		      \bigl(A^{c}_{\sigma}\,\delta^{d}_{z}\delta^{\eta}_{\rho}
		           -A^{c}_{\rho}\,\delta^{d}_{z}\delta^{\eta}_{\sigma}\bigr)\Bigr]
		   \\
		  &=\sqrt{|g|}\,\kappa_{ab}\Bigl(
		      g^{\mu\eta}g^{\nu\sigma}\,C^{b}_{cz}\,A^{c}_{\sigma}(x)
		     -g^{\mu\rho}g^{\nu\eta}\,C^{b}_{cz}\,A^{c}_{\rho}(x)\Bigr)
		   \\
		  &=\sqrt{|g|}\,\kappa_{ab}\,C^{b}_{cz}
		    \bigl(g^{\mu\eta}A^{c\nu}(x)-g^{\nu\eta}A^{c\mu}(x)\bigr),
		\end{align*}	
		which is clearly $(\mu,\nu)$-antisymmetric. As this term is multiplied by
		$\Delta^a_{\mu\nu}$ which has zero antisymmetric part on $(\mu,\nu)$ due to
		the results obtained in $\mathrm{i})$, the whole term annihilates, and so
		\begin{align*}
			\sigma^*(\iota_\zeta\Omega_{L_{YM}}) &=
			\left[
				C^c_{zb}A^b_\nu(x)p_c^{\eta\nu}(x)+
				\partial_\nu p_z^{\eta\nu}(x)
			\right]\dd^{n+1}x=0\\
			&\implies 
			\boxed{
			D_\nu p_z^{\eta\nu}(\sigma(x))=0}\,,
		\end{align*}
		i.e. Yang--Mills equations. Recall that $F^a_{\mu\nu}(\sigma(x))$ and hence
		$p_a^{\mu\nu}(\sigma(x))$ depends on $S^a_{\mu\nu}$ only through its antisymmetric
		part; hence Yang--Mills equations do not restric $S^a_{(\mu\nu)}$ either. These
		degrees of freedom are not physical, are just an artifact of expanding the intrinsic
		Yang--Mills equations to non-holonomic sections, which do not represent physical
		fields. 

		\item $\zeta=\dfrac{\partial}{\partial x^{\eta}}$: This case can be done in a
		similar way to case $\mathrm{iii})$ of Subsection
		\ref{sub:intrinsic_yang_mills_equations} but there is an argument that will save
		us some computations. Imagine there is a vector field $\xi\in\mathfrak{X}(X)$ such
		that $\sigma_*\xi=\zeta$, in that case
		\[
			\sigma^*(\iota_\zeta\Omega_{L_{YM}})=\sigma^*(\iota_{\sigma_*\xi}\Omega_{L_{YM}})
			=\iota_\xi(\sigma^*\Omega_{L_{YM}})=0
		\]
		where last equality comes from the fact that we are pullbacking a $(n+2)$-form down
		to $X$, hence it annihilates. Now, take $\partial_\eta$ as a basis for our vector
		fields in $X$, we know that
		\[
			\sigma_*\bigg(\underbrace{\parder{}{x^\eta}}_{:=\xi}\bigg)=
			\underbrace{\parder{}{x^\eta}}_{=\zeta}+
			\partial_\eta A^a_\mu(x)\parder{}{A^a_\mu}+\partial_\eta S^a_{\mu\nu}(x)
			\parder{}{A^a_{\mu\nu}}\implies
		\]
		\[
			\implies \zeta=\sigma_*(\xi)-\partial_\eta A^a_\mu(x)\parder{}{A^a_\mu}
			-\partial_\eta S^a_{\mu\nu}(x)\parder{}{A^a_{\mu\nu}}\equiv\sigma_*(\xi)+B
		\]
		and so our third case transforms into
		\begin{align*}
			\sigma^*(\iota_\zeta\Omega_{L_{YM}})=\cancel{\sigma^*(\iota_{\sigma^*(\xi)}
			\Omega_{L_{YM}})}+\sigma^*(\iota_B\Omega_{L_{YM}}),
		\end{align*}
		and we just have to realise that $B$ is `spanned' by our previous $\mathrm{i})$
		and $\mathrm{ii})$ cases; to summarize, this horizontal case gives no additional
		information.
\end{enumerate}
\section{Hamilton--De Donder--Weyl equations on \texorpdfstring{$M_1$}{M1}}
Consider a section $\sigma\in\Gamma(\pi_{ZX})$ and let $s:=\leg_{L_{YM}}\circ\,\sigma
\in\Gamma(\pi_{M_1X})$; then in coordinates we write
\[
	s(x)=\left(x^\mu,A^a_\mu(x),s_a^{\mu\nu}(x)\right),\quad\text{with }s_a^{(\mu\nu)}=0.
\]
Hamilton--De Donder--Weyl equations \eqref{eq:Hamilton-equations} then read
\[
	\frac{\partial A^a_\mu(x)}{\partial x^\nu}=\frac{\partial H}{\partial p_a^{\mu\nu}}
	(s(x)),\quad\sum_{\nu=0}^n\frac{\partial s_a^{\mu\nu}(x)}{\partial x^\nu}=-
	\frac{\partial H}{\partial A^a_\mu}(s(x)),
\]
where $H=-p=L_{YM}+\frac{1}{2}C^a_{bc}A^b_\mu A^c_\nu p_a^{\mu\nu}$ and 
$L_{YM}(\hat{z})\equiv-\frac{1}{4\sqrt{|g|}}p^a_{\mu\nu}p_a^{\mu\nu}$ for $\hat{z}\in M_1$.
From the first equation we have
\begin{align*}
	\partial_\nu A^a_\mu(x)&=-\frac{1}{2\sqrt{|g|}}s^a_{\mu\nu}(x)+\frac{1}{2}C^a_{bc}
	A^b_\mu(x)A^c_\nu(x)\\
	&\implies s^a_{\mu\nu}(x)=\sqrt{|g|}\left(
		-2\partial_\nu A^a_\mu(x)+C^a_{bc}A^b_\mu(x)A^c_\nu(x)
	\right).
\end{align*}
Since we are in $M_1$, $s_a^{\mu\nu}(x)=s_a^{[\mu\nu]}(x)$ and so this last equation is
equivalent to
\[
\boxed{
	s^a_{[\mu\nu]}(x)=\sqrt{|g|}\left(
	\partial_\mu A^a_\nu(x)-\partial_\nu A^a_\mu(x)+C^a_{bc}A^b_\mu(x)A^c_\nu(x)
	\right)=p^a_{\mu\nu}(x)
	}\,,
\]
namely, the antisymmetric components of the `momenta' of the section $s$ must
be equal to the multimomenta of our field, while the symmetric part is exactly zero
(this part is non-zero in general if working in $Z^*$).
From the second equation we get
\[
\boxed{
	\partial_\nu s_a^{\mu\nu}(x)+C^c_{ab}A^b_\nu(x)s_c^{\mu\nu}(x)\equiv D_\nu s_a^{\mu\nu}(x)=0
	}\,,
\]
and since from the first equation $s_a^{\mu\nu}(x)$ equals the multimomenta, we recover
once again Yang--Mills equations.
\section{The connection approach to Yang--Mills equations}
\label{sec:the_connection_approach_to_yang_mills_equations}
Let us now dive into the connection approach to the e.o.m. As mentioned in the main text,
this approach is the natural one in order to formulate a constraint algorithm and
construct a Hamiltonian description equivalent to the Lagrangian one in the case of almost
regular Lagrangians (e.g. Yang--Mills). Let us consider a general connection $\Gamma$
in $\pi_{ZX}$ with horizontal projector given by
\[
	\hh_z\left(\parder{}{x^\mu}\bigg|_z\right)=\parder{}{x^\mu}\bigg|_z+
	(\Gamma_\mu)^a_\rho(z)\parder{}{A^a_\rho}\bigg|_z+(\Gamma_\mu)^a_{\rho\sigma}(z)
	\parder{}{A^a_{\rho\sigma}}\bigg|_z,\quad\forall z\in Z.
\]
To shorter this section, we will use Corollary \ref{cor:some_iotah_computations}
in which we have archived some useful computations regarding $\iota_\hh$. Once noted this and
using $\Omega_{L_{YM}}$ expression of Section
\ref{sec:poincare_cartan_forms_legendre_transformations_and_the_constraint_algorithm},
we compute Equation \eqref{eq:geometric-De-Donder} (which lives in $Z$), obtaining
(recall we are solving for $\Gamma$ components)
\begin{align*}
\iota_{\hh}\Omega_{L_{YM}}
\;&=
(\cancel{n}+1)C^{c}_{ab}\,A^{b}_{\nu}\,p_{c}^{\mu\nu}\,
\dd A^{a}_{\mu}\wedge\dd^{n+1}x
\;-\;
\bigl((\Gamma_{\nu})^{a}_{\mu}-A^{a}_{\mu\nu}\bigr)\,
\dd p_{a}^{\mu\nu}\wedge\dd^{n+1}x
\\
&+\theta^{a}_{\mu}\wedge\dd^{n+1}x
\left(
\partial_{\nu}p_{a}^{\mu\nu}
+(\Gamma_{\nu})^{b}_{\rho}\,
\frac{\partial p_{a}^{\mu\nu}}{\partial A^{b}_{\rho}}
+(\Gamma_{\nu})^{b}_{\rho\sigma}\,
\frac{\partial p_{a}^{\mu\nu}}{\partial A^{b}_{\rho\sigma}}
\right)
\\
&+\bcancel{n\,\theta^{a}_{\mu}\wedge\dd p_{a}^{\mu\nu}\wedge\dd^{n}x_{\nu}}
\;\overset{!}{=}\;n\,\Omega_{L_{YM}}
\;=\;\cancel{n\,C^{c}_{ab}\,A^{b}_{\nu}\,p_{c}^{\mu\nu}\,
\dd A^{a}_{\mu}\wedge\dd^{n+1}x}
\\
&+\bcancel{n\,\theta^{a}_{\mu}\wedge\dd p_{a}^{\mu\nu}\wedge\dd^{n}x_{\nu}}
\\[4pt]
\Longleftrightarrow\;\;&
C^{c}_{ab}\,A^{b}_{\nu}\,p_{c}^{\mu\nu}\,
\dd A^{a}_{\mu}\wedge\dd^{n+1}x
-\bigl((\Gamma_{\nu})^{a}_{\mu}-A^{a}_{\mu\nu}\bigr)\,
\dd p_{a}^{\mu\nu}\wedge\dd^{n+1}x
\\
&+\left(
\partial_{\nu}p_{a}^{\mu\nu}
+(\Gamma_{\nu})^{b}_{\rho}\,
\frac{\partial p_{a}^{\mu\nu}}{\partial A^{b}_{\rho}}
+(\Gamma_{\nu})^{b}_{\rho\sigma}\,
\frac{\partial p_{a}^{\mu\nu}}{\partial A^{b}_{\rho\sigma}}
\right)
\dd A^{a}_{\mu}\wedge\dd^{n+1}x=0.
\end{align*}
Now we use Section \ref{eq:explicit_computation_of_multimomentum} to expand the
$\dd p_a^{\mu\nu}$ term and obtain
\begin{align*}
	\iota_\hh\Omega_{L_{YM}}=n\Omega_{L_{YM}}\iff
	&\biggl(
		C^{c}_{ab}\,A^{b}_{\nu}\,p_{c}^{\mu\nu}\,
		+\partial_{\nu}p_{a}^{\mu\nu}
		+(\Gamma_{\nu})^{b}_{\rho}\,
		\frac{\partial p_{a}^{\mu\nu}}{\partial A^{b}_{\rho}}
		+(\Gamma_{\nu})^{b}_{\rho\sigma}\,
		\frac{\partial p_{a}^{\mu\nu}}{\partial A^{b}_{\rho\sigma}}\\
		&-C^c_{da}\bigl((\Gamma_{\sigma})^{b}_{\rho}-A^{b}_{\rho\sigma}\bigr)
		W^{\rho\sigma,\mu\tau}_{bc}A^d_\tau\biggr)\dd A^{a}_{\mu}\wedge\dd^{n+1}x\\
		&-\bigl((\Gamma_{\nu})^{a}_{\mu}-A^{a}_{\mu\nu}\bigr)W_{ab}^{\mu\nu,\rho\sigma}\dd
		A^b_{\rho\sigma}\wedge\dd^{n+1}x=0,
\end{align*}
where $W_{ab}^{\mu\nu,\rho\sigma}$ is the Hessian computed on Section
\ref{eq:explicit_hessian} (that is $(\mu,\nu)$ and $(\rho,\sigma)$-antisymmetric).
Now since both terms are independent of each other ($\dd A^a_\mu$ and
$\dd A^b_{\rho\sigma}$ are independent on $Z$), both must be zero. From the second term we
get
\[
	\bigl((\Gamma_{\nu})^{a}_{\mu}(z)-A^{a}_{\mu\nu}(z)\bigr)W_{ab}^{\mu\nu,\rho\sigma}(z)=0
	\iff\boxed{(\Gamma_{[\nu})^{a}_{\mu]}(z)=A^{a}_{[\mu\nu]}(z)}
\]
for all $z\in Z$, where the iff condition comes from the antisymmetric part of the hessian
being invertible. The symmetric part $(\Gamma_{(\nu})^a_{\mu)}(z)$ remains free, i.e., does
not add any further restriction; this
is the actual gauge freedom, consequence of the Yang--Mills Lagrangian depending on
$A^a_{\mu\nu}$ only through its antisymmetrization.

While from the first term (the fifth term inside this form annihilates as it contains exactly
the first boxed result) we get
\begin{equation*}
\boxed{\;
\frac{\partial p_{a}^{\mu\nu}(z)}{\partial x^{\nu}}
+(\Gamma_{\nu})^{b}_{\rho}(z)\,
\frac{\partial p_{a}^{\mu\nu}(z)}{\partial A^{b}_{\rho}}
+(\Gamma_{\nu})^{b}_{\rho\sigma}(z)\,
W_{ab}^{\mu\nu,\rho\sigma}(z)
+C^{c}_{ab}\,A^{b}_{\nu}(z)\,p_{c}^{\mu\nu}(z)=0
}\,,
\end{equation*}
for all $z\in Z$. When pullbacked by an integrable holonomic section, this last equation
recovers Yang--Mills equations.
Since $W_{ab}^{\mu\nu,\rho\sigma}$ is invertible under the antisymmetric condition on
$(\mu,\nu)$ and $(\rho,\sigma)$, we can always solve for $(\Gamma_\nu)^b_{[\rho\sigma]}$
\textbf{for every point in} $\bm{Z}$ (with $(\Gamma_\nu)^b_{[\rho\sigma]}$ once again free).
Hence, the constraint algorithm of Subsection
\ref{subsect:the_singular_case} ends in the first step $Z_1=Z$ and so does its
Hamiltonian counterpart (on $M_1=\leg_{L_{YM}}(Z)$) as we mentioned in Subsection
\ref{subsect:almost_regular_lagrangians_and_the_constraint_algorithm}.

\section{Bianchi identity}
Apart from e.o.m. (on-shell identity), as already mentioned on the main text we have
an off-shell identity coming from the fact that the exterior covariant derivative of
the curvature of our connections is exactly zero, which is the Bianchi identity
$D^\omega\Omega^\omega=0$. If we choose a trivialization, this identity yields
a (local) off-shell identity on $X$
\begin{align*}
	&D^\omega\Omega^\omega=\dd\Omega^\omega+[\omega,\Omega^\omega]=0\\
	\implies s_\alpha^*(D^\omega\Omega^\omega)=&D_{A_\alpha}F_\alpha=dF_\alpha
	+[A_\alpha,F_\alpha]=0,
\end{align*}
where
\[
	F_\alpha=\frac{1}{2}F^a_{\mu\nu}(x)\dd x^\mu\wedge\dd x^\nu\otimes e_a,
\]
\[
	F^a_{\mu\nu}(x)=\partial_\mu A^a_\nu(x)-\partial_\nu A^a_\mu(x)+C^a_{bc}A^b_\mu(x)
	A^c_\nu(x)=F^a_{[\mu\nu]}(x).
\]
Substituting in the previous expression is straightforward to obtain
\[
\boxed{
	\partial_{[\alpha}F^a_{\mu\nu]}(x)+C^a_{bc}A^b_{[\alpha}(x)F^c_{\mu\nu]}(x)=0,
	\quad\forall x\in X
	}
\]
which we recall is an off-shell identity. Its intrinsic version is just $D_A F=0$ as already
anticipated in Remark \ref{rem:YM_invariant_form}.

\section{Symmetries and conserved quantities}
We now proceed computing conserved quantities following Chapter
\ref{chap:symmetries_and_Noether_theorems}. We will denote $j^1\xi_C$ to the 1-jet
prolongation of our gauge symmetries $\xi_C$ (instead of $\xi_C^{(1)}$). Recall that our
symmetries have components given by Equation \eqref{eq:xi_C_local} and jet prolongation
components given by Equation \eqref{eq:j1xi_C_local}. Some useful computations will be
archived on Section \ref{sub:useful_identities_for_gauge_symmetries_computations}.
\subsection{Invariance of the Lagrangian: the \texorpdfstring{$F$}{F} function}
\label{sub:invariance_of_the_YM_lagrangian}
Let us check that our gauge symmetries $\xi_C$ are indeed symmetries of the Yang--Mills
Lagrangian in the sense of Definition \ref{def:symmetries_of_the_lagrangian}. For that, we
compute the $F$ function defined in \eqref{eq:F_function}, whose local expression
\eqref{eq:F_local_expression} reads
\[
	F=j^1\xi_C(L_{YM})+\left(\parder{\xi_C^\mu}{x^\mu}+A^a_{\mu\nu}
	\parder{\xi_C^\nu}{A^a_\mu}\right)L_{YM}.
\]
In our case $\xi_C$ is $\pi$-vertical by Theorem \ref{thm:gauge_prolongation}, i.e.
$\xi_C^\mu=0$, so the whole parenthesis annihilates and we are simply left with
$F=j^1\xi_C(L_{YM})$.

Using now the Leibniz rule (vectors act as derivation over functions) together with the fact
that neither $\kappa_{ab}$ nor $g^{\mu\nu}$ depend on $A^a_\mu$ nor on
$A^a_{\mu\nu}$, and substituting \eqref{eq:j1xi_C_of_F}, we get
\begin{align*}
	F&=j^1\xi_C\left(-\frac{1}{4}\sqrt{|g|}\,F_a^{\mu\nu}F^a_{\mu\nu}\right)
	=-\frac{1}{4}\sqrt{|g|}\,\kappa_{ac}\,g^{\mu\rho}g^{\nu\sigma}\,
	j^1\xi_C\left(F^c_{\rho\sigma}F^a_{\mu\nu}\right)\\
	&=-\frac{1}{4}\sqrt{|g|}\,\kappa_{ac}\,g^{\mu\rho}g^{\nu\sigma}
	\Bigl(j^1\xi_C\left(F^c_{\rho\sigma}\right)F^a_{\mu\nu}
	+F^c_{\rho\sigma}\,j^1\xi_C\left(F^a_{\mu\nu}\right)\Bigr)\\
	&=-\frac{1}{4}\sqrt{|g|}\,\kappa_{ac}\,g^{\mu\rho}g^{\nu\sigma}
	\Bigl(C^c_{hi}F^h_{\rho\sigma}\xi^i(x)F^a_{\mu\nu}
	+F^c_{\rho\sigma}C^a_{jk}F^j_{\mu\nu}\xi^k(x)\Bigr)\\
	&=-\frac{1}{4}\sqrt{|g|}\Bigl(C_{ahi}F^{h\mu\nu}F^a_{\mu\nu}\xi^i(x)
	+C_{cjk}F^{c\mu\nu}F^j_{\mu\nu}\xi^k(x)\Bigr).
\end{align*}
After some dummy index relabelling and using total antisymmetry for $C_{abc}$,
both terms cancel each other
\begin{align*}
	F&\equiv-\frac{1}{4}\sqrt{|g|}\Bigl(C_{jci}F^{c\mu\nu}F^j_{\mu\nu}\xi^i(x)
	+C_{cji}F^{c\mu\nu}F^j_{\mu\nu}\xi^i(x)\Bigr)\\
	&=-\frac{1}{4}\sqrt{|g|}\Bigl(\cancel{C_{jci}F^{c\mu\nu}F^j_{\mu\nu}\xi^i(x)}
	-\cancel{C_{jci}F^{c\mu\nu}F^j_{\mu\nu}\xi^i(x)}\Bigr)=0.
\end{align*}
Hence
\[
	\boxed{F=0=j^1\xi_C(L_{YM})}\,,
\]
and so, by Proposition \ref{prop:Fzero_iff_symmetry}, in the case that $\xi_C$ satisfies
the boundary conditions $\mathrm{ii})$ and $\mathrm{iii})$ of Definition
\ref{def:symmetries_of_the_lagrangian} (trivially satisfied whenever $\partial X=\emptyset$),
\textbf{$\xi_C$ is a symmetry of the Yang--Mills Lagrangian}.

Even more, using Equation \eqref{eq:lengthy_computation}, since $F=0$ and $\xi_C$ is
vertical, we get $\pounds_{j^1\xi_C}\Theta_{L_{YM}}=0$ \textbf{exactly}
(and not only modulo $\mathcal{I(C)}$ as Definition
\ref{def:symmetries_of_the_lagrangian} requires), i.e. our gauge symmetry of the Lagrangian
is exact, not only modulo divergence terms.
\subsection{Noether's gauge conserved current and charge}
By using Theorem \ref{thm:noether_theorem_symmetries_of_the_lagrangian} or equivalently
the multisymplectic momentum map of Section \ref{sect:momentum_map}, we can compute
the conserved current associated with a gauge symmetry $\xi_C$ of our Lagrangian.
Given some $\xi\in\mathfrak{gau}(P)$ and using Equation
\eqref{eq:useful_local_expression_for_the_preserved_quantity} (recall that $\xi_C$
is vertical) we have
\begin{equation}
\label{eq:conserved_current}
\boxed{
	J^\xi\equiv\iota_{j^1\xi_C}\Theta_{L_{YM}}
	\overset{\eqref{eq:useful_local_expression_for_the_preserved_quantity}}{=}
	\parder{L_{YM}}{A^a_{\mu\nu}}(\xi_C)^a_\mu\dd^nx_\nu
	\overset{\ref{eq:pd_of_L_wrt_field_components}}{=}p_a^{\mu\nu}
	(D_\mu\xi^a(x))\dd^nx_\nu
	}\,.
\end{equation}

Let us see now that indeed it is conserved in the sense of Definition
\ref{def:conserved_quantity_for_EL_equations}, i.e. that, given a solution $\phi_\omega$ of
the Euler--Lagrange equations \eqref{eq:euler-lagrange_equations}, we have
$(j^2\phi_\omega)^*\dd J^\xi=0$ (recall that, exactly as it happened with the Euler--Lagrange
equations, $\dd J^\xi$ contains $\dd A^a_{\mu\nu}$ terms and so we have to pullback it with
the second jet prolongation of $\phi_\omega$).
 
Since we have just seen that $\pounds_{j^1\xi_C}\Theta_{L_{YM}}=0$ exactly, the proof of
Proposition \ref{prop:dJ_equals_iota_omega} applies verbatim to our gauge symmetries, giving
\[
	\dd J^\xi=\iota_{j^1\xi_C}\Omega_{L_{YM}}.
\]
Recalling then the expression of $\Omega_{L_{YM}}$ obtained in Section
\ref{sec:poincare_cartan_forms_legendre_transformations_and_the_constraint_algorithm},
\[
	\Omega_{L_{YM}}=C^{c}_{ab}A^{b}_{\nu}\,p_{c}^{\mu\nu}\,\theta^{a}_{\mu}\wedge\dd^{n+1}x
	+\theta^{a}_{\mu}\wedge\dd p_{a}^{\mu\nu}\wedge\dd^{n}x_{\nu},
\]
together with the local expression \eqref{eq:j1xi_C_local} of the jet prolongation of our
gauge symmetries,
\begin{align*}
	j^1\xi_C&=(D_\mu\xi)^a(x)\parder{}{A^a_\mu}+\der{(D_\mu\xi)^a(x)}{x^\nu}
	\parder{}{A^a_{\mu\nu}}\\
	&=\Bigl(\partial_\mu\xi^a(x)+C^a_{bc}A^b_\mu\xi^c(x)\Bigr)\parder{}{A^a_\mu}\\
	&+\Bigl[\partial_\nu\partial_\mu\xi^a(x)+C^a_{bc}\bigl(A^b_{\mu\nu}\xi^c(x)
	+A^b_\mu\partial_\nu\xi^c(x)\bigr)\Bigr]\parder{}{A^a_{\mu\nu}},
\end{align*}
and using that $j^1\xi_C$ is vertical (hence it annihilates both $\dd^{n+1}x$ and
$\dd^nx_\nu$) while $\iota_{j^1\xi_C}\theta^a_\mu=(D_\mu\xi)^a(x)$ and
$\iota_{j^1\xi_C}\dd p_a^{\mu\nu}=j^1\xi_C(p_a^{\mu\nu})$, we obtain
\begin{align*}
	(j^2\phi_\omega)^*\left(\dd J^\xi\right)&=(j^2\phi_\omega)^*\left(
	\iota_{j^1\xi_C}\Omega_{L_{YM}}\right)\\
	&=(j^2\phi_\omega)^*\Bigl[(D_\mu\xi)^a(x)\,C^{c}_{ab}A^{b}_{\nu}\,p_{c}^{\mu\nu}\,
	\dd^{n+1}x+(D_\mu\xi)^a(x)\,\dd p_{a}^{\mu\nu}\wedge\dd^{n}x_{\nu}\\
	&\qquad\qquad\quad-\cancel{j^1\xi_C\left(p_a^{\mu\nu}\right)\theta^{a}_{\mu}\wedge
	\dd^{n}x_{\nu}}\Bigr],
\end{align*}
where the last term annihilates as it is the pullback of a contact form by an holonomic
section. Applying now \eqref{eq:pullbacking_dform} to $p_a^{\mu\nu}$ and using
$\dd x^\alpha\wedge\dd^nx_\nu=\delta^\alpha_\nu\,\dd^{n+1}x$, we finally get
\begin{align*}
	(j^2\phi_\omega)^*\left(\dd J^\xi\right)&=(D_\mu\xi)^a(x)\,C^{c}_{ab}A^{b}_{\nu}(x)\,
	p_{c}^{\mu\nu}(x)\,\dd^{n+1}x\\
	&+(D_\mu\xi)^a(x)\,\partial_\nu p_a^{\mu\nu}(x)\,\dd^{n+1}x\\
	&=(D_\mu\xi)^a(x)\underbrace{\left(\partial_\nu p_a^{\mu\nu}(x)
	+C^{c}_{ab}A^{b}_{\nu}(x)p_{c}^{\mu\nu}(x)\right)}
	_{=\,D_\nu p_a^{\mu\nu}(x)\,=\,0\text{ by }\eqref{eq:Yang-Mills_equations}}\dd^{n+1}x=0,
\end{align*}
namely, $J^\xi$ is a conserved current for every $\xi\in\mathfrak{gau}(P)$, the whole
computation reducing precisely to the Yang--Mills equations
\eqref{eq:Yang-Mills_equations}.

Using Proposition \ref{prop:dJ_equals_iota_omega}, we also know that
\[
		dJ^\xi=\iota_{j^1\xi_C}\Omega_{L_{YM}},
\]
i.e. $\bm{J^\xi}$ \textbf{is a Hamiltonian} $\bm{n}$\textbf{-form} (up to closed forms) with Hamiltonian vector field
$j^1\xi_C$ (up to vectors in the kernel of $\flat_{\Omega_{L_{YM}}}$).


Let us finally note that, once pullbacked by a solution of the Yang--Mills equations, our
current becomes exact on $X$. Indeed
\[
	(j^1\phi_\omega)^*J^\xi=p_a^{\mu\nu}(x)\,(D_\mu\xi)^a(x)\,\dd^nx_\nu\in\Omega^n(X),
\]
and expanding the covariant derivative and using the Leibniz rule,
\[
	p_a^{\mu\nu}(x)(D_\mu\xi)^a(x)=\partial_\mu\left(p_a^{\mu\nu}(x)\xi^a(x)\right)
	-\left(D_\mu p_a^{\mu\nu}(x)\right)\xi^a(x),
\]
so that on-shell ($D_\mu p_a^{\mu\nu}(x)=0$) our current is a total divergence.
Using then the $(\mu,\nu)$-antisymmetry of $p_a^{\mu\nu}$ we have
\begin{equation}
	\label{eq:superpotential}
	\boxed{
	(j^1\phi_\omega)^*J^\xi=\partial_\mu\left(p_a^{\mu\nu}(x)\xi^a(x)\right)\dd^nx_\nu
	=\dd K^\xi,\qquad
	K^\xi:=-\frac{1}{2}\,p_a^{\mu\nu}(x)\,\xi^a(x)\,\dd^{n-1}x_{\mu\nu}
	}\,,
\end{equation}
and hence the associated conserved charge reduces to a boundary integral (where
we will call $\bm{K}$ the \textbf{superpotential}). The proof of Equation
\eqref{eq:superpotential} is done in Subsection
\ref{sub:proof_of_superpotential_definition}.

Let us make this last statement explicit by computing the charge associated with $J^\xi$
after the slicing, i.e. its instantaneous momentum map
\eqref{eq:instantaneous_momentum_map_local} for a given Cauchy surface $M$. Let then
$\gamma\in\tilde{Z}$ be the restriction to $M$ of $j^1\phi_\omega$ for $\phi_\omega$ a
solution of the Yang--Mills equations \eqref{eq:Yang-Mills_equations}; we have
\begin{align*}
	\tilde{J}^\xi(\gamma)&=\int_M\gamma^*\left(\iota_{j^1\xi_C}\Theta_{L_{YM}}\right)
	+\int_{\partial M}\cancel{\gamma^*\left(\iota_{j^1\xi_C}\Pi\right)}
	=\int_M\gamma^*J^\xi\\[4pt]
	&\overset{\text{on-shell}}{=}\int_M\gamma^*\left(\dd K^\xi\right)
	\overset{\text{(Stokes)}}{=}\int_{\partial M}\gamma^*K^\xi
	=-\frac{1}{2}\int_{\partial M}\gamma^*\left(p_a^{\mu\nu}(x)\,\xi^a(x)\,
	\dd^{n-1}x_{\mu\nu}\right),
\end{align*}
where the boundary term of \eqref{eq:instantaneous_momentum_map_local} annihilates because
$\Pi$ can be proved to be a purely horizontal $n$-form while $j^1\xi_C$ is $\pi_{ZX}$-vertical.
\textbf{The whole charge is therefore carried by the superpotential} $\bm{K^\xi}$, and this
has a neat physical reading:
\begin{itemize}
	\item If $\xi$ generates a \textbf{localizable} symmetry in the sense of Definition
	\ref{def:localizable_symmetry} -- and, as we remarked in Section
	\ref{sec:initial_clarifications}, we may always multiply $\xi$ by a bump function
	supported in the interior of $M$ -- then $\xi^a(x)|_{\partial M}=0$ and hence
	\[
		\tilde{J}^\xi(\gamma)=0\,,
	\]
	for that now bumped $\xi$ of compact support.
	Genuine gauge transformations carry \textbf{no} conserved charge at all; they are pure
	redundancy of the description and not physical symmetries.
	\item If, on the contrary, $\xi$ generates a \textbf{global} (rigid) transformation, i.e.
	it does not vanish on $\partial M$, then $\tilde{J}^\xi(\gamma)$ is a genuinely non-trivial
	conserved charge, entirely determined by the asymptotic data of the field on $\partial M$
	(this is precisely the colour charge of the Yang--Mills field).
\end{itemize}
This dichotomy is nothing but the manifestation, in our concrete Yang--Mills setting, of
Noether's second Theorem \ref{thm:second_noether_theorem}: the charges associated with
localizable symmetries are trivial on-shell, while for rigid transformations they are given
by
\[
	\boxed{
	\tilde{J}^\xi(\gamma)=\xi^a\left(-\frac{1}{2}\int_{\partial M}\gamma^* p_a^{\mu\nu}(x)
	\dd^{n-1}x_{\mu\nu}\right)\equiv \xi^a\,Q_a
	}\,.
\]

\section{Yang--Mills graded Lie algebra}
In this section we will perform some computations regarding the graded Lie algebra structure
manifested in Theorem \ref{thm:graded-lie-algebra}.
\subsection{Top forms \texorpdfstring{$\ell=n$}{l=n}}
\label{sub:top_forms_ell_n_}
We begin our computations with the trivial case of top Hamiltonian forms $\ell=n$, i.e.,
we work with the usual Lie bracket on $\mathfrak{g}$ for our Hamiltonian vector fields
and with the multisymplectic bracket of Definition \ref{def:multisymplectic_bracket}
for our Hamiltonian $n$-forms. As we already mentioned in the previous section, for any
$\xi\in\mathfrak{gau}(P)$, $J^\xi$ is a Hamiltonian $n$-form with Hamiltonian vector field
given by $j^1\xi_C$.
Hence, consider two different $\xi,\zeta\in\mathfrak{gau}(P)$ and let us compute the bracket
of their associated Hamiltonian $n-$forms $J^\xi$ and $J^\zeta$ respectively.

Before going into the actual computation, we already know the result. Indeed, from
Proposition \ref{prop:momentum_map_algebra} we have Equation \eqref{eq:momentum_map_algebra}
(where $\mathfrak{gau}(P)$ plays the role of the (infinite-dimensional) Lie algebra),
and so \textbf{the algebra closes up to close forms in general}
\[
	\{J^\xi,J^\zeta\}=J^{[\xi,\zeta]}-\dd\left(\iota_{j^1\xi_C}\iota_{j^1\xi_C}
	\Theta_{L_{YM}}\right),
\]
were the bracket is the pointwise bracket defined on Equation \eqref{eq:pointwise_bracket},
i.e. locally given by $[\xi,\zeta]^a(x)=C^a_{bc}\xi^b(x)\zeta^c(x)$.
Since $\Theta_{L_{YM}}\in\Omega^{n+1}_1 Y$ (which is easily seen from its local expression
given in Section
\ref{sec:poincare_cartan_forms_legendre_transformations_and_the_constraint_algorithm}),
$\iota_{j^1\xi_C}\iota_{j^1\xi_C}\Theta_{L_{YM}}=0$ and \textbf{the algebra closes exactly}.
In components
\begin{align*}
	\{J^\xi,J^\zeta\}&=\iota_{j^1\xi_C}\iota_{j^1\zeta_C}\Omega_{L_{YM}}=\iota_{j^1\xi_C}\dd J^\zeta
	\overset{\eqref{eq:conserved_current}}{=}\iota_{j^1\xi_C}\big(D_\mu\zeta^a(x)\,\dd
	p_a^{\mu\nu}\wedge\dd^nx_\nu\\
	&+p_a^{\mu\nu}\,\dd(D_\mu\zeta^a(x))\wedge\dd^nx_\nu\big)=D_\mu\zeta^a(x)\,j^1\xi_C
	(p_a^{\mu\nu})\dd^nx_\nu\\
	&+p_a^{\mu\nu}\,C^a_{bc}\zeta^c(x)j^1\xi_C(A^a_\mu)\dd^nx_\nu
	\overset{\eqref{eq:j1xi_C_of_p}}{=}\big(-C^b_{ad}p_b^{\mu\nu}\xi^d(x)D_\mu\zeta^a(x)\\
	&+C^a_{bc}p_a^{\mu\nu}\zeta^c(x)D_\mu\xi^b(x)\big)\dd^nx_\nu\equiv
	C^a_{bc}p_a^{\mu\nu}\big(\zeta^c(x)D_\mu\xi^b(x)-\xi^c(x)D_\mu\zeta^b(x)\big)\\
	&=C^a_{bc}p_a^{\mu\nu}D_\mu(\xi^b(x)\zeta^c(x))\dd^nx_\nu=
	p_a^{\mu\nu}D_\mu(C^a_{bc}\xi^b(x)\zeta^c(x))\dd^nx_\nu=
	J^{[\xi,\zeta]},
\end{align*}
where in the last equality we used the pointwise bracket of Equation
\eqref{eq:pointwise_bracket} and Equation \eqref{eq:conserved_current}
\[
	\boxed{
	\{J^\xi,J^\zeta\}=J^{[\xi,\zeta]},\quad\forall\xi,\zeta\in\gau(P)	
	}\,.
\]

\subsection{The generalization of Hamiltonian forms}
\label{sub:the_generalization_of_hamiltonian_forms}
The main problem of lowering to $\ell<n-1$ is that of
actually finding meaningful multivector gauge symmetries.
Physically, Definition \ref{def:conserved_quantity_for_EL_equations} does not
cover situations were the divergence on shell of the current is not
zero but a potential term.

Let $\alpha\in\mathcal{H}^n(Z)$ be a top Hamiltonian form with associated
Hamiltonian vector field $X_\alpha\in\mathfrak{X}_H(Z)$; then an easy
computation shows that $\alpha$ is a preserved quantity for the Euler--Lagrange
equations; take any $\phi_\omega\in\Gamma(\pi)$ solution to the
Euler--Lagrange equations \eqref{eq:euler-lagrange_equations}, then
\[
	(j^1\phi_\omega)^*(\dd\alpha)=(j^1\phi_\omega)^*(\iota_{X_\alpha}
	\Omega_{L_{YM}})\overset{\eqref{eq:intrinsic_EL_equations}}{=}0.
\]Of course, $(j^1\phi_\omega)^*(\dd\alpha)$ can be written locally
as $(j^1\phi_\omega)^*(\dd\alpha)=\dd\alpha^\mu(x)\wedge\dd^nx_\mu$
for some set of functions $\alpha^\mu(x)\in C^\infty(X)$, and so
conservation of our current implies that these components are
divergenceless, namely
\[
	(j^1\phi_\omega)^*(\dd\alpha)=\partial_\mu\alpha^\mu(x)\dd^{n+1}x=0
	\implies\partial_\mu\alpha^\mu(x)=0,\quad\forall x\in X. 
\]
How can we obtain instead situations where we have
$\partial_\mu\alpha^\mu(x)=B(x)$? Take any
$\beta\in\Omega^{n+1}_1 Y$, i.e., a semibasic form on $Z$. This form has
to have then the local expression $\beta=B(z)\dd^{n+1}x$ for some
function $B\in C^\infty(Z)$. What happens if we now relax
the Hamiltonian form definition to one modulo this kind
of semibasic forms?
\[
	\dd\alpha=\iota_{X_\alpha}\Omega_{L_{YM}}+\beta,\quad\beta\in\Omega^{n+1}_1Z.
\]
Then
\[
	(j^1\phi_\omega)^*(\dd\alpha)=\cancel{(j^1\phi_\omega)^*(\iota_{X_\alpha}
	\Omega_{L_{YM}})}+(j^1\phi_\omega)^*(\beta)=0
\]
\[
	\implies\partial_\mu\alpha^\mu(x)=B(j^1\phi_\omega(x))\equiv B(x),
	\quad\forall x\in X. 
\]
This is the approach taken in \cite{leonGradedPoissonGraded2025}.
Given our Yang--Mills Lagrangian, as computed in Section
\ref{sec:poincare_cartan_forms_legendre_transformations_and_the_constraint_algorithm}
our Hamiltonian function on $M_1=\leg_{L_{YM}}(Z)$ has the form
\[
	H(x^\mu,A^a_\mu,\hat{p}^a_{\mu\nu}=p^a_{\mu\nu})=
	L_{YM}+\frac{1}{2}C^a_{bc}A^b_\mu A^c_\nu p_a^{\mu\nu},
\]
with $L_{YM}=-\frac{1}{4\sqrt{|g|}}p_a^{\mu\nu}p^a_{\mu\nu}$. This Hamiltonian function defines
a Hamiltonian section as $h_1=(\mu_{\tilde{M}_1})^{-1}$ (recall $\tilde{M}_1\cong M_1$
from Remark \ref{rem:M1_equivalence}), and considering the pullback by it of the canonical
multisymplectic $(n+2)$-form on $\Omega^{n+2}_2 Y$, i.e., $\Omega_{M_1}=h_1^*\Omega$ given locally by
\begin{align}
	\label{eq:Omega_in_M1}
	\Omega_{M_1}&=\dd H\wedge \dd^{n+1}x-\dd p_a^{\mu\nu}\wedge\dd A^a_\mu\wedge\dd^nx_\nu\\
	&=\parder{H}{A^a_\mu}\,\dd A^a_\mu\wedge\dd^{n+1}x
	+\parder{H}{p_a^{\mu\nu}}\,\dd p_a^{\mu\nu}\wedge \dd^{n+1}x
	-\dd p_a^{\mu\nu}\wedge\dd A^a_\mu\wedge\dd^nx_\nu
\end{align}
then, the graded Lie algebra of Subsection \ref{subsec:graded} applies
verbatim with $(M_1,\Omega_{M_1})$ instead of $(Z,\Omega_L)$
and for a chosen a LT $s_\alpha$, the following forms
\[
	\A^a:=A^a_\mu\dd x^\mu,\quad\F_a:=\frac{1}{2}p_a^{\mu\nu}\dd^{n-1}x_{\mu\nu},
\]
where $\A\in\Omega^1(M_1;\mathfrak{g})$ and $\F\in\Omega^{n-1}(M_1;\mathfrak{g}^*)$
(note we already introduced $\F_a$ on $X$ in section \ref{sec:yang_mills_equations})
belong to this semibasic, more broad sense of Hamiltonian forms
-called \textbf{special Hamiltonian forms} \cite{leonGradedPoissonGraded2025}-
under the following multivector fields
\[
	X_{\F}=\frac{1}{2}\parder{}{A^a_\mu}\wedge\parder{}{x^\mu},\quad
	X_{\A}=\frac{1}{2n}\parder{}{p_a^{\mu\nu}}\wedge\parder{}{^{n-1}x_{\mu\nu}},
\]
where $\parder{}{^nx^\mu}:=\iota_{\dd x^\mu}\parder{}{^{n+1}x}\equiv(-1)^\mu\parder{}
{x^0}\wedge\dots\wedge\widehat{\parder{}{x^\mu}}\wedge\dots\wedge\parder{}{x^{n}}$.

Physically (in a motivational sense) is important to notice that, considering the
section $\sigma\in\Gamma(\pi_{M_1X})$ defined as $\sigma:=\leg_{L_{YM}}\circ\,
j^1\phi_\omega$ for some $\phi_\omega\in\Gamma(\pi)$, then
\[
	\sigma^*\A=A^a_\mu(x)\dd x^\mu=s_\alpha^*\omega=A_\alpha,
	\sigma^*\F=\frac{1}{2}p_a^{\mu\nu}(x)\dd^{n-1}x_{\mu\nu}
	\propto \star(s_\alpha^*\Omega^\omega)=\star F_\alpha
\]
where the second relation is an equality up to some $\kappa$ and
$\sqrt{|g|}$ factors.

We now proceed computing its multisymplectic brackets $\{\cdot,\cdot\}^*$
defined on Subsection \ref{subsec:graded}. The first observation is that
the bracket $\{\widehat{\A}^a,\widehat{\A}^b\}^*$ vanish identically
by Definition \ref{def:graded-bracket} ($2\deg\A^a=2<n$, the case of interest)
while $\{\widehat{\A}^a,\widehat{\F}^b\}^*$ and
$\{\widehat{\F}^a,\widehat{\F}^b\}^*$ are $0$ and $(n-2)$-forms respectively
(degrees $n$ and $2$ on the new grading \eqref{eq:grading:fix}). Although
this observation, the bracket $\{\widehat{\F}^a,\widehat{\F}^b\}^*$ also happens
to vanish; we just have to note that the multivector $X_{\F_a}\wedge X_{\F_b}$
has two $\parder{}{A^a_\mu}$ `inside', while $\Omega_{M_1}$ in Equation
\eqref{eq:Omega_in_M1} has at most one $\dd A^a_\mu$. The only bracket hence
left to compute is $\{\widehat{\A}^a,\widehat{\F}^b\}^*$
\begin{align*}
	\{\widehat{\A}^a,\widehat{\F}^b\}^*&=-(-1)^{|\hat{\A}^a|}\iota_{X_{\A^a}\wedge X_{\F_b}}\Omega_{M_1}
	=(-1)^n\iota_{X_{\F_b}}\iota_{X_{\A^a}}\Omega_{M_1}
	=(-1)^n\cancel{(-1)^{2n}}\iota_{X_{\A^a}}\iota_{X_{\F_b}}\Omega_{M_1}\\
	&=(-1)^n\frac{1}{2}\iota_{X_{\A^a}}\iota_{\parder{}{x^\rho}}\iota_{\parder{}{A_\rho^b}}\Omega_{M_1}
	=(-1)^n\frac{1}{2}\iota_{X_{\A^a}}\iota_{\parder{}{x^\rho}}
	\left[
		\parder{H}{A_\rho^b}\dd^{n+1}x+\dd p_b^{\rho\tau}\wedge\dd^{n}x_{\tau}
	\right]
	\\
	&=(-1)^n\frac{1}{2}\iota_{X_{\A^a}}\left[\parder{H}{A_\rho^b}\dd^n x_\rho-\dd p_b^{\rho\tau}
	\wedge\dd^{n-1}x_{\tau\rho}\right]\\
	&=(-1)^n\frac{1}{4n}\iota_{\parder{}{^{n-1}x_{\mu\nu}}}\iota_{\parder{}{p_a^{\mu\nu}}}\left[
	\parder{H}{A_\rho^b}\dd^n x_\rho+\dd p_b^{\rho\tau}\wedge\dd^{n-1}x_{\rho\tau}\right]\\
	&=(-1)^n\frac{1}{4n}\iota_{\parder{}{^{n-1}x_{\mu\nu}}}\left[
	\delta^a_b(\delta_\mu^\rho\delta_\nu^\tau-\delta_\mu^\tau\delta_\nu^\rho)\,\dd^{n-1}x_{\rho\tau}
	\right]
	=(-1)^n\frac{1}{\cancel{4}2n}\,\delta^a_b\,\iota_{\parder{}{^{n-1}x_{\mu\nu}}}\left[\cancel{2}\,\dd^{n-1}x_{\mu\nu}\right]\\
	&=(-1)^n\frac{1}{2n}\delta^a_b\sum_{\mu,\nu}1=(-1)^n\frac{1}{\cancel{2}n}\cancel{2}\delta^a_b\sum_{\mu>\nu}1
	=(-1)^n\frac{\cancel{n}(n+1)}{2\cancel{n}}\,\delta^a_b
	=(-1)^n\frac{n+1}{2}\,\delta^a_b
\end{align*}
\[
	\implies
	\boxed{
		\{\widehat{\A}^a,\widehat{\F}^b\}^*=(-1)^n\frac{n+1}{2}\,\delta^a_b,\quad
		\forall a,b\in\{1,\dots,N\}
	}\,.
\]
\section{Table summary and further work}
As a way of closing this thesis, we have collected some of the main concepts and results
covered in this work in the form of a `dictionary' that serves as a `translator' between
particle physicists and differential geometers. The dictionary is presented in Table
\ref{tab:dictionary}. With this in place, we conclude by outlining several directions
for future work that lie beyond the scope of this thesis but that we find
interesting and/or promising as potential research directions 
\begin{itemize}
	\item Apply multisymplectic formalism to gravity theories through the
	Einstein--Hilbert Lagrangian.
	\item Apply Noether theorems to general spacetime translations and recover the
	Energy-Momentum tensor.
	\item Delve into the physical consequences of Cartan symmetries. 
	\item Relate multisymplectic formulation of classical field theories with
	$k$-symplectic and $k$-cosymplectic ones. The advantage of these formulations is
	that they are more manageable when performing numerical integration computations.
	\item Deepen in the relationship among the graded Lie algebra and covariantly closed
	forms ($D_A\alpha=0$).
	\item Explore the existence of (special) Hamiltonian forms in intermediate degrees
	and analyze their corresponding conserved currents. 
\end{itemize}

\let\emph\emphoriginal

\printbibliography
\appendix
\allowdisplaybreaks
\chapter{Proof of vertical endomorphism intermediate condition }
\label{app:proof_of_belonging}
Let us prove that indeed $T_{j^1_x\phi}\pi_{ZY}(V_i)-
T_x\phi\circ T_{j^1_x\phi}\pi_{ZX}(V_i)\in\mathcal{V}_{\phi(x)}\pi$
\begin{proof}
Indeed
\[
T_{j^1_x\phi}\pi_{ZY}(V_i)-T_x\phi\circ 
T_{j^1_x\phi}\pi_{ZX}(V_i)\in\mathcal{V}_{\phi(x)}\pi\iff
\]
\[
\iff T_{\phi(x)}\pi\left[T_{j^1_x\phi}\pi_{ZY}(V_i)-T_x\phi\circ 
T_{j^1_x\phi}\pi_{ZX}(V_i)\right]=0
\]
\[
\iff T_{\phi(x)}\pi\circ T_{j^1_x\phi}\pi_{ZY}(V_i)-
T_{\phi(x)}\pi\circ 
T_x\phi\circ T_{j^1_x\phi}\pi_{ZX}(V_i)=0
\]
and
\[
	\pi\circ\pi_{ZY}=\pi_{ZX}\implies
	T_{\phi(x)}\pi\circ T_{j^1_x\phi}\pi_{ZY}=T_{j^1_x\phi}
	\pi_{ZX},
\]
\[
	\pi\circ\phi=\mathrm{id}_X\implies T_{\phi(x)}\pi\circ T_x\phi=
	\mathrm{id}_{T_xX}.
\] 
The proof is clearer though working in components. Stating first the 
following notation convention
\[
	\partial_\mu\equiv\frac{\partial}{\partial x^\mu},\quad
	\partial_i\equiv\frac{\partial}{\partial y^i},\quad
	\partial_i^\mu\equiv\frac{\partial}{\partial z^i_\mu},
\]
we consider some $V\in T_{j^1_x\phi}Z$ locally given by 
$V=V^\mu(j^1_x\phi)\partial_\mu|_{j^1_x\phi}+V^i(j^1_x\phi)
\partial_i|_{j^1_x\phi}+V^i_\mu(j^1_x\phi)\partial_i^\mu|_{j^1_x\phi}$
(in the following we will omit the point $j^1_x\phi$ in the expressions
for the sake of clarity, except in the pushforward maps, where we
will write it).
Is then straightforward to compute
\[
	T_{j^1_x\phi}\pi_{ZY}(V)=V^\mu T_{j^1_x\phi}\pi_{ZY}(\partial_\mu)+
	V^i T_{j^1_x\phi}\pi_{ZY}(\partial_i)+
	V^i_\mu T_{j^1_x\phi}\pi_{ZY}(\partial_i^\mu)
\]
\[
	=V^\mu\partial_\mu+V^i\partial_i,
\]
\[
	T_{j^1_x\phi}\pi_{ZX}(V)=V^\mu T_{j^1_x\phi}\pi_{ZX}(\partial_\mu)+
	V^i T_{j^1_x\phi}\pi_{ZX}(\partial_i)+V^i_\mu T_{j^1_x\phi}\pi_{ZX}
	(\partial^\mu_i)=V^\mu\partial_\mu,
\]
\[
	T_x\phi\circ T_{j^1_x\phi}\pi_{ZX}(V)=V^\mu T_x\phi(\partial_\mu)=
	V^\mu\partial_\mu+V^\mu\partial_\mu\phi^i\partial_i,
\]
and hence 
\[
T_{j^1_x\phi}\pi_{ZY}(V)-T_x\phi\circ T_{j^1_x\phi}\pi_{ZX}(V)=
(V^i-V^\mu\partial_\mu\phi^i)\partial_i.
\]
which is clearly $\pi-$vertical.
\end{proof}

\chapter{Some manipulations of the extremality condition}
\label{app:some_manipulations_of_the_extremality_condition}
As already mentioned in the main text, we can rewrite the extremality
condition as
\[
	\frac{dS(\phi_t)}{dt}\bigg|_{t=0}=\int_C(j^1\phi)^*
	\pounds_{\xi_Y^{(1)}}(\mathcal{L})\overset{!}{=}0,
\]
from which
\begin{equation}
	\label{eq:first_expression}
	\begin{aligned}
		\int_C(j^1\phi)^*\pounds_{\xi_Y^{(1)}}(\mathcal{L})&=
		\int_C(j^1\phi)^*\pounds_{\xi_Y^{(1)}}(L)\eta+
		\int_C(j^1\phi)^*L\pounds_{\xi_Y^{(1)}}(\eta) \\
		&=\int_C(j^1\phi)^*\iota_{\xi_Y^{(1)}}(dL)\eta+
		\int_C(j^1\phi)^*L\pounds_{\xi_Y^{(1)}}(\eta).
	\end{aligned}	
\end{equation}
In components, we have that $dL=\frac{\partial L}{\partial x^\mu}
dx^\mu+\frac{\partial L}{\partial y^i}dy^i+\frac{\partial L}
{\partial z^i_\mu}dz^i_\mu$ while $\xi_Y^{(1)}=\xi_Y^\mu(x)
\partial_\mu+\xi_Y^i(x,y)\partial_i+\left(\frac{d\xi_Y^i(x,y)}
{dx^\mu}-z^i_\nu \frac{d\xi_Y^\nu(x)}
{dx^\mu}\right)\partial^i_\mu$, and so
\[
	\iota_{\xi_Y^{(1)}}(dL)=\frac{\partial L}{\partial x^\mu}
	\xi_Y^\mu(x)+\frac{\partial L}{\partial y^i}\xi_Y^i(x,y)+
	\frac{\partial L}{\partial z^i_\mu}\left(\frac{d\xi_Y^i(x,y)}
	{dx^\mu}-z^i_\nu \frac{d\xi_Y^\nu(x)}{dx^\mu}\right)
\]
while
\[
	\pounds_{\xi_Y^{(1)}}(\eta)=\iota_{\xi_Y^{(1)}}
	(\cancelto{0}{d\eta})+d(\iota_{\xi_Y^{(1)}}\eta)=
	d(\xi_Y^\mu(x)d^n x_\mu)=\frac{\partial\xi_Y^\mu(x)}
	{\partial x^\mu}\eta=\frac{d\xi_Y^\mu(x)}
	{d x^\mu}\eta.
\]
Using these two last results back into \eqref{eq:first_expression}
we have
\[
	\int_C(j^1\phi)^*\pounds_{\xi_Y^{(1)}}(\mathcal{L})=
	\int_C(j^1\phi)^*\frac{\partial L}{\partial x^\mu}
	\xi_Y^\mu(x)\eta+\int_C(j^1\phi)^*\frac{\partial L}
	{\partial y^i}\xi_Y^i(x,y)\eta
\]
\[
	+\bigg[\int_C(j^1\phi)^*\frac{\partial L}{\partial z^i_\mu}
	\left(\frac{d\xi_Y^i(x,y)}{dx^\mu}-z^i_\nu\frac{d\xi_Y^\nu(x)}
	{dx^\mu}\right)\eta\bigg]+\int_C(j^1\phi)^*L\frac{d\xi_Y^\mu(x)}
	{d x^\mu}\eta.
\]
We now try to group total derivative terms $(d/dx^\mu)$ by adding 
the corresponding total derivative of $z^i_\mu$. Note that these
terms -since we are total deriving $z^i_\mu$- will contain
coordinates of the second jet bundle, and so we have to pullback
them with $(j^2\phi)^*$ instead of just $(j^1\phi)^*$
\[
	\int_C(j^1\phi)^*\pounds_{\xi_Y^{(1)}}(\mathcal{L})=
	\int_C(j^1\phi)^*\frac{\partial L}{\partial x^\mu}
	\xi_Y^\mu(x)\eta+\int_C(j^1\phi)^*\frac{\partial L}
	{\partial y^i}\xi_Y^i(x,y)\eta
\]
\[
	+\bigg[\int_C(j^2\phi)^*
	\frac{d}{dx^\mu}\left(\xi_Y^i(x,y)-z^i_\nu\xi_Y^\nu(x)\right)
	\frac{\partial L}{\partial z^i_\mu}\eta
	+\int_C(j^2\phi)^*\frac{dz^i_\nu}{dx^\mu}\xi_Y^\nu(x)
	\frac{\partial L}{\partial z^i_\mu}\eta
	\bigg]
\]
\[
	+\int_C(j^1\phi)^*L\frac{d\xi_Y^\mu(x)}
	{d x^\mu}\eta.
\]
Let us now use Stokes' theorem in this last integral 
term to decompose it as a boundary integral plus another integral.
For this, we use the following trick
\[
	L\frac{d\xi_Y^\mu(x)}{dx^\mu}=\frac{d(L\xi_Y^\mu(x))}{dx^\mu}-
	\xi_Y^\mu(x)\frac{dL}{dx^\mu}
\]
obtaining
\[
	\int_C(j^1\phi)^*L\frac{d\xi_Y^\mu(x)}
	{d x^\mu}\eta=
	\int_C(j^2\phi)^*\frac{d(L\xi_Y^\mu(x))}
	{dx^\mu}\eta-\int_C(j^2\phi)^*\xi_Y^\mu(x)
	\frac{dL}
	{dx^\mu}\eta
\]
\[
	=\int_C\frac{\partial((j^2\phi)^*L\xi_Y^\mu(x))}
	{\partial x^\mu}\eta-\int_C(j^2\phi)^*\xi_Y^\mu(x)
	\frac{dL}
	{dx^\mu}\eta
\]
\[
	=\int_C d\left((j^1\phi)^*(L\xi_Y^\mu(x)
	d^nx_\mu)\right)-\int_C(j^2\phi)^*\xi_Y^\mu(x)
	\frac{dL}
	{dx^\mu}\eta
\]
\[
	=\int_{\partial C} (j^1\phi)^*L\xi_Y^\mu(x)
	d^nx_\mu-\int_C(j^2\phi)^*\xi_Y^\mu(x)
	\frac{dL}
	{dx^\mu}\eta,
\]
where we used Stokes' theorem in the last equality.
Expanding now the total derivative of the Lagrangian
\[
	\int_C(j^1\phi)^*L\frac{d\xi_Y^\mu(x)}
	{d x^\mu}\eta=\int_{\partial C} (j^1\phi)^*L\xi_Y^\mu(x)d^nx_\mu
	-\int_C(j^1\phi)^*\xi_Y^\mu(x)\frac{\partial L}{\partial x^\mu}\eta
	\]
\[
	-\int_C(j^1\phi)^*\xi_Y^\mu(x)z^i_\mu\frac{\partial L}{\partial y^i}\eta
	-\int_C(j^2\phi)^*\xi_Y^\mu(x)z^i_{\mu\nu}\frac{\partial L}
	{\partial z^i_\nu}\eta.
\]
Lastly we note that as jet coordinates are independent from one another
we have that $z^i_{\mu\nu}=z^i_{\nu\mu}=\frac{dz^i_\mu}{dx^\nu}$, hence
writing back this integral term in the extremality condition we have
\[
	\int_C(j^1\phi)^*\pounds_{\xi_Y^{(1)}}(\mathcal{L})=
	\cancel{\int_C(j^1\phi)^*\frac{\partial L}{\partial x^\mu}
	\xi_Y^\mu(x)\eta}+\int_C(j^1\phi)^*\frac{\partial L}
	{\partial y^i}\xi_Y^i(x,y)\eta
\]
\[
	+\bigg[\int_C(j^2\phi)^*
	\frac{d}{dx^\mu}\left(\xi_Y^i(x,y)-z^i_\nu\xi_Y^\nu(x)\right)
	\frac{\partial L}{\partial z^i_\mu}\eta
	+\bcancel{\int_C(j^2\phi)^*\frac{dz^i_\nu}{dx^\mu}\xi_Y^\nu(x)
	\frac{\partial L}{\partial z^i_\mu}\eta}
	\bigg]
\]
\[
	+\int_{\partial C} (j^1\phi)^*L\xi_Y^\mu(x)d^nx_\mu
	-\cancel{\int_C(j^1\phi)^*\xi_Y^\mu(x)\frac{\partial L}{\partial x^\mu}\eta}
	-\int_C(j^1\phi)^*\xi_Y^\mu(x)z^i_\mu\frac{\partial L}{\partial y^i}\eta
\]
\[
	\bcancel{-\int_C(j^2\phi)^*\frac{dz^i_\mu}{dx^\nu}\xi_Y^\mu(x)
	\frac{\partial L}{\partial z^i_\nu}\eta}.	
\]

Regrouping terms we have
\[
	\int_C(j^1\phi)^*\pounds_{\xi_Y^{(1)}}(\mathcal{L})=
	\int_C(j^1\phi)^*\frac{\partial L}
	{\partial y^i}(\xi_Y^i(x,y)-z^i_\mu\xi_Y^\mu(x))\eta
\]
\[
	+\bigg[\int_C(j^2\phi)^*
	\frac{d}{dx^\mu}\left(\xi_Y^i(x,y)-z^i_\nu\xi_Y^\nu(x)\right)
	\frac{\partial L}{\partial z^i_\mu}\eta\bigg]
	+\int_{\partial C} (j^1\phi)^*L\xi_Y^\mu(x)d^nx_\mu.
\]
Once again we use Stokes' theorem now over the big bracket $[...]$
term using the following trick
\[
	\frac{d}{dx^\mu}\left(\xi_Y^i(x,y)-z^i_\nu\xi_Y^\nu(x)\right)
	\frac{\partial L}{\partial z^i_\mu}=
\]
\[
	\frac{d}{dx^\mu}\left([\xi_Y^i(x,y)-z^i_\nu\xi_Y^\nu(x)]
	\frac{\partial L}{\partial z^i_\mu}\right)
	-[\xi_Y^i(x,y)-z^i_\nu\xi_Y^\nu(x)]\frac{d}{dx^\mu}\left(
	\frac{\partial L}{\partial z^i_\mu}\right)
\]
obtaining
\[
	\bigg[\int_C(j^2\phi)^*
	\frac{d}{dx^\mu}\left(\xi_Y^i(x,y)-z^i_\nu\xi_Y^\nu(x)\right)
	\frac{\partial L}{\partial z^i_\mu}\eta\bigg]=
\]
\[
	\int_C(j^2\phi)^*
	\frac{d}{dx^\mu}\left([\xi_Y^i(x,y)-z^i_\nu\xi_Y^\nu(x)]
	\frac{\partial L}{\partial z^i_\mu}\right)\eta
	-\int_C(j^2\phi)^*
	[\xi_Y^i(x,y)-z^i_\nu\xi_Y^\nu(x)]\frac{d}{dx^\mu}\left(
	\frac{\partial L}{\partial z^i_\mu}\right)\eta=
\]
\[
	\int_C\frac{\partial}{\partial x^\mu}
	\left((j^2\phi)^*[\xi_Y^i(x,y)-z^i_\nu\xi_Y^\nu(x)]
	\frac{\partial L}{\partial z^i_\mu}\right)\eta
	-\int_C(j^2\phi)^*
	[\xi_Y^i(x,y)-z^i_\nu\xi_Y^\nu(x)]\frac{d}{dx^\mu}\left(
	\frac{\partial L}{\partial z^i_\mu}\right)\eta=
\]
\[
	\int_Cd\left((j^1\phi)^*[\xi_Y^i(x,y)-z^i_\nu\xi_Y^\nu(x)]
	\frac{\partial L}{\partial z^i_\mu}d^nx_\mu\right)
	-\int_C(j^2\phi)^*
	[\xi_Y^i(x,y)-z^i_\nu\xi_Y^\nu(x)]\frac{d}{dx^\mu}\left(
	\frac{\partial L}{\partial z^i_\mu}\right)\eta
\]
\[
	=\int_{\partial C}(j^1\phi)^*[\xi_Y^i(x,y)-z^i_\nu\xi_Y^\nu(x)]
	\frac{\partial L}{\partial z^i_\mu}d^nx_\mu
	-\int_C(j^2\phi)^*
	[\xi_Y^i(x,y)-z^i_\nu\xi_Y^\nu(x)]\frac{d}{dx^\mu}\left(
	\frac{\partial L}{\partial z^i_\mu}\right)\eta
\]
(Stokes' theorem was applied one again in the last equality)
which back into the extremality condition gives
\[
	\int_C(j^1\phi)^*\pounds_{\xi_Y^{(1)}}(\mathcal{L})=
	\int_C(j^1\phi)^*\frac{\partial L}
	{\partial y^i}(\xi_Y^i(x,y)-z^i_\mu\xi_Y^\mu(x))\eta
	+\int_{\partial C} (j^1\phi)^*L\xi_Y^\mu(x)d^nx_\mu
\]
\[
	+\bigg[
	\int_{\partial C}(j^1\phi)^*[\xi_Y^i(x,y)-z^i_\nu\xi_Y^\nu(x)]
	\frac{\partial L}{\partial z^i_\mu}d^nx_\mu
	-\int_C(j^2\phi)^*
	[\xi_Y^i(x,y)-z^i_\nu\xi_Y^\nu(x)]\frac{d}{dx^\mu}\left(
	\frac{\partial L}{\partial z^i_\mu}\right)\eta
	\bigg].
\]
Grouping integrals on $C$ and boundary integrals we finally obtain
\begin{equation}
\begin{aligned}
	\label{eq:first_expression_extremality}	
	\int_C(j^1\phi)^*\pounds_{\xi_Y^{(1)}}(\mathcal{L})=
	\int_C(j^2\phi)^*\left[\frac{\partial L}
	{\partial y^i}-\frac{d}{dx^\mu}
	\frac{\partial L}{\partial z^i_\mu}\right]
	(\xi_Y^i(x,y)-z^i_\nu\xi_Y^\nu(x))\eta \\
	+\int_{\partial C}(j^1\phi)^*\left[ [(\xi_Y^i(x,y)-z^i_\nu\xi_Y^\nu(x)]
	\frac{\partial L}{\partial z^i_\mu}+L\xi_Y^\mu(x)\right]d^nx_\mu.
\end{aligned}
\end{equation}
Once obtained this form, what is left to do its to prove that
the boundary integral is equal to
\[
	\int_{\partial C}(j^1\phi)^*(\iota_{\xi_Y^{(1)}}\Theta_L)
\]
for $\Theta_L$ the Poincaré-Cartan
$(n+1)-$form of Definition \ref{def:poincare-cartan-forms} with
coordinate expression
\[
\Theta_L
=
\left(
L - z^i_\mu \frac{\partial L}{\partial z^i_\mu}
\right)d^{n+1}x
+
\frac{\partial L}{\partial z^i_\mu}\,dy^i \wedge d^n x_\mu.
\]
As $dy^i=\theta^i+z^i_\nu dx^\nu$, by definition of the contact forms,
we can substitute this expression back into $\Theta_L$ coordinate
form obtaining
\begin{align*}
\Theta_L
&=
\left(
L - z^i_\mu \frac{\partial L}{\partial z^i_\mu}
\right)d^{n+1}x
+
\frac{\partial L}{\partial z^i_\mu}\,\theta^i\wedge d^n x_\mu
+
\frac{\partial L}{\partial z^i_\mu}z^i_\nu\,dx^\nu\wedge d^n x_\mu \\
&=
\left(
L\cancel{- z^i_\mu \frac{\partial L}{\partial z^i_\mu}}
\cancel{+z^i_\mu \frac{\partial L}{\partial z^i_\mu}}
\right)d^{n+1}x
+
\frac{\partial L}{\partial z^i_\mu}\,\theta^i\wedge d^n x_\mu \\
&=
L\,d^{n+1}x
+
\frac{\partial L}{\partial z^i_\mu}\,\theta^i\wedge d^n x_\mu,
\end{align*}
as we already had in the main text. Contracting with $\xi_Y^{(1)}$
we then have
\begin{align*}
\iota_{\xi_Y^{(1)}}\Theta_L
&=
\iota_{\xi_Y^{(1)}}\!\left(L\,d^{n+1}x\right)
+
\iota_{\xi_Y^{(1)}}\!\left(\frac{\partial L}{\partial z^\mu_i}\,\theta^i\wedge d^n x_\mu\right) \\
&=
L\,\iota_{\xi_Y^{(1)}}(d^{n+1}x)
+
\frac{\partial L}{\partial z^\mu_i}\,
\iota_{\xi_Y^{(1)}}\!\left(\theta^i\wedge d^n x_\mu\right)
=L\xi_Y^\mu d^nx_\mu+
\frac{\partial L}{\partial z^\mu_i}\,
\left(\xi_Y^i-z^i_\nu\xi_Y^\nu\right)d^n x_\mu\\
&-\xi_Y^\nu\frac{\partial L}{\partial z^\mu_i}\,
\theta^i\wedge d^{n-1} x_{\mu\nu}.
\end{align*}
Is then clear that, as the last term is a contact form, when pulling
back with $(j^1\phi)^*$ we obtain
\[
	(j^1\phi)^*\iota_{\xi_Y^{(1)}}\Theta_L=(j^1\phi)^*\left[(\xi_Y^i-z^i_\nu\xi_Y^\nu)
	\frac{\partial L}{\partial z^i_\mu}+L\xi_Y^\mu\right]d^nx_\mu
\]
and the result follows upon integration on $\partial C$ on both sides.

\chapter{Proof of Proposition \ref{prop:holonomicity_iff_cartan}}
\label{app:proof_of_holonomicity_iff_cartan}
Although this proof itself is not central to the scope of
this master's thesis, we will do the full demonstration in
order to correctly internalize the concepts
of holonomicity and the Cartan distribution and to refresh
their skills in local computations of the pushforward map
and the interior product.	

\begin{proof}
($\Rightarrow$) By hypothesis $\sigma$ is holonomic so by definition
$\exists\phi\in\Gamma(\pi)$ such that
$\sigma=j^1\phi$ and so in local coordinates for a fixed chart we
have $\sigma(x^\mu)=j^1\phi(x^\mu)= (x^\mu,\phi^i(x),
\partial_\mu\phi^i(x))$.
Let $x\in X$ be an arbitrary but fixed point and
$w\in T_xX$ an arbitrary but fixed vector above that point with
coordinates $w=w^\mu(x)\partial_\mu|_x$ for a fixed chart. We then
have
\[
	T_x\sigma(w)=T_xj^1\phi(w)=w^\mu(x)T_xj^1\phi(\partial_\mu|_x)=
\]	
\[
	=w^\mu(x)\left[\frac{\partial x^\nu\circ j^1\phi(x)}{\partial
	x^\mu}\frac{\partial}{\partial x^\nu}\bigg|_{\sigma(x)}+\frac
	{\partial y^i\circ j^1\phi(x)}{\partial x^\mu}\frac{\partial}
	{\partial y^i}\bigg|_{\sigma(x)}+\frac{\partial z^i_\nu\circ
	j^1\phi(x)}{\partial x^\mu}\frac{\partial}{\partial z^i_\nu}
	\bigg|_{\sigma(x)}\right]=
\]
\[
	=w^\mu(x)\left[\frac{\partial}{\partial x^\mu}\bigg|_{\sigma(x)}
	+\partial_\mu\phi^i(x)\frac{\partial}
	{\partial y^i}\bigg|_{\sigma(x)}+\partial_\mu\partial_\nu\phi^i(x)
	\frac{\partial}{\partial z^i_\nu}\bigg|_{\sigma(x)}\right]\equiv v.
\]
In general $v\in T_{\sigma(x)}Z$, but we want to show now that
specifically $v\in\mathcal{D_{C}}(\sigma(x))\subseteq T_{\sigma(x)}Z$,
or equivalently that $\iota_v[\theta^i(\sigma(x))]=0$ $\forall i\in
\{1,\dots,m\}$ as the contact 1-forms with local coordinates
$\theta^i=dy^i-z^i_\mu dx^\mu$ generate the ideal $\mathcal{C}$.
Let $i\in\{1,\dots,m\}$ be arbitrary but fixed, then
\[
	\iota_v[\theta^i(\sigma(x))]=\iota_v[dy^i|_{\sigma(x)}-
	\partial_\mu\phi^i(x)dx^\mu|_{\sigma(x)}]=\iota_v
	[dy^i|_{\sigma(x)}]-\partial_\mu\phi^i(x)\iota_v
	[dx^\mu|_{\sigma(x)}]=
\]
\[
	=w^\mu(x)\partial_\mu\phi^i(x)-\partial_\mu\phi^i(x)w^\mu(x)=0,
\]
hence $v\in\mathcal{D_{C}}(\sigma(x))$ and so $T_x\sigma(T_xX)
\subseteq T_{\sigma(x)}Z$.

($\Leftarrow$) Now, by hypothesis $T_x\sigma(T_xX)\subseteq
T_{\sigma(x)}Z\,\forall x\in X$ and so if we take an arbitrary but
fixed $x\in X$ with any given chart such that $\sigma$ takes the
local coordinates $\sigma(x^\mu)=(x^\mu,\sigma^i(x^\mu),
\sigma^i_\nu(x^\mu))$ and an arbitrary but fixed non zero $w=w^\mu(x)
\partial_\mu|_x\in T_xX$ (trivial if $w=0$) then  
\[
	T_x\sigma(w)=w^\mu(x)T_x\sigma(\partial_\mu|_x)=
\]	
\[
	=w^\mu(x)\left[\frac{\partial x^\nu\circ \sigma(x)}{\partial
	x^\mu}\frac{\partial}{\partial x^\nu}\bigg|_{\sigma(x)}+\frac
	{\partial y^i\circ \sigma(x)}{\partial x^\mu}\frac{\partial}
	{\partial y^i}\bigg|_{\sigma(x)}+\frac{\partial z^i_\nu\circ
	\sigma(x)}{\partial x^\mu}\frac{\partial}{\partial z^i_\nu}
	\bigg|_{\sigma(x)}\right]=
\]
\[
	=w^\mu(x)\left[\frac{\partial}{\partial x^\mu}\bigg|_{\sigma(x)}
	+\partial_\mu\sigma^i(x)\frac{\partial}
	{\partial y^i}\bigg|_{\sigma(x)}+\partial_\mu\sigma_\nu^i(x)
	\frac{\partial}{\partial z^i_\nu}\bigg|_{\sigma(x)}\right]\equiv v
	\in\mathcal{D_C}(\sigma(x)),
\]
and so, $\forall i\in\{1,\dots,m\}$, $\iota_v[\theta^i(\sigma(x))]=0$.
But then for an arbitrary but fixed $i\in\{1,\dots,m\}$ we have
\[
\iota_v[\theta^i(\sigma(x))]=\partial_\mu\sigma^i(x)w^\mu(x)-w^\mu(x)
\sigma^i_\mu(x)=0\overset{w\neq 0}{\iff}\sigma^i_\mu(x)=\partial_\mu
\sigma^i(x),\forall\mu\in\{0,\dots,n\}.
\]
As $x$ and $i$ were arbitrary, then $\sigma$ is by definition holonomic
and this ends the proof. 
\end{proof}
\chapter{Proof of Proposition \ref{prop:relation_Hamiltonian_forms_and_functions}}
\label{app:proof_of_proposition_ref_prop_relation_hamiltonian_forms_and_functions}
\begin{proof}
	$d\tilde{\alpha}$ is a $1-$form on $\tZ$ and as such is totally determined by
	its act over $\mathfrak{X}(\tZ)$ at each single point $\gamma\in\tZ$.
	Let $\tilde{\xi}\in\mathfrak{X}(\tZ)$ and $\gamma\in\tZ$ be fixed but
	arbitrary vector field and point on $\tZ$ respectively, 
	then if $c_{\tZ}(t)$ is any curve on $\tZ$ such that $c_{\tZ}(0)=\gamma$ and
	$\dot{c}_{\tZ}(0)=\tilde{\xi}(\gamma)$, noting that as $\tilde{\alpha}$ is a
	$0-$form $d\tilde{\alpha}(\tilde{\xi})=\pounds_{\tilde{\xi}}\tilde{\alpha}=
	\tilde{\xi}(\tilde{\alpha})$ by Cartan's magic formula we have that
	\[
		(d\tilde{\alpha})(\tilde{\xi})|_\gamma=\tilde{\xi}(\gamma)(\tilde{\alpha})=
		\frac{d}{dt}\left(\tilde{\alpha}\circ c_{\tZ}(t)\right)\bigg|_{t=0}
		\overset{\ref{def:integration_upon_Cauchy_of_forms}}{=}\frac{d}{dt}\left[
		\int_M c_{\tZ}(t)^*\alpha -\int_{\partial M}c_{\tZ}(t)^*\beta\right]\bigg|_{t=0}
	\]
	\[
		=\int_M\frac{d}{dt}\left(c_{\tZ}(t)^*\alpha\right)\big|_{t=0}
		-\int_{\partial M}\frac{d}{dt}\left(c_{\tZ}(t)^*\beta\right)\big|_{t=0},
	\]
	where when invoking Definition \ref{def:integration_upon_Cauchy_of_forms}
	we have to note that $c_{\tZ}(t)\in\tZ$ $\forall t\in\RR$ and as so $\tilde{\alpha}
	\circ c_{\tZ}(t)=\tilde{\alpha}|_{c_{\tZ}(t)}$.
	
	Now, if $\xi\in\mathfrak{X}(Z)$ is a vector field that induces $\tilde{\xi}$
	(i.e. we can define $\tilde{\xi}(\gamma)(u)=:\xi(\gamma(u))$) with flow $\Phi_t$,
	then for a fixed $t\in\RR$ $c_{\tZ}(t)(u)=(\Phi_t\circ\gamma)(u)$ $\forall u\in M$
	and so
	\[
		\frac{d}{dt}\left(c_{\tZ}(t)^*\alpha\right)\big|_{t=0}=
		\frac{d}{dt}\left(\gamma^*\circ\Phi_t^*\alpha\right)\big|_{t=0}=
		\gamma^*\frac{d}{dt}\left(\Phi_t^*\alpha\right)\big|_{t=0}=
		\gamma^*\pounds_{\xi}\alpha
	\] 
	(the same happens for the $\beta$ integral). Using this fact we finally have
	\[
		(d\tilde{\alpha})(\tilde{\xi})|_\gamma=\int_M\gamma^*(\pounds_{\xi}\alpha)-
		\int_{\partial M}\gamma^*(\pounds_{\xi}\beta)i=\int_M\gamma^*(\iota_{\xi}d\alpha)
	\]
	\[
		+\int_M\gamma^*(d\iota_{\xi}\alpha)
		-\int_{\partial M}\gamma^*(\iota_{\xi}d\beta)-
		\overset{\text{(Stokes' theorem)}}
		{\cancel{\int_{\partial M}\gamma^*(d\iota_{\xi}\beta)}}
		\overset{\text{(Stokes' theorem)}}{=}\int_M\gamma^*(\iota_{\xi}d\alpha)
	\]
	\[
		+
		\cancel{\int_{\partial M}\gamma^*(\iota_{\xi}\alpha)}
		-\cancel{\int_{\partial M}\gamma^*(\iota_{\xi}d\beta)}=\int_M\gamma^*\iota_\xi
		\iota_{X_\alpha}\Omega_L
		=\iota_{\tilde{X}_{\tilde{\alpha}}}\widetilde{\Omega_L}(\tilde{\xi})|_\gamma,
	\]
	and hence, by the argumentation of the beginning of the proof we have that 
	$d\tilde{\alpha}=\iota_{\tilde{X}_{\tilde{\alpha}}}\widetilde{\Omega_L}$,
	which finishes the proof.
\end{proof}

\chapter{Some computations related to symmetries}
\label{app:some_symmetries_computations}

Let us first prove that $(j^1\phi_t)(\partial X)\subseteq B$ is satisfied
in the case of $\xi_Y^{(1)}$ being tangent to $B$. The main result used to
prove this is the fact that the flow of $\xi_Y^{(1)}$, $\Phi_t^Z$ satisfies
(cf. \cite{saundersGeometryJetBundles1989})
\[
	\Phi_t^Z(j_x^1\phi)=j^1_{\Phi_t^X(x)}(\Phi_t^Y\circ\phi\circ\Phi_{-t}^X).
\]

Then, as $\xi_Y$ is projectable,
$\Phi_t^X(\partial X)\subseteq\partial X$, and the equality is obtained
by noting that $\Phi_t^X$ is a diffeomorphism using the definition of a flow.
For a fixed but arbitrary $x\in\partial X$, we consider $x^\prime:=
(\Phi_t^X)^{-1}(x)\in\partial X$ and so
\[
	(j^1\phi_t)(x)=(j^1\phi_t)(\Phi_t^X(x^\prime))=(j^1_{\Phi_t^X(x^\prime)})
	(\Phi_t^Y\circ\phi\circ\Phi_{-t}^X)=\Phi_t^Z(j^1_{x^\prime}\phi).
\]
As by definition of an admissible solution $(j^1\phi)(\partial X)\subseteq B$
and $x^\prime\in\partial X$, we have that $j_{x^\prime}\phi\in B$ and then
the result follows from imposing $\xi_Y^{(1)}$ being tangent to $B$, the
result we were chasing (note that if a vector field is tangent to a subset, then its flow send
points of the subset to points of the subset).

Let us now prove that $\iota_{\xi_Y^{(1)}}\Theta_L$ is exact on $B$ if
$\pounds_{\xi_Y^{(1)}|_B}\Pi=0$. We recall the definition
of $\Pi$ as a $n-$form satisfying $i_B^*\Theta_L=d\Pi$, then, as $\xi_Y^{(1)}$
is tangent to $B$ by the previous hypothesis, we have
\[
	i_B^*(\iota_{\xi_Y^{(1)}}\Theta_L)=\iota_{\xi_Y^{(1)}|_B}(i_B^*\Theta_L)=
	\iota_{\xi_Y^{(1)}|_B}(d\Pi)=\pounds_{\xi_Y^{(1)}|_B}(\Pi)-
	d(\iota_{\xi_Y^{(1)}|_B}\Pi),
\]
and the exactness on $B$ is then fulfilled if $\pounds_{\xi_Y^{(1)}|_B}(\Pi)=0$.

In the following we write the proof of Proposition \ref{prop:useful_proposition_for_checking_symmetries}
\begin{proof}
\begin{itemize}
		\item Case of $\bm{\xi_Y^{(1)}}$ \textbf{being projectable}

		If $\bm{\xi_Y^{(1)}}$ is projectable then the third and fourth term of
		\eqref{eq:lengthy_computation} vanish as $\xi_Y^\mu$ does not depend on $y$
		while the first and the second terms vanish as $F=0$ so $\pounds_{\xi^{(1)}_{Y}}\Theta_L=0$;
		isolating $\xi_Y^{(1)}(L)$ in \eqref{eq:F_local_expression} gets us
		\[
			\xi_Y^{(1)}(L)=-\frac{\partial\xi_Y^\mu(x)}{\partial x^\mu}L=
			-\frac{d\xi_Y^\mu(x)}{dx^\mu}L.
		\]
		\item Case of $n=0$

		The big bracket of third term in Equation \eqref{eq:lengthy_computation} is literally
		zero as $\mu$ and $\nu$ index must be both equal (to $0$); forth term contains
		$d^{n-1}x_{\nu\mu}$ which vanish for $n=0$ while as before
		the first and the second terms vanish as $F=0$. As $\xi^\mu_Y$ does not depend on $z$
		isolating $\xi_Y^{(1)}(L)$ in \eqref{eq:F_local_expression} get us
		\eqref{eq:useful_tool_for_checking_symmetries}.
\end{itemize}	
\end{proof}

\chapter{Some proofs on gauge theory}
\label{app:some_proofs_on_gauge_theory}
Let us proof Equation \eqref{eq:adjoint_bundle_is_vertical_equivariant_vector_fields_on_P}
\begin{proof}
As we have previously denoted, given $\xi\in\mathfrak{g}=T_e G$, we have that $\xi_P(p)$ is
defined by the derivative of the curve $c(t)=p\exp(t\xi)$ on $t=0$, then
\begin{align*}
	(R_g)_*(\xi_P(p))&=\der{}{t}\Big|_{t=0}(R_g\circ c)(t)=\der{}{t}\Big|_{t=0}c(t)g=
	\der{}{t}\Big|_{t=0}p\exp(t\xi)g\\
	&=\der{}{t}\Big|_{t=0}(pg)(g^{-1}\exp(t\xi)g).
\end{align*}
Now, consider the curve $\tilde{\gamma}(t):=g^{-1}\exp(t\xi)g$; this curve is equal to $C_g(\exp(t\xi))$, and so
$\tilde{\gamma}(0)=e$ and $\dot{\tilde{\gamma}}(0)=\der{}{t}\big|_{t=0}C_g(\exp(t\xi))
=(dC_g)_e(\xi)=Ad_g(\xi)$. Using the uniqueness of the exponential curve, then $\tilde{\gamma}(t)
=\exp(t\Ad_g\xi)$. 
As this is valid for every $g\in G$, is also valid for $g^{-1}$ and
so $g^{-1}\exp(t\xi)g=\exp(t\Ad_{g^{-1}}\xi)$.

Finally
\begin{align}
	\label{eq:proof_adjoint_bundle_is_vertical_proof}
	(R_g)_*(\xi_P(p))&=\der{}{t}\Big|_{t=0}(pg)(g^{-1}\exp(t\xi)g)
	=\der{}{t}\Big|_{t=0}(pg)\exp(t\Ad_{g^{-1}}\xi)\\
	&=(\Ad_{g^{-1}}\xi)_P(pg).
\end{align}
\end{proof}
Let us now prove that $\mathfrak{gau}(P)$ vectors are vertical and $G-$invariant
\begin{proof}
	Consider a $1-$parameter family of diffeomorphisms $\Phi_t$ such that $\Phi_0=
	\mathrm{id}_P$. For each $p\in P$, consider the curve $\Phi_t(p)$ on $P$ and
	define the vector field $\xi\in\mathfrak{X}(P)$ given by $\xi(p):=\der{}{t}
	\big|_{t=0}\Phi_t(p)$. Then $\xi$ satisfies
	\begin{enumerate}
	\item $\bm{\xi}$ \textbf{is vertical}, i.e. $\xi\in\mathcal{V}\pi_P$

	Using that $\Phi_t$ covers the identity on $M$, we have $\pi_P\circ\Phi_t=\pi_P$.
	Differentiating at $t=0$ we have that,
	\[
	(\pi_{P})_{*,p}\,\xi(p)=
	(\pi_{P})_{*,p}\left(\der{}{t}\Big|_{t=0}\Phi_t(p)\right)
	=
	\der{}{t}\Big|_{t=0}\big(\pi_P(\Phi_t(p))\big)
	=
	\der{}{t}\Big|_{t=0}\big(\pi_P(p)\big)
	=0.
	\]
	Hence $\xi(p)\in\ker(\pi_{P})_{*,p}=\mathcal{V}_p\pi_P$ and so $\xi$
	is vertical by definition.
	\item $\bm{\xi}$ \textbf{is} $\bm{G-}$\textbf{invariant}, i.e. $(R_g)_*\xi=\xi$

	From $G$-equivariance of $\Phi_t$, $\Phi_t(pg)=\Phi_t(p)g$,
	then we have:
	\begin{align*}
		(R_g)_{*,p}\,\xi(p)&=\der{}{t}\Big|_{t=0} (R_g\circ \Phi_t)(p)
		=\der{}{t}\Big|_{t=0} \big(\Phi_t(p)\cdot g\big)\\
		&=\der{}{t}\Big|_{t=0} \Phi_t(p\cdot g)=\xi(p\cdot g).
	\end{align*}
	Therefore $(R_g)_*\xi = \xi$ and so $\xi$ is $G$-invariant by definition.
	\end{enumerate}
\end{proof}

Now, we proof that if $\Phi\in\Gau(P)$, then $\Phi\cdot\omega$ is a PFB connection for
every $\omega$ PFB connection, i.e., $\Phi^*\omega$ is $\Ad-$equivariant an
reproduce fundamental fields
\begin{proof}
\begin{enumerate}
	\item $\bm{\mathrm{Ad}}$\textbf{-equivariance}

		A principal connection must satisfy $R_g^* \omega = \text{Ad}_{g^{-1}} \circ \omega$.
		For the pullbacked form $\Phi^*\omega$, we compute:
		\[
		R_g^* (\Phi^* \omega) = (\Phi \circ R_g)^* \omega,
		\]
		and since $\Phi$ is $G$-equivariant, ($\Phi\circ R_g=R_g\circ\Phi$), so we have:
		\[
		(\Phi\circ R_g)^* \omega=(R_g\circ\Phi)^*\omega=\Phi^*(R_g^*\omega).
		\]
		Using now the equivariance of the original connection $\omega$:
		\[
		\Phi^*(R_g^*\omega)=\Phi^*(\text{Ad}_{g^{-1}}\circ\omega)=\text{Ad}_{g^{-1}}\circ
		(\Phi^*\omega)
		\]
		where the last equality follows from the fact that the linear operator
		$\text{Ad}_{g^{-1}}$ commutes with the pullback. Thus, $\Phi^*\omega$ is equivariant.
	\item \textbf{Reproducing the fundamental fields}

		Let $u\in P$ be a fixed element in $P$ but arbitrary.
		A connection must satisfy $\omega(A_P(u)) = A$ for any $A \in \mathfrak{g}$, where
		$A_P$ is the fundamental vector field. By definition of the pullback:
		\[
		(\Phi^*\omega)_u(A_P(u))=\omega_{\Phi(u)} (d\Phi)_u(A_P(u)))
		\]
		Because $\Phi$ is a vertical, $G$-equivariant map
		$\Phi(u\cdot\exp(tA))=\Phi(u)\cdot\exp(tA)$. Taking the derivative
		at $t=0$:
		\[
		d(\Phi)_u(A_P(u))=\frac{d}{dt}\Big|_{t=0}\Phi(u \cdot \exp(tA))=
		\frac{d}{dt}\Big|_{t=0} \Phi(u)\cdot\exp(tA)=(A_P)_{\Phi(u)}
		\]
		Substituting this into the expression for $\Phi^*\omega$:
		\[
		\omega_{\Phi(u)} (d(\Phi)_u(A_P(u)))=\omega_{\Phi(u)}
		((A_P)_{\Phi(u)})=A
		\]
		The last step follows because $\omega$ is a connection. Therefore, $\Phi^*\omega$
		reproduce the fundamental fields.
\end{enumerate}
Since $\Phi^* \omega$ satisfies both conditions, it is a PFB connection.
\end{proof}
Let us now prove Proposition \ref{prop:pullback_omega_gauge}
\begin{proof}
Write $\mu\colon P\times G\to P$ for the right action on $P$, so 
$\Phi_{\tau}(p)=\mu(p,\tau(p))=(\mu\circ(\mathrm{id}_{P},\tau))(p)$.
For $v\in T_{p}P$ the Leibniz rule for $\mu$ gives
\[
  T_p\Phi_\tau(v)=T_{(p,\tau(p))}\mu(v,T_p\tau (v))=T_p[\mu(\cdot,\tau(p))](v)+
  T_{\tau(p)}[\mu(p,\cdot)](T_p\tau(v)).
\]
The first term is equal to $\bigl(R_{\tau(p)}\bigr)_{*}(v)$, while for the second we
notice that $\mu(p,\cdot)=\mu(pg,\cdot)\circ L_{g^{-1}}$ and hence by chain rule
\[
  T_{\tau(p)}[\mu(p,\cdot)](T_p\tau(v))=T_{\tau(p)}[\mu(pg,\cdot)\circ L_{g^{-1}}]
  (T_p\tau(v)).
\]
In order to make the proof easier, let us denote $g\equiv\tau(p)$ and $w\equiv T_p\tau(v)
\in T_g G$,
then
\[
  T_{g}[\mu(pg,\cdot)\circ L_{g^{-1}}](w)=T_e[\mu(pg,\cdot)]\left(T_g(L_{g^{-1}})(w)\right);
\]
the first summand being the derivative in the $P$ slot and the second the derivative in
the $G$ slot. Now, by definition of the Maurer--Cartan form, $T_g(L_{g^{-1}})(w)=
\theta_G|_g(w)$ while $T_e(\mu(pg,\cdot))(X)=X_P(pg)$ (a fact that is not difficult to prove),
hence 
\[
  T_{p}\Phi_{\tau}(v)
  =\bigl(R_{\tau(p)}\bigr)_{*}(v)
  \;+\;\Bigl(\theta_{G}|_g\bigl(T_{p}\tau(v)\bigr)\Bigr)_{P}\!\bigl(p\,\tau(p)\bigr).
\]
Applying $\omega$ and using the two defining properties of Definition \ref{def:connection},
\[
  \bigl(\Phi_{\tau}^{*}\omega\bigr)_{p}(v)
  =\omega_{p\tau(p)}\bigl((R_{\tau(p)})_{*}v\bigr)+\theta_{G}|_g\bigl(T_{p}\tau(v)\bigr)
  =\Ad_{\tau(p)^{-1}}\bigl(\omega_{p}(v)\bigr)+\bigl(\tau^{*}\theta_{G}\bigr)_{p}(v),
\]
which is \eqref{eq:pullback_omega_gauge}. That $\Phi_{\tau}^{*}\omega$ is a connection
follows from equivariance of $\Phi_{\tau}$: a $G$-equivariant map intertwines
fundamental vector fields, $T\Phi_{\tau}(\xi_{P}(p))=\xi_{P}(\Phi_{\tau}(p))$, so
$(\Phi_{\tau}^{*}\omega)(\xi_{P})=\xi$; and $\Phi_{\tau}\circ R_{g}=R_{g}\circ\Phi_{\tau}$
gives
\[
  R_{g}^{*}\Phi_{\tau}^{*}\omega=(\Phi_{\tau}\circ R_{g})^{*}\omega
  =(R_{g}\circ\Phi_{\tau})^{*}\omega=\Phi_{\tau}^{*}R_{g}^{*}\omega
  =\Ad_{g^{-1}}\!\circ\Phi_{\tau}^{*}\omega .
\]
Equation \eqref{eq:gauge_transf_A_active} follows by pulling
\eqref{eq:pullback_omega_gauge} back along $s_{\alpha}$ and using
$s_{\alpha}^{*}\tau^{*}=(\tau\circ s_{\alpha})^{*}=\tau_{\alpha}^{*}$.
\end{proof}

Let us now prove Proposition \ref{prop:connections_as_sections_explicit}
\begin{proof}
The orbit of $j^{1}_{x}s$ is $\{(pg,(R_{g})_{*}H_{p})\}_{g\in G}$ by
Lemma~\ref{lem:jets_are_complements}; freeness of the action means it hits
$\pi_{1,0}^{-1}(p)$ only for $g=e$, and transitivity on $P_{x}$ means it
reaches every $p'\in P_{x}$. So an orbit is exactly an equivariant
assignment of complements along the fiber. Smoothness of $\phi_{\omega}$
and of its inverse is local and can be proven to follow from
\eqref{eq:horizontal_lift_local}; the same happens for its fiber coordinates
being $A_\mu^a(x)$ which can be reached once stating the identification
$\phi_\omega(x)\longleftrightarrow A_\alpha(x)$ by some linear algebra results.
For \eqref{eq:horizontal_lift_local},
apply $\omega$: by Definition~\ref{def:connection}(i),
\[
  \omega\Bigl(T_{x}s_{\alpha}(v)-\bigl(A_{\alpha}(v)\bigr)_{P}\Bigr)
  =(s_{\alpha}^{*}\omega)(v)-A_{\alpha}(v)=A_{\alpha}(v)-A_{\alpha}(v)=0 ,
\]
so the displayed vector is horizontal, and it projects to $v$ because the
subtracted term is vertical. The last statement is the
Remark following Definition~\ref{def:connection}: the difference of two
connection forms is horizontal and $\Ad$-equivariant, hence descends to a
$1$-form on $X$ with values in $\ad P$.
\end{proof}

We now prove Theorem \ref{thm:gauge_prolongation} in glory detail
\section{Proof of Theorem \ref{thm:gauge_prolongation}}
\label{app:the_big_proof}
\begin{proof}
Each $\Phi_{t}$ lies in $\Gau(P)$, so $(\Phi_{t})_{Y}$ is a one parameter
group of vertical automorphisms of $\pi$ by
Lemma~\ref{lem:PhiC_welldefined} and \eqref{eq:PhiY_def}.

\textbf{i)} Its generator $\xi_{C}$ is therefore a $\pi$-vertical vector field, and
for the local expression, let $\tau_{t}\in C^{\infty}_{\Ad}(P,G)$ be the
equivariant map of $\Phi_{t}$ given by
Proposition~\ref{prop:gauge_group_equivariant_maps}, so $\Phi_{t}(p)=p\,
\tau_{t}(p)$ with $\tau_{0}(p)\equiv e$. Differentiating at $t=0$,
\[
  \xi(p)=\der{}{t}\Big|_{0}\,\Phi_t(p)=\der{}{t}\Big|_{0}\,p\,\tau_{t}(p)
        =\Bigl(\der{}{t}\Big|_{0}\tau_{t}(p)\Bigr)_{\!P}(p)
  \quad\Longrightarrow\quad
  \der{}{t}\Big|_{0}\tau_{t}=\widehat\xi ,
\]
i.e.\ the equivariant map $\widehat\xi$ of
Lemma~\ref{lem:gau_three_descriptions} is exactly the velocity of
$\tau_{t}$. By \eqref{eq:PhiY_def} the curve on $Y=C(P)$
$t\mapsto(\Phi_{t})_{Y}(\phi_{\omega}(x))$ is $t\mapsto
\phi_{\Phi_{t}^{*}\omega}(x)$, whose coordinates in the frame $s_{\alpha}$
are, by Proposition~\ref{prop:pullback_omega_gauge},
\[
  A^{(t)}_{\alpha}:=s_\alpha^*(\Phi_t^*\omega)
  =\Ad_{\tau_{t,\alpha}^{-1}}\!\circ A_{\alpha}+\tau_{t,\alpha}^{*}\theta_{G},
  \qquad \tau_{t,\alpha}:=\tau_{t}\circ s_{\alpha},
  \quad \tau_{0,\alpha}\equiv e,
  \quad \der{}{t}\Big|_{0}\tau_{t,\alpha}=\xi_{\alpha},
\]
where $\xi_\alpha:=\hat{\xi}\circ s_\alpha\equiv\xi^a(x)e_a$. Note that the
dependence of the components $\xi^a$ on $x$ is made explicit; this is done in order
to remark the correspondence of Lemma \ref{lem:gau_three_descriptions} in which
our 'gauge vectors´ are seen as sections of the adjoint bundle and hence, as elements
of an infinite-dimensional Lie algebra ($\mathfrak{gau}(P)$ in particular); when pullbacked
to $X$, the infinite nature of the vector is reflected on this $x$ dependence.

The first term of last equation differentiates to $\ad_{-\xi_{\alpha}}A_{\alpha}
=[A_{\alpha},\xi_{\alpha}]$. For the second, by definition of the
Maurer--Cartan form,
\[
  \bigl(\tau_{t,\alpha}^{*}\theta_{G}\bigr)_{x}(v)
  =\bigl(L_{\tau_{t,\alpha}(x)^{-1}}\bigr)_{*}\,T_{x}\tau_{t,\alpha}(v),
  \qquad v\in T_{x}X;
\]
since $\tau_{0,\alpha}$ is constant we have $T_{x}\tau_{0,\alpha}=0$, so by
the product rule only the derivative of the second factor survives at
$t=0$, giving $T_{x}\xi_{\alpha}(v)=\dd\xi_{\alpha}(v)$. Hence
\begin{equation}\label{eq:infinitesimal_gauge}
  \delta_{\xi}A_{\alpha}
  :=\der{}{t}\Big|_{0}A^{(t)}_{\alpha}
  =\dd\xi_{\alpha}+[A_{\alpha},\xi_{\alpha}]
  =:D_{A_{\alpha}}\xi_{\alpha},
\end{equation}
which in components is (note we are using the coordinate version of
\eqref{eq:D_adjoint_components} as $\xi_\alpha$ comes from $\hat{\xi}$ which is
Ad-equivariant)
\begin{equation}\label{eq:infinitesimal_gauge_components}
  (\delta_{\xi}A)^{a}_{\mu}(x)
  =\partial_{\mu}\xi^{a}(x)+C^{a}_{bc}A^{b}_{\mu}(x)\xi^{c}(x)
  =+\bigl(D_{\mu}\xi\bigr)^{a},
\end{equation}
and this is \eqref{eq:xi_C_local}, since the fiber coordinates of
$\phi_{\omega}$ are the $A^{a}_{\mu}$.

\textbf{ii)} For the kernel, note that \eqref{eq:xi_C_local} is \emph{affine} in the
fiber coordinates: $\xi_{C}$ vanishes identically iff both the constant
term $\partial_{\mu}\xi^{a}$ and the linear term
$C^{a}_{bc}A^{b}_{\mu}\xi^{c}$ vanish for all values of $A^{b}_{\mu}$, i.e.\
iff $\dd\xi_{\alpha}=0$ and $\ad_{\xi_{\alpha}}=0$; that is, iff
$\xi_{\alpha}$ is locally constant (i.e. constant in a neighborhood of every point in
$U_\alpha)$) with values in the centre
$\mathfrak{z}(\mathfrak{g})$ (i.e. vector fields on $\mathfrak{g}$ commuting with
every other Lie algebra vector field). Central elements are $\Ad$-invariant for $G$
connected (cf. \cite{hall_lie_2015}), so by \eqref{eq:xi_transition} such a $\xi$ is a
globally well defined locally constant $\mathfrak{z}(\mathfrak{g})$-valued function
on $X$, and the kernel vanishes when $\mathfrak{z}(\mathfrak{g})=0$. For the bracket, both
$\xi_{C}$ and $\zeta_{C}$ are vertical, so (for holonomic basis)
$[\xi_{C},\zeta_{C}]^a_\mu
=\xi_{C}\bigl(({D_{\mu}\zeta})^{a}\bigr)-\zeta_{C}\bigl((D_{\mu}\xi)^{a}\bigr)$
and, since $\partial(D_{\mu}\zeta)^{a}/\partial A^{b}_{\nu}
=C^{a}_{bc}\zeta^{c}\delta^{\nu}_{\mu}$,
\[
  [\xi_{C},\zeta_{C}]^a_\mu
  =C^{a}_{bc}\bigl(\partial_{\mu}\xi^{b}\zeta^{c}
   -\partial_{\mu}\zeta^{b}\xi^{c}\bigr)
  +C^{a}_{bc}C^{b}_{de}A^{d}_{\mu}\bigl(\xi^{e}\zeta^{c}-\zeta^{e}\xi^{c}\bigr).
\]
The first bracket equals $C^{a}_{bc}\partial_{\mu}(\xi^{b}\zeta^{c})$ by
the skew symmetry $C^{a}_{cb}=-C^{a}_{bc}$, and the second equals
$C^{a}_{db}C^{b}_{pq}A^{d}_{\mu}\xi^{p}\zeta^{q}$ by the Jacobi identity in
the form
$C^{a}_{db}C^{b}_{pq}+C^{a}_{pb}C^{b}_{qd}+C^{a}_{qb}C^{b}_{dp}=0$.
Together these give exactly
$\bigl(D_{\mu}[\xi,\zeta]\bigr)^{a}$ with $[\xi,\zeta]^{a}
=C^{a}_{bc}\xi^{b}\zeta^{c}$, i.e.\ $[\xi_{C},\zeta_{C}]=[\xi,\zeta]_{C}$.
(This is where the sign of Remark~\ref{rem:bracket_sign} is used: had we
taken the vector field bracket on $\gau(P)$, the map would be an
anti-homomorphism.)

\textbf{iii)} Equation \eqref{eq:j1xi_C_local} is the definition of the $1$-jet
prolongation of a projectable vector field (Definition \ref{def:1jet-prolongation}):
the fiber component of
$j^{1}\xi_{C}$ along $\partial/\partial A^{a}_{\mu,\nu}$ is the total
derivative $\dd/\dd x^{\nu}$ of the fiber component of $\xi_{C}$, because
$\xi_{C}$ is vertical (so there is no $-A^{a}_{\mu,\rho}\,\partial_{\nu}
(\text{base component})$ correction). Expanding the total derivative gives
\eqref{eq:j1xi_C_local}.
\end{proof}
\chapter{Auxiliary computations of Chapter \ref{chap:multisymplectic_yang_mills_computations}}
\label{app:auxiliary}
As already mentioned in the main text, on this Appendix we will archive some tedious
computations that perse do not add any value to the thesis. We will adopt the same notations
as the ones declared on Section \ref{sec:initial_clarifications}.

First, let us compute explicitly the multimomenta in terms of the coordinates on $Z$
\section{Multimomenta}
\label{eq:explicit_computation_of_multimomentum}
	\begin{align*}
	p_{a}^{\mu\nu}
	&:=
	\frac{\partial L_{YM}}{\partial A^{a}_{\mu\nu}}
	=
	\frac{\partial}{\partial A^{a}_{\mu\nu}}
	\left[
	-\frac{1}{4}\,\kappa_{bc}\,
	F^{b}_{\rho\sigma}F^{c}_{\tau\gamma}\,
	g^{\rho\tau}g^{\sigma\gamma}\sqrt{|g|}
	\right]\\
	&=
	-\frac{1}{4}\,\kappa_{bc}\,
	\frac{\partial F^{b}_{\rho\sigma}}
	{\partial A^{a}_{\mu\nu}}\,
	F^{c}_{\tau\gamma}\,
	g^{\rho\tau}g^{\sigma\gamma}\sqrt{|g|}
	-
	\frac{1}{4}\,\kappa_{bc}\,
	F^{b}_{\rho\sigma}\,
	\frac{\partial F^{c}_{\tau\gamma}}
	{\partial A^{a}_{\mu\nu}}\,
	g^{\rho\tau}g^{\sigma\gamma}\sqrt{|g|}\\
	&=
	-\frac{1}{4}\,\kappa_{bc}\,
	\left(
	\delta^{b}_{a}\delta^{\mu}_{\sigma}\delta^{\nu}_{\rho}
	-
	\delta^{b}_{a}\delta^{\mu}_{\rho}\delta^{\nu}_{\sigma}
	\right)
	F^{c}_{\tau\gamma}\,
	g^{\rho\tau}g^{\sigma\gamma}\sqrt{|g|}\\
	&\qquad
	-\frac{1}{4}\,\kappa_{bc}\,
	F^{b}_{\rho\sigma}\,
	\left(
	\delta^{c}_{a}\delta^{\mu}_{\gamma}\delta^{\nu}_{\tau}
	-
	\delta^{c}_{a}\delta^{\mu}_{\tau}\delta^{\nu}_{\gamma}
	\right)
	g^{\rho\tau}g^{\sigma\gamma}\sqrt{|g|}\\
	&=
	-\frac{1}{4}\,\kappa_{ac}\,
	\delta^{\mu}_{\sigma}\delta^{\nu}_{\rho}\,
	F^{c}_{\tau\gamma}\,
	g^{\rho\tau}g^{\sigma\gamma}\sqrt{|g|}
	+
	\frac{1}{4}\,\kappa_{ac}\,
	\delta^{\mu}_{\rho}\delta^{\nu}_{\sigma}\,
	F^{c}_{\tau\gamma}\,
	g^{\rho\tau}g^{\sigma\gamma}\sqrt{|g|}\\
	&\qquad
	-\frac{1}{4}\,\kappa_{ba}\,
	F^{b}_{\rho\sigma}\,
	\delta^{\mu}_{\gamma}\delta^{\nu}_{\tau}\,
	g^{\rho\tau}g^{\sigma\gamma}\sqrt{|g|}
	+
	\frac{1}{4}\,\kappa_{ba}\,
	F^{b}_{\rho\sigma}\,
	\delta^{\mu}_{\tau}\delta^{\nu}_{\gamma}\,
	g^{\rho\tau}g^{\sigma\gamma}\sqrt{|g|}\\
	&=
	-\frac{1}{4}\,\kappa_{ac}\,
	F^{c}_{\tau\gamma}\,
	g^{\nu\tau}g^{\mu\gamma}\sqrt{|g|}
	+
	\frac{1}{4}\,\kappa_{ac}\,
	F^{c}_{\tau\gamma}\,
	g^{\mu\tau}g^{\nu\gamma}\sqrt{|g|}\\
	&\qquad
	-\frac{1}{4}\,\kappa_{ba}\,
	F^{b}_{\rho\sigma}\,
	g^{\rho\nu}g^{\sigma\mu}\sqrt{|g|}
	+
	\frac{1}{4}\,\kappa_{ba}\,
	F^{b}_{\rho\sigma}\,
	g^{\rho\mu}g^{\sigma\nu}\sqrt{|g|}\\
	&=
	-\frac{1}{4}\,\kappa_{ac}\,
	F^{c\,\nu\mu}\sqrt{|g|}
	+
	\frac{1}{4}\,\kappa_{ac}\,
	F^{c\,\mu\nu}\sqrt{|g|}
	-
	\frac{1}{4}\,\kappa_{ab}\,
	F^{b\,\nu\mu}\sqrt{|g|}
	+
	\frac{1}{4}\,\kappa_{ab}\,
	F^{b\,\mu\nu}\sqrt{|g|}\\
	&\equiv
	\kappa_{ab}\,F^{b\,\mu\nu}\sqrt{|g|}.
	\end{align*}
In this Section we will also compute $dp_a^{\mu\nu}$; first
\begin{align*}
\dd F_{\rho\sigma}^{b}
&=
\dd A_{\sigma\rho}^{b}
-
\dd A_{\rho\sigma}^{b}
+
C^{b}_{cd}
\left(
(\dd A_{\rho}^{c})A_{\sigma}^{d}
+
A_{\rho}^{c}(\dd A_{\sigma}^{d})
\right)\\
&=
\dd A_{\sigma\rho}^{b}
-
\dd A_{\rho\sigma}^{b}
+
C^{b}_{cd}
\left(
(\dd A_{\rho}^{c})A_{\sigma}^{d}
-
A_{\rho}^{d}(\dd A_{\sigma}^{c})
\right)
\end{align*}
hence from the multimomentum expression we have
\begin{align*}
\dd p_{a}^{\mu\nu}
&=
\sqrt{|g|}\,\kappa_{ab}\,g^{\mu\rho}g^{\nu\sigma}
\left[
\dd A_{\sigma\rho}^{b}
-
\dd A_{\rho\sigma}^{b}
+
C^{b}_{cd}
\left(
(\dd A_{\rho}^{c})A_{\sigma}^{d}
-
A_{\rho}^{d}(\dd A_{\sigma}^{c})
\right)
\right]
\\
&\qquad
+
\kappa_{ab}F^b_{\rho\sigma}\,\partial_\alpha(g^{\mu\rho}g^{\nu\sigma}\sqrt{|g|})\dd x^\alpha\\
&=\kappa_{ab}F^b_{\rho\sigma}\,\partial_\alpha(g^{\mu\rho}g^{\nu\sigma}\sqrt{|g|})\dd x^\alpha+
C^b_{cd}\frac{\partial^{2}L_{YM}}{\partial A^{a}_{\mu\nu}\,\partial
A^{b}_{\rho\sigma}}
A^c_\sigma\dd A^d_\rho\\
&+\frac{\partial^{2}L_{YM}}{\partial A^{a}_{\mu\nu}\,\partial A^{b}_{\rho\sigma}}
\dd A^b_{\rho\sigma},
\end{align*}
where in the last equality we used Section \ref{eq:explicit_hessian}.

We can also compute the partial derivative of the Lagrangian w.r.t. to $A^a_\mu$ $Z$-coordinate
\section{Partial derivative of the Lagrangian w.r.t. the field components}
\label{eq:pd_of_L_wrt_field_components}
\begin{align*}
\frac{\partial L_{YM}}{\partial A^{a}_{\mu}}
&=
\frac{\partial}{\partial A^{a}_{\mu}}
\left[
-\frac{1}{4}\kappa_{bc}F^{b}_{\rho\sigma}F^{c}_{\tau\epsilon}
g^{\rho\tau}g^{\sigma\epsilon}\sqrt{|g|}
\right]
\\
&=
-\frac{1}{4}\kappa_{bc}
\frac{\partial F^{b}_{\rho\sigma}}{\partial A^{a}_{\mu}}
F^{c}_{\tau\epsilon}
g^{\rho\tau}g^{\sigma\epsilon}\sqrt{|g|}
-\frac{1}{4}\kappa_{bc}F^{b}_{\rho\sigma}
\frac{\partial F^{c}_{\tau\epsilon}}{\partial A^{a}_{\mu}}
g^{\rho\tau}g^{\sigma\epsilon}\sqrt{|g|}
\\
&=
-\frac{1}{4}\kappa_{bc}
\left(
C^{b}_{ad}\delta^{\mu}_{\rho}A^{d}_{\sigma}
+
C^{b}_{da}\delta^{\mu}_{\sigma}A^{d}_{\rho}
\right)
F^{c}_{\tau\epsilon}
g^{\rho\tau}g^{\sigma\epsilon}\sqrt{|g|}
\\
&\qquad
-\frac{1}{4}\kappa_{bc}F^{b}_{\rho\sigma}
\left(
C^{c}_{ad}\delta^{\mu}_{\tau}A^{d}_{\epsilon}
+
C^{c}_{da}\delta^{\mu}_{\epsilon}A^{d}_{\tau}
\right)
g^{\rho\tau}g^{\sigma\epsilon}\sqrt{|g|}
\\
&=
-\frac{1}{4}\kappa_{bc}C^{b}_{ad}
\delta^{\mu}_{\rho}A^{d}_{\sigma}
F^{c}_{\tau\epsilon}
g^{\rho\tau}g^{\sigma\epsilon}\sqrt{|g|}
-\frac{1}{4}\kappa_{bc}C^{b}_{da}
\delta^{\mu}_{\sigma}A^{d}_{\rho}
F^{c}_{\tau\epsilon}
g^{\rho\tau}g^{\sigma\epsilon}\sqrt{|g|}
\\
&\qquad
-\frac{1}{4}\kappa_{bc}F^{b}_{\rho\sigma}
C^{c}_{ad}\delta^{\mu}_{\tau}A^{d}_{\epsilon}
g^{\rho\tau}g^{\sigma\epsilon}\sqrt{|g|}
-\frac{1}{4}\kappa_{bc}F^{b}_{\rho\sigma}
C^{c}_{da}\delta^{\mu}_{\epsilon}A^{d}_{\tau}
g^{\rho\tau}g^{\sigma\epsilon}\sqrt{|g|}
\\
&=
-\frac{1}{4}\kappa_{bc}C^{b}_{ad}
A^{d}_{\sigma}F^{c}_{\tau\epsilon}
g^{\mu\tau}g^{\sigma\epsilon}\sqrt{|g|}
-\frac{1}{4}\kappa_{bc}C^{b}_{da}
A^{d}_{\rho}F^{c}_{\tau\epsilon}
g^{\rho\tau}g^{\mu\epsilon}\sqrt{|g|}
\\
&\qquad
-\frac{1}{4}\kappa_{bc}C^{c}_{ad}
F^{b}_{\rho\sigma}A^{d}_{\epsilon}
g^{\rho\mu}g^{\sigma\epsilon}\sqrt{|g|}
-\frac{1}{4}\kappa_{bc}C^{c}_{da}
F^{b}_{\rho\sigma}A^{d}_{\tau}
g^{\rho\tau}g^{\sigma\mu}\sqrt{|g|}
\\
&\equiv
-\frac{1}{4}\kappa_{bc}C^{b}_{ad}
A^{d}_{\nu}F^{c\,\mu\nu}\sqrt{|g|}
-\frac{1}{4}\kappa_{bc}C^{b}_{da}
A^{d}_{\nu}F^{c\,\nu\mu}\sqrt{|g|}
\\
&\qquad
-\frac{1}{4}\kappa_{bc}C^{c}_{ad}
F^{b\,\mu\nu}A^{d}_{\nu}\sqrt{|g|}
-\frac{1}{4}\kappa_{bc}C^{c}_{da}
F^{b\,\nu\mu}A^{d}_{\nu}\sqrt{|g|}
\\
&=
-\frac{1}{2}\kappa_{bc}C^{b}_{ad}
A^{d}_{\nu}F^{c\,\mu\nu}\sqrt{|g|}
-\frac{1}{2}\kappa_{bc}C^{c}_{ad}
F^{b\,\mu\nu}A^{d}_{\nu}\sqrt{|g|}
\\
&\equiv
-\sqrt{|g|}\,\kappa_{cd}C^{c}_{ab}A^{b}_{\nu}F^{d\,\mu\nu}
=
-\cancel{\sqrt{|g|}}\,C^{c}_{ab}A^{b}_{\nu}
\frac{p_{c}^{\mu\nu}}{\cancel{\sqrt{|g|}}}
=
-C^{c}_{ab}A^{b}_{\nu}p_{c}^{\mu\nu}.
\end{align*}
\section{Hessian}
\label{eq:explicit_hessian}
Let us now compute the Hessian of our Lagrangian as a function in $Z$
\begin{align*}
\frac{\partial^{2}L_{YM}}
{\partial A^{a}_{\mu\nu}\,\partial A^{b}_{\eta\theta}}
&\equiv W_{ab}^{\mu\nu,\eta\theta}=
\frac{\partial}{\partial A^{b}_{\eta\theta}}
\left(
p_a^{\mu\nu}
\right)
=
\frac{\partial}{\partial A^{b}_{\eta\theta}}
\left[
\sqrt{|g|}\,\kappa_{ac}\,
F^{c\,\mu\nu}
\right]
=
\sqrt{|g|}\,\kappa_{ac}\,
\frac{\partial F^{c\,\mu\nu}}
{\partial A^{b}_{\eta\theta}}
\\
&=
\sqrt{|g|}\,\kappa_{ac}\,
g^{\mu\sigma}g^{\nu\tau}
\frac{\partial F^{c}_{\sigma\tau}}
{\partial A^{b}_{\eta\theta}}
=
\sqrt{|g|}\,\kappa_{ac}\,
g^{\mu\sigma}g^{\nu\tau}
\left(
\delta^{c}_{b}\delta^{\eta}_{\tau}\delta^{\theta}_{\sigma}
-
\delta^{c}_{b}\delta^{\eta}_{\sigma}\delta^{\theta}_{\tau}
\right)
\\
&=
\sqrt{|g|}\,\kappa_{ab}\,
g^{\mu\sigma}g^{\nu\tau}
\left(
\delta^{\eta}_{\tau}\delta^{\theta}_{\sigma}
-
\delta^{\eta}_{\sigma}\delta^{\theta}_{\tau}
\right)
=
\sqrt{|g|}\,\kappa_{ab}
\left(
g^{\mu\theta}g^{\nu\eta}
-
g^{\mu\eta}g^{\nu\theta}
\right),
\end{align*}
which is clearly $(\eta,\theta)$ and $(\mu,\nu)$-antisymmetric.
\section{Some horizontal projector computations}
Recalling Equation \eqref{eq:interior_product_of_h} for the definition of the interior
product $\iota_\hh$ (not a real interior product as it does not lower the degree of forms)
let us prove the following proposition
\begin{proposition}
	Let $\Gamma$ be any connection on $\pi_{ZX}$ let $f\in C^\infty(Z)$ be any
	smooth function on $Z$, then the following expression holds valid locally
	\begin{equation}
		\label{eq:ihf}
		\iota_\hh\dd f(z)=\left(\partial_\mu f(z)+(\Gamma_\mu)^a_\rho(z)\parder{f(z)}{A^a_\rho}
		+(\Gamma_\mu)^a_{\rho\sigma}(z)\parder{f(z)}{A^a_{\rho\sigma}}\right)dx^\mu|_z
	\end{equation}
\end{proposition}
\begin{proof}
Let $V\in\mathfrak{X}(Z)$ be any vector field on $Z$ and let us define its components
through its local expression
\[
		V(z)=V^\mu(z)\parder{}{x^\mu}\bigg|_z+V^a_\mu(z)\parder{}{A^a_\mu}\bigg|_z+
		V^a_{\mu\nu}(z)\parder{}{A^a_{\mu\nu}}\bigg|_z.
\]
We know also that
\[
	\hh(V)(z)=V^\mu(z)\left(\parder{}{x^\mu}\bigg|_z+(\Gamma_\mu)^a_\rho(z)
	\parder{}{A^a_\mu}\bigg|_z+(\Gamma_\rho)^a_{\rho\sigma}(z)
	\parder{}{A^a_{\rho\sigma}}\bigg|_z\right),
\]
and so we compute
\begin{align*}
	(\iota_\hh\dd f)(V)&=\dd f\left(\hh(V)\right)=[h(V)](f)\\
	&=V^\mu\left(\parder{f}{x^\mu}+(\Gamma_\mu)^a_\rho
	\parder{f}{A^a_\mu}+(\Gamma_\rho)^a_{\rho\sigma}
	\parder{f}{A^a_{\rho\sigma}}\right)\\
	&=\iota_V\left[
	\left(
	\parder{f}{x^\mu}+(\Gamma_\mu)^a_\rho
	\parder{f}{A^a_\mu}+(\Gamma_\rho)^a_{\rho\sigma}
	\parder{f}{A^a_{\rho\sigma}}
	\right)dx^\mu
	\right],
\end{align*}
from which we read $\iota_\hh\dd f$ expression.
\end{proof}
\begin{corollary}
\label{cor:some_iotah_computations}
	In particular, we have the following expressions
	\begin{itemize}
		\item $\iota_\hh(\dd x^\mu)=\dd x^\mu,\quad$
		$\iota_\hh(\dd A^a_\mu)=(\Gamma_\alpha)^a_\mu dx^\alpha,\quad$
		$\iota_\hh(\dd A^a_{\mu\nu})=(\Gamma_\alpha)^a_{\mu\nu} dx^\alpha$.
		\item
		\[
			\iota_\hh(\dd p_a^{\mu\nu})=\left(\partial_\alpha p_a^{\mu\nu}+
			(\Gamma_\alpha)^a_\rho\parder{p_a^{\mu\nu}}{A^a_\rho}
			+(\Gamma_\alpha)^a_{\rho\sigma}\parder{p_a^{\mu\nu}}{A^a_{\rho\sigma}}\right)
			dx^\alpha.
		\]
		while noting that $\iota_\hh(\alpha\wedge\beta)=\iota_\hh(\alpha)\wedge\beta+
		\alpha\wedge\iota_\hh(\beta)$ for any $\alpha,\beta\in\Omega(Z)$, we also have
		that $\iota_\hh\eta=(n+1)\eta$.
	\end{itemize}
\end{corollary}
\section{Useful identities for gauge symmetries computations}
\label{sub:useful_identities_for_gauge_symmetries_computations}
We begin computing the contractions $\iota_{j^1\xi_C}(\dd F^a_{\mu\nu})$ and
$\iota_{j^1\xi_C}(\dd p_a^{\mu\nu})$, which equal $j^1\xi_C(p_a^{\mu\nu})$ and
$j^1\xi_C(p_a^{\mu\nu})$ respectively.

For the field strength, using the local expression \eqref{eq:j1xi_C_local} of the jet
prolongation of our gauge symmetries together with the partial derivatives of the $F$ tensor
computed in Section \ref{sec:initial_clarifications}, we have
\begin{align*}
	j^1\xi_C(F^a_{\mu\nu})&=\left[\partial_\rho\xi^d(x)+C^d_{ef}A^e_\rho\xi^f(x)\right]
	\parder{F^a_{\mu\nu}}{A^d_\rho}\\
	&+\left[\partial_\sigma\partial_\rho\xi^d(x)+C^d_{ef}\left(A^e_{\rho\sigma}\xi^f(x)
	+A^e_\rho\partial_\sigma\xi^f(x)\right)\right]\parder{F^a_{\mu\nu}}{A^d_{\rho\sigma}}\\
	&=\left[\partial_\rho\xi^d(x)+C^d_{ef}A^e_\rho\xi^f(x)\right]
	C^a_{dg}\left(\delta^\rho_\mu A^g_\nu-\delta^\rho_\nu A^g_\mu\right)\\
	&+\left[\partial_\sigma\partial_\rho\xi^d(x)+C^d_{ef}\left(A^e_{\rho\sigma}\xi^f(x)
	+A^e_\rho\partial_\sigma\xi^f(x)\right)\right]
	\left(\delta^a_d\delta^\rho_\nu\delta^\sigma_\mu-\delta^a_d\delta^\rho_\mu
	\delta^\sigma_\nu\right)\\
	&=C^a_{dg}\Bigl(\partial_\mu\xi^d(x)A^g_\nu-\partial_\nu\xi^d(x)A^g_\mu
	+C^d_{ef}A^e_\mu A^g_\nu\xi^f(x)
	-C^d_{ef}A^e_\nu A^g_\mu\xi^f(x)\Bigr)\\
	&+\cancel{\left(\partial_\mu\partial_\nu\xi^a(x)
	-\partial_\nu\partial_\mu\xi^a(x)\right)}\\
	&+C^a_{ef}\Bigl(A^e_{\nu\mu}\xi^f(x)-A^e_{\mu\nu}\xi^f(x)
	+A^e_\nu\partial_\mu\xi^f(x)-A^e_\mu\partial_\nu\xi^f(x)\Bigr).
\end{align*}
Relabeling now the dummy indices $d\to f$, $g\to e$ on the two terms of the first parenthesis
containing $\partial\xi$, and using the antisymmetry $C^a_{fe}=-C^a_{ef}$ on the last two
terms of the last parenthesis, all of them cancel pairwise
\[
	\cancel{C^a_{fe}\,\partial_\mu\xi^f(x)A^e_\nu}
	-\bcancel{C^a_{fe}\,\partial_\nu\xi^f(x)A^e_\mu}
	-\cancel{C^a_{fe}A^e_\nu\partial_\mu\xi^f(x)}
	+\bcancel{C^a_{fe}A^e_\mu\partial_\nu\xi^f(x)}=0,
\]
and so we are left with
\begin{align*}
	j^1\xi_C(F^a_{\mu\nu})&=C^a_{dg}C^d_{ef}\left(A^e_\mu A^g_\nu-A^e_\nu A^g_\mu\right)\xi^f(x)
	+C^a_{ef}\left(A^e_{\nu\mu}-A^e_{\mu\nu}\right)\xi^f(x).
\end{align*}
In order to manipulate this expression, we will use the Jacobi identity for the
structure constants in the form
\[
	-C^a_{gd}C^d_{ef}=C^a_{ed}C^d_{fg}+C^a_{fd}C^d_{ge},
\]
and the definition of the field strength to write
\[
	A^e_{\nu\mu}-A^e_{\mu\nu}=F^e_{\mu\nu}-C^e_{bc}A^b_\mu A^c_\nu,
\]
obtaining (after some dummy index relabelling)
\begin{align*}
	j^1\xi_C(F^a_{\mu\nu})&\equiv C^a_{ef}F^e_{\mu\nu}\xi^f(x)
	+\Bigl(-C^e_{bc}C^a_{ef}A^b_\mu A^c_\nu
	+C^a_{dg}C^d_{ef}A^e_\mu A^g_\nu
	-C^a_{dg}C^d_{ef}A^e_\nu A^g_\mu\Bigr)\xi^f(x)\\
	&\equiv C^a_{ef}F^e_{\mu\nu}\xi^f(x)
	+\Bigl(-C^d_{eg}C^a_{df}
	+C^a_{dg}C^d_{ef}
	-C^a_{de}C^d_{gf}\Bigr)A^e_\mu A^g_\nu\,\xi^f(x),
\end{align*}
and after some index manipulation, it can been seen that the last big parenthesis equals
to zero by Jacobi identity on $\mathfrak{g}$.
Hence we obtain
\begin{equation}
	\label{eq:j1xi_C_of_F}
	\boxed{
	j^1\xi_C(F^a_{\mu\nu})=C^a_{bc}F^b_{\mu\nu}\xi^c(x)
	}.
\end{equation}
For the multimomenta $p_a^{\mu\nu}$ we use the explicit expression obtained in Section
\ref{eq:explicit_computation_of_multimomentum}, namely
$p_a^{\mu\nu}=\sqrt{|g|}\kappa_{ab}g^{\mu\rho}g^{\nu\sigma}F^b_{\rho\sigma}$, so that
\begin{align*}
	j^1\xi_C(p_a^{\mu\nu})&=\sqrt{|g|}\,\kappa_{ab}\,
	g^{\mu\rho}g^{\nu\sigma}\,j^1\xi_C(F^b_{\rho\sigma})\\
	&=\sqrt{|g|}\,\kappa_{ab}\,g^{\mu\rho}g^{\nu\sigma}\,C^b_{ef}F^e_{\rho\sigma}\xi^f(x),
\end{align*}
where in the first equality we used that neither $\kappa_{ab}$ nor $g^{\mu\nu}$ depend on
$A^a_\mu$ nor on $A^a_{\mu\nu}$ and that $j^1\xi_C$ is vertical.
Inverting now the multimomenta expression,
\[
	\frac{1}{\sqrt{|g|}}\kappa^{ca}g_{\alpha\mu}g_{\beta\nu}p_a^{\mu\nu}=F^c_{\alpha\beta}
	\implies F^e_{\rho\sigma}=\frac{1}{\sqrt{|g|}}\kappa^{e\theta}g_{\rho\alpha}g_{\sigma\beta}
	p_\theta^{\alpha\beta},
\]
we obtain
\begin{align*}
	j^1\xi_C(p_a^{\mu\nu})&=\frac{\cancel{\sqrt{|g|}}}{\cancel{\sqrt{|g|}}}\,
	\kappa_{ab}\,\kappa^{e\theta}\,
	\overbrace{g^{\mu\rho}g_{\rho\alpha}}^{=\delta^\mu_\alpha}
	\overbrace{g^{\nu\sigma}g_{\sigma\beta}}^{=\delta^\nu_\beta}\,
	C^b_{ef}\,p_\theta^{\alpha\beta}\,\xi^f(x)\\
	&=\kappa_{ab}\kappa^{e\theta}C^b_{ef}\,p_\theta^{\mu\nu}\xi^f(x)
	=\kappa^{e\theta}C_{aef}\,p_\theta^{\mu\nu}\xi^f(x)\\
	&\overset{(\dagger)}{=}-\kappa^{e\theta}C_{eaf}\,p_\theta^{\mu\nu}\xi^f(x)
	=-C^\theta_{af}\,\xi^f(x) p_\theta^{\mu\nu},
\end{align*}
where in $(\dagger)$ we have used that $C_{abc}=\kappa_{ad}C^d_{bc}$ is totally
antisymmetric, and hence
\begin{equation}
	\label{eq:j1xi_C_of_p}
	\boxed{
	j^1\xi_C(p_a^{\mu\nu})=-C^b_{ad}\,p_b^{\mu\nu}\,\xi^d(x)
	}.
\end{equation}
\subsection{Proof of superpotential definition}
\label{sub:proof_of_superpotential_definition}
We will now derive explicitly Equation \eqref{eq:superpotential},
in which the step of identifying the superpotential was omitted.
We begin proving the following identity
\begin{equation}
	\label{eq:crucial_identity_superpotential}
	\dd x^\alpha\wedge\dd^{n-1}x_{\mu\nu}=\delta^\alpha_\nu\dd^nx_\mu-
	\delta^\alpha_\mu\dd^nx_\nu.
\end{equation}
From our definitions of the forms $\dd^nx_\mu$ and $\dd^{n-1}x_{\mu\nu}$
at the beginning of Section \ref{sect:lagrangian_setting}, we know that
$\dd x^\alpha\wedge\dd^{n}x_\mu=\delta^\alpha_\mu\dd^{n+1}x$ and so
$\iota_{\partial_\nu}(\dd x^\alpha\wedge\dd^{n}x_\mu)=
\delta^\alpha_\mu\dd^nx_\nu$. If we instead use the derivation rule
for the interior product we have
\[
	\delta^\alpha_\mu\dd^nx_\nu=
	\iota_{\partial_\nu}(\dd x^\alpha\wedge\dd^{n}x_\mu)=
	\delta^\alpha_\nu\dd^nx_\mu-\dd x^\alpha\wedge\dd^{n-1}
	x_{\mu\nu}
\]
and solving for $\dd x^\alpha\wedge\dd^{n-1}x_{\mu\nu}$
yields Equation \eqref{eq:crucial_identity_superpotential}.

Now, consider any antisymmetric tensor with values
in $X$, namely, $T^{\mu\nu}(x)=T^{[\mu\nu]}(x)$
(in particular, we can choose $T^{\mu\nu}(x)
\equiv p_a^{\mu\nu}(x)\xi^a(x)$ since $p_a^{\mu\nu}(x)$
is antisymmetric in $(\mu,\nu)$), then
\begin{align*}
	d\left(T^{\mu\nu}(x)\dd^{n-1}x_{\mu\nu}\right)
	&=\partial_\alpha T^{\mu\nu}(x)\dd x^\alpha\wedge
	\dd^{n-1}x_{\mu\nu}
	\overset{\eqref{eq:superpotential}}{=}
	\partial_\alpha T^{\mu\nu}(x)
	\big[\delta^\alpha_\nu\dd^nx_\mu
	-\delta^\alpha_\mu\dd^nx_\nu\big]\\
	&=\partial_\nu T^{\mu\nu}(x)\dd^nx_\mu
	-\partial_\mu T^{\mu\nu}(x)\dd^nx_\nu\\
	&\overset{\mu\leftrightarrow\nu}{\equiv}
	\partial_\mu T^{\nu\mu}(x)\dd^nx_\nu
	-\partial_\mu T^{\mu\nu}(x)\dd^nx_\nu\\
	&=-\partial_\mu T^{\mu\nu}(x)\dd^nx_\nu
	-\partial_\mu T^{\mu\nu}(x)\dd^nx_\nu=
	-2\partial_\mu T^{\mu\nu}(x)\dd^nx_\nu
\end{align*}
where in the anterior to last equality
we used that $T^{\mu\nu}(x)$ is antisymmetric.
Then $\partial_\mu T^{\mu\nu}(x)\dd^nx_\nu=
\dd\left(-\frac{1}{2}T^{\mu\nu}(x)
\dd^{n-1}x_{\mu\nu}\right)$ and
Equation \eqref{eq:superpotential} follows
immediately.
\chapter[A dictionary between the physical and the geometric language][Physics--geometry dictionary]{A dictionary between the physical and the geometric language}
\label{app:dictionary}
 
The following table summarises some concepts covered in the thesis as a dictionary
between the vocabulary of theoretical physics and the objects of differential geometry
that implement it in the multisymplectic framework. The first column fixes
the concept, the second states it as a physicist would, and the third gives
its geometric incarnation.
\begin{landscape}
 
\footnotesize                          
\setlength{\tabcolsep}{6pt}
\renewcommand{\arraystretch}{1.25}
 
\begin{longtable}{@{}
    >{\raggedright\arraybackslash}p{0.18\linewidth}
    >{\raggedright\arraybackslash}p{0.38\linewidth}
    >{\raggedright\arraybackslash}p{0.38\linewidth}@{}}
 
\caption[Dictionary between the physical and the geometric language]{%
    Dictionary between the physical and the geometric language.}
\label{tab:dictionary}\\
\toprule
\textbf{Concept} & \textbf{Theoretical physics POV} & \textbf{Differential geometry POV} \\
\midrule
\endfirsthead
 
\caption[]{\textit{(continued)}}\\
\toprule
\textbf{Concept} & \textbf{Theoretical physics POV} & \textbf{Differential geometry POV} \\
\midrule
\endhead
 
\midrule
\multicolumn{3}{r@{}}{\footnotesize\itshape Continued on the next page}\\
\endfoot
 
\bottomrule
\endlastfoot
 
 
\multicolumn{3}{@{}l}{}\\*
Field & $\phi^i(x)$ as a spacetime collection of functions & $\phi\in\Gamma(\pi)$ as a section of the configuration bundle\\
Field variation & $\delta\phi^i=\phi^{\prime i}(x)-\phi^i(x)$& $\phi_t:=\Phi_t^Y\circ\phi\circ\Phi_{-t}^X$\\
Lagrangian density & $\mathcal{L}(x^\mu,\phi^i,\partial_\mu\phi^i)$ & $\mathcal{L}=L(x^\mu,\phi^i,\partial_\mu\phi^i)\dd^{n+1}x$\\
Principle of extremal action & $\frac{\delta S(\phi)}{\delta\phi}\overset{!}{=}0$ & $\der{S(\phi_t)}{t}\Big|_{t=0}\overset{!}{=}0$\\
Euler--Lagrange equations & $\parder{L}{\phi^i}-\partial_\mu
\left(\parder{L}{[\partial_\mu\phi^i]}\right)=0$ & $(j^2\phi)^*\left(\parder{L}{y^i}
-\der{}{x^\mu}\parder{L}{z^i_\mu}\right)=0$\\
On-shell vs off-shell & - & Pullback through $j^1\phi$ vs identity on the whole $Z$\\
Connection & $\nabla$, $\bm{A}$... & $T_zZ=\mathcal{V}_z\pi\oplus H_z$\\
Poincaré--Cartan forms & $p\,\dd q- H\,\dd t$ & $\Theta_L$ and $\Omega_L$\\
Observables & $f\in C^\infty(T^*Q)$, $\{f,g\}=\omega(X_f,X_g)$ (symplectic mechanics) & $\alpha\in\mathcal{H}^0(Z)$ and $X_\alpha\in\mathfrak{X}^{n+1}(Z)$ with multisymplectic bracket\\
Fixed time commutator & Equal time Poisson brackets (once quantized, equal time commutation relations) & $\{\tilde{\alpha}_1,\tilde{\alpha}_2\}=\iota_{\tilde{X}_1}\iota_{\tilde{X}_2}\widetilde{\Omega_L}$\\
Hamiltonian formulation & Canonical formalism (non covariant) & Hamilton--De Donder--Weyl theory (fully covariant)\\
Canonical Quantization & Second quantization / Path Integral & Geometric quantization\\
Symmetries & $\delta_\xi\mathcal{L}=0$ (modulo divergence terms) & $\pounds_{\xi_Y^{(1)}}\Theta_L-\dd\alpha\in\mathcal{I(C)}$\\
Internal gauge d.o.f. & $U_\alpha\times G$ & PFB and LTs\\
Lie algebra bracket & $[T_a,T_b]=i\,f_{abc}T_c$ (hermitian generators) & $[e_a,e_b]=C^c_{ab}e_c$\\
Choosing a gauge & Choosing internal basis on each $\{x\}\times\mathfrak{g}$& Choosing a section $s_\alpha:U_\alpha\to P$ or LT\\
Gauge transformations & $\delta_\xi$ & $\Phi\in\Gau(P)$, $\xi\in\gau(P)$\\
Vector potential and field strength & $A_\mu=A^a_\mu(x)\,e_a$, $F_{\mu\nu}=\frac{1}{2}F^a_{\mu\nu}(x)\,e_a$
& $A_\alpha=s_\alpha^*\omega$, $F_\alpha=s_\alpha^*\Omega^\omega$\\
Yang--Mills Lagrangian & $\mathcal{L}=-\frac{1}{4g^2}F^a_{\mu\nu}F_a^{\mu\nu}$ ($g$ coupling constant) &
$\mathcal{L}=-\frac{1}{4}F^a_{\mu\nu}F_a^{\mu\nu}\eta$\\
Covariant derivative & $D_\mu=\partial_\mu\pm igA_\mu$& 
$D_A\phi^a=d\phi^a+C^a_{bc}A^b\wedge\phi^c$, $D^\omega$\\
Vector potential change of trivialization & $A_\mu(x)^\prime=\Omega(x)A_\mu(x)\Omega^\dagger(x)
+\frac{i}{g}\Omega(x)\partial_\mu\Omega^\dagger(x)$& $A_\beta=\Ad_{g^{-1}}\circ A_\alpha+g^*\theta_G$\\
Gauge fixing & $\partial_\mu A^\mu\overset{!}{=}0$, $\nabla\cdot\vec{A}\overset{!}{=}0$...
into $\mathcal{L}_{GF}$ + $\mathcal{L}_{FP}$ (Faddeev--Popov term) & Choosing a unique representative on each $\Gamma(C(P))/\Gau(P)$\\
Yang--Mills equations & $\nabla_\nu F^{b\mu\nu}(x)+C^b_{cd}A^c_\nu(x)F^{d\mu\nu}(x)=0$ &
$D_A\star F=0$\\
Bianchi identity & $\partial_{[\alpha}F^a_{\mu\nu]}(x)+C^a_{bc}A^b_{[\alpha}(x)F^c_{\mu\nu]}(x)=0$&
$D_AF=0$\\
Yang--Mills charge & Color charge (for $\mathrm{SU}(3)$) & Integral of superpotential \\
\end{longtable}
\end{landscape}
\end{document}